\documentclass[a4paper,11pt]{article}
\pdfoutput=1 

\usepackage{jheparxiv}

\usepackage[T1]{fontenc} 

\usepackage{tensor}
\usepackage{amsmath}
\usepackage{amssymb}
\usepackage{amsthm}
\usepackage{mathrsfs}
\usepackage[normalem]{ulem}
\usepackage{bbm}
\usepackage{graphicx}
\usepackage{stmaryrd}
\usepackage{twistor}
\usepackage{xcolor}
\usepackage{mathtools}
\usepackage{xfrac}
\usepackage{comment}
\usepackage{caption}
\usepackage{subcaption}

\newcommand{\what}[1]{\widehat{#1} }
\newcommand{\Sym}{\op{Sym}}
\newcommand{\norm}[1]{\left\| #1 \right\|}
\newcommand{\Lap}{\triangle}
\newcommand{\vac}[1]{|#1 \rangle}
\newcommand{\mc}{\mathcal}
\newcommand{\ep}[1]{\left\langle #1 \right\rangle} 
\renewcommand{\i}{\mathrm{i}} 
\renewcommand{\b}{\mathrm{b}}
\renewcommand{\c}{\mathrm{c}}
\newcommand{\mf}{\mathfrak}
\newcommand{\op}{\operatorname}
\newcommand{\Oo}{\mathscr{O}}
\newcommand{\mbb}{\mathbb}
\newcommand{\msf}[1]{\mathsf{#1}}

\newcommand{\cI}{\mathcal{I}}

\newcommand{\ip}{\,\lrcorner\,}
\newcommand{\g}{\mathfrak{g}}
\newcommand{\cS}{\mathcal{S}}

\newcommand{\wt}[1]{\widetilde{#1}}
\newcommand{\F}{\mathbb{F}}
\newcommand{\br}[1]{\overline{#1}} 

\newcommand{\til}[1]{\widetilde{#1}} 
\newcommand{\lap}{\triangle}
\newcommand{\sT}{\mathsf{T}}

\newcommand{\sg}{\mathsf{g}}

\newcommand{\cK}{\mathcal{K}}
\renewcommand{\bz}{\bar z}
\newcommand{\CS}{\mathcal{CS}}
\newcommand{\vho}{\varrho}
\newcommand{\sa}{\mathsf{a}}
\renewcommand{\sb}{\mathsf{b}}
\renewcommand{\sc}{\mathsf{c}}

\newtheorem{theorem}{Theorem}

\title{Burns space and holography}

\author[a]{Kevin Costello,}
\author[b]{Natalie M. Paquette}
\author[c]{and Atul Sharma}

\affiliation[a]{Perimeter Institute for Theoretical Physics,\\ 51 Caroline Street, Waterloo, Ontario, Canada\vspace{0.1cm}}
\emailAdd{kcostello@perimeterinstitute.ca}

\affiliation[b]{Department of Physics,\\
University of Washington, Seattle, WA, 98195, USA\vspace{0.1cm}}
\emailAdd{npaquett@uw.edu}

\affiliation[c]{Center for the Fundamental Laws of Nature \& Black Hole Initiative,\\Harvard University, Cambridge, MA, 02138, USA \vspace{0.1cm}}
\emailAdd{atulsharma@fas.harvard.edu}

\begin{document} 

\abstract{We elaborate on various aspects of our top-down celestial holographic duality wherein the semiclassical bulk spacetime is a 4d asymptotically flat, self-dual K\"ahler geometry known as Burns space. The bulk theory includes an open string sector comprising a 4d WZW model and a closed string sector called ``Mabuchi gravity'' capturing fluctuations of the K{\"a}hler potential. Starting with the type I topological B-model on the twistor space of flat space, we obtain the twistor space of Burns space from the backreaction of a stack of $N$ coincident D1 branes, while the chiral algebra is obtained from (a twist of) the brane worldvolume theory. One striking consequence of this duality is that all loop-level scattering amplitudes of the theory on Burns space can be expressed as correlation functions of an explicit 2d chiral algebra.  

We also present additional large-$N$ checks, matching several 2 and 3-point amplitudes and their collinear expansions in the WZW$_4$ sector, and the mixed WZW$_4$-Mabuchi sector, of the bulk theory to the corresponding 2 and 3-point vacuum correlators and operator product expansions in the dual chiral algebra. Key features of the duality, along with our main results, are summarized in the introduction.}

\maketitle
\flushbottom

\section{Introduction \& Conclusion}
\label{sec:intro}

The holographic principle is by now widely believed to be intrinsic to any complete description of quantum gravity. Holography has been easiest to formulate in (asymptotically) anti-de Sitter space, with asymptotia at spatial infinity and observables equivalent to those, e.g. correlation functions, in a boundary conformal field theory. Most difficult, but also most relevant for the real world, is the case of de Sitter space, with its temporal infinities and the attendant subtleties in defining observables. Formulating a holographic principle in asymptotically flat spacetimes, with null asymptotia and a well-defined S-matrix, may now serve as a Goldilocks case for concrete formulations of the holographic principle, particularly in the interesting case when the gravitating spacetime is four dimensional, which falls outside of the purview of the high-dimensional matrix models understood so far (see, e.g., \cite{BFSS, IKKT}). 

There is a renewed surge of interest in holography for asymptotically flat spacetimes on the heel of refined studies of asymptotic symmetry algebras (see \cite{Strominger:2017zoo} for a review), most recently exhibiting 2d chiral algebraic structures at the classical level \cite{Guevara:2021abz, Strominger:2021mtt,Himwich:2021dau,Adamo:2021lrv}. Related suites of proposals collectively known as ``celestial holography'' have posited that the form of a putative holographic dual for such theories takes the form of an exotic, non-unitary conformal field theory (often called the celestial CFT or CCFT) supported on the celestial sphere at asymptotic null infinity. For reviews and results, see \cite{Raclariu:2021zjz, Pasterski:2021raf, Pasterski:2021rjz} and references therein. 

There are many challenges to interpreting a would-be CFT dual at null infinity, ranging from the puzzling (``How can we understand a CFT exhibiting an integral shift symmetry in its conformal weights, which is the avatar of 4d translation invariance in an eigenbasis of boost weights?'') to the fundamental (``How can 4d gravitational dynamics in flat spacetimes be encoded locally on the celestial sphere, particularly in such a way as to guarantee stringent CFT axioms?'') Various bottom-up approaches have been suggested to ameliorate these issues, for instance by leveraging judicious transforms on tempered distributions called shadow and light transforms in the latter case \cite{Pasterski:2017kqt, Atanasov:2021cje, Sharma:2021gcz}.

In this paper, we elaborate on an explicit, example top-down duality engineered from string theory, after the fashion of the AdS/CFT correspondence. This example was constructed with the aim of sharpening and addressing some of these questions. While we do not claim the strong features exhibited by our duality are necessary to describe or construct a CCFT (or other type of holographic dual for flat spacetime) outside of this toy example, we do stress that it highlights some \textit{sufficient} conditions for such a construction to exist. In particular, one is guaranteed to obtain a 2d chiral algebra with a local, associative OPE if its dual theory is integrable. Further, in 4d the condition of integrability geometrizes to a theory admitting a local holomorphic lift to twistor space. Recent advances in computing and cancelling gauge anomalies in 6d holomorphic theories \cite{Costello:2019jsy, Costello:2021bah} enable us to study a class of such 4d integrable theories at the quantum level by appealing to their twistorial uplift. Furthermore, the holomorphic theory on twistor space can, in favorable cases, be identified with sectors of a topological string theory, to which the methods of twisted holography \cite{CG, Costello:2020jbh}, especially homological algebra \cite{PW}, can be applied. Important earlier work, closely related to twisted holography, on the protected boundary chiral algebra subsector of 4d $\mathcal{N}=4$ and its bulk dual can be found in \cite{Bonetti:2016nma}.

When the 4d ``bulk'' theories are not integrable, one must contend with the rich non-analyticities of scattering amplitudes (outside special classes of scattering problems, such as the tree level MHV sector in gauge theory), in which the local twistorial perspective will be insufficient (and a non-local twistorial approach may not be the most expedient). Mapping such scattering problems to the celestial sphere results in non-localities, in violation of the standard axioms of the chiral algebras, which must be reckoned with in any CCFT.\footnote{Though see \cite{Bakalov_2003} for a proposed nonlocal mathematical extension of vertex algebras pre-dating celestial holography.} Relatedly, CCFT states will acquire non-vanishing anomalous scaling dimensions. Furthermore, in classically integrable 4d theories, there is a map between obstructions to the integrability of its quantization, and gauge anomalies in the local twistorial uplift. These non-vanishing anomalies in 6d encoding 4d non-integrability at the quantum level have recently been understood to lead to violations in associativity of (the quantum deformation of) celestial chiral algebras \cite{Costello:2022upu, Bittleston:2022jeq, Fernandez:2023abp,Zeng:2023qqp}. It will be crucial to address these issues when formulating the holographic dual of a general 4d bulk theory.

Still, we believe it is worth understanding in detail even the non-generic situation in which the 4d bulk theory is integrable, since there a ``celestial'' chiral algebra of asymptotic symmetries can be described explicitly without any modifications of the standard CFT axioms. Whether or how the above issues can be addressed in a general theory, admitting true black hole solutions and chaotic dynamics, is beyond the scope of this work. We will construct such a well-defined holographic duality from the top-down, in type I string theory. In the remainder of this introduction, we will motivate and sketch this construction.

It is perhaps not surprising that the setting for a holographic duality for integrable theories lies in the context of the topological string (see, e.g., \cite{Aganagic:2003qj} for relevant earlier work) or twisted holography, and it is precisely in the setting of twisted theories, which are either holomorphic or topological theories, where the technique of Koszul duality has the most teeth: in the twisted world, the symmetries of a system enhance dramatically to infinite-dimensional algebras, providing constraints powerful enough to fix some or all of the resulting dynamics. In the developing subject of twisted holography, it is expected that one may even be able to use this enormous symmetry to compute or fix observables to all orders in a $1/N$ expansion. The precise duality we study in this paper arises as follows.

We study the topological B-model of the type I string, compactified on twistor space.\footnote{Defining the topological B-model on a non-Calabi Yau manifold requires a slight modification of the usual topological string construction, which we describe in the main text.} As is well-known, the vacuum state of the type I string includes both open and closed string sectors described in the topological string framework by a holomorphic Chern-Simons theory and a Kodaira-Spencer theory of gravity, respectively. When these sectors are coupled, this theory has no gauge anomaly \cite{Costello:2019jsy} if and only if the gauge group is $\SO(8)$. Consequently, it reduces to an integrable theory in 4d spacetime. Roughly, the holomorphic Chern-Simons theory reduces to a 4d WZW model \cite{Donaldson:1985zz, Losev:1995cr} with target manifold the group manifold $\SO(8)$, while the Kodaira-Spencer theory reduces to a theory, described by the Mabuchi functional \cite{mabuchi1986k}, of dynamical fluctuations of the K{\"a}hler potential governing the 4d K{\"a}hler metric. We stress that in 4d flat spacetime, the $\SO(8)$ WZW$_4$ model has vanishing scattering amplitudes, as befits an integrable theory: \textit{WZW$_4$ theories for general $G$ are not integrable at the quantum level. Equivalently, their local holomorphic uplifts to twistor space cannot be consistently quantized}. Again, while the 4d models for other $G$ are interesting and perfectly well-behaved quantum field theories, we do not know how to define a dual 2d chiral algebra for these theories. Indeed, we expect naive attempts to define a chiral algebra of asymptotic symmetries for these theories will result in failures of associativity when quantum effects are incorporated.

To get a full-fledged holographic duality, we add a large number of D-branes in 6d, which backreact on the twistor space of flat space and so, in turn, backreact on 4d flat spacetime. In particular, we add $N \gg 1$ Euclidean B-model D1-branes wrapping the zero section $\mathbb{CP}^1$ of the twistor fibration. This $\mathbb{CP}^1$ is precisely the celestial sphere from the 4d point of view. The worldvolume theory of these branes in the B-model is a 2d chiral algebra, arising from dimensional reduction of holomorphic Chern-Simons theory \cite{Witten:1992fb, CG}, supported on the celestial sphere, as desired in a celestial holographic duality. Furthermore, we explicitly compute the backreaction in the 6d Kodaira-Spencer theory, amounting to a deformation of the background complex structure, and find that we obtain the so-called twistor space of Burns space \cite{Lebrun91explicitself-dual}. Reducing to 4d, the deformed spacetime metric is the asymptotically flat (Euclidean) Burns metric \cite{burns1986twistors}. Therefore, the 4d ``bulk'' theory in our holographic duality can be described semiclassically as a coupled $\SO(8)$ WZW$_4$  model + Mabuchi gravity on Burns space. Crucially, our 4d theories acquire nontrivial scattering amplitudes (i.e. nonvanishing scattering at 4 points and higher) in this curved background, which can be matched to chiral correlators. In particular, we can explicitly match collinear limits of this scattering with the OPE of the dual chiral algebra per standard celestial reasoning. We do so in the $N \rightarrow \infty$ limit in this paper, though we expect the duality holds at finite $N$, and we plan to study this in future work.

Our construction employs the twistor space associated to analytically continued spacetime, in which one may work without specializing to a particular spacetime signature (or only doing so at a later stage). We remark that our full nonperturbative duality\footnote{Perturbatively, one may continue to Lorentzian signature and compute scattering amplitudes in the Burns space metric \cite{Hawking:1979hw}, and our duality will still hold.} is perhaps most naturally formulated in Euclidean signature spacetime. One way to see this is to notice that, as is familiar from the self-dual sectors of gauge theory and gravity, evaluating the action on nontrivial saddles in Lorentzian signature requires the introduction of complex fields. However, in any signature, the theories under consideration in this paper are non-unitary. While a CCFT dual may be expected to be non-unitary on general grounds, a non-unitary bulk spacetime theory is a priori less desirable; indeed, we view it as a serious shortcoming of our stringent construction, closely related to the existence of a twistorial uplift/associative OPE, that must be overcome in other top-down constructions of flat space holography.

The WZW$_4$ model, plus Mabuchi gravity, is an interacting and non-renormalizable system. One surprising consequence of this conjectural  holographic duality is that it implies a non-perturbative (finite-$N$) isomorphism between scattering in  this theory on a curved spacetime (Burns space) and correlation functions in a 2d chiral algebra!  That such a strong dynamical statement holds in an asymptotically flat spacetime is a testament to the underlying simplicity of the 4d theory.

Many questions remain about this duality, and putative dualities for generic theories in asymptotically flat spacetimes. In future work, we aim to study, for instance, non-planar corrections to this duality and more aspects of its embedding into full 10d string theory. A smaller, immediate puzzle is whether or not the other integrable/twistorial theories studied in \cite{Costello:2021bah,Costello:2022upu,Costello:2022wso,Bittleston:2022nfr} admit 10d uplifts to string theory, and hence their own twisted holographic correspondences. Similarly, seeing as Burns space is only scalar-flat but not Ricci-flat, the study of celestial chiral algebras on Ricci-flat self-dual backgrounds like ALE spaces also promises to be a rich avenue for exploration \cite{Bittleston:2023bzp}. 

The most pressing question in this program is the incorporation of complete gravitational dynamics, including all metric degrees of freedom, into the 4d spacetime. Our holographic duality enjoys a standard gravitational description in 6d, with fluctuating metric components corresponding to complex structure deformations and gauged diffeomorphisms. The passage from 6d to 4d involves a gauge fixing of these diffeomorphisms, leaving us with a Liouville-like closed string sector, as we describe in more detail in the main text, and in particular not a covariant theory. Again, this is perhaps unsurprising due to the close connections between integrable or exactly-solvable models in various dimensions, but it deprives us of the richness of gravitational physics in this toy model. We expect that incorporating genuine gravitational dynamics in a 4d spacetime will require a potentially quite dramatic modification of standard CFT axioms, or perhaps a holographic formulation of a different type such as an adaptation of matrix theory to lower dimensions. 

\subsection{Summary of the paper}
This work is a companion to our Letter \cite{Costello:2022jpg}. In the remainder of this paper we expand on this duality in detail.    
\begin{enumerate}
	\item In section \ref{sec:bmodel} we introduce the topological B-model on twistor space.   The topological B-model is typically studied on Calabi-Yau manifolds and, with a careful treatment of boundary conditions, can be generalized to complex $3$-folds equipped with a meromorphic volume form.  A result of Pontecorvo \cite{pontecorvo} shows that twistor spaces of scalar-flat K\"ahler $4$-manifolds are equipped with a meromorphic $3$-form, and so are the natural twistor spaces on which to place the topological B-model.  We show that the closed string sector of the topological B-model gives rise to Mabuchi gravity: a theory of a dynamical scalar field interpreted as the K\"ahler potential.  The equations of motion of Mabuchi gravity imply that the associated K\"ahler metric has vanishing scalar curvature. We show, following \cite{Bittleston:2020hfv,Penna:2020uky}, that the open string sector is the WZW$_4$ model \cite{Donaldson:1985zz, Nair:1991bf, Losev:1995cr}. Our analysis focuses on the type I topological B-model \cite{Costello:2019jsy}, whose open string sector consists of holomorphic Chern-Simons for $\SO(8)$. 
	\item In section \ref{sec:burns} we consider the backreaction of a stack of D1 branes in the topological B-model on the twistor space of flat $\R^4$.  The analysis follows that in \cite{CG}.  We find that the backreacted twistor space contains $\SL_2(\C)$, and is the twistor space of the asymptotically flat Burns metric \cite{burns1986twistors,Lebrun91explicitself-dual}, a certain scalar-flat K\"ahler metric on the complex manifold $\Oo(-1) \to \CP^1$.   The backreaction changes the topology of the four-dimensional spacetime by introducing a non-contractible $S^2$. This was first anticipated in \cite{Hartnoll:2004rv} and is somewhat similar to the Gopakumar-Vafa \cite{Gopakumar:1998ki} geometric transition.

	\item In section \ref{sec:dual}, we analyze the holographic dual chiral algebra.  This is the theory living on a stack of $N$ D1 branes wrapping a curve in the twistor space of flat $\R^4$. The algebra is the BRST reduction of a collection of free symplectic bosons by $\op{Sp}(N)$, and can also be realized as the chiral algebra associated by \cite{Beem:2013sza} to a family of $\mathcal{N}=2$ SCFTs in dimension $4$.
 
    An important subtlety in our analysis is that this chiral algebra has point defects associated to the locus where the meromorphic volume form on twistor space has poles. Conformal blocks of the chiral algebra in the presence of these defects are shown to be in one-to-one correspondence with the Hilbert space of the bulk theory on Burns space.
    
	\item In section \ref{sec:bdry}, we describe the relationship between our Burns space holography and familiar $\AdS_3 \times S^3$ holography. The twistor space of Burns space contains $\SL_2(\C) \simeq \AdS_3 \times S^3$.  The holographic dual chiral algebra lives on the boundary of $\AdS_3$.  We identify where this boundary lives in twistor space, and we show how  the twistor lines of Burns space  are copies of $S^2$ which foliate $\AdS_3$.
	\item In section \ref{sec:twistor_dictionary} we analyze the holographic dictionary on twistor space in more detail. We explicitly identify the states on twistor space which match the large-$N$ states in our chiral algebra, and we show that the defects in our chiral algebra can be seen from the holographic dual theory on twistor space. 

	In this section we prove one of our main results: the tree level scattering amplitudes of WZW$_4$ plus Mabuchi gravity on Burns space have a surprising symmetry enhancement, from the $\SU(2) \times \U(1)$ isometry group of the Burns geometry to a larger $\SU(2) \times \SU(2)$ group, and indeed to an infinite-dimensional algebra.  This allows us to identify these amplitudes with vacuum correlators in the chiral algebra, in the absence of defects.  

	We conjecture  (with quite strong evidence) that this result continues to hold at loop level.  This conjecture, together with the holographic correspondence and the Penrose transform, implies something rather remarkable: \emph{all} scattering amplitudes of WZW$_4$ plus Mabuchi gravity on Burns space can be expressed as the correlators of a rather simple chiral algebra.    Since the chiral algebra is relatively simple, this implies that these scattering amplitudes are all in principle computable.

		At tree level, the amplitudes for WZW$_4$ coincide with the all-plus amplitudes for Yang-Mills theory. This result implies that these Yang-Mills amplitudes on Burns space match planar correlators in the chiral algebra.

	\item In sections \ref{sec:dict} and \ref{sec:ops}, we initiate the analysis of the holographic duality directly on Burns space, as opposed to its twistor space.  We identify the scattering states on Burns space holographically dual to the single-trace operators in the large $N$ chiral algebra.  For WZW$_4$ states, these are certain explicit solutions of the Laplace equation.  For states of Mabuchi gravity, they solve instead a fourth-order equation involving the Paneitz operator.  We phrase the conjectured duality as a match between scattering amplitudes of explicit states on Burns space and chiral algebra correlators. 

	\item In section \ref{sec:tests}, we turn to tests of the duality directly on Burns space, in the planar limit.    We compute the tree level two-point function of states in the WZW$_4$ sector by the standard holographic method \cite{Witten:1998qj}. This had already been computed by Hawking, Page and Pope \cite{Hawking:1979hw} by a different method, yielding the same result.\footnote{Hawking et al.\ were computing the two-point scattering amplitude of a conformally-coupled scalar on $\CP^2$, but since this manifold is conformally equivalent to Burns space, and the scalar curvature of Burns space vanishes, the result is the same.}   The result is a certain Bessel function, whose series expansion matches perfectly the two-point function of states in the chiral algebra.   

		We also compute certain terms in the collinear limits of 3-point WZW$_4$ amplitudes, and find that they match exactly with the OPE coefficients of the dual chiral algebra.  A similar calculation is performed for the OPE and 3-point amplitude involving two WZW$_4$ states and one ``graviton'' (K{\"a}hler scalar). 

\end{enumerate}


\section{Topological strings on twistor space}
\label{sec:bmodel}

In this section, we provide a concise review of B-model topological strings on twistor spaces of self-dual spacetimes. The B-model can only be studied on Calabi-Yau manifolds, but twistor spaces are not Calabi-Yau. Nevertheless, the twistor space of a spacetime with a self-dual \emph{K\"ahler} metric comes equipped with a meromorphic 3-form. Excising the polar divisor of this 3-form from the twistor space yields a non-compact Calabi-Yau 3-fold.  The topological string makes sense with defects (or holomorphic boundary conditions) along the divisor.   We will study a type I version of the open+closed topological B-model on two particular examples of such 3-folds: the twistor spaces of flat space and Burns space.

\subsection{Twistor geometry}
\label{sec:twistor}

A metric on a 4-manifold is termed \emph{self-dual} (SD) if its Weyl tensor is self-dual. Self-dual metrics span integrable subsectors of Einstein as well as conformal gravity \cite{Mason:1991rf}. Integrability in four dimensions is often 
intimately linked to the existence of a twistor space: a complex 3-fold that encodes 4-dimensional physics in terms of complex analytic geometry. As a smooth manifold, the twistor space $\pi : Z\to M$ of a smooth, oriented 4-manifold $M$ with Riemannian metric $g$ is given by the bundle of pointwise metric- and orientation-compatible almost complex structures. It has $\SO(4)/\U(2)\simeq\CP^1$ fibers and is diffeomorphic to the unit sphere bundle $S(\Lambda^-)$ of the rank 3 bundle of anti-self-dual (ASD) 2-forms $\Lambda^-\to M$. Alternatively, it is diffeomorphic to the projective 2-spinor bundle of $M$.

$Z$ can be equipped with the Atiyah-Hitchin-Singer almost complex structure. To construct this, one splits the tangent space of every point $p\in Z$ into vertical and horizontal components using the Levi-Civita connection of $g$ induced on $S(\Lambda^-)$. The almost complex structure at $p$ is taken to be the direct sum of the standard complex structure on $\CP^1$ with the almost complex structure on the tangent space $T_{\pi(p)}M\otimes\C$ parametrized by $p$. It is well-known that this becomes an integrable almost complex structure on $Z$ if and only if $g$ is self-dual \cite{Penrose:1976js, Atiyah:1978wi}; see also \cite{Woodhouse:1985id} for a review. The $\CP^1$ fiber over a point $x\in M$ is known as the \emph{twistor line} corresponding to $x$ and will be denoted $L_x$. When this almost complex structure is integrable, these twistor lines become holomorphic rational curves in $Z$ cut out by the so-called \emph{incidence relations}. They have normal bundle $\CO(1)\oplus\CO(1)$. Moreover, the antipodal map on each $L_x$ gives rise to an antiholomorphic involution without fixed points, giving $Z$ a real structure.

When $g$ is self-dual as well as Ricci-flat, the twistor space $Z$ is called a \emph{nonlinear graviton} \cite{Penrose:1976js}. In this case, $(M,g)$ becomes a hyperk\"ahler manifold. Such metrics describe self-dual Einstein gravity without a cosmological constant. Instead, in the rest of this work we will mainly be interested in the case when $g$ is not necessarily Ricci-flat but is only required to be scalar-flat and K\"ahler.

If $g$ is a K\"ahler metric on a 4-manifold $M$, the results of \cite{derd:1983,lebrun1986topology} show that it is self-dual\footnote{With respect to the orientation in which the K\"ahler form is anti-self-dual.} if and only if it is scalar-flat, i.e., has zero Ricci scalar $R=0$. So we will usually refer to such self-dual metrics as \emph{scalar-flat K\"ahler}. As observed by Hitchin in \cite{hitchinkahler}, the twistor space $Z$ of any SD 4-manifold $(M,g)$ is spin, i.e., its canonical bundle $K_Z$ always admits a square root $K_Z^{1/2}$. If $M$ itself is spin, then $K_Z$ also admits a fourth root $K_Z^{1/4}$. Furthermore, a theorem by Pontecorvo \cite{pontecorvo} shows that if $g$ is self-dual as well as K\"ahler, then it gives rise to a globally holomorphic section $\check{\omega}$ of $K_Z^{-1/2}$. This can be inverted to produce a meromorphic section $\Omega = \check{\omega}^{-2}$ of $K_Z$. 

Another standard fact is that the restriction of $K_Z$ to every twistor line is given by $\CO(-4)\to\CP^1$. As a result, $K_Z^{-1/2}|_{L_x}\simeq\CO(2)$, so the section $\check{\omega}|_{L_x}$ has two zeroes on every twistor line $L_x$. These correspond to a pair of almost complex structures on the tangent space $T_xM\otimes\C$. In fact, Pontecorvo constructs the section $\check{\omega}$ in such a way that one of these zeroes parametrizes precisely the integrable almost complex structure $J$ with respect to which $g$ is K\"ahler, and the other zero corresponds to the conjugate almost complex structure $-J$. As we vary $L_x$, the zeroes of $\check{\omega}$ sweep out a quadric in $Z$ that acts as the polar divisor of $\Omega$. The K\"ahler form $\omega$ on $M$ is recovered by performing a contour integral
\be\label{omegacirc}
\omega(x) = \frac{1}{2\pi}\oint_{\Gamma\subset L_x}\Omega
\ee
where the contour $\Gamma$ separates the zeroes of $\check{\omega}|_{L_x}$ and circles the zero at $(x,J|_x)$ clockwise. Our orientation and normalization conventions will be such that the resulting K\"ahler form is a real ASD 2-form on $M$ and satisfies $\omega^2 = -2\,\vol_g$. 

For the reader's convenience, we provide a more extensive review of this construction using local coordinates on twistor space in appendix \ref{app:twistor}. See also \cite{Hartnoll:2004rv} for a comparable review.

\paragraph{Example: Flat space.} The paradigmatic example of the twistor correspondence is the twistor space of flat space $\R^4$ with its Euclidean metric. This helps us in setting up our local coordinates and spinor conventions. We will mainly follow the conventions of \cite{Costello:2021bah, Sharma:2022arl}.

Let $x^\mu$ denote the standard coordinates on $\R^4$. We can define double null coordinates $x^{\al\dal}$, $\al=1,2,\dal=\dot1,\dot2$, by setting
\be
x^{\al\dal} = \frac{1}{\sqrt2}\begin{pmatrix}x^0+\im x^3&&x^2+\im x^1\\-x^2+\im x^1&&x^0-\im x^3\end{pmatrix}\,.
\ee
The indices $\al,\dal$ are spinor indices of $\SU(2)$ of opposite chirality. In terms of these, the Euclidean metric can be expressed as
\be\label{flatmetric}
\d s^2 = \delta_{\mu\nu}\d x^\mu\d x^\nu =  \eps_{\al\beta}\,\eps_{\dal\dot\beta}\,\d x^{\al\dal}\,\d x^{\beta\dot\beta}\,.
\ee
Here, $\eps_{\al\beta},\eps_{\dal\dot\beta}$ are $2\times2$ Levi-Civita symbols that we will ubiquitously use to raise, lower and contract spinor indices with the conventions
\be\label{conventions}
\lambda^\al = \eps^{\al\beta}\lambda_\beta\,,\qquad\lambda_\beta = \lambda^\al\eps_{\al\beta}\,,\qquad\mu^{\dal} = \eps^{\dal\dot\beta}\mu_{\dot\beta}\,,\qquad\mu_{\dot\beta} = \mu^{\dal}\eps_{\dal\dot\beta}
\ee
etc. Spinor contractions will be abbreviated using square and angle brackets
\be
\la\lambda\,\kappa\ra \vcentcolon= \lambda^\al\kappa_\al\,,\qquad [\mu\,\nu]\vcentcolon=\mu^{\dal}\nu_{\dal}
\ee
that are invariant under $\SL_2(\C)$ rotations of the dotted or undotted indices.

We will mostly work with the following complex coordinates built out of $x^{\al\dal}$,
\be\label{utox}
u^{\dal} \vcentcolon= x^{1\dal}\,,\qquad \hat u^{\dal} \vcentcolon= x^{2\dal} = (-\overline{u^{\dot2}},\overline{u^{\dot1}})\equiv(-\bar u^2,\bar u^1)\,.
\ee
In the last equality, we are using the convention that the complex conjugate of a dotted spinor is an undotted spinor. In contrast, the spinor $\hat u^{\dal}$ built from the complex conjugates transforms in the same $\SU(2)$ representation as $u^{\dal}$ due to the property $\widehat{Lu} = L\hat u$ valid for matrices $L\in\SU(2)$. Because of this, the map $u^{\dal}\mapsto\hat u^{\dal}$ is known as \emph{quaternionic conjugation}. In these coordinates, the Euclidean metric reads
\be\label{euclid}
\d s^2 = 2\,\|\d u\|^2 = 2\left(|\d u^{\dot1}|^2 + |\d u^{\dot2}|^2\right)\,.
\ee
Here $\|u-v\|^2 = |u^{\dot1}-v^{\dot1}|^2+|u^{\dot2}-v^{\dot2}|^2$ is the Euclidean norm. We note the useful relations
\be
[\hat u\,u] = \|u\|^2\,,\qquad x^\mu x_\mu=2\|u\|^2\,.
\ee
The factor of $2$ here is a convenient convention that is common in twistor theory \cite{Penrose:1984uia,Penrose:1986uia}. The associated volume form is
\be\label{orientation}
\d^4x\equiv\d x^0\wedge\d x^1\wedge\d x^2\wedge\d x^3 = \d u^{\dot1}\wedge\d\bar u^1\wedge\d u^{\dot2}\wedge\d\bar u^2
\ee
which defines our orientation convention.

\medskip

The twistor space of $\R^4$ is traditionally denoted $\PT$. It is the total space of a rank 2 holomorphic vector bundle over the Riemann sphere,
\be
\PT = \CO(1)\oplus\CO(1)\to\CP^1\,.
\ee
Let $z$ be an affine coordinate along its $\CP^1$ base, and let $v^{\dal}=(v^{\dot1},v^{\dot2})$ denote coordinates along the $\C^2$ fibers. Under $z\mapsto1/z$, the fiber coordinates transform as
\be\label{zto1/z}
z\mapsto\frac{1}{z}\implies v^{\dal}\mapsto\frac{v^{\dal}}{z}\,.
\ee
This holomorphic transition function, defined on the annulus $z\neq0,\infty$, endows $\PT$ with the structure of a complex 3-fold. Equivalently, it is also standard to view $\PT$ as an open subset of $\CP^3$ given by $\PT=\CP^3-\CP^1$. As an application of this second viewpoint, one defines line bundles $\CO(n)\to\PT$ as restrictions of the standard line bundles $\CO(n)\to\CP^3$ to this subset, or equivalently as pull-backs from $\CP^1$.

Every point $x\in\R^4$ corresponds to a holomorphic global section of $\CO(1)\oplus\CO(1)\to\CP^1$,
\be
L_x\;:\;v^{\dal} = u^{\dal}+z\hat u^{\dal}\,.
\ee
These are taken to be the twistor lines. The line bundles $\CO(n)$ restrict to the standard line bundles $\CO(n)\to\CP^1$ on each $L_x$. The antipodal map of $\CP^1$ extends to $\PT$ as a fixed-point-free antiholomorphic involution
\be
(z,v^{\dal})\mapsto\left(-\frac{1}{\bar z},-\frac{\hat v^{\dal}}{\bar z}\right)\,,\qquad \hat v^{\dal} \vcentcolon=(-\bar v^2,\bar v^1)\,.
\ee
The twistor lines $L_x$ are invariant under this map, i.e., if $(z,v^{\dal})$ lies on $L_x$, then so does its antipodal point. Since knowing two points on a (projective) line completely determines the line, we obtain a diffeomorphism $\R^4\times\CP^1\to\PT$ given by $(x,z)\mapsto(v^{\dal},z) = (u^{\dal}+z\hat u^{\dal},z)$. Moreover, this identifies the twistor lines $L_x$ as the fibers of the projection $\R^4\times\CP^1\to\R^4$.

\medskip

The Euclidean metric \eqref{euclid} is K\"ahler in the complex coordinates $u^{\dal}$. It has K\"ahler form
\be\label{omegaflat}
\omega = \im\,\eps_{\dal\dot\beta}\,\d u^{\dal}\wedge\d\hat u^{\dot\beta} = \im\,\bigl(\d u^{\dot1}\wedge\d\bar u^1 + \d u^{\dot2}\wedge\d\bar u^2\bigr)\,.
\ee
This is an ASD 2-form and satisfies $\frac12\,\omega^2 = -\d^4x$. It is instructive to obtain this K\"ahler form from twistor space via Pontecorvo's theorem. The canonical bundle of $\PT$ is $\Oo(-4)$. To get the K\"ahler form $\omega $, we take $\check{\omega} = z$ as our choice of global section of $K_{\PT}^{-1/2} = \CO(2)$. The associated meromorphic 3-form is given by
\be
\Omega  = \check{\omega}^{-2} = \frac{\d z\wedge\d v^{\dot1}\wedge\d v^{\dot2}}{z^2}\,.
\ee
Using the transition function \eqref{zto1/z}, we see that this has poles of order $2$ at $z=0$ and $z=\infty$ each. 

We can use the diffeomorphism $v^{\dal}=u^{\dal}+z\hat u^{\dal}$ to pull back $\Omega$ to $\R^4\times\CP^1$,
\be
\Omega  = \frac{\d z\wedge(\d u^{\dot1}-z\,\d\bar u^2)\wedge(\d u^{\dot2}+z\,\d\bar u^1)}{z^2}\,.
\ee
Integrating it in $z$ along a contour that surrounds the pole at $z=0$ picks out the K\"ahler form of flat space:
\be
\frac{1}{2\pi}\oint_{|z| = 1}\Omega = \frac{1}{2\pi\im}\oint_{|z| = 1}\frac{\d z}{z}\wedge\omega = \omega \,.
\ee
In future sections, we will come across similar calculations in the context of more interesting scalar-flat K\"ahler manifolds.


\subsection{Open strings and the WZW$_4$ model}
\label{sec:open}
As we have seen, twistor spaces of scalar-flat K\"ahler manifolds are complex $3$-folds equipped with meromorphic volume forms.  The topological B-model can be studied on any complex $3$-fold with a holomorphic volume form, and careful choices of boundary conditions \cite{CG, Costello:2021bah} allow one to define the topological B-model on $3$-folds equipped with a meromorphic volume form as well.  This suggests a general correspondence between topological string theory and field theories on scalar-flat K\"ahler manifolds. In this subsection and the next, we will describe the $4$-dimensional theories corresponding to the open and closed topological B-model.  

We first consider the open string sector. Let $Z$ be a complex 3-fold equipped with a globally holomorphic 3-form $\Omega$. Holomorphic Chern-Simons theory on $Z$ has the classical action
\be\label{hcs}
S_\text{hCS}[\cA] = \frac{1}{(2\pi\im)^3}\int_Z\Omega\wedge\tr\left(\frac12\,\cA\wedge\dbar\cA + \frac{1}{3}\,\cA\wedge \cA\wedge \cA\right)\,.
\ee
It is a gauge theory governing the integrability of smooth partial connections $\cA\in\Omega^{0,1}(Z,\g)$ on complex vector bundles $E\to Z$. Here, $\dbar$ is the antiholomorphic exterior derivative on $Z$, and $\g$ is the Lie algebra of the gauge group. When $Z$ is Calabi-Yau, holomorphic Chern-Simons arises as the string field theory of open strings in the topological B-model with target space $Z$ \cite{Witten:1992fb}. We will focus on a type I analogue of the B-model introduced in \cite{Costello:2019jsy}. Its gauge group $G$ is determined by a stack of space-filling D5 branes. The gauge group is arbitrary at the classical level, but will be chosen to be $\SO(8)$ for the open string sector to couple to the closed string sector in an anomaly-free manner at the quantum level.

We will be interested in the case when $Z$ is the twistor space of some SD 4-manifold $M$. This case was studied for flat twistor space $Z=\PT$ in \cite{Costello_talk, Bittleston:2020hfv,Penna:2020uky,Costello:2021bah}, but their analysis generalizes straightforwardly to curved twistor spaces.\footnote{One generally needs a K\"ahler metric to write the worldsheet theory of the $(2,2)$ $\sigma$-model, although the B-model topological twist does not require this. It is also standard to use a K\"ahler metric to impose harmonic gauges in the string field actions, although again this is not strictly necessary. For a compact self-dual 4-manifold $M$, twistor space is K\"ahler if and only if $M=S^4$ or $\CP^2$ \cite{hitchinkahler}. More generally, it is bimeromorphic to K\"ahler if and only if it is Moishezon \cite{lebrun1992twistors,lebrun1992bi,campana}. We will always work on spacetimes (eg.\ $\C^2$ or $\CO(-1)$) obtained from removing points from compact 4-manifolds (viz.\ $S^4$ or $\CP^2$) possessing K\"ahler twistor spaces. This corresponds to removing projective lines from the twistor spaces, and the K\"ahler structure restricts naturally to the resulting geometry.} 
The canonical bundle of a general twistor space $Z$ is nontrivial. However, as we reviewed in the previous section, if the self-dual metric on $M$ is also K\"ahler, then $Z$ does admit a meromorphic 3-form $\Omega$ with double poles on a quadric $\check{\omega}=0$. So we can still study holomorphic Chern-Simons on $Z-\{\check{\omega}=0\}$ by taking this $\Omega$ as the 3-form in the action \eqref{hcs}. This is equivalent to studying \eqref{hcs} on $Z$ but with ``boundary conditions'' on the partial connection $\cA$ that ensure a well-defined variational principle. In order for the holomorphic Chern-Simons Lagrangian to be free of poles, we will demand the boundary conditions
\be\label{abdry}
\cA \in \Omega^{0,1}(Z,\CO(-D)\otimes\End\,E)\,,
\ee
that is, $\cA$ vanishes holomorphically to first order on the divisor $D = \{\check{\omega}=0\}$.  Geometrically, this means that the holomorphic bundle built from $\cA$ is fixed on the divisor $D$. If we assume that our background bundle $E$ is trivialized on $D$, then the bundles obtained from varying $E$ are also trivialized.

We can write $\cA=\check{\omega}\,\varphi$ for some smooth $\varphi\in\Omega^{0,1}(Z,K_Z^{1/2}\otimes\End\,E)$. Using $\Omega=\check{\omega}^{-2}$, the Lagrangian in terms of $\varphi$ is found to be (up to normalization)
\be
\frac12\int  \op{tr} \varphi \dbar \varphi + \frac{1}{3} \int \check{\omega} \op{tr} \varphi^3 \,.
\ee
The integrands in both terms in this equation are $(0,3)$ forms valued in the canonical bundle $K_Z$. This is because $\varphi$ is twisted by $K_Z^{1/2}$ and $\check{\omega}$ is twisted by $K_Z^{-1/2}$.  It is important to note that the kinetic term in this exression is non-degenerate, and that the interaction term tends to zero on Pontecorvo's quadric where $\check{\omega}$ vanishes.

Smooth automorphisms of $E$ induce gauge transformations
\be\label{asym}
\cA \mapsto h^{-1}\cA h + h^{-1}\dbar h\,,\qquad h\in\Gamma(\text{Aut}\,E)\,.
\ee
Holomorphic Chern-Simons on $Z$ as defined above is invariant under those gauge transformations for which the transformed field continues to satisfy the boundary condition \eqref{abdry}. We will also restrict attention to vector bundles $E\to Z$ whose restrictions $E|_{L_x}$ to every twistor line $L_x$ are trivial.\footnote{The moduli space of $\SO(8)$ bundles on $\CP^1$ has two components classified by $\Z_2$, the fundamental group of $\SO(8)$. So bundles on $\PT$ that restrict to non-trivial bundles on $\CP^1$ may indeed occur. Unfortunately, their spacetime interpretation is far from clear. A preliminary line of attack for Penrose transforming bundles that are non-trivial on a finite number of twistor lines is described in \cite{Sparling:1990}, and such bundles have in fact already started to occur in celestial holography in the guise of twistorial monopoles \cite{Garner:2023izn}! Alternatively, one can use $\Spin(8)$ as the gauge group (which a priori isn't ruled out by chiral anomaly cancellation), and all stable $\Spin(8)$ bundles on $\CP^1$ are trivial.} 

Each twistor line carries two canonical points we call $0$ and $\infty$, where Pontecorvo's quadric intersects $L_x$.  The boundary conditions for our gauge field mean that the bundle is trivialized on the quadric, and so at the points $0,\infty \in L_x$.  Because the bundle is trivial on $L_x$, the trivialization at $0$ extends uniquely to a trivialization on all of $L_x$ (and similarly for the trivialization at $\infty$).  The two trivalizations differ by a point $\sg(x) \in G$. This means that we have associated a space-time field
\be
\sg \in \text{Maps}(M,G)
\ee
to a field configuration on twistor space.  

An equivalent way to think of $\sg$ is that it is the value of an open holomorphic Wilson line wrapping $L_x$.  An open Wilson line is gauge invariant because the gauge field vanishes at $0$ and $\infty$.

Explicitly, we can find a frame $H\in\text{Maps}(Z,G)$ for sections of $E$ that trivializes $\cA$ when restricted to each $L_x$, i.e.,
\be\label{aLxH}
\cA|_{L_x} = H^{-1}\dbar|_{L_x}H\qquad\forall\;x\in M\,.
\ee
$H$ can be gauge fixed to equal the identity matrix in a neighborhood of $\infty$ on each $L_x$, and equal $\sg(x)$ on a neighborhood of $0$.
One can decompose $\cA$ into horizontal and vertical $(0,1)$-forms with respect to the Levi-Civita connection induced on $Z$. The horizontal part of $\cA$ enters quadratically in the action and can be integrated out by imposing its equation of motion. As reviewed in appendix \ref{app:twac}, in the frame \eqref{aLxH} we can partially solve for $\cA$ using this equation of motion. This yields
\be\label{atoH}
\cA = H^{-1}(\dbar+\pi_{0,1}A)H\,,
\ee
written in terms of a $\g$-valued spacetime 1-form $A$ (trivially pulled back to $Z$ using the projection $\pi:Z\to M$). The notation $\pi_{p,q}\al$ denotes the $(p,q)$-part of any $(p+q)$-form $\al$ on $Z$. 

The new field $A$ acts as a spacetime gauge field; alternatively it may be thought of as the zero mode of $\cA$ under KK reduction of holomorphic Chern-Simons along the $\CP^1$ fibers of $Z\to M$. Higher KK modes drop out and end up never contributing to the reduction. Using the boundary conditions on $\cA$ and $H$ mentioned above, we obtain
\be
A = -\dbar\sg\,\sg^{-1}\,,
\ee
fixing $A$ in terms of a single ``positive helicity'' degree of freedom $\sg$. Here and in what follows, whenever $\dbar$ acts on a spacetime object, it represents the dbar operator on $M$ in the complex structure associated to its K\"ahler metric $g$.

\medskip

The compactification of \eqref{hcs} along the fibers of $Z\to M$ is also briefly reviewed in appendix \ref{app:twac}. It is performed by plugging the solution \eqref{atoH} for $\cA$ into \eqref{hcs} and integrating fiberwise. The dynamical dependence of \eqref{atoH} is purely along $M$ up to factors of the frame $H$; and even the frame is pure gauge except at the two poles of $\Omega$ where it can be completely fixed in terms of $\sg$. As a result, the integral over the fibers can be performed explicitly without generating any KK modes. This results in a 4-dimensional Wess-Zumino-Witten (WZW$_4$) model on the scalar-flat K\"ahler manifold $M$,
\be\label{Swzw}
S_{\text{WZW}_4}[\sg, \omega] = -\frac{\im}{8\pi^2}\int_M\omega\wedge\tr\!\left(\sg^{-1}\p\sg\wedge\sg^{-1}\dbar\sg\right) + \frac{\im}{24\pi^2}\int_{M\times[0,1]}\omega\wedge\tr\!\left(\tilde\sg^{-1}\d\tilde\sg\right)^3\,.
\ee
This action is a functional of the dynamical field $\sg$ as well as the choice of background K\"ahler form $\omega$ on $M$. Its first term is a standard kinetic term, with $\p$ and $\dbar$ denoting holomorphic and antiholomorphic exterior derivatives on $M$. The second is a 5-dimensional Wess-Zumino term. $\tilde\sg$ is an extension of $\sg$ to $M\times[0,1]$ representing a homotopy of $\sg$ to a fixed reference profile $\sg_0$. We will implicitly choose $\sg_0 = \mathbbm{1}$ for the rest of this work. $\d\tilde\sg$ represents the 5-dimensional exterior derivative of $\tilde\sg$ on $M\times[0,1]$, while $\omega$ continues to denote the 4-dimensional K\"ahler form on $M$ in both terms.

The WZW$_4$ model has been studied in many contexts \cite{Donaldson:1985zz,Losev:1995cr,Nair:1991bf}. As reviewed here, it describes the effective spacetime dynamics of open topological strings on twistor space. In the past, it provided a classical action principle for the self-dual Yang-Mills equation on $M$ \cite{Mason:1991rf}, although the relation to gauge theory does not persist beyond tree level.\footnote{Nonetheless, WZW$_4$ seems to be related to 4d $\cN=2$ heterotic strings at tree as well as loop level \cite{Ooguri:1991ie}.} The equation of motion of $\sg$ reads
\be\label{geom}
\omega\wedge\dbar(\sg^{-1}\p\sg) = 0\,.
\ee
This is known as Yang's equation \cite{Yang:1977zf}. On its support, the gauge field $A=-\dbar\sg\,\sg^{-1}$ solves the self-dual Yang-Mills equation $F_A^- = 0$ on $M$. As in the case for a 2-dimensional WZW model, the derivation of \eqref{geom} involves cancelling certain contributions from the kinetic term against contributions from the Wess-Zumino term. For the action and the variational problem to be well-defined, one requires that the Wess-Zumino term be independent of the 5d extension $\tilde\sg$ of the 4d field $\sg$. This imposes the quantization condition \cite{Losev:1995cr}
\be
\frac{\omega}{2\pi}\in H^2(M,\Z)\,,
\ee
showing that $\omega$ acts as a 4d analogue of the Kac-Moody level familiar from 2d WZW models.

For perturbative calculations, it is convenient to work locally on the group manifold of $G$ in terms of an adjoint-valued scalar $\phi$ by writing
\be\label{phidef}
\sg = \e^{\phi}\,,\qquad\phi\in\text{Maps}(M,\g)\,.
\ee
As its 5-dimensional extension, one can take $\tilde\sg = \e^{t\phi}$, where $t\in[0,1]$ is a coordinate along the interval. Performing such a field redefinition and integrating out $t$ converts the action \eqref{Swzw} to
\be\label{phiac}
S_{\text{WZW}_4}[\phi,\omega] = -\frac{\im}{8\pi^2}\int_M\omega\wedge\tr\left(\p\phi\wedge\dbar\phi - \frac{1}{3}\,\phi\,[\p\phi,\dbar\phi] + \mathrm{O}(\phi^4)\right)\,,
\ee
where $[\p\phi,\dbar\phi]\equiv\p\phi\wedge\dbar\phi+\dbar\phi\wedge\p\phi$ is standard notation. The field equation of $\phi$ reads
\be\label{phieom}
\omega\wedge\bigg(\p\dbar\phi + \frac{1}{2}\,[\p\phi,\dbar\phi] + \mathrm{O}(\phi^3)\bigg) = 0\,.
\ee
The corresponding linearized field equation is $\omega\wedge\p\dbar\phi=0$, which is just the Laplace equation associated to the K\"ahler metric on $M$. 

We will use this form of the action in later sections to build Feynman rules and derive some simple 2- and 3-point tree amplitudes of WZW$_4$ on Burns space. Remembering the classical equivalence of this model with (a gauge-fixed formulation of) SD Yang-Mills, we will often refer to these amplitudes as tree level all-plus ``gluon'' amplitudes, though we emphasize that WZW$_4$ is not actually a gauge theory.


\subsection{Closed strings and Mabuchi gravity}
\label{sec:closed}

The closed string sector of the topological B-model governs complex structure deformations of a Calabi-Yau 3-fold $Z$. Locally, a complex structure deformation is described by a Beltrami differential $\mu\in\Omega^{0,1}(Z,T^{1,0}Z)$ that deforms its dbar operator
\be
\dbar\mapsto\dbar + \mu\,.
\ee
The deformed almost complex structure is integrable if and only if $\mu$ solves the Maurer-Cartan equation
\be\label{mceq}
\dbar\mu + \frac{1}{2}\,[\mu,\mu] = 0\,.
\ee
Here, $[\mu,\mu]$ denotes the wedge product on the $(0,1)$-form factors in $\mu$ and the Lie bracket on the $(1,0)$-vector field factors. 

The Bershadsky-Cecotti-Ooguri-Vafa (BCOV) \cite{Bershadsky:1993cx} theory -- also known as Kodaira-Spencer gravity -- is an action principle that gives rise to this Maurer-Cartan equation as its equation of motion. It arises as the string field theory describing closed strings in the B-model \cite{Bershadsky:1993cx}. The BCOV action is nonlocal in nature,
\be
S_\text{BCOV}[\mu] = \frac{1}{(2\pi\im)^3}\int_Z\frac12\,\dbar(\mu\ip\Omega)\wedge\p^{-1}(\mu\ip\Omega) + \frac16\, (\mu^3\!\ip\Omega)\wedge\Omega\,,
\ee
where $\p$ is the holomorphic exterior derivative on $Z$ and $\lrcorner$ denotes interior product. The field $\eta$ is constrained by
\be\label{peta}
\p\eta=0\,,\qquad\eta\vcentcolon=\mu\ip\Omega\in\Omega^{2,1}(Z)\,.
\ee
This implies that the $(2,1)$-form $\eta$ as defined here is locally $\p$-exact,
\be
\eta = \p\gamma\,,\qquad\gamma\in\Omega^{1,1}(Z)\,,
\ee
so that $\p^{-1}\eta=\gamma$ modulo $\ker\p$. In spite of having a nonlocal kinetic term, the BCOV action has a perfectly well-behaved equation of motion and perturbative expansion. It is also invariant under smooth diffeomorphisms of $Z$ generated by exponentiating the linearized transformations
\be
\delta \mu = \dbar\xi + \mu\ip\p\xi - \xi\ip\p \mu
\ee
for some $\xi\in\Omega^0(Z,T^{1,0}Z)$ satisfying $\p(\xi\ip\Omega)=0$.

Again, we specialize to the case when $Z$ is the twistor space of a scalar-flat K\"ahler spacetime.  Our boundary conditions are simply that the $(2,1)$ form $\eta$ remains smooth at the divisor $\check{\omega}=0$. This implies that $\mu$ must vanish to second order at $\check{\omega}=0$.  The BCOV Lagrangian is then smooth everywhere and leads to a well-defined variational problem.  Its kinetic term is unchanged when written in terms of $\eta$, so that it has a well-defined propagator and pertubative expansion.  Its diffeomorphism symmetry is also reduced to the consideration of those diffeomorphisms of $Z$ that preserve these boundary conditions. At the level of linearized diffeomorphisms, this simply requires that $\xi$ vanish to second order at $\check{\omega}=0$.

\medskip

The reduction of the BCOV action to spacetime gives rise to a theory of scalar-flat K\"ahler fluctuations of the background scalar-flat K\"ahler metric \cite{Costello:2021bah}. Compactification along the twistor lines is now much more involved, so we only provide an executive summary in appendix \ref{app:bcovred}. 

The dynamical field on spacetime is found to be a single scalar field $\rho(x)$. In the linearized theory, this is related to the Beltrami differential by a Penrose integral formula
\be
\rho(x) = \int_{L_x}\gamma\,.
\ee
This scalar field is interpreted as a perturbation of the background K\"ahler potential of $M$
\be
K\mapsto\cK = K+\rho\,.
\ee
The function $\rho$ depends on the choice of $(1,1)$ form $\gamma$ with $\partial \gamma = \eta$.  However, the closed $1$-form $\d \rho$ on spacetime does not depend on the choice of $\gamma$.  This means that $\rho$ is defined up to the addition of a constant. The deformed K\"ahler metric takes the form
\be
\d s^2 = 2\,\frac{\p^2\cK}{\p u^{\dal}\p\hat u^{\dot\beta}}\,\d u^{\dal}\,\d\hat u^{\dot\beta}
\ee
written in local complex coordinates $u^{\dal},\hat u^{\dal}$ that are modeled after the flat space coordinates of \eqref{utox}.

When $\mu$ obeys its equation of motion on twistor space, the deformed K\"ahler metric on spacetime continues to be scalar-flat. The Ricci scalar can be viewed either as a trace of the Ricci tensor with respect to the metric, or as a trace of the Ricci form with respect to the K\"ahler form. Using the latter viewpoint, scalar-flatness is best imposed as the orthogonality of the K\"ahler form and the Ricci form:
\be
\varpi\wedge\cP = 0\,,
\ee
where $\varpi$ and $\cP$ are the deformed metric's K\"ahler and Ricci forms,
\begin{align}
\varpi &= \omega+\im\,\p\dbar\rho = \im\,\p\dbar\cK\,,\label{defomega}\\
\cP &= P - \im\,\p\dbar\log\frac{\varpi^2}{\omega^2}\,,\label{defricci}
\end{align}
and $P=-\im\,\p\dbar\log\det(\p_{u^{\dal}}\p_{\hat u^{\dot\beta}}K)$ denotes the Ricci form of the background K\"ahler metric on $M$. 

There is a standard action functional for this equation known as the Mabuchi functional, first introduced in \cite{mabuchi1986k} for the study of constant scalar curvature K\"ahler metrics. We will propose the following variant of this functional as the gravitational sector of our duality:
\be\label{mabuchi}
S_\text{M}[\rho] = -\frac{1}{8\pi^2}\int_M\varpi^2\log\frac{\varpi^2}{\omega^2} + \im\,P\wedge\p\rho\wedge\dbar\rho\,.
\ee
For the scalar-flat case under consideration, this is equivalent to the Mabuchi functional (as displayed for instance in the recent review \cite{phong2007lectures}) up to boundary terms. So we continue to refer to it as the Mabuchi action. A few integrations by parts show that this can also be written in a background invariant form reminiscent of a $\U(1)$ WZW$_4$ model,
\be\label{mabuchiK}
S_\text{M}[\rho] = -\frac{\im}{8\pi^2}\int_M\cP\wedge\p\cK\wedge\dbar\cK\,.
\ee
Since boundary terms tend to be a tricky point in holography, we will choose to treat \eqref{mabuchi} as our bulk action in what follows. 
We note that all forms of the action are invariant under the gauge symmetry of $\rho$ which shifts it by a constant. They are also invariant under K\"ahler transformations $\cK(u,\hat u)\mapsto\cK(u,\hat u) + f(u) + \tilde f(\hat u)$ up to boundary terms. 

Remembering $\varpi=\omega+\im\,\p\dbar\rho$ and background scalar-flatness $\omega\wedge P=0$, it is straightforward to compute the variation of \eqref{mabuchi} with respect to $\rho$,
\be
\begin{split}
    \delta S_\text{M} &= -\frac{1}{4\pi^2}\int_M\im\,\p\dbar(\delta\rho)\wedge\varpi\left(\log\frac{\varpi^2}{\omega^2}+1\right)-\delta\rho\,\varpi\wedge P\\
    &= \frac{1}{4\pi^2}\int_M\delta\rho\;\varpi\wedge\cP\,,
\end{split}
\ee
having dropped any boundary terms. Hence, the critical points of the Mabuchi functional correspond to scalar-flat K\"ahler metrics on $M$ obtained from perturbing a fixed background scalar-flat K\"ahler metric. 

The field content of Mabuchi gravity also couples to the WZW$_4$ model in a canonical fashion. One simply makes the replacement
\be
\omega\mapsto\varpi\,,\qquad S_{\text{WZW}_4}[\phi,\omega]\mapsto S_{\text{WZW}_4}[\phi,\varpi]
\ee
in the action \eqref{Swzw}. At the level of twistor space, this coupling arises from replacing $\dbar\mapsto\dbar+\mu$ in the holomorphic Chern-Simons action \eqref{hcs}. If $\rho$ is taken to be a globally defined scalar field, $\omega$ and $\varpi=\omega+\im\,\p\dbar\rho$ live in the same K\"ahler class. In particular, the K\"ahler class of $\varpi$ continues to be integral, which was essential for the WZW$_4$ action to be well-defined. In what follows, we will take this to be the case.

To better understand the perturbative expansion of Mabuchi gravity, let us expand \eqref{mabuchi} to cubic order in $\rho$. Recalling the formula $\lap\rho = 4\im\omega\wedge\p\dbar\rho/\omega^2$ for the Laplacian on K\"ahler 4-manifolds, we can first compute the expansion of $\log\varpi^2/\omega^2$,
\begingroup
\allowdisplaybreaks
    \begin{align}
    \log\frac{(\omega+\im\,\p\dbar\rho)^2}{\omega^2} &= \log\left(1+\frac{2\im\omega\wedge\p\dbar\rho}{\omega^2} - \frac{(\p\dbar\rho)^2}{\omega^2}\right) = \log\left(1+\frac{\lap\rho}{2} - \frac{(\p\dbar\rho)^2}{\omega^2}\right)\nonumber\\
    &= \frac{\lap\rho}{2} - \frac{(\p\dbar\rho)^2}{\omega^2} - \frac{(\lap\rho)^2}{8} + \frac{\lap\rho}{2}\frac{(\p\dbar\rho)^2}{\omega^2} + \frac{(\lap\rho)^3}{24} + \mathrm{O}(\rho^4)\,,
\end{align}
\endgroup
To cubic order, the corresponding expansion of the Mabuchi action reads
\be\label{rhoac}
S_\text{M}[\rho] = -\frac{1}{64\pi^2}\int_M(\lap\rho)^2\omega^2 + 8\im P\wedge\p\rho\wedge\dbar\rho - 4(\lap\rho)(\p\dbar\rho)^2 - \frac{1}{6}(\lap\rho)^3\omega^2 + \mathrm{O}(\rho^4)\,,
\ee
where we have dropped globally exact terms like $(\p\dbar\rho)^2$ and $\omega\wedge\p\dbar\rho$ from the expanded Lagrangian. In flat space, the background Ricci form vanishes, $P=0$. Up to normalization conventions for the K\"ahler form, this action then reduces to the one derived in \cite{Costello:2021bah} by compactifying BCOV theory along the $\CP^1$ fibers of the twistor space of $\R^4$. On a more general background, it is also appended with the term $P\wedge\p\rho\wedge\dbar\rho$ proportional to the background Ricci form. This extra term is crucial for deriving the correct linearized wavefunctions for $\rho$, which will enter some of the tests of our holographic duality in later sections.

The relative normalization between the WZW$_4$ action and the Mabuchi action is fixed by demanding that the coupled theory lifts to an anomaly-free, holomorphic theory on twistor space. This is accomplished through a Green-Schwarz mechanism that we review below. Alternatively, on spacetime, it is determined by demanding that the one loop 4-point $\phi$ amplitude vanish in the coupled theory in flat space. The interested reader may refer to \cite{Costello:2021bah} for more details of the precise normalizations.


\subsection{Anomaly cancellation and renormalizability}
\label{sec:coupled}
Holomorphic Chern-Simons theory and BCOV theory are naturally coupled, which is immediate from their origins as open (resp., closed) string sectors of topological string theory.
Holomorphic Chern-Simons theory on a complex three-fold suffers from a one-loop gauge anomaly, associated to the diagram in Figure \ref{fig:anomaly}.
\begin{figure}
\centering
	\includegraphics[scale=0.22]{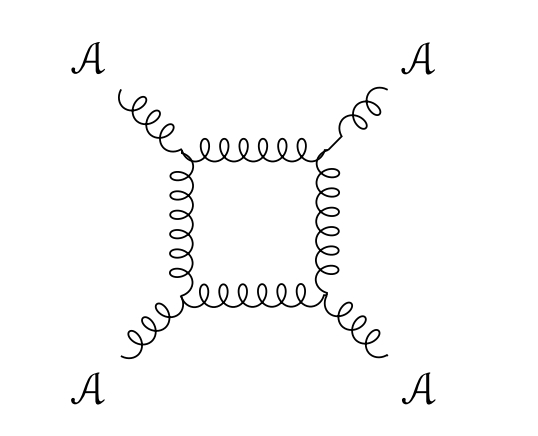}
	\caption{The gauge anomaly for holomorphic Chern-Simons. \label{fig:anomaly} }	
\end{figure}
In \cite{Costello:2019jsy}, it was shown that, for certain gauge groups, this anomaly can be cancelled by a Green-Schwarz mechanism, involving the exchange of a closed string field. This is depicted in Figure \ref{fig:gs}
.
\begin{figure}
\centering
\includegraphics[scale=0.21]{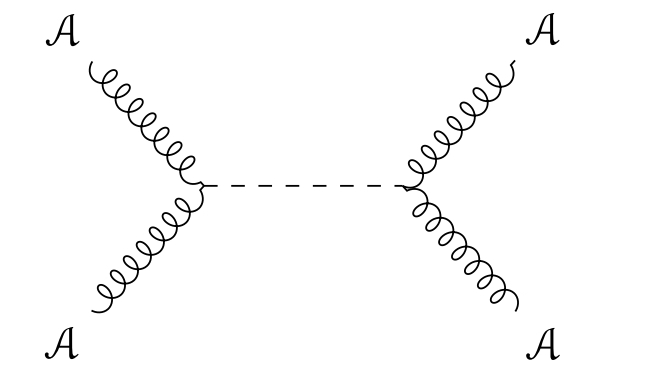}
	\caption{The Green-Schwarz mechanism, showing cancellation of the one-loop open string anomaly by a tree level diagram involving exchange of closed string fields.  \label{fig:gs}} 
\end{figure}
(This tadpole cancellation mechanism was previously known in the topological A-model from the world-sheet perspective \cite{Walcher:2007qp}.)  In \cite{Costello:2019jsy} it was shown that with holomorphic Chern-Simons gauge group $\SO(8)$, anomaly cancellation occurs \emph{at all orders\footnote{There is a folklore belief that anomalies only occur at one loop, but this is false. A counterexample is given in \cite{CWYII}. }        in perturbation theory}.  

Further, the constraint that all anomalies are cancelled fixes all counter-terms uniquely.  This is worth noting, because BCOV theory and holomorphic Chern-Simons theory are both non-renormalizable by power counting. Intuitively, one should think that renormalizability of these theories is due to the large amount of gauge symmetry. In contrast to ordinary gauge theory and gravity, the group of gauge transformations in these theories which preserve the zero field configuration is infinite dimensional. It is the group of holomorphic gauge transformations  or the group of holomorphic volume-preserving diffeomorphisms respectively. 

This result has an important consequence for the four-dimensional theories we are considering, which was emphasized in \cite{Costello:2021bah}.
\begin{theorem}
	Mabuchi gravity coupled to WZW$_4$ for the group $\SO(8)$ admits a canonically defined quantization on a scalar-flat K\"ahler manifold, despite being non-renormalizable. 
\end{theorem}
This canonically defined quantization is characterized by asking that it lifts, at the quantum level, to a local theory on twistor space.  From the spacetime perspective, any theory that lifts to twistor space in a local manner has the following features:
\begin{enumerate} 
	\item It has no scattering amplitudes on flat space.
	\item On flat space, the analytically-continued RG flow is periodic, with imaginary period $2 \pi \i$. (This means that all coupling constants only depend on the logarithm $\lambda$ of the scale by periodic functions $e^{n \lambda}$ for $n$ an integer).  
	\item On any scalar-flat K\"ahler manifold, all correlation functions are meromorphic functions of the complexified spacetime coordinates. 
\end{enumerate}
These features strongly constrain all interactions and counter-terms.  For instance, in \cite{Costello:2021bah} it was shown that the cubic term in the Mabuchi functional is distinguished (among interactions of dimension $6$) by the constraint that it does not generate any logarithmic OPEs.  


\subsection{The Mabuchi theory as a gravitational theory}

Mabuchi gravity is a theory of a dynamical scalar field, which has the geometric interpretation as the K\"ahler potential of a metric. The diffeomorphism symmetry, however, is \emph{not} gauged. 

Gravitational actions where diffeomorphism symmetry is not gauged are familiar in other contexts. For instance, the Liouville theory in dimension $2$ is a Lagrangian for a scalar field $\phi$ which can also be interpreted as a model for quantum gravity (e.g. \cite{Polyakov:1981rd}).  If one interprets $e^{\phi}$ as a conformal factor multiplying a flat metric, then Liouville's equation for $\phi$ is equivalent to Einstein's equation.  

In Liouville theory, as in Mabuchi theory, one does not gauge diffeomorphism symmetries, so it is a gauge-fixed form of Einstein gravity. Earlier work coupling Liouville and Mabuchi actions in dimension $2$ appears in \cite{Ferrari:2011we}. 

One might object that Mabuchi gravity should not really be thought of as ``gravitational'' for this reason.  One might further object that standard holography, where the boundary CFT has a well-behaved stress tensor, cannot hold unless the bulk gravitational theory gauges diffeomorphism symmetry. 

Our holographic duality is saved because it is \emph{not} a standard holographic duality on four-dimensional spacetime. It only becomes a standard holographic duality on twistor space.  In the next section, we will show that the twistor space of the Burns metric is essentially $\SL_2(\C) \simeq \mbb{H}^3 \times S^3$, i.e., Euclidean AdS$_3 \times S^3$ (up to certain extra boundary divisors).  On twistor space, our duality is a completely standard holographic duality. Moreover, diffeomorphisms \emph{are} gauged on twistor space! It is only the boundary conditions at $0,\infty$ on the twistor $\CP^1$'s that allow us to fully gauge-fix the diffeomorphism symmetry.   


\section{Burns space from brane backreaction}
\label{sec:burns}

To construct our holographic duality, we start with the twistor space of $\R^4$. Recall that this was given by the bundle $\PT = \CO(1)\oplus\CO(1)\to\CP^1$. We will wrap a stack of $N$ D1 branes on the zero section of this bundle. Wrapping D1 branes on the zero section is a natural choice from the point of view of celestial holography: by construction, the zero section is the celestial sphere of the origin of flat four-dimensional spacetime (i.e. prior to backreacting). We will compute the backreaction as a deformation of complex structure.  The resulting backreacted geometry will get identified with the twistor space of a scalar-flat K\"ahler spacetime known as Burns space.


\subsection{The physical string setup}
\label{sec:branes}

Before passing to the twisted holographic setup, it is useful to outline the 10d physical string uplift of our basic model. We will restrict ourselves to the string theory compactified on the twistor space of flat space (i.e. before computing the backreaction arising from the addition of $N$ D-branes), and will pursue the computation of the backreaction to the twistor space of Burns space only in the corresponding topological string theory. We leave to future work the study of the open/closed duality in the full physical string theory.

The ingredients for our basic setup have been anticipated in \cite{Costello:2021bah}. The starting point is a $3$-fold $X$ which is a fibration $\pi: X \rightarrow \CP^1$ where the fiber is an elliptic K3 surface. As explained in \cite{Costello:2021bah}, we require that the fibration must be chosen to ensure the following condition on the canonical bundle
\begin{equation}
    K_X \simeq \pi^* \CO(2).
\end{equation}
It is not hard to show that one can construct such bundles over the base $\CP^1$ (see, e.g., the suggestion by D. Maulik in \cite{Costello:2021bah} for one example).   
Then we can build a kind of twistor space  
\begin{equation}
    \mathbb{PT}_X = \pi^*(\CO(1)\oplus \CO(1)) \rightarrow X.
\end{equation}
As a smooth manifold, $\mbb{PT}_X = X \times \R^4$.  This fibers (holomorphically) over the ordinary twistor space $\mbb{PT} = \CP^1 \times \R^4$, and the fibers are elliptic K3 surfaces. Points in complexified spacetime are $\CP^1$'s in ordinary twistor space. In our modified twistor space $\PT_X$, points in complexified spacetime are copies of $X$ inside $\PT_X$.  Because the normal bundle to $X$ in $\PT_X$ is $\pi^\ast \Oo(1) \oplus \pi^\ast \Oo(1)$, there is a four-dimensional moduli space of such.  

Our model is given by the type I string compactified on $\mathbb{PT}_X$. As usual, we must include 32 space-filling D9 branes plus an O9 plane. In flat space, the D9 branes contribute an $\SO(32)$ gauge group, though in our curved geometry we must turn on nontrivial gauge fluxes on the D9 worldvolume to cancel the resulting anomaly. At a generic point in the vector bundle moduli space, this breaks the gauge group to $\SO(8)$. (This is perhaps most easily seen by S-dualizing to the $\SO(32)$ heterotic string. When compactifying the heterotic string on a Calabi-Yau, one must turn on vector bundle moduli corresponding to flat gauge bundles. Such a vector bundle at a generic point in the moduli space of a K3 fibration will break the gauge group from $\SO(32) \rightarrow \SO(8)$, i.e. has second Chern class equal to 16 \cite{Aspinwall:1996mn}). For tadpole cancellation on a compact geometry, one must also wrap a prescribed number of D5 branes on $X$. We will elide the details of this computation by considering a decompactification limit of our geometry, so that we may add an arbitrary number $N$ of D5s on $X$. In this limit, the $\SO(8)$ gauge group becomes a flavor symmetry group as its coupling goes to zero, presaging the role it will play in our chiral algebra.

From this basic model, one can consider a variety of duality frames. We will content ourselves with making contact with the IIB frame sketched in \cite{Costello:2021bah}, though it is interesting to note than an explicit F-theory uplift of this geometry, with non-Higgsable SO(8) gauge group, can be constructed \footnote{We are grateful to M. Kim for very useful discussions on this.}. 

First, we may consider four T-dualities along the K3 fiber of $X$ (for example, we may consider the Kummer, or $T^4/\mathbb{Z}_2$, locus in the K3 moduli space to explicitly perform this operation). This produces the type $\text{I}'$ string on $\mathbb{PT}_X$, where the D9/O9 system has become a D5/O5 system wrapping the twistor space of $\R^4$, namely the directions $\CO(1)\oplus \CO(1) \rightarrow \CP^1$, and the D5 branes have become $N$ D1 branes on the base $\CP^1$. Notice that since the D1 and D5 branes are parallel, this is a 1/4-BPS system, preserving 8 supercharges. 

An additional two T-dualities can be performed on the elliptic fiber directions of the K3 fiber to reach a frame with a D7/O7 system and $N$ D3 branes. The type I string on the $\mathbb{PT}_X$ geometry is equivalent, as always, to an orientifold of the IIB string. From the IIB perspective, we have a D7/O7 system hosting an SO(8) gauge group, with $N$ D3 branes hosting an $\Sp(N)$ gauge group.\footnote{The geometric part of the $\mathbb{Z}_2$ involution can be taken to act naturally on the elliptic fiber of $\mathbb{PT}_X$, yielding a $T^2/\mathbb{Z}_2 \simeq \CP^1$ fibered over $\CP^1 \rightarrow \CP^1$.}
To this 10d duality frame, we may turn on an Omega-background along the (decompactification limit of the) four fiber directions (i.e. the K3 fiber directions in the $\mathbb{PT}_X$ geometry). This procedure results in an effective compactification down to the 6 twistor space directions, and an effective Euclidean D1/D5/O5- system. From this IIB frame one may pass to the topological string B-model ``subsector'', or equivalently the type I topological string.
It is this 6d type I topological string theory on $\CO(1) \oplus \CO(1) \rightarrow \mathbb{CP}^1$, with noncompact D5 branes supporting an $\SO(8)$ flavor symmetry and compact D1 branes supporting an $\Sp(N)$ gauge symmetry, that we will study in the remainder of this work.


\subsection{The backreacted twistor geometry}
\label{sec:backreact}

For the rest of this section, we revert to working in the simpler setting of topological strings. Let us start with $\PT = \CO(1)\oplus\CO(1)\to\CP^1$ and wrap $N\geq0$ D1 branes along the zero section $v^{\dal}=0$ (see section \ref{sec:twistor} for a review of our notation). This zero section is the twistor line $L_0$ corresponding to the origin of $\R^4$. We study the brane backreaction by adding a source term to the equation of motion \eqref{mceq} of BCOV theory on $\PT$.

Following \cite{CG}, the bulk-brane coupling adds a source term
\be
\frac{N}{2\pi\im}\int_{L_0}\p^{-1}(\mu\ip\Omega)\,,\qquad\Omega = \frac{\d z\,\d^2v}{z^2}\,,
\ee
to the action of BCOV theory.  As a result, the D1 branes backreact to generate a nontrivial complex structure deformation $\dbar\mapsto \dbar+\mu$ determined by the Maurer-Cartan equation with a delta function source
\be\label{dbarV}
\dbar \mu + \frac{1}{2}\,[\mu,\mu] = N\,(2\pi\im)^2\,\bar\delta^2(v)\,z^2\frac{\p}{\p z}\,.
\ee
Here, $\dbar = \d\bar z\,\p_{\bar z}+\d\bar v^\al\,\p_{\bar v^\al} = \d\bar z\,\p_{\bar z}+\d\hat v^{\dal}\,\p_{\hat v^{\dal}}$ is the antiholomorphic exterior derivative on $\PT$, and $\bar\delta^2(v) = \delta(v^{\dot1})\delta(v^{\dot2})\,\d\bar v^1\wedge\d\bar v^2$ denotes the complex delta distribution with support at $v^{\dal}=0$. The factor of $z^2$ on the right implies that the source vanishes at $z=0,\infty$. It ensures consistency with the boundary conditions that $\mu$ (and hence the left hand side) vanishes to second order at $z=0,\infty$ each. In particular, replacing $z\mapsto 1/z$ maps $z^2\p_z\mapsto-\p_z$, showing that $1/z$ acts like a more natural coordinate on our D1 brane worldvolume. We will occasionally return to this observation in later sections. It was mainly used in \cite{Costello:2022jpg} to simplify the expressions of certain scattering amplitudes.

The solution for the Beltrami differential that we are looking for is given by
\be\label{musol}
\mu = -N\,\frac{[\hat v\,\d\hat v]}{\|v\|^4}\,z^2\frac{\p}{\p z}\,.
\ee
It is clear that $\mu$ vanishes to second order at $z=0$. Under $z\mapsto 1/z$, we map $v^{\dal}\mapsto v^{\dal}/z$, so that $\mu$ also vanishes to second order at $z=\infty$. Since $[\hat v\,\d\hat v] = \overline{[v\,\d v]} = \bar v^2\d\bar v^1-\bar v^1\d\bar v^2$, we observe that $-[\hat v\,\d\hat v]/\|v\|^4$ is the Bochner-Martinelli kernel of the dbar operator on each $\C^2$ fiber at fixed $z$. Hence,
\be
\dbar\,\frac{[\hat v\,\d\hat v]}{\|v\|^4} = \dbar|_z\,\frac{[\hat v\,\d\hat v]}{\|v\|^4} = -(2\pi\im)^2\,\bar\delta^2(v)\,.
\ee
At the same time, $[\mu,\mu]\propto[\hat v\,\d\hat v]\wedge [\hat v\,\d\hat v] = 0$. So \eqref{musol} solves \eqref{dbarV}. $\mu$ is also divergence-free with respect to the holomorphic volume form $\d z\,\d^2v/z^2$, i.e., it satisfies the constraint \eqref{peta}. This is checked in two steps:
\begin{align}
&\eta = \mu\ip\frac{\d z\wedge\d^2v}{z^2} = N\,\frac{[\hat v\,\d\hat v]\wedge\d^2v}{\|v\|^4}\\
\implies&\p\eta = 2N\,\frac{[\hat v\,\d\hat v]\wedge[\hat v\,\d v]\wedge\d^2v}{\|v\|^6} = 0\,.
\end{align}
Our derivation of the solution for $\mu$ is completely analogous to the backreaction computed starting from the resolved conifold $\CO(-1)\oplus\CO(-1)$ in \cite{CG}, up to the factor of $z^2$ that comes from starting with $\CO(1)\oplus\CO(1)$ instead.

Define the set
\be\label{ZN}
Z_N \vcentcolon= \PT-\{v^{\dal}=0\}
\ee
and equip it with the deformed complex structure $\dbar+\mu$. We will refer to $Z_N$ with this complex structure as our backreacted geometry, with the convention that $Z_0$ corresponds to zero backreaction. $Z_N$ can again be covered by the two patches $z\neq0$ and $z\neq\infty$, so let us work in the patch $z\neq0$. 

A smooth function $f(z,v)$ is holomorphic in the deformed complex structure if and only if $(\dbar+\mu)f = 0$. Equivalently, it is holomorphic if it is annihilated by the deformed $(0,1)$-vector fields
\be\label{01vec}
\frac{\p}{\p\bar z}\,,\qquad\frac{\p}{\p\hat v^{\dal}} + \frac{N\hat v_{\dal}}{\|v\|^4}\,z^2\frac{\p}{\p z}\,.
\ee
These annihilate $v^{\dal}$ as well as the combinations
\be\label{wsol}
w^{\dal} = \frac{v^{\dal}}{z} - \frac{N\hat v^{\dal}}{\|v\|^2}\,,
\ee
so that any three out of the four quantities $v^{\dal},w^{\dal}$ can provide holomorphic coordinates on the patch $\{z\neq0\}\subset Z_N$. Contracting \eqref{wsol} with $\hat v_{\dal}$ and using the relation $[\hat v\,v]=\|v\|^2$, we can also invert \eqref{wsol} to solve for $z$, 
\be\label{zvw}
z = \frac{\|v\|^2}{[\hat v\,w]}\,.
\ee
This is no longer a holomorphic coordinate as it depends non-holomorphically on $v^{\dal}$. A similar analysis can also be performed in the ``antipodal patch'' $\{z\neq\infty\}$, but we will provide a more global description shortly.

Since $Z_N$ has three complex dimensions, the four holomorphic quantities $v^{\dal},w^{\dal}$ must be interrelated. Indeed, they are easily seen to satisfy
\be\label{conifold}
[v\,w]\equiv v^{\dot2}w^{\dot1} - v^{\dot1}w^{\dot2} = N\,.
\ee
Whenever $N\neq0$, we can rewrite this as the condition
\be\label{vwmat}
\frac{1}{\sqrt N}\begin{pmatrix}w^{\dot1}&&v^{\dot1}\\w^{\dot2}&&v^{\dot2}\end{pmatrix}\in\SL_2(\C)\,,
\ee
i.e., that this $2\times2$ matrix has unit determinant. Therefore, the map $(v^{\dal},z)\mapsto(v^{\dal},w^{\dal})$ identifies this patch of $Z_N$ with the group manifold of $\SL_2(\C)$. The complex structure of $Z_N$ is then simply the standard complex structure of $\SL_2(\C)$.

This mechanism of obtaining $Z_N$ from $\PT$ by brane backreaction is a twistorial analogue of the twisted holographic backreaction computed in \cite{CG, Costello:2020jbh}, and is also redolent of the geometric transition from the resolved conifold to the deformed conifold.

\paragraph{Deformed twistor lines.} We wish to interpret $Z_N$ as the twistor space of some asymptotically flat spacetime. To do this, we need to find a 4-parameter family of rational curves that are holomorphic with respect to $\dbar+\mu$ and have normal bundle $\CO(1)\oplus\CO(1)$. They are also required to be invariant under a suitable real structure on $Z_N$. These will act as our deformed twistor lines. 

Let us first quote the result. In the coordinates $v^{\dal},w^{\dal}$, the deformed twistor lines are captured by the new incidence relations
\begin{equation}\label{hcurve}
    v^{\dal}(\sigma) = u^{\dal} + \sigma\hat u^{\dal}\,,\qquad w^{\dal}(\sigma) = \hat u^{\dal} + \frac{u^{\dal}}{\sigma}\left(1+\frac{N}{\|u\|^2}\right)\,,
\end{equation}
where $\sigma$ is a stereographic coordinate on $\CP^1\simeq S^2$.  These satisfy $[v(\sigma)\,w(\sigma)]=N$ and are parametrized by the four real moduli contained in $u^{\dal}\in\C^2-0$. As before, $\hat u^{\dal}=(-\bar u^2,\bar u^1)$. Along such a curve, the coordinate $z$ varies non-holomorphically:
\be
z(\sigma) = \sigma-\frac{N\sigma}{N+(1+|\sigma|^2)\|u\|^2}\,.
\ee
As $N\to0$, we get $z(\sigma)\to\sigma$ and the curves revert to the twistor lines of $\PT$. 

The twistor lines \eqref{hcurve} are defined as written for $u \neq 0$.  However, they do make sense in the $\norm{u} \to 0$ limit. If we send $\norm{u} \to 0$ while making the reparametrization $\sigma \mapsto \frac{\sigma}{ \norm{u}}$, the twistor line \eqref{hcurve} becomes
\begin{equation} \label{hcurve0}
	v^{\dal}(\sigma) = \frac{\sigma\hat u^{\dal}} {\|u\|} \,,\qquad w^{\dal}(\sigma) = \frac{N u^{\dal}}{\sigma\|u\|}\,,
\end{equation}
which is well-defined in the $\|u\|\to0$ limit. This twistor line depends only on the phase of $u/\|u\|$, but a change of phase can also be absorbed into a reparametrization of $\sigma$.  Therefore it only depends on the projectivization of the 2-spinor $u^{\dal} \in \C^2-0$. 

What we have learned is that twistor lines \eqref{hcurve} are associated, not to points in $\C^2$, but to points in its blow-up $\til{\C}^2$ at the origin.  In other words, the backreaction on twistor space has changed the topology of spacetime by blowing up its origin. We will see this again from the perspective of the spacetime metric shortly. 

More generally, wrapping D1 branes along the twistor line of a point $x\in\C^2$ will induce a geometric transition that blows up $\C^2$ at $x$. Hence, our analysis concretely realizes the expectations outlined in the prescient work of Hartnoll and Policastro \cite{Hartnoll:2004rv}.

\medskip
Being patchwise diffeomorphic to $\SL_2(\C)$, the deformed geometry $Z_N$ comes equipped with a natural real structure induced by the antiholomorphic map
\be\label{burnsreal}
(v^{\dal},w^{\dal})\mapsto(-\hat w^{\dal},\hat v^{\dal})\,.
\ee
Due to the property $\hat{\hat v}^{\dal}=-v^{\dal}, \hat{\hat w}^{\dal}=-w^{\dal}$, this map is an involution on $\C^4$. Complex conjugation shows that $v^{\dal}w_{\dal}=N\iff(-\hat w^{\dal})\hat v_{\dal} = N$, so it descends to a well-defined involution on $\SL_2(\C)$. The set of fixed points of \eqref{burnsreal} would have been $\{w^{\dal}=\hat v^{\dal}\}$, but such points satisfy $[v\,w]=[v\,\hat v]=-\|v\|^2<0$. Since $N>0$, they cannot satisfy $[v\,w]=N$, so this subset of $\C^4$ does not intersect $\SL_2(\C)$. Hence, our involution is fixed-point-free. 

Crucially, our curves \eqref{hcurve} are invariant under this involution. This is best seen by reparametrizing the curves through a change of coordinate
\be
\sigma\mapsto\sigma\,\sqrt{1+\frac{N}{\|u\|^2}}\,.
\ee
The equations of the curve now read\footnote{This symmetric form of the curves was communicated to us by Lionel Mason.}
\be
v^{\dal}(\sigma) = u^{\dal} + \sigma\hat u^{\dal}\sqrt{1+\frac{N}{\|u\|^2}}\,,\qquad w^{\dal}(\sigma) = \hat u^{\dal} + \frac{u^{\dal}}{\sigma}\sqrt{1+\frac{N}{\|u\|^2}}\,.
\ee
In this parametrization, we observe that if our curve passes through $(v^{\dal},w^{\dal})$ at a given point $\sigma=\sigma_0$, then it also passes through $(-\hat w^{\dal},\hat v^{\dal})$ at the antipodal point $\sigma=-1/\bar\sigma_0$. So our curves are the required \emph{real} twistor lines. For practical calculations however, it will be easiest to stick with the parametrization \eqref{hcurve}.

\medskip

Actually, there are some further subtleties in this identification. Taken at face value, the curves in \eqref{hcurve} map the poles $\sigma=0,\infty$ to the $[v\,w]\to0$ boundary of $\SL_2(\C)$. To obtain a genuine twistor space, we need to enlarge $Z_N$ so as to also include such boundary points. This is resolved by the observation that $Z_N$ embeds into the twistor space of $\CP^2$. 

As we will see shortly, the backreaction on spacetime produces the Burns geometry, which is a manifold conformally equivalent to $\CP^2$ with a point at $\infty$ removed (cf.\ the discussion around equation \eqref{burnstocp2}). As a smooth manifold, the latter is diffeomorphic (in an orientation-reversing manner) to the blow-up of $\C^2$ at the origin, which includes $\C^2-0$ as a coordinate patch. Since twistor space  as a complex manifold (without any volume form) only depends on the conformal structure on the $4$-manifold, the full twistor space of the Burns geometry is the same as that of $\CP^2$, with a point at $\infty$ removed.

The twistor space of $\CP^2$ equipped with its Fubini-Study metric is the flag variety $\F$ of points in lines in $\CP^2$ \cite{Atiyah:1978wi,Ward:1980am}
\be\label{Fdef}
\F = \bigl\{([V^k],[W_k])\in\CP^2\times\CP^2\,\bigl|\,V^kW_k=0\bigr\}\,,\qquad k\in\{1,2,3\}\,.
\ee
There exists a (non-holomorphic) twistor fibration $\F\to\CP^2$ whose projection map is given by a cross product of 3-vectors
\be
([V],[W])\mapsto [\bar V\times W]\equiv[\eps^{klm}\bar V_lW_m]
\ee
with $\eps^{klm}$ the Levi-Civita symbol, and $\bar V_l\equiv\overline{V^l}$ denoting componentwise complex conjugation. The fibers of this fibration are rationally embedded $\CP^1$'s that get identified with twistor lines. Hence, $\F$ is diffeomorphic to a sphere bundle over $\CP^2$ as required.

$Z_N$ embeds holomorphically into $\F$ via the map
\be\label{ZNtoF}
(v^{\dal},w^{\dal})\mapsto [V^k] = [v^{\dal},-\sqrt N]\,,\;[W_k]=[w_{\dal},\sqrt N]\,.
\ee
So, $Z_N$ is the locus in the flag variety where both $V^3$ and $W_3$ are non-zero. To obtain the full twistor space, we need to adjoin the locus where exactly one of $V^3, W_3$ is zero. This is the quadric $V^3W_3=0$, except for the locus $V^3=W_3=0$ that forms the twistor line of the point at infinity $[0,0,1]\in\CP^2$ (the twistor lines are described below in more detail). This will play the role of Pontecorvo's quadric in the backreacted geometry. Because of $V^kW_k=0$, this is the same as the quadric $V^1W_1+V^2W_2=v^{\dal}w_{\dal}=0$.

The locus where $V^3 = 0$ arises when the vector $v^{\dal}$ tends to $\infty$ while $w^{\dal}$ remains finite.  By adjoining this locus, we include the $\sigma = \infty$ point in the parametrized curve \eqref{hcurve}. Similarly, the locus $W_3 = 0$ arises when $w^{\dal} \to \infty$ so that $\sigma \to 0$ in \eqref{hcurve}.

The divisor $\{V^3 = 0, W_3\neq0\}$ is a copy of the blow-up of $\C^2$ at $0$. To see this, we note that on this locus, we can set $W_3 = 1$.  Then the region is parameterized by $v^{\dal}$ and $w^{\dal}$ satisfying $[v\,w] = 0$, where $v$ is projective.  If $w \neq 0$, then $v$ must be the projectivization of $w$. If $w = 0$, then $v$ is arbitrary; so we have inserted into the $\C^2$ parametrized by $w$ a copy of $\CP^1$ at $0$. Similarly for the locus $\{W^3 = 0,V^3\neq0\}$. 

Thus, the full twistor space we are interested in is obtained by adjoining to our backreacted space $\SL_2(\C)$ two copies of the blow-up of $\C^2$.  Every twistor line intersects each of these divisors in a point.   We continue to refer to this completed twistor space as $Z_N$ by a mild abuse of notation.  

The flag variety also includes the locus $V^3 = W_3 = 0$. On this locus, $v,w$ are projective vectors and $[v\,w] = 0$. Thus, $v$ is proportional to $w$, and this locus is a copy of $\CP^1$. It is the diagonal $\CP^1$ in $\CP^1_v \times \CP^1_w$.  This $\CP^1$ locus is not part of the twistor space of Burns space: it is the $\CP^1$ corresponding to the point in $\CP^2$ we have removed. It will be part of the asymptotic boundary.

We can now view our deformed incidence relations \eqref{hcurve} as twistor lines of $\CP^2$. To see this, let
\begin{align}
    U^k&=\left(u^{\dal},\frac{\|u\|^2}{\sqrt N}\right) = \left(u^{\dot1},u^{\dot2},\frac{\|u\|^2}{\sqrt N}\right)\,,\\
    \bar U_k &\equiv\overline{U^k} = \left(-\hat u_{\dal},\frac{\|u\|^2}{\sqrt N}\right) = \left(\bar u^1,\bar u^2,\frac{\|u\|^2}{\sqrt N}\right)\,.
\end{align}
$[U^k]$ gives a point on $\CP^2$, and $[\bar U_k]$ is its complex conjugate. Equivalently, we are using $\sqrt N u^{\dal}/\|u\|^2$ as affine coordinates on a patch of $\CP^2$ (these are related to the standard affine coordinates by an orientation reversing inversion $u^{\dal}\mapsto \sqrt{N}u^{\dal}/\|u\|^2$). The twistor line associated to $U^k$ is a rational curve in the flag variety $\F$. It is abstractly given by the set
\be\label{hcurveCP2}
\{([V^k],[W_k])\in\CP^2\times\CP^2\,|\,V^k W_k = U^kW_k=V^k\bar U_k=0\}\,.
\ee
This can be rationally parametrized in terms of a stereographic coordinate $\sigma\in\CP^1$ by composing \eqref{hcurve} with the map \eqref{ZNtoF}. As we are working with homogeneous coordinates, we are now able to write this composition in a way that is well-defined at $\sigma=0,\infty$,
\be
    [V^k] = \left[u^{\dal}+\sigma\hat u^{\dal},-\sqrt N\right]\,,\quad
    [W_k] = \left[\biggl(1+\frac{N}{\|u\|^2}\biggr)u_{\dal}+\sigma\hat u_{\dal},\sqrt N\sigma\right]\,.
\ee
This parametrization manifestly satisfies $V^k W_k = U^kW_k=V^k\bar U_k=0$. 

We can also show that these are \emph{real} twistor lines of $\CP^2$, i.e., they are preserved by the natural fixed-point-free antiholomorphic involution on $\F$ that maps $([V^k],[W_k])\mapsto([\bar W^k],[\bar V_k])$, where $\bar V_k\equiv\overline{V^k}$, $\bar W^k\equiv\overline{W_k}$. To verify this, one computes
\be
    [\bar W^k] = \left[\biggl(1+\frac{N}{\|u\|^2}\biggr)\hat u^{\dal}-\bar\sigma u^{\dal},\sqrt N\bar\sigma\right]\,,\quad
    [\bar V_k] = \left[-\hat u_{\dal}+\bar\sigma u_{\dal},-\sqrt N\right]\,,
\ee
from which one can derive that $\bar W^k\bar V_k = U^k\bar V_k=\bar W^k\bar U_k=0$. So both $([V^k],[W_k])$ and its conjugate $([\bar W^k],[\bar V_k])$ lie on the same twistor line. We conclude that we have found the correct twistor lines required for constructing a self-dual spacetime with a metric of Euclidean signature. Since the twistor construction is conformally invariant, this metric will necessarily be conformal to the Fubini-Study metric.

In this picture, it is also clear that the twistor lines \eqref{hcurve} correspond to points in the patch $\{U^3\neq0\} \simeq \C^2-0\subset\CP^2-\{[0,0,1]\}$. Points like $(v^{\dal},w^{\dal}) = (\hat\zeta^{\dal}/\|\zeta\|,N\zeta^{\dal}/\|\zeta\|)$ also belong to the flag variety. But they project down to the set of points $[U^k] = [\zeta^{\dal},0]$ in $\CP^2$. Viewed as points on the blow-up of $\C^2$ at $0$, these are just the points on the exceptional divisor of the blow-up, $[\zeta^{\dal}]\in\CP^1$. This is why they did not belong to any of our twistor lines \eqref{hcurve}. Instead, the appropriate twistor lines are found by setting $U^k=(\zeta^{\dal},0)$ in the unparametrized definition \eqref{hcurveCP2}. We can also easily guess a parametrization for them as maps $\CP^1-\{0,\infty\}\to\SL_2(\C)$,
\be
v^{\dal}(\sigma) = \frac{\sigma \hat\zeta^{\dal}}{\|\zeta\|}\,,\qquad w^{\dal}(\sigma) = \frac{N\zeta^{\dal}}{\sigma\|\zeta\|}\,.
\ee
These agree with \eqref{hcurve0} if one sets $u^{\dal}=t\,\zeta^{\dal}$ with $t\in\C$, $[\zeta^{\dal}]\in\CP^1$, and absorbs the phase of $t$ in a reparametrization of $\sigma$. If needed, they may be written in a more symmetric fashion by rescaling $\sigma\mapsto\sigma\sqrt{N}$. Together with the lines \eqref{hcurve}, these foliate $Z_N$. 

\paragraph{The backreacted spacetime geometry.} 
Finally, let us obtain the scalar-flat K\"ahler manifold that $Z_N$ corresponds to. We will derive the metric on $\C^2-0\subset\til\C^2$, and it will be shown to extend to the exceptional divisor of the blow-up in the next section. 

The moduli space of the curves \eqref{hcurve} is $\C^2-0$ and has complex coordinates $u^{\dal}$. Since $Z_N$ embeds into the twistor space of $\CP^2$, the SD metric one obtains on $\C^2-0$ from the Penrose transform will be in the same conformal class as Fubini-Study (though with the opposite orientation). To fix a K\"ahler metric in this conformal class, we need to prescribe a meromorphic 3-form on $Z_N$.  

There is a unique up to scale holomorphic volume form on $\SL_2(\C)$ which is invariant under both the left and right actions of $\SL_2(\C)$.    It is conveniently described by means of the Poincar\'{e} residue. Given an analytic hypersurface $f(z_i)=0$ embedded in $\C^n$ -- with $z_i$ the coordinates on $\C^n$ -- a natural holomorphic volume form on it is given by the residue of $\d^nz/f$ at $f=0$. For $\SL_2(\C)\subset Z_N$, a global $(3,0)$-form is obtained by computing a residue of the volume form of $\C^4$ along $\SL_2(\C)$,
\be\label{3-form}
\Omega_N = -\text{Res}_{Z_N}\frac{\d^2v\wedge\d^2w}{[v\,w]-N} = -\frac{\d^2v\wedge\d w^{\dot1}}{v^{\dot1}}\,,
\ee
where the overall sign defines the orientation in which $\Omega_N$ reduces to $z^{-2}\d z\wedge\d v^{\dot1}\wedge\d v^{\dot2}$ as $N\to0$ (see \eqref{omegazv} below). 

This can also be written as a weightless meromorphic 3-form on the embedding of $Z_N$ in $\F$, 
\be\label{omegaF}
\Omega_N = \text{Res}_\F\,\frac{N}{(V^3W_3)^2}\frac{\D^2V\wedge\D^2W}{V\cdot W}\,,
\ee
with manifest second order poles on the quadric $V^3W_3=0$. In this expression, $\D^2V = \frac12\,\eps_{ijk}V^i\d V^j\wedge\d V^k$ and $\D^2W = \frac12\,\eps^{ijk}W_i\d W_j\wedge\d W_j$ are weighted holomorphic top forms on the two copies of $\CP^2$ involved in the definition \eqref{Fdef}, and the residue is taken at the hypersurface $V\cdot W=0$ cutting out $\F$. Pulling back by the embedding \eqref{ZNtoF}, and using $V^3W_3 = -(V^1W_1+V^2W_2)=-v^{\dal}w_{\dal}$ on $\F$, we see that \eqref{omegaF} reduces to \eqref{3-form} away from the polar quadric $V^3W_3=0$.

An important property of $\Omega_N$ is that its periods measure D1 brane charge \cite{Aganagic:2003db,Aganagic:2003qj}. This is the B-model analogue of D branes carrying RR charges in type II strings. Indeed, suppose we rewrite $\Omega_N$ in the undeformed coordinates $z,v^{\dal}$. Substituting for $w^{\dot1}$ from \eqref{wsol}, we can reexpress \eqref{3-form} as
\be\label{omegazv}
\Omega_N = \d^2v\wedge\left(\frac{\d z}{z^2} + \frac{N\,\D\hat v}{\|v\|^4}\right)\,.
\ee
If we integrate this over any 3-cycle surrounding the brane locus $v^{\dal}=0$, we find a multiple of the brane charge. For instance, integrating it over $\|v\|=1$ at fixed $z$ yields
\be
\int_{z,\|v\|=1}\Omega = N\int_{\|v\|=1}\d^2v\wedge\D\hat v = (2\pi\im)^2N
\ee
where the 3-sphere integral may be performed using standard techniques. These periods are necessarily quantized for the path integral of holomorphic Chern-Simons \eqref{hcs} to be invariant under gauge transformations with non-zero winding number. This provides us a twistorial reason for the quantization of $N$. 

We can use $\Omega_N$ to construct the backreacted spacetime geometry. As in flat space, the incidence relations \eqref{hcurve} provide a diffeomorphism from the projective spinor bundle of $\C^2-0$ to our backreacted twistor space $Z_N$. Pulling back $\Omega_N$ by this diffeomorphism yields
\be
\Omega_N = \frac{\d^2v}{\sigma^2}\wedge\left(\d\sigma+\frac{N[u\,\d u]}{\|u\|^4}\right)\,.
\ee
This is derived by expressing $\d w^{\dal}$ in terms of $\d\sigma,\d u^{\dal},\d v^{\dal}$, so that all $\d v^{\dal}$ dependence drops when wedging against $\d^2v$. As expected, it has a pair of second-order poles in $\sigma$, located at $\sigma=0,\infty$ in our choice of frame for the projective spinor bundle. Explicitly substituting
\be
\d^2v = \d v^{\dot1}\wedge\d v^{\dot2} = (\d u^{\dot1}-\sigma\,\d\bar u^2-\bar u^2\d\sigma)\wedge(\d u^{\dot2}+\sigma\,\d\bar u^1+\bar u^1\d\sigma)
\ee
and applying the contour integral formula \eqref{omegacirc}, we obtain the 2-form
\be\label{omegaburns}
\omega = \frac{1}{2\pi}\oint_{\sigma=0}\Omega_N = \im\left(\eps_{\dal\dot\beta}\,\d u^{\dal}\wedge\d\hat u^{\dot\beta} + N\,\frac{[u\,\d u]\wedge[\hat u\,\d\hat u]}{\|u\|^4}\right)
\ee
on $\C^2-0$. It is hermitian and satisfies $\d\omega=0$, so it gives the K\"ahler form of a hermitian metric on $\C^2-0$. The Penrose transform ensures that it is a self-dual (and hence scalar-flat) metric. Introducing the holomorphic and antiholomorphic derivatives $\p = \d u^{\dal}\,\p_{u^{\dal}}$, $\dbar = \d\hat u^{\dal}\,\p_{\hat u^{\dal}}$ on $\C^2$, we find
\be
\omega = \im\,\p\dbar K\,,\qquad K = \|u\|^2 + N\log\|u\|^2\,.
\ee
In the limit $N\to0$ of zero backreaction, $K$ reverts to the K\"ahler potential of the flat metric. 


\subsection{Burns space}
\label{sec:ahs}

Let $x^{\al\dal} = (u^{\dal},\hat u^{\dal})$ be complex coordinates on $\C^2-0$ as before (see \eqref{utox} for our conventions). The K\"ahler metric obtained from the K\"ahler form \eqref{omegaburns} is known as the \emph{Burns metric} \cite{burns1986twistors}. It is given by
\be\label{burns}
    \d s^2 = 2\left(\|\d u\|^2 + N\,\frac{|[u\,\d u]|^2}{\|u\|^4}\right)\,,
\ee
where $|[u\,\d u]|^2 = |u^{\dot2}\d u^{\dot1}-u^{\dot1}\d u^{\dot2}|^2$ as usual. This is Riemannian for $N>0$. 

Its isometry group is $\U(2)$, generated by rotations on the dotted index:
\be
u^{\dal}\mapsto L^{\dal}{}_{\dot\beta}u^{\dot\beta}\,,\qquad L\in\U(2)\,.
\ee 
This is a $\U(1)\times\SU(2)$ subgroup of the $\Spin(4)\simeq\SU(2)\times\SU(2)$ isometry group of flat space: the left-handed $\SU(2)$ is broken down to $\U(1)$, whereas the right-handed $\SU(2)$ remains unbroken and acts on the dotted indices in the usual way. The two $\U(1)$ subgroups of this isometry group generate circle actions in the $u^{\dot1}$ and $u^{\dot2}$ complex planes, each of which can be used to rewrite the Burns metric as a metric on a circle bundle over AdS$_3$. This is reviewed later on in this section.

Following the general discussion in section \ref{sec:open} and \ref{sec:closed}, the type I B-model compactifies to Mabuchi gravity plus an $\SO(8)$ WZW$_4$ model in the classical background of the Burns metric. For the purposes of holography, we would like to work in the large $N$ limit. To obtain a meaningful limit, one needs to perform the change of coordinates
\be\label{utoNu}
u^{\dal}\mapsto\sqrt{N}u^{\dal}\,.
\ee
Doing this, one obtains an alternative expression for the Burns metric,
\be
\d s^2 = 2N\left(\|\d u\|^2 + \frac{|[u\,\d u]|^2}{\|u\|^4}\right)\,.
\ee
The corresponding K\"ahler potential becomes $K=N(\|u\|^2+\log\|u\|^2)$. 

In this coordinate system, the WZW$_4$ action (including the coupling to the K\"ahler perturbation $\rho$) rescales as
\be
S_{\text{WZW}_4}[\phi,\omega+\im\,\p\dbar\rho] \mapsto N S_{\text{WZW}_4}\!\left[\phi,\omega|_{N=1}+\frac{\im\,\p\dbar\rho}{N}\right]
\ee
where $\omega|_{N=1}$ is the canonically normalized K\"ahler form obtained by setting $N=1$ in \eqref{omegaburns}. This is the form of the action employed in our Letter \cite{Costello:2022jpg}. The quantization of $N$ follows from the quantization of the K\"ahler form required by the WZW$_4$ model. Rescaling $\phi\mapsto\phi/\sqrt{N}$ then turns $1/\sqrt N$ into the topological string coupling. The $\phi^3$ coupling becomes $1/\sqrt{N}$, which is the open string coupling constant, while the $\rho\phi^2$ coupling is $1/N$, the closed string coupling constant. This allows for a Green-Schwarz anomaly cancellation of the gauge and gravitational anomalies in twistor space \cite{Costello:2019jsy,Costello:2021bah}.

The rescaling of the Mabuchi functional can also be evaluated straightforwardly without any need to rescale $\rho$. The K\"ahler form changes to $\omega\mapsto\omega|_{N=1}N$ while the Laplacian maps to $\lap\mapsto\lap|_{N=1}N^{-1}$ (see \eqref{burnslap}), so the kinetic term in \eqref{rhoac} immediately becomes independent of $N$. The $\rho^3$ coupling is seen to be $1/N$ as expected. 

Ideally, this is the form of the action one should use to take the large $N$ limit, but in what follows we will continue to work with the original metric \eqref{burns}. This will prove useful since most of our calculations will turn out to be expressible as formal power series in $N$. This does not cause any trouble with taking a large $N$ limit because to work in the large $N$ limit simply means working at tree level in the bulk in most scenarios of interest. 

\medskip

For now, let us continue to explore the geometry associated to \eqref{burns}. The square root of the determinant of the Burns metric is
\be
\sqrt{|g|} = 1+\frac{N}{\|u\|^2}\,.
\ee
Using this, we can compute its Ricci form
\be\label{ricciburns}
\begin{split}
    P &= -\im\,\p\dbar\log\sqrt{|g|}\\
    &= \frac{\im N}{\|u\|^2(N+\|u\|^2)}\left(\frac{[u\,\d u]\wedge[\hat u\,\d\hat u]}{\|u\|^2}+\frac{[\hat u\,\d u]\wedge[u\,\d\hat u]}{N+\|u\|^2}\right)\,.
\end{split}
\ee
This is non-vanishing, so the metric is not Ricci-flat as expected. On the other hand, noticing that the K\"ahler form \eqref{omegaburns} can be reexpressed as
\be\label{omegaburns1}
\omega = \im\left(1+\frac{N}{\|u\|^2}\right)\left(\frac{[u\,\d u]\wedge[\hat u\,\d\hat u]}{\|u\|^2}-\frac{[\hat u\,\d u]\wedge[u\,\d\hat u]}{N+\|u\|^2}\right)\,,
\ee
it becomes immediate that the Ricci form satisfies $\omega\wedge P=0$. This confirms that the Burns metric is scalar-flat, i.e., has Ricci scalar $R=0$. Importantly, this tells us that the Burns metric has cosmological constant $\Lambda=0$,\footnote{As opposed to the more general possibility of SD metrics of constant scalar curvature $R=24\Lambda\neq0$.} distinguishing it from the asymptotically AdS geometries obtained from backreaction in past holographic setups.

The singularity at the origin $u^{\dal}=0$ is not a genuine curvature singularity, as is easily corroborated by computing the squared Ricci tensor or the Kretschmann scalar,
\be
R^{\mu\nu}R_{\mu\nu} = \frac{4N^2}{(N+\|u\|^2)^4}\,,\qquad R^{\mu\nu\rho\sigma}R_{\mu\nu\rho\sigma} = \frac{32N^2}{(N+\|u\|^2)^4}\,,
\ee
which are both finite at the origin. It simply signals that the spacetime is incomplete. 

In fact, we can easily extend \eqref{burns} from a metric on $\C^2-0$ to a metric on $\widetilde\C^2$, the blow-up of $\C^2$ at the origin, by identifying $\widetilde\C^2$ with the total space of the tautological line bundle $\CO(-1)\to\CP^1$. To see this, let $\zeta^{\dal}$ denote homogeneous coordinates on the base $\CP^1$. And let $t\in\C$ be a holomorphic coordinate along the fibers of $\CO(-1)$. We can perform the coordinate transformation
\be\label{utsub}
u^{\dal} = t\,\zeta^{\dal}
\ee
which is bijective for $u^{\dal}\neq0$ and maps to the exceptional divisor $\{t=0$, $[\zeta^{\dal}]\in\CP^1\}$ at the origin $u^{\dal}=0$. Under non-zero rescalings $\zeta^{\dal}\mapsto s\zeta^{\dal}$, the fiber coordinate transforms as $t\mapsto s^{-1}t$ and $u^{\dal}$ remains invariant. Note that since $t=0$ is the origin of $\C^2$, the metric \eqref{burns} is naively ill-defined there. But since $[\zeta\,\zeta]=0$, one finds that $[u\,\d u]=t\,[\zeta\,\d\zeta]$, so the substitution \eqref{utsub} converts \eqref{burns} into
\be\label{burnst}
\d s^2 = 2\left(\|\zeta\,\d t+t\,\d\zeta\|^2+N\,\frac{|[\zeta\,\d\zeta]|^2}{\|\zeta\|^4}\right)\,.
\ee
Since $t$ drops out of the second term, this resulting metric is non-singular at the locus $t=0$. Hence, we can extend the metric freely to $t=0$ and thus to the total space of $\CO(-1)$. 

Let $\D\zeta\vcentcolon=[\zeta\,\d\zeta]$ denote the standard weight $2$ frame for the holomorphic cotangent bundle of $\CP^1$. We can complete this into a frame $\D\zeta,\D t$ for the holomorphic cotangent bundle of $\CO(-1)$ by using the standard Chern connection 1-form on $\CO(-1)$,
\be
\D t \vcentcolon= \d t + \frac{[\hat\zeta\,\d\zeta]}{\|\zeta\|^2}\,t\,.
\ee
This frame is designed to be conformally covariant under local non-zero rescalings:
\be\label{locres}
(\zeta^{\dal},t)\mapsto\left(s(\zeta,t)\,\zeta^{\dal},\frac{t}{s(\zeta,t)}\right)\implies (\D\zeta,\D t)\mapsto\left(s(\zeta,t)^2\,\D\zeta,\frac{\D t}{s(\zeta,t)}\right)\,.
\ee
The metric \eqref{burnst} can then be compactly expressed as
\be\label{burnst1}
\d s^2 = 2\left\{\|\zeta\|^2|\D t|^2 + \left(N+|t|^2\|\zeta\|^2\right)\frac{|\D\zeta|^2}{\|\zeta\|^4}\right\}\,.
\ee
It is seen to be invariant under the local rescalings \eqref{locres}, so it is well-defined everywhere on $\CO(-1)$. 

In our analysis, we have been slightly cavalier about length dimensions by viewing both $u^{\dal}$ and $N$ as dimensionless. We can restore standard length dimensions for the coordinates $u^{\dal}$ by replacing $N$ by a length scale $N\ell^2$. Here, $\ell\sim\sqrt{\alpha'}$ is essentially the string length and is physically interpreted as setting the size of the blow-up at the origin. As $t\to0$ with $N>0$, \eqref{burnst1} reduces to $N/2$ times the usual stereographic metric $4|\D\zeta|^2/\|\zeta\|^4$ on the exceptional divisor $\CP^1$. Upon changing $N\mapsto N\ell^2$, this shows that the radius of the exceptional divisor is precisely $\sqrt{\frac N2}\,\ell$. The limit $N\to0$ of zero backreaction is equivalent to turning off the blow-up. Since the topological string is classically exact in $\alpha'$, we will suppress $\ell$ consistently throughout this work by setting it to 1.

\medskip

The Burns metric is a complete, \emph{asymptotically flat}, zero scalar curvature metric on $\widetilde\C^2$ \cite{burns1986twistors,lebrun1988counter,lebrun1988poon}, and our discussion of its twistor space has been adapted from the more extensive discussion presented in \cite{Lebrun91explicitself-dual}. We conclude that the completed spacetime geometry obtained from D1 brane backreaction in the twistorial B-model is that of $\widetilde\C^2$ equipped with this metric. We refer to it variously as \emph{Burns space} or \emph{Burns geometry}.

\paragraph{Asymptotic flatness.} In Euclidean signature, a metric on a connected, non-compact 4-manifold $M$ is called asymptotically flat (more appropriately asymptotically Euclidean) if it asymptotes to the Euclidean metric with a falloff of at least $\mathrm{O}(r^{-2})$ with respect to some 4-radius $r=\sqrt{x^\mu x_\mu}$. More precisely, suppose there exists a compact set $K\subset M$ such that $M-K$ is diffeomorphic to the complement of a closed ball in $\R^4$. Let $x^\mu$ denote coordinates on this $\R^4$. If we can find a large enough $K$ for which the metric components in these coordinates have the falloff
\be
g_{\mu\nu} = \delta_{\mu\nu} + \mathrm{O}(r^{-2})\,,
\ee
then the metric $g$ is said to be \emph{asymptotically flat} \cite{Lebrun91explicitself-dual,lebrun1992twistors}. This property is manifest for the Burns metric in the complex coordinates $u^{\dal}\in\C^2-0$ employed in this section, as its difference $2Nu_{\dal}\hat u_{\dot\beta}\d u^{\dal}\d\hat u^{\dot\beta}/\|u\|^4$ from the flat metric has the componentwise falloff $\mathrm{O}(\|u\|^2/\|u\|^4) = \mathrm{O}(\|u\|^{-2})$.

\paragraph{Relation to $\CP^2$ and spacetime foam.} At the level of the topology of spacetime, the backreaction has precipitated a geometric transition from flat space $\C^2\simeq\R^4$ to its blow-up at the origin $\widetilde\C^2\simeq\CO(-1)$! In fact, it was already conjectured in the early work on twistor strings \cite{Hartnoll:2004rv} that wrapping D1 branes on twistor lines corresponds to blowing up the corresponding points in spacetime; our results give a concrete realization of the simplest case of this proposal. The work of \cite{Hartnoll:2004rv} was originally motivated from the subject of spacetime foam, where $\CP^2$ plays a key role. 

The Burns metric is seen to be conformal to the Fubini-Study metric on $\CP^2$ by performing the inversion $u^{\dal}\mapsto \sqrt{N}u^{\dal}/\|u\|^2$ \cite{lebrun1988poon},
\be\label{burnstocp2}
\begin{split}
    \d s^2_\text{Burns} &\mapsto 2N\left(1+\frac{1}{\|u\|^2}\right)^2\frac{\|\d u\|^2+|[u\,\d u]|^2}{(1+\|u\|^2)^2}\\
    &= 2N\left(1+\frac{1}{\|u\|^2}\right)^2\;\d s^2_\text{Fubini-Study}\,.
\end{split}
\ee
The conformal factor is clearly $\sqrt{2N}\,(1+\|u\|^{-2})$. Though we are working in an affine patch, one can check that this conformal transformation extends everywhere on $\CP^2$ except for its singularity $u^{\dal}=0$. Introducing homogeneous coordinates $[U^k] = [u^{\dal},1]$ on $\CP^2$, this is seen to be the point $[0,0,1]$, which one may interpret as the ``point at infinity''. We conclude that Burns space is conformally diffeomorphic to $\CP^2$ minus a point equipped with its Fubini-Study metric.\footnote{It is easily checked that $\CP^2$ minus a point is biholomorphic to $\CO(1)$ in an orientation-preserving manner. The diffeomorphism to $\CO(-1)$ -- the dual to $\CO(1)$ -- is instead orientation-reversing and not holomorphic.}

Historically, because of this fact, Burns space made an appearance in models of spacetime foam \cite{Hawking:1979hw,Hawking:1979pi}. These works studied Euclidean analogues of scattering amplitudes on gravitational instantons like $\CP^2$. In particular, they studied conformally coupled scalars on $\CP^2$. Due to the conformal equivalence between $\CP^2$ and Burns space, some of the results of their perturbative analysis will beautifully carry over to our context modulo appropriate conformal rescalings.

\medskip

In the rest of this section, we review some further interesting properties of Burns space which are potentially relevant to celestial holography but can be safely skipped on first reading.

\paragraph{Realization as an Einstein-Maxwell instanton.} We have taken the viewpoint that the Burns metric is a classical background of Mabuchi gravity. It has a non-zero Ricci form \eqref{ricciburns} so is clearly not Ricci-flat, therefore it does not represent a purely gravitational instanton in the usual sense. Since Burns space is asymptotically flat, this is consistent with the observation of \cite{Gibbons:1979xn} that any finite action, SD Ricci-flat spacetime that is asymptotically Euclidean must be $\R^4$. Nonetheless, the Burns metric (and more generally any scalar-flat K\"ahler metric) does act as an Einstein-Maxwell instanton \cite{Flaherty:1978wh,Dunajski:2006vs}, supported by the non-self-dual Maxwell flux
\be
F = \omega + P
\ee
where $\omega$ and $ P$ are its K\"ahler and Ricci forms given in \eqref{omegaburns} and \eqref{ricciburns} respectively.

To see this, let us compute the Riemann curvature of Burns space. Introduce a complex orthonormal coframe $\theta^{\al\dal}$ for the Burns metric,\footnote{Such complexified coframes are convenient to work with but do not obey the reality conditions \eqref{realtheta} assumed in appendix \ref{app:twistor}. If needed, they can always be rotated to ones that do by complexified local Lorentz transformations on the spinor indices.}
\be\label{btet}
\theta^{1\dal} = \d u^{\dal} - \frac{N\hat u^{\dal}}{\|u\|^4}\,[u\,\d u]\,,\qquad \theta^{2\dal} = \d\hat u^{\dal}\,.
\ee
It satisfies $\d s^2 = \eps_{\al\beta}\,\eps_{\dal\dot\beta}\,\theta^{\al\dal}\,\theta^{\beta\dot\beta}$. It is standard to use $\Sigma^{\al\beta}=\theta^\al{}_{\dal}\wedge\theta^{\beta\dal}$ and $\tilde\Sigma^{\dal\dot\beta} = \theta_\al{}^{\dal}\wedge\theta^{\al\dot\beta}$ as the bases of ASD and SD 2-forms on spacetime \cite{Capovilla:1991qb}, see appendix \ref{app:twistor} for more details. The ASD spin connection is determined from $\d\Sigma^{\al\beta} = 2\,\Gamma^{(\al}{}_{\gamma}\wedge\Sigma^{\beta)\gamma}$,
\be
\Gamma_{12} = \Gamma_{21} = \frac{N[u\,\d\hat u]}{2\,\|u\|^2(N+\|u\|^2)}\,,\qquad\Gamma_{11} = \Gamma_{22} = 0\,.
\ee
The associated ASD Riemann curvature 2-form is easily confirmed to be pure Ricci and self-dual
\be\label{Phiburns}
\begin{split}
    &R_{\al\beta} = \d\Gamma_{\al\beta} + \Gamma_\al{}^\gamma\wedge\Gamma_{\gamma\beta} = \Phi_{\al\beta\dal\dot\beta}\tilde\Sigma^{\dal\dot\beta}\,,\\ 
    &\Phi_{\al\beta\dal\dot\beta} = \frac{No_{(\al}\iota_{\beta)}u_{(\dal}\hat u_{\dot\beta)}}{\|u\|^2(N+\|u\|^2)^2}\,.
\end{split}
\ee
Here, $o_\al=(1,0)$ and $\iota_\al=(0,1)$ are a constant spinor basis. This gives another proof of the fact that the metric is self-dual and scalar-flat. 

The Einstein-Maxwell equations read
\begin{align}
    &R_{\mu\nu}-\frac{1}{2}g_{\mu\nu}R = T_{\mu\nu} = F_{\mu\rho}F_\nu{}^\rho - \frac{1}{4}g_{\mu\nu}F^2\,,\label{gremax}\\
    &\d F = \d\star F = 0\,,\label{Femax}
\end{align}
where $R_{\mu\nu}$ is the Ricci tensor, $F$ is the Maxwell field strength, and $\star$ is the Hodge dual. Decompose $F$ into ASD and SD parts, $F = F_{\al\beta}\Sigma^{\al\beta} + \tilde F_{\dal\dot\beta}\tilde\Sigma^{\dal\dot\beta}$. In an orthonormal frame like ours, the spinor equivalent of the Ricci tensor decomposes as $R_{\al\dal\beta\dot\beta} = 2\Phi_{\al\beta\dal\dot\beta} - \frac{1}{4}\eps_{\al\beta}\eps_{\dal\dot\beta}R$ \cite{Wald:1984rg}, whereas that of the Maxwell stress tensor reduces to $T_{\al\dal\beta\dot\beta} = -2F_{\al\beta}\tilde F_{\dal\dot\beta}$. The Einstein equation \eqref{gremax} then decomposes into two irreducible equations
\be\label{emax}
\Phi_{\al\beta\dal\dot\beta} + F_{\al\beta}\tilde F_{\dal\dot\beta} = 0\,,\qquad R=0\,.
\ee
Now, in our orientation with volume form $\sqrt{|g|}\,\d u^{\dot1}\wedge\d\bar u^1\wedge\d u^{\dot2}\wedge\d\bar u^2$, the K\"ahler form $\omega$ is ASD, a standard fact in K\"ahler geometry. On the other hand, the Ricci form $ P$ is SD, which can be checked directly but is mainly a consequence of scalar-flatness. Setting $F=\omega+ P$, we find
\be\label{maxwell}
\begin{split}
    F_{\al\beta}\Sigma^{\al\beta} &= \omega\implies F_{\al\beta} = \im\,o_{(\al}\iota_{\beta)}\,,\\
    \tilde F_{\dal\dot\beta}\tilde\Sigma^{\dal\dot\beta} &=  P\implies\tilde F_{\dal\dot\beta} = \frac{\im Nu_{(\dal}\hat u_{\dot\beta)}}{\|u\|^2(N+\|u\|^2)^2} \,.
\end{split}
\ee
The equations \eqref{emax} are then immediately satisfied by \eqref{Phiburns} and \eqref{maxwell} by construction. At the same time, $F=\omega+ P$ solves the Maxwell equations \eqref{Femax} because $\omega$ and $P$ are co-closed by virtue of being closed and either ASD or SD.

This shows that Burns space can be supported within Einstein gravity via a Maxwell flux. The Maxwell background \eqref{maxwell} is not self-dual, as $F_{\al\beta}\neq0$. So this is necessarily a solution of the non-self-dual theory. Nonetheless, this fact has been used for generating new solutions of supergravity in the works \cite{Bobev:2011kk,Bobev:2012af}, which may be of interest for building more general setups for flat holography.

\paragraph{Hyperbolic foliation.} The Burns metric also admits other geometrical incarnations. One that may be of potential interest for holography is as a metric on the total space of a circle bundle over an open subset of Euclidean AdS$_3$ \cite{Lebrun91explicitself-dual}. 

This can be understood by a judicious change of coordinates. For simplicity, let us revert to working on $\C^2-0$. The Burns metric has a pair of circle isometries: rotations in the $u^{\dot1}$ complex plane or the $u^{\dot2}$ complex plane. In the symplectic structure associated to the K\"ahler form $\omega$, both of these are generated by hamiltonian vector fields. Introduce new coordinates $t\in(0,2\pi)$, $q\in(0,\infty)$, $(x,y)\in\R^2$ defined by the relations
\be
t = \arg(u^{\dot1})\,,\qquad q^2 = 2\,|u^{\dot1}|^2\left(1+\frac{N}{\|u\|^2}\right),\qquad u^{\dot2} = \frac{x+\im y}{\sqrt2}\,.
\ee
$t$ provides a coordinate along the orbits of the circle action rotating the $u^{\dot1}$ complex plane, and $q^2$ is precisely the Hamiltonian generating this circle action. These coordinates are well-defined away from the fixed point set of this circle action. The latter is a copy of $\R^2$ given by $\{u^{\dot1}=0\}$. The coordinate patch on $\C^2-0$ obtained from removing this $\R^2$ (which in particular includes the origin) is $\R^4-\R^2 \simeq \text{AdS}_3\times S^1$. The phase $t$ acts as a coordinate along $S^1$, while $(x,y,q)$ provide Poincar\'e coordinates on AdS$_3$. Actually, $q$ only covers the full range $(0,\infty)$ when $(x,y)\neq0$. When $(x,y)=0$, one sees from its definition that $q^2 = 2(|u^{\dot1}|^2+N)>2N$. So this coordinate patch covers $S^1$ times the open subset $\text{AdS}_3-\{x=y=0,q\leq N\}$. 

In these coordinates, the Burns metric \eqref{burns} transforms into
\be
\d s^2 = q^2\left(\frac1V\left(\d t+A\right)^2 + V\,\frac{\d x^2+\d y^2+\d q^2}{q^2}\right)\,.
\ee
The overall conformal factor of $q^2$ ensures asymptotic flatness. The data $(V,A)$ is independent of $t$ and describes a magnetic monopole on AdS$_3$ centered at $(x,y,q)=(0,0,\sqrt{2N})$. The scalar potential $V$ of this monopole is found to be
\be
V = 1+\frac12\,\bigg(\frac{q^2+x^2+y^2+2N}{\sqrt{(q^2+x^2+y^2+2N)^2-8Nq^2}}-1\bigg)\,.
\ee
It solves the AdS$_3$ Laplace equation and approaches $V\to1$ in the flat limit $N\to0$. Recalling the expression for the hyperbolic distance between the points $(x,y,q)$ and $(0,0,\sqrt{2N})$,
\be
\beta = \cosh^{-1}\!\left(\frac{x^2+y^2+q^2+2N}{2q\sqrt{2N}}\right)\,,
\ee
we can express $V$ in terms of the fundamental solution $1/(\e^{2\beta}-1)$ of the AdS$_3$ Laplace equation,
\be
V = 1 + \frac{1}{\e^{2\beta}-1}\,.
\ee
The associated 3-dimensional vector potential reads
\be
A = \frac{(x\,\d y-y\,\d x)}{2(x^2+y^2)}\,\bigg(\frac{q^2+x^2+y^2-2N}{\sqrt{(q^2+x^2+y^2+2N)^2-8Nq^2}}-1\bigg)\,.
\ee
This satisfies $\d A = \star_{\text{AdS}_3}\d V$ and vanishes as $N\to0$. 

In this sense, the Burns metric and its multicentered cousins found by LeBrun in \cite{Lebrun91explicitself-dual} act as hyperbolic analogues of the Gibbons-Hawking metrics \cite{Hawking:1976jb,Gibbons:1978tef}. This form of the metric may be helpful in relating holography on Burns space to more standard AdS$_3$/CFT$_2$ dualities, potentially realizing the hopes of \cite{Ball:2019atb}. It could also help in finding multicentered generalizations of asymptotically flat holography, perhaps in analogy with recent generalizations of twisted holography that include multiple stacks of D1 branes \cite{Budzik:2022hcd}.


\section{The holographic dual chiral algebra}
\label{sec:dual}

In this section, we describe the celestial holographic dual to the B-model on the twistor space of Burns space. By compactifying the latter along twistor lines, one obtains a duality relating WZW$_4$ + Mabuchi gravity on Burns space to this celestial CFT.

The dual CFT is a purely chiral system, living on $\CP^1$ with two defects placed at $z = 0$ and $z = \infty$.  We will first describe the system without the defects.

The chiral algebra is the algebra living on $N$ D1 branes in the type I topological B-string.  This algebra is an orientifold of the chiral algebra living on D1 branes in the type II topological B-string, which is the algebra studied in \cite{CG}.   The orientifold procedure was computed in \cite{Costello:2019jsy}. 

An alternative description of the chiral algebra is that it is the algebra associated by the superconformal localization procedure of \cite{Beem:2013sza, Oh:2019bgz} to the family of 4d $\mathcal{N}=2$ supersymmetric gauge theories with $\Sp(N)$ gauge symmetry, $\SO(8)$ flavour symmetry  and matter as described below. These are non-unitary chiral algebras with negative Kac-Moody levels and therefore, by the usual Sugawara relation, negative central charges. 

The algebra is built by performing BRST reduction on a collection of free symplectic bosons. The gauge group for the BRST reduction is $ \Sp( N )$. 

Let $V$ denote the fundamental representation of dimension $2N$, and let $i$ denote an index for a basis of $V$.  The symplectic pairing is $\omega_{ij}$.\footnote{This is not to be confused with the K{\"a}hler potential.} The fundamental fields of the dual are symplectic bosons carrying conformal weights $(\frac12,0)$ each:
\begin{equation}
\begin{split}
I_{r i} &\in \Omega^{\frac12, 0}(\CP^1, \C^8 \otimes V) \\
X^{\dot{\alpha}}_{ij}  & \in \Omega^{\frac12,0} (\CP^1, \wedge^2_0 V \otimes \C^2). 
\end{split}
\end{equation}
The fields $X^{\dot{\alpha}}$ transform under the 4d Lorentz group as right-handed spinors. They describe D1-D1 strings. The notation $\wedge^2_0 V$ means we are using the trace-free exterior square, so that
\begin{equation} \label{tracefree}
	\omega^{ij} X^{\dot{\alpha}}_{ij}  = 0\,.  
\end{equation}
$\omega^{ij}$ here is the inverse of $\omega_{ij}$, with the convention $\omega^{ij}\omega_{jk}=\delta^i_k$, and can be used to raise indices. This is the matter content summarized in our last work \cite{Costello:2022jpg}.

Occasionally, it will prove convenient for our holographic calculations to keep the trace part of  $X^{\dal}_{ij}$ explicitly in the chiral algebra, rather than only taking the trace-free exterior square. This is analogous to the choice to work with $\U(N)$ rather than $\SU(N)$ gauge group in 4d $\mathcal{N}=4$ super Yang-Mills. It corresponds to the center-of-mass degrees of freedom in the bulk, or dually the $E[1,0], E[0, 1]$ (see the definition \eqref{CFT}) generators in the chiral algebra, which can be consistently removed. We will find it convenient to keep the center of mass modes nonvanishing when building our holographic dictionary in section \ref{sec:ops}, and use them in section \ref{sec:gravope} to test the gravitational sector of our duality.

The matter fields $I_{ri}$ describe D1-D5 strings. The index $r$ in $I_{ri}$ is an $\SO(8)$ index for a basis of $\C^8$.  The Lagrangian for the free theory is
\begin{equation} 
\int \omega^{ij} \delta^{rs}  I_{r i} \dbar I_{s j }  + \int \eps_{\dot{\alpha} \dot{\beta}} \omega^{ik} \omega^{jl}     X^{\dot{\alpha}}_{ij}\dbar X^{\dot{\beta}}_{kl}. 
\end{equation}  
This leads to the OPE (where we have absorbed some factors of $2 \pi \i$ into rescaling the fields)
\begin{align}
		I_{ri} (z)\, I_{sj} (z') &\sim \frac{\delta_{rs} \omega_{ij} }{z - z' } \,, \label{IIope}	\\
		X^{\dot{\alpha}}_{ i j }(z)\,  X^{\dot{\beta}}_{ k l }(z') &\sim \frac{ \eps^{\dot{\alpha} \dot{\beta} }  } { z - z'}\left( \omega_{i[k|} \omega_{j|l]} - \frac{\omega_{ij} \omega_{kl}}{2N}   \right)\,. \label{XXope}
\end{align}
The tensor structure in the $XX$ OPE ensures that both sides are trace-free. 

When we want to keep the center of mass modes in play, we simply modify the $XX$ OPE by dropping the trace-free condition:
\be\label{XXcom}
X^{\dot{\alpha}}_{ i j }(z)\,  X^{\dot{\beta}}_{ k l }(z') \sim \frac{ \eps^{\dot{\alpha} \dot{\beta} }\omega_{i[k|} \omega_{j|l]} } { z - z'}\,. 
\ee
For the rest of this section, we continue to work with a trace-free $X$, but our considerations can be easily modified to include the center of mass modes.\footnote{For example, including the center of mass modes does not affect cancellation of the $\Sp(N)$ gauge anomaly because the trace $\omega^{ij}X^{\dal}_{ij}$ is $\Sp(N)$ invariant and does not couple to the $\Sp(N)$ gauge field.}

We then perform the BRST reduction with respect to $\Sp(N)$.  This means we introduce $\b-\c$ ghosts transforming in the adjoint representation of $\Sp(N)$.  The operator $\b$ is of ghost number $-1$ and spin $1$, $\c$ is of ghost number $1$ and spin $0$, and they have the OPE
\begin{equation} 
	\b^{\sa}(z) \c_{\sb}(z') \sim \frac{\delta^\sa_\sb }{z - z'}\,,
\end{equation}
where $\sa,\sb,\dots$ are indices in the adjoint of $\mathfrak{sp}(N)$. The BRST current is 
\begin{multline}
	J_\text{BRST} = \frac12\, f^{\sb\sc}{}_\sa :\! \b^\sa\c_\sb \c_\sc\!:  + \frac{1}{2}\,\delta^{rs} T^{\sa ij} :\!\c_\sa I_{ri} I_{sj}\!:\\ + \frac{1}{2}:\!\c_\sa X^{\dot{\alpha}}_{ij} X^{\dot{\beta}}_{kl}\!: \eps_{\dot{\alpha} \dot{\beta} } (T^{\sa ik} \omega^{jl} - T^{\sa jk } \omega^{il} - T^{\sa il} \omega^{jk} + T^{\sa jl} \omega^{ik} ) 
\end{multline}
where $T^{\sa}$ are generators of $\mf{sp}(N)$, and $T^{\sa ij}$ is the matrix element for their action in the fundamental representation of $\mf{sp}(N)$. 

The BRST procedure produces the ADHM constraint equation
\begin{equation}\label{adhm}
    [X_{\dot 1}, X_{\dot 2}] + \delta^{rs}I_rI_s = 0
\end{equation} which enables us to commute $X$ operators past each other at the cost of introducing a double-trace term. Since the latter contributions are suppressed in the planar limit, we will be able to freely commute $X$'s in our computations that follow. 

Let us check that this system is free of the BRST anomaly.  If $\mathcal{X}$ denotes an element of $\mf{sp}(N)$, the anomaly which prevents the BRST operator squaring to zero is given by
\begin{equation} 
2 \op{Tr}_{\mf{sp}(N)} (\cX^2)  - \op{Tr}_{\C^8 \otimes V \oplus \C^2 \otimes \wedge^2_0 V} (\cX^2).  
\end{equation} 
We can check whether or not this expression vanishes by evaluating it for $\cX$ in a rank $1$ subalgebra, which we take to be 
\begin{equation} 
 \mf{sl}(2) =\mf{sp}(1)    \subset \mf{sp}(N). 
\end{equation}
 Therefore, let us decompose the fundamental of $\mf{sp}(N)$ as a representation of $\mf{sl}(2)$: 
\begin{equation} 
V = W \oplus \C^{2N - 2} 
\end{equation}  
where $W$ is the fundamental of $\mf{sl}(2)$. 

Then,
\begin{equation}
\begin{split}
\wedge^2_0 V &=  (2N-2) W \\ 
 \mf{sp}(N) &= S^2 V = S^2 W \oplus (2 N-2) W 
\end{split} 
\end{equation}
We are taking the trace in the virtual representation 
\begin{equation} 
2 \wedge^2_0 V \oplus 8 V  \ominus 2 S^2 V  
\end{equation}
We see that the factors of $(2 N-2) S_{\pm}$ cancel, and we are left with
\begin{equation} 
8 W \ominus 2 S^2 W
\end{equation}
If we take $\cX = h$  to be a basis for the Cartan of $\mf{sl}(2)$, then the trace of $\cX^2$ is
\begin{equation} 
16 - 16 = 0
\end{equation} 
as desired.


\subsection{Large $N$ BRST cohomology}

States on Burns space will be dual to BRST invariant single-trace operators of the chiral algebra in the large $N$ limit. The BRST cohomology can be computed at large $N$ using tools from homological algebra \cite{loday1984cyclic, tsygan1983homology}.  Details are provided in appendix \ref{app:BRST_computation}. Massive states like giant gravitons \cite{Budzik:2021fyh} are also expected to play an interesting role at finite $N$, but they lie beyond the scope of this work.

In what follows, for each $\dal=\dot1,\dot2,$ we will view $(X^{\dal})^i{}_j \equiv \omega^{ik}X^{\dal}_{kj}$ as endomorphisms of the fundamental representation $V$ of $\Sp(N)$. Similarly, $I_r{}^i\equiv\omega^{ij}I_{rj}$ and $I_{ri}$ will be viewed as a vector or covector in the fundamental representation, with $r$ an $\SO(8)$ index.  A basis for the BRST cohomology is 
\begin{equation}\label{CFT} 
	\begin{split} 
		J_{rs} [m,n] &= \omega\bigl( I_{r} , (X^{\dot1})^{(m}(X^{\dot2})^{n)} I_s\bigr)\,,\\
		E[m,n] &= \op{Tr} \bigl( (X^{\dot1})^{(m}(X^{\dot2})^{n)}\bigr) \,,\\ 
		F[m,n] &= \op{Tr} \left( [X\,\p X] (X^{\dot1})^{(m}(X^{\dot2})^{n)} \right) + \text{terms with ghosts}\,.
	\end{split}
\end{equation}
For example, $J_{rs}[1,0]=I_{ri}(X^{\dot1})^i{}_jI_s{}^j$, $E[2,0]=\omega^{ij}(X^{\dot1})_{ik}(X^{\dot1})^{k}{}_j$, etc.   

In these expressions, $(X^{\dot1})^{(m}(X^{\dot2})^{n)}$ is shorthand notation for the symmetrized product of these $n+m$ matrices. Also, $[X\,\p X]=X^{\dal}\p X_{\dal}$ uses the notation for spinor contractions introduced earlier. The currents $F[m,n]$ receive corrections containing ghosts that ensure BRST closure, but we have suppressed these as we will not be needing them in the sequel. (In particular, $F[0,0]$ is the stress tensor operator).

The currents $J_{rs}[m,n]$ are anti-symmetric in the $r,s$ indices, because $\omega$ is anti-symmetric and because
\begin{equation} 
	\omega (I_r, X^{\dot{\alpha}} I_s ) = \omega ( X^{\dot{\alpha}} I_r, I_s ). 
\end{equation}
Therefore we can view $r,s$ as being an index in the adjoint for $\mf{so}(8)$. We will often suppress the $\Sp(N)$ inner product and write these currents simply as $J_{rs}[m,n] = I_r (X^{\dot1})^{(m}(X^{\dot2})^{n)}I_s$, interpreting $I_r$ as a covector and $I_s$ as a vector of $\Sp(N)$. To lighten notation, the composite index $rs$ will also often be replaced with an adjoint index $a$ for $\mf{so}(8)$.


\subsection{Defect boundary conditions from Koszul duality}
\label{sec:chiralbdy}

To complete the duality, we would like to prescribe boundary conditions on our chiral algebra fields near the defects at $z=0,\infty$. We will obtain boundary conditions on the composite currents $J[k,l], E[k,l], F[k,l]$ with relative ease using the Koszul duality \footnote{One can also use Koszul duality to study maps between modules of Koszul dual algebras, which has been explored recently from a twistorial perspective in \cite{Garner:2023izn}. It would be interesting to study Koszul dual modules further, but we will not do so in this work.} approach espoused in \cite{Costello:2020jbh,PW,Costello:2022wso}. Boundary conditions on the fundamental fields $X,I$ that give rise to these boundary conditions on the composite currents are described in the next section.

Koszul duality helps determine the most general couplings of the D1-brane chiral algebra to topological strings in the bulk. The chiral algebra couples in a standard way to the fields of the type I B-model on $\C^3$, with volume form 
\begin{equation} 
	\d z\, \d^2 v \,.
\end{equation}
We are interested in a chiral algebra which couples to the fields of the type I B-model on $\PT$, the twistor space of $\R^4$, with the meromorphic volume form
\begin{equation} 
	\Omega = z^{-2}\, \d z \,\d^2 v\,. 
\end{equation}
Building such couplings will lead us to impose appropriate boundary conditions on the chiral algebra operators.

Let us first recall how the chiral algebra couples to $\SO(8)$ holomorphic Chern-Simons theory on $\C^3$ with the volume form $\d z\, \d^2 v$. If 
\begin{equation} 
	\mc{A} \in \Omega^{0,1}(\C^3, \mf{so}(8) ) 
\end{equation}
is the gauge field of holomorphic Chern-Simons theory, the coupling (at tree level) is by
\begin{equation} \label{Acoup}
\sum_{m,n \ge 0} \frac{1}{m! n!} \int_{v^{\dal} = 0}\d z  \; J[m,n]\;\partial_{v^{\dot1}}^m \partial_{v^{\dot2}}^n \mc{A} \,.
\end{equation}
Including $z = \infty$, flat space $\C^3$ naturally completes to the resolved conifold 
\begin{equation} 
	\Oo(-1)\oplus\Oo(-1) \to\CP^1\,. 
\end{equation}
Since $\Oo(-1)$ is the square root of the canonical bundle on $\CP^1$, the operator $\partial_{v^{\dal}}$ is of spin $-\frac{1}{2}$. Since $J[m,n]$ is of spin $1 + \frac{1}{2} (m+n)$, the whole expression is of spin $1$ and is therefore something that can be integrated against the $(0,1)$ form $\mc{A}$.  

Away from the locus $z = 0, \infty$, this formula can be adapted to twistor space. This patch of twistor space is isomorphic to a patch of $\C^3$ when we use the coordinates $z, z^{-1} v^{\dal}$. The point is that the volume form on twistor space in these coordinates is 
\begin{equation} 
	\Omega=\d z\wedge \d (z^{-1} v^{\dot1} )\wedge \d (z^{-1} v^{\dot2})\, . 
\end{equation}
Therefore, the theory on the D1 brane couples to holomorphic Chern-Simons theory on this patch by the formula
\begin{equation} \label{Jcoup}
	\sum_{m,n \ge 0} \frac{1}{m! n!} \int_{v^{\dal} = 0}\d z \;z^{m+n} J[m,n]\;\partial_{v^{\dot1}}^m \partial_{v^{\dot2}}^n \mc{A} 
\end{equation}
which is obtained by replacing $v^{\dal}\mapsto z^{-1}v^{\dal}$ in \eqref{Acoup}. This makes sense, because on twistor space $\partial_{v^{\dal}}$ is of spin $\frac{1}{2}$, but $z \partial_{v^{\dal}}$ is of spin $-\frac{1}{2}$.\footnote{A description in homogeneous coordinates may be helpful for practitioners of twistor theory. If one introduces homogeneous coordinates $[\mu^{\dal},\lambda_\al]$ on $\PT=\CP^3-\CP^1$ by setting $z=\lambda_2/\lambda_1, v^{\dal}=\mu^{\dal}/\lambda_1$, then the fields $X,I$ live on the twistor line $\mu^{\dal}=0$ and are valued in $\CO(-1)$ each. The composite current $J[m,n]$ is $\CO(-m-n-2)$-valued. So it couples to the bulk through the coupling \begin{equation*}
    \frac{1}{m!n!}\int_{\mu^{\dal}=0}\la\lambda\,\d\lambda\ra\;(\lambda_1\lambda_2)^{m+n}J[m,n]\;\p_{\mu^{\dot1}}^m\p_{\mu^{\dot2}}^n\cA\,.
\end{equation*}
This reduces to the coupling in \eqref{Jcoup} if one sets $\lambda_\al=(1,z)$, $\mu^{\dal}=v^{\dal}$. Similar expressions can be written for couplings to $E[m,n]$, $F[m,n]$.} 

A similar, but more complicated, analysis applies to the closed string fields. The closed string field in the topological string is an element $\eta \in \Omega^{2,1}$ which is $\partial$-closed. We can write 
\begin{equation} 
	\eta = \partial \gamma
\end{equation}
for a $(1,1)$-form $\gamma$, defined up to $\gamma\sim\gamma+\p\al$ for some $\al\in\Omega^{0,1}$. Then the coupling between the $E,F$ towers and $\eta$ is given by
\begin{equation} 
	\begin{split} 
		& \sum_{m,n\geq0}\frac{1}{m! n!}\int E[m,n]\;\mc{L}_{\partial_{v^{\dot1}}} ^m \mc{L}_{\partial_{v^{\dot2}}}^n \gamma \\
		&+  \sum_{m,n\geq0}\frac{1}{m! n!}\int \d z\;  F[m,n]\;\mc{L}_{\partial_{v^{\dot1}}} ^m \mc{L}_{\partial_{v^{\dot2}}}^n \iota_{\partial_{v^{\dal}}} \mc{L}_{\partial_{v_{\dal}}}\gamma. 
	\end{split}
\end{equation}
Here, $\cL_\xi\gamma$ denotes Lie derivative with respect to a vector $\xi$, and $\iota_\xi\gamma\equiv\xi\ip\gamma$ is standard abbreviation for the interior product. It is easy to check that this coupling transforms by a total derivative under $\gamma\mapsto\gamma+\p\alpha$.

This is the expression for the coupling on $\C^3$. When we transform to twistor space, by the analysis above, we find that the coupling becomes
\begin{equation} \label{EFcoup}
	\begin{split} 
		&  \sum_{m,n\geq0}\frac{1}{m! n!}\int z^{m+n}E[m,n]\;\mc{L}_{\partial_{v^{\dot1}}} ^m \mc{L}_{\partial_{v^{\dot2}}}^n \gamma \\
		&+  \sum_{m,n\geq0}\frac{1}{m! n!}\int \d z\;  z^{m+n+2}F[m,n]\;\mc{L}_{\partial_{v^{\dot1}}} ^m \mc{L}_{\partial_{v^{\dot2}}}^n \iota_{\partial_{v^{\dal}}} \mc{L}_{\partial_{v_{\dal}}}\gamma. 
	\end{split}
\end{equation}
Recalling the expression \eqref{CFT} for $E,F$ in terms of $X$, we see that each appearance of $X^{\dal}$ introduces a $z \partial_{v^{\dal}}$ into the coupling.

This analysis tells us that the chiral algebra we have constructed is not quite the Koszul dual to the algebra of bulk operators. Recall \cite{Costello:2020jbh, PW} that the Koszul dual chiral algebra is the universal algebra that couples to the bulk system.  Since the coupling between our chiral algebra and the bulk system has zeroes at $z = 0,\infty$, it is not the most general system we can couple in a gauge invariant way.  The most general system will allow the currents $J,E,F$ to have the following poles at $z = 0,\infty$:
\begin{equation} 
	\begin{split} 
		J[m,n] & \ \ \text{ pole of order } n+m+1 \text { at } z = 0,\infty \\
E[m,n] & \ \ \text{ pole of order } n+m \text { at } z = 0,\infty \\
F[m,n] & \ \ \text{ pole of order } n+m+2 \text { at } z = 0,\infty  
	\end{split}
\end{equation}
Since $\cA$ has a first order zero and $\gamma$ is regular at $z=0,\infty$, the couplings \eqref{Jcoup} and \eqref{EFcoup} remain regular for these configurations of poles.

The modes $J_l[m,n]$, $E_l[m,n]$, $F_l[m,n]$ (where $l$ refers to the spin of the mode) are defined by the expansions\footnote{We are using conventions where modes with negative mode numbers $l$ (for large enough $|l|$) are the \emph{annihilation} operators, which is the opposite of the usual convention in physics.}
\begin{equation}
    \begin{split}
        J[m,n](z) &= \sum_l\frac{J_l[m,n]}{z^{1+\frac{m+n}{2}-l}}\\
        E[m,n](z) &= \sum_l\frac{J_l[m,n]}{z^{\frac{m+n}{2}-l}}\\
        F[m,n](z) &= \sum_l\frac{J_l[m,n]}{z^{2+\frac{m+n}{2}-l}}\,,
    \end{split}
\end{equation}
where $l$ runs over $\Z + \frac{m+n}{2}$. In terms of these, the vacuum vector at $z = 0$ satisfies
\begin{equation} 
	\begin{split} 
		J_l[m,n]  \vac{v_0} & = 0 \text { for } 2l +  m + n < 0 \\
E_l[m,n]  \vac{v_0} & = 0 \text { for } 2l +  m + n < 0 \\
F_l[m,n]  \vac{v_0} & = 0 \text { for } 2l  + m + n < 0\,.
 \label{eqn:vaccuum_ideal}
	\end{split}
\end{equation}
These relations define a module for the mode algebra of the chiral algebra.  More formally, the relations \eqref{eqn:vaccuum_ideal} define a left ideal in the mode algebra of the chiral algebra, and quotienting by this left ideal gives the desired module.

What is not obvious, however, is that this module is of the desired size.   Let us explain what we mean by this.  A (topological) basis for the mode algebra of the chiral algebra is given by ordered products of the generators $J_l[m,n]$, $E_l[m,n]$, $F_l[m,n]$  in some ordering prescription.  We would like a basis of the module to be given, in a similar way, by products of the generators 
\begin{equation} 
	\begin{split} 
		J_l[m,n] & \text { for } 2l +  m + n \ge  0 \\
		E_l[m,n]   &  \text { for } 2l +  m + n \ge 0 \\
		F_l[m,n]   &  \text { for } 2l  + m + n \ge 0\,.
		\label{eqn:vaccuum_module}
	\end{split} 
\end{equation}
These products should be taken in some chosen ordering for the currents listed above. 

It could, however, happen that the module is too small, or even zero.  This can occur if the commutator of two generators satisfying the equality \eqref{eqn:vaccuum_ideal} is a product of expressions which satisfy the opposite inequalities \eqref{eqn:vaccuum_module}, forcing certain composite states to be zero.   

In the Appendix \ref{app:boundary_size}, we will verify that the module is indeed of the correct size.   


\subsection{Microscopic definition of the boundary conditions}\label{sec:bdydef}

So far, we have defined the boundary conditions for the chiral algebra by prescribing the poles that elements of the BRST cohomology acquire.   In this section we will see how to define the boundary condition before we take BRST cohomology, by considering the behaviour of the fundamental fields $I,X$ and the $\b-\c$ ghosts.

Recall that we require $J[m,n]$ to have a pole at $z = 0,\infty$ of order $m+n+1$, whereas $E[m,n]$ has a pole of order $m+n$ and $F[m,n]$ has a pole of order $m+n+2$.  The poles in $E,F$ are easy to obtain: we simply allow the $X$ fields to have a first order pole at $z = 0,\infty$. 

The pole for $J[m,n]$ is more challenging to obtain, because we would need to give $I$ a pole of order $\frac{1}{2}$.  Fortunately, as Davide Gaiotto explained to us, it is possible to give $I$ a pole of order $\frac{1}{2}$, by giving a kind of Ramond puncture\footnote{This is an example of the notion of a twisted module known from vertex algebra literature \cite{frenkel2004twisted}.} at $z = 0$, $z = \infty$.

Normally, a Ramond puncture is defined by choosing a square root of the canonical bundle on a punctured surface that does not extend across the puncture.  Because we want a Ramond-like puncture only for the $I$ fields but not for the $X$ fields, we will do something slightly different.

Our Ramond-like puncture is defined by considering the system on $\C^\times$, and placing a defect on a line stretching between $0$ and $\infty$. Along this defect, the field $I$ has a branch cut and changes sign. The other fields are unaffected. 

Once we have this branch cut, it makes sense to give $I$ a pole of order $\frac{1}{2}$.  We also give $X$ a pole of order $1$, and  the $\b$-ghost a pole of order $2$. 

Let us describe this more formally in terms of modes.  The field $I$ has modes $I_l$ of integer spin (suppressing $\SO(8)$ and $\Sp(N)$ indices), as do $\b,\c$. The field $X$ has modes of half-integer spin. The vector $\vac{v_0}$ is annihilated by the modes:
\begin{equation} 
	\begin{split} 
		I_l \vac{v_0} &= 0 \text { for } l < 0 \\
		X^{\dal}_l \vac{v_0} &= 0 \text { for } l < -\frac{1}{2} \\
		\b_l \vac{v_0} &= 0 \text { for } l < -1 \\
		\c_l \vac{v_0} &= 0 \text { for } l < 0 
	\end{split}
\end{equation}
For this to make sense, we need to know that the ideal of elements annihilating $\vac{v_0}$ is closed under the BRST operator.  The BRST operator is schematically of the form
\begin{equation} 
	\begin{split} 
		Q I_l &= \sum_{r+s = l} \c_r I_s \,,\qquad Q X^{\dal}_l = \sum_{r+s = l} \c_r X^{\dal}_s \\
		Q \c_l &=  \sum_{r+s = l} \c_r c_s  \,,\qquad
		Q \b_l = \sum_{r+s = l} \left( :\!I_r I_s\!: + \eps_{\dal\dot\beta} :\!X^{\dal}_r X^{\dot\beta}_s\!:    \right). 
	\end{split}
\end{equation}
It is easy to check that if the left hand side of this expression annihilates $\vac{v_0}$, then so does the right hand side.   
With these boundary conditions, we have the desired poles in the currents $J[m,n]$, $E[m,n]$, $F[m,n]$. For instance, the lowest spin component of $J[m,n]$ that survives on the boundary is of the form
\begin{equation} 
	I_{0} ( X^{\dot1}_{-\frac{1}{2}})^r ( X^{\dot2}_{-\frac{1}{2}})^s I_{-1} 
\end{equation}
which is of spin $ - \frac{1}{2}(r+s)$.  This is what we would find from a current of spin $1+ \frac{1}{2}(r+s)$ with a pole of order $r+s+1$.   

We verify in the appendix \ref{app:BRST_boundary} that, at large $N$, the BRST cohomology of this microscopic description of the module reproduces the description given earlier.


\subsection{Conformal blocks of the chiral algebra with defects} 
\label{sec:blocks}

Given a chiral algebra, perhaps with defects, we say a \emph{conformal block} is a way of defining correlation functions of local operators in a way consistent with all OPEs, and with all boundary behaviour determined by the defects.  

In this section we will find that the chiral algebra has an infinite-dimensional space of conformal blocks, which we will identify with the Hilbert space of the 4d system on a small $S^3$.   We only compute the conformal blocks at infinite $N$.  In 4d, we compute the Hilbert space via an ansatz that matches it with the space of local operators in the free theory.    It is important to note that there is no Hamiltonian in our analysis: it is purely kinematic.  

Let us first explain why we would expect it to be true.  In \cite{Costello:2022wso}, it was shown that for a theory on flat space coming from a holomorphic theory on twistor space, conformal blocks for the celestial chiral algebra are isomorphic to local operators.  This was derived both by a formal QFT argument, and an explicit calculation (see also \cite{Bu:2022dis}). 

By the state-operator correspondence, the space of local operators in the 4d theory is the Hilbert space on $S^3$.   We therefore need to show how to generalize the argument of \cite{Costello:2022wso} to a 4d geometry which is not flat.  As in \cite{Costello:2022wso}, we will give both a formal argument and an explicit calculation.

Since our chiral algebra is generated by the towers $J,E,F$, correlation functions are determined by the correlation functions with only $J,E,F$ insertions. We will determine the conformal blocks by analyzing the ambiguity in the correlation functions of these insertions.

As a warm-up, let us first consider correlators of the operator $J = J[0,0]$.  This is a usual Kac-Moody current for $\mf{so}(8)$. If we do not change the boundary conditions, then the correlators
\begin{equation} 
	\ep{J (z_1) \dots J(z_n)}  
\end{equation}
are meromorphic functions of $z_1,\dots,z_n$. They have poles only at $z_i = z_j$, and the polar part is determined by the OPE.  Because $J$ is of spin $1$, the correlator has a second-order zero  at $z_i = \infty$.  This $n$-point function is completely determined by the $(n-1)$-point function.  The only ambiguity in defining the correlation function is present in the definition of the zero-point function $\ep{1}$, which can be an arbitrary constant. 

This means that in this case the vector space of conformal blocks is one dimensional.\footnote{The reader may notice that the correlation functions are over-determined.  This is related to the fact that the Kac-Moody algebra has \emph{derived} conformal blocks, in the sense introduced by Beilinson-Drinfeld \cite{beilinson2004chiral}. }  

Let us see what changes when we allow $J(z)$ to have a first order pole at zero and infinity. The pole at infinity is to be taken in addition to the second order zero prescribed by the fact that $J$ is of spin $1$, so that in total we have a first order zero at $\infty$.

Suppose, in this case, we have defined the $k$-point function
\begin{equation} 
	\ep{J(z_1) \dots J(z_k)}  
\end{equation}
for $k < n$, in a way consistent with the OPE. Let us consider the ambiguity in defining the $n$-point function. 

The poles in the correlator $\ep{J (z_1) \dots J(z_n)}$ at $z_i = z_j$ are determined by the OPEs and the $k$-point correlators for $k < n$.  This means that the $n$-point function is determined up to the addition of a meromorphic function $\cF(z_1,\dots, z_n)$ which is regular at $z_i = z_j$.  The function $\cF$ can have a first order pole at $z_i = 0$ and must have a first order zero at $z_i = \infty$.
This means that $\cF$ is proportional to
\begin{equation} 
	\cF(z_1,\dots,z_n) \propto \frac{1}{z_1 \dots z_n}. 
\end{equation}

In this case, we see that for each $n$, we are free to add an arbitrary multiple of $z_1^{-1} \cdots z_n^{-1}$ to the correlation function. This expression can depend also on the Lie algebra indices of the current, but in a symmetric way. We find that the full space of conformal blocks is the symmetric algebra
\begin{equation} 
	\Sym^\ast \g = \oplus_{n \ge 0} \Sym^n \g 
\end{equation}

Now let us turn to the conformal blocks of our full chiral algebra generated by $J[r,s]$, $E[r,s]$, $F[r,s]$.  It is natural to expect, following the analysis above, that the ambiguity in defining the $n$-point correlation function -- which one can call the $n$-point conformal block -- is the $n$th symmetric power of the $1$-point conformal block.  If so, the computation of the conformal blocks amounts to computing the $1$-point conformal block, which is quite easy.   Now let us check this property:
\begin{proposition}
	The full space of conformal blocks is isomorphic to the symmetric algebra on the conformal blocks for the one-point function. \label{prop:conformalblocks}
\end{proposition}
\begin{proof}
The proof is similar to the proof in the case of the Kac-Moody algebra, except that we cannot simply work by induction on the number of insertions. This is because in the larger algebra, the OPE of two operators $E,F,J$ will in general be expressed in terms of normally ordered products of a number of such operators.  The $(n-1)$-point function where one of the insertions is a normally ordered product can be expressed in terms of a $k$-point function without normally ordered products, where $k \ge n$.  Therefore the inductive argument will not work.  

	Instead, we need to introduce a measure of complexity of the operators which decreases when we take the normally ordered product. We say the \emph{width} of a single-trace operator is the number of fundamental fields $I,X$ required to build it. Thus, $F[r,s]$, $J[r,s]$ have width $r+s+2$ and $E[r,s]$ has width $r+s$. The OPE between single-trace operators of width $l_1,l_2$ is a sum of expressions of width $\le l_1 + l_2 -2$. 

	Suppose we are studying the correlation functions of $n$ single-trace operators $E,F,J$ of total width $l$.  Suppose by induction we have already defined a consistent set of correlation functions with $k$ insertions of single-trace operators of total width $m$ where $k + m < n + l$.  Then,  the $n$-point function of total width $l$ is defined up to the addition of a function $\cF(z_1,\dots,z_n)$ which is regular at $z_i = z_j$. Thus, $\cF$ is a sum of products of the monomials $z_1^{k_1} \dots z_n^{k_n}$ where the allowed values of $k_1,\dots,k_n$ are determined by the spins of the fields and the boundary conditions -- that is, by the conformal blocks of the $1$-point functions.  
\end{proof} 
Finally, the one-point conformal blocks are easy to determine. We will write the one-point conformal blocks as a representation of the symmetry group of the system, which is $\SO(8) \times \SL_2(\C) \times \C^\times$.  Here $\SL_2(\C) \times \C^\times$ is the complexification of the isometry group of Burns space, and this group acts on twistor space in a way fixing the boundary divisors.  In terms of the currents $E,F,J$, $\C^\times$ is the rotation of the $z$-plane. The currents $E[m,n]$ for $m+n$ fixed transform in the irreducible spin $(m+n)/2$ representation of $\SL_2(\C)$, and similarly for $F[m,n]$, $J[m,n]$. Finally, $\SO(8)$ only acts on $J[m,n]$, which transforms in the adjoint representation.

Let $J[r]$, $E[r]$, $F[r]$ denote the collection of currents $J[m,n], E[m,n], F[m,n]$ with $m+n = r$.  Let $S_+$ denote the fundamental representation of $\SL_2(\C)$.  Let $S_-$ denote the two-dimensional representation of $\C^\times$ with weights $(\frac{1}{2},-\frac{1}{2})$. (The terminology is because these representations come from the spin representations of $\op{Spin}(4)$). 

Then, the one-point function of $J[r]$, $E[r]$, $F[r]$ live in the following representations:
\begin{equation} 
	\begin{split} 
		J[r] & \in \Sym^r S_+ \otimes \Sym^{r} S_- \otimes \mf{so}_8		\\
		E[r] & \in \Sym^r S_+ \otimes \Sym^r S_- \\
		F[r] & \in \Sym^r S_+ \otimes \Sym^{r} S_- 
	\end{split}
\end{equation}
In order to compare with what we find for the field theory on Burns space, it will be useful to write this in terms of sections of vector bundles on $\CP^1$.  We find the one-point conformal blocks of the $E,F,J$ towers are:
\begin{equation} 
	\begin{split} 
		J & \in H^0 (\CP^1,  \mf{so}_8 \otimes \Sym^\ast ( S_+ \otimes \Oo(1) ) )\\
		E & \in  H^0 (\CP^1,   \Sym^\ast ( S_+ \otimes \Oo(1) ) )\\
		F & \in  H^0 (\CP^1,   \Sym^\ast ( S_+ \otimes \Oo(1) ) ).
	\end{split} 
\end{equation}

\subsection{Comparison of conformal blocks with the Hilbert space in four dimensions}
We want to show that the space of conformal blocks of the chiral algebra is isomorphic to the Hilbert space of the four-dimensional dual theory on a small $3$-sphere surrounding the ``core'' of Burns space.  In the coordinates $u^{\dal}$ where the Burns space K\"ahler potential is
\begin{equation} 
	\norm{u}^2 + N\log \norm{u}^2 
\end{equation}
the relevant $S^3$ is just that where $\norm{u}^2 = \eps$, $\eps$ small.  On flat space, the Hilbert space for this $S^3$ is the space of local operators.  On Burns space, where the apparent singularity in the metric has been resolved by blowing up the origin, this Hilbert space represents the space of operators wrapping the $S^2$ that replaces the origin. 

We note that this Hilbert space does not have a Hamiltonian, because a neighbourhood of this $S^3$ does not have a time-translation symmetry.  

We will show that the Hilbert space and the space of conformal blocks are isomorphic as representations of the $\SO(8) \times \SU(2) \times \U(1)$ symmetry algebra present on both systems. As we explained earlier, the interpretation of this is the following.  To define the correlation functions of the chiral algebra, we need to choose a conformal block.  To define the scattering amplitudes on Burns geometry, we are free to modify the theory by placing an operator in the center of Burns space.  Our analysis shows that this additional data is the same in each case.  

The Hilbert space is built in the usual way. First, we build the phase space of the model, and then we apply geometric quantization.   Since there is no Hamiltonian, the phase space is simply a symplectic manifold defined from the Cauchy data.    We will model the Cauchy data as simply that of a collection of free scalars in $\mf{so}(8)$.  One might object that by doing this we will be missing more global features, such as Skyrmions, which are states with topological charge in $\pi_3(\SO(8))$. However, we will show shortly that Skyrmions do not arise in our model. 

The phase space is also insensitive to the metric. Therefore there is no change if we work with the flat metric, rather than the Burns metric.   Since the free theory is conformally invariant, this Hilbert space can then be computed by the state-operator correspondence: it is the same as  the space of local operators in the free theory.

Finally, the space of local operators in the free theory can be computed using twistor space and matched with what we found above from studying conformal blocks.  (This is a version of the argument given in \cite{Costello:2022wso} relating conformal blocks to local operators).  The Hilbert space is the symmetric algebra on the single-particle Hilbert space, which is identified with the space of local operators which are linear functionals in the fields.  To prove the correspondence with conformal blocks, we therefore need to match the space of those local operators which are linear in the fields, with the one-point conformal blocks. 

The space of local operators which are linear on the fields is the dual vector space to the vector space of solutions to the equations of motion of the free theory.  For a free scalar field valued in $\mf{so}_8$, the space of solutions to the equations of motion is (by the Penrose transform)
\begin{equation} 
	H^1(\PT, \Oo(-2) ) \otimes \mf{so}_8. 
\end{equation}
By pushing forward the structure sheaf of $\PT$ to $\mbb{CP}^1$, this is the same as
\begin{equation} 
	H^1(\mbb{CP}^1, \Oo(-2) \otimes \Sym^\ast (S_+ \otimes \Oo(-1) ) ) \otimes \mf{so}_8. 
\end{equation}
Serre duality tells us that the dual vector space of this is
\begin{equation} 
	H^1(\mbb{CP}^1, \Oo(-2) \otimes \Sym^\ast (S_+ \otimes \Oo(-1) ) )^\vee\otimes \mf{so}_8 = H^0(\mbb{CP}^1, \Sym^\ast (S_+ \otimes \Oo(1) ))\otimes \mf{so}_8   
\end{equation}
which is exactly the one-point conformal blocks of the $J$-tower of currents.

A similar analysis applies to the closed string fields.  The space of field configurations $\phi$ satisfying the fourth-order equation $\Lap \Lap \phi = 0$ is the same as the space of pairs of solutions to the ordinary Laplace equation $\Lap \phi_i = 0$, $i = 1,2$.  If we take $\phi$ up to the addition of a constant, then we take $\phi_1$ up to the addition of a constant.  This means that the space of linear local operators for the  closed string fields is
\begin{equation} 
	H^1(\mbb{CP}^1, \Oo(-1) \otimes \Sym^{> 0} (S_+ \otimes\Oo(-1) ) ) \oplus  H^1(\mbb{CP}^1, \Oo(-1) \otimes \Sym^{\ast} (S_+ \otimes \Oo(-1) ) ) 
\end{equation}
By Serre duality, this matches exactly with the contribution of the $E$ and $F$ towers to the one-point conformal blocks.  The fact that the first factor above contains only the positive symmetric powers matches the fact that there is no $E[0,0]$ current (recall from \eqref{CFT} that $E[0,0]=1$). 

\paragraph{Absence of Skyrmions.} Finally, we should mention why we do not need to consider topologically nontrivial states, i.e.\ Skyrmions, in the open string sector. It turns out that they can not arise in WZW$_4$, as we will see from a twistor space analysis. 

Consider the twistor space of $\R^4 - 0$.  This is the complement of the zero section in $\PT$, and so is a fibration over $\CP^1$ with fibers $\C^2 - 0$.  Twistor space has two patches, lying over the loci where $z \neq 0$ and $z \neq \infty$ in $\CP^1$.  Over these patches twistor space is $\C \times (\C^2 - 0)$. A field configuration on these patches is a holomorphic bundle, and every such bundle is topologically trivial.

A topologically nontrivial field configuration on spacetime would come equipped with topologically nontrivial gluing data on the overlap $\C^\times\times(\C^2 - 0)$ of the two patches.  The gluing data is a holomorphic map 
\begin{equation*} 
	\C^\times\times(\C^2 - 0) \to \SO(8,\C). 
\end{equation*}
Associated to this is a class in $\pi_3(\SO(8,\C))$, obtained by restricting this gluing map to $S^3 \subset \C^2 -0$. This class is the topological charge of the Skyrmion in spacetime. Our goal is to show that this class is zero.

$\SO(8,\C)$ is an affine algebraic variety, and hence a closed subvariety of some $\C^n$.  Hartog's theorem tells us that any holomorphic map from $\C^2 - 0$ to $\C^n$ extends to a map from all of $\C^2$.  Therefore, the gluing data gives a topologically trivial class in $\pi_3(\SO(8,\C))$.

\section{Where does celestial CFT live?}
\label{sec:bdry}

Up till now, our analysis has been fairly top-down and string theoretic. But for a genuine holographic interpretation, we must also connect to the plethora of bottom-up approaches developed in the literature \cite{Pasterski:2021raf}. To do this, we need to construct a physical picture of the holographic boundary where the celestial dual lives. In this section, we make precise the identification of this boundary as a boundary of the twistor space of Burns space.

This is accomplished by exhibiting the twistor space of Burns space as AdS$_3\times S^3$. The boundary of this twistorial AdS$_3$ bulk will act as our Euclidean proxy for the celestial sphere. This is where the boundary data for solutions of linearized free field equations lives. In our case, this spherical boundary also gets appended with two extra copies of the non-negative reals $[0,\infty]$ living at its north and south poles. At the level of celestial CFT, these play the role of point defects living at the poles.  Thereafter, the rough idea will be to foliate AdS$_3$ by twistor lines. $S^2$ compactification along the twistor lines yields Burns space. This converts familiar AdS$_3\times S^3$ holography into ``exotic'' celestial holography!


\subsection{Looking for AdS$_3$} 

As we saw in section \ref{sec:backreact}, the twistor space of Burns space is obtained from the complete flag variety $\F=\F(1,2,3)$ of points in lines in $\CP^2$ by subtracting a single copy of $\CP^1$. This led us to conclude that Burns space itself is diffeomorphic to $\CP^2$ minus a point. It is the $\CP^2$ viewpoint that most naturally helps us understand the boundary structure of our holographic duality. The removed point corresponds precisely to spatial infinity $\|u\|\to\infty$ on Burns space (in the sense of Euclidean signature one-point compactifications).

Recall that the flag variety $\F$ was the set of points $([V^k],[W_k])\subset\CP^2\times\CP^2$ satisfying $V\cdot W = 0$. This was the twistor space of $\CP^2$, which we coordinatized by homogeneous coordinates $[U^k]$. Complex conjugation was denoted $U^k\mapsto \bar U_k\equiv\overline{U^k}$. The twistor line corresponding to $[U^k]$ was cut out by the equations $V\cdot W = U\cdot W = V\cdot\bar U=0$. In what follows, we will use the notation $V^k = (V^{\dal},V^3)$, $W_k=(W_{\dal},W_3)$, $U^k = (U^{\dal},U^3)$ etc.\ for convenience.

Introduce affine coordinates $u^{\dal} = U^{\dal}/U^3$ on $\CP^2$. To obtain Burns space in these coordinates, one performs the diffeomorphism $u^{\dal}\mapsto \sqrt{N}u^{\dal}/\|u\|^2$. Being an inversion, this maps the origin $u^{\dal}=0$ of this affine patch to the spatial infinity of Burns space. In homogeneous coordinates, this is what we called the \emph{point at infinity} $[U^k]=[0,0,1]$. When constructing Burns space, we excised this point from $\CP^2$. To study its conformal boundary, it will prove more useful to compactify Burns space into $\CP^2$ by adding back this point. Doing this, we can instead treat $[0,0,1]$ as the spacetime boundary. In the flag variety, the twistor line associated to $[0,0,1]$ is the copy of $\CP^1$ given by\footnote{Not to be confused with the notation for the chiral algebra field $I_{ri}$.}
\be
I = \{V^3 = W_3 = 0 \}\,.
\ee
We call this the \emph{line at infinity}.

The brane backreaction generates a copy of $\SL_2(\C)$ cut out by $[v\,w]=N$. The boundary of $\SL_2(\C)$ would be the set of points $(v,w)\in\C^4$ satisfying $[v\,w]=0$. We can append these to $\SL_2(\C)$ by an appropriate compactification. In our setup, $\SL_2(\C)$ is compactified by embedding it in the flag variety via the map $V^{\dal} = -v^{\dal}V^3/\sqrt{N}$, $W_{\dal} = w_{\dal}W_3/\sqrt{N}$ given in \eqref{ZNtoF}. The image of this map is the open set $V^3W_3\neq0$. The complement of this image in $\F$ is the boundary divisor $V^3W_3=0$. This is a union of the following sets:
\be
    D_0 = \{W_3=0\}\,,\qquad
    D_\infty = \{V^3=0\}\,.
\ee
Each of these occurs with multiplicity one. $D_0$ and $D_\infty$ are each a copy of the blow-up of $\CP^2$ at a point. Their intersection equals the line at infinity: $D_0\cap D_\infty=I$.

The twistor space of Burns space is given by $Z_N = \F-I = \SL_2(\C)\cup(D_0\cup  D_\infty - I)$.   From the 4d perspective, the divisors $D_0-I$, $D_\infty-I$ of $Z_N$ are not part of asymptotic infinity. Each of them maps isomorphically onto $\til\C^2$, the complex manifold underlying Burns space.

\medskip

In \cite{CG}, the authors compactified the deformed conifold $\SL_2(\C)$ to a quadric in $\CP^4$. The resulting geometry had an $\SL_2(\C)\times\SL_2(\C)$ symmetry. Quotienting out an $\SU(2)$ subgroup of one of the copies of $\SL_2(\C)$ led to Euclidean AdS$_3 \simeq \SL_2(\C)/\SU(2)$. Their holographic dual lived on the boundary of this AdS$_3$. 

In our current setting, we have instead compactified the $\SL_2(\C)$ by embedding it in the flag variety. The $\U(1)_\text{left}\times\SU(2)_\text{right}$ isometries of Burns space lift to an action of the complexification $\C^\times\times\SL_2(\C)$ on the flag variety:
\be
([V],[W])\mapsto \bigl([L^{\dal}{}_{\dot\beta}V^{\dot\beta},s^{-1}V^3]\,,\;[L_{\dal}{}^{\dot\beta}W_{\dot\beta},sW_3]\bigr)\,,\quad L\in\SL_2(\C)\,,\;s\in\C^\times\,.
\ee
The $\C^\times$ is our defect conformal group. We want to quotient $\F$ by $\SU(2)\subset\SL_2(\C)$ and show that the resulting geometry is again AdS$_3$. We would also like to visualize the projection of the twistor lines onto this AdS$_3$.

On the $\SL_2(\C)$ patch $V^3W_3\neq0$, the $\SU(2)$ action is free and has $S^3$ orbits. Quotienting out $S^3$ gives AdS$_3$ (without its boundary) as the orbit space. Poincar\'e coordinates on this AdS$_3$ can be constructed out of the natural $\SU(2)$ invariants in the game. Using the conventions of \cite{CG}, we introduce the Poincar\'e coordinates $x\in\CP^1$, $\vho>0$,
\be\label{xy}
\begin{split}
    x &= \frac{[\hat w\,v]}{\|w\|^2} = -\frac{W_3\hat W^{\dal}V_{\dal}}{V^3\hat W^{\dal}W_{\dal}}\,,\\
    \vho &= \frac{1}{\|w\|^2} = \frac{|W_3|^2}{N\hat W^{\dal}W_{\dal}}\,.
\end{split}
\ee
The boundary of AdS$_3$ would usually be the copy of $\CP^1$ obtained by taking a union of the limits $\vho\to0$ and $\vho\to\infty$.

The boundary divisors $D_0$ and $ D_\infty$ project down to copies of $[0,\infty]$ in $\F/\SU(2)$, so they cannot cover the $S^2$ boundary of AdS$_3$. We would have naively expected the line at infinity $I$ to then map to the boundary of AdS$_3$. A similar projection does give the correct boundary of AdS$_3$ in the case of the $\CP^4$ compactification of $\SL_2(\C)$ \cite{CG}. But unfortunately, in our flag variety compactification, $I$ is the orbit of the point $([1,0,0],[0,1,0])\in\F$ under the $\SU(2)$ action. So the $\SU(2)$ quotient collapses it to a point. Hence, the line at infinity cannot by itself be the boundary that we are looking for. This is consistent with the fact that $I$ is the fixed point set of the left-handed $\C^\times$ action. It is only acted on by the right-handed $\SL_2(\C)$ symmetries which constitute the flavor symmetry of our dual chiral algebra, whereas we expect the $\C^\times$ to be the nontrivial defect conformal symmetry of our chiral algebra sphere.

This analysis shows that, unlike the $\CP^4$ compactification, the flag variety compactification of $\SL_2(\C)$ simply does not contain the boundary divisor relevant to a conventional notion of holography. Nonetheless, it turns out that this situation can be easily ameliorated. The correct boundary is instead found by blowing up $I$. 


\subsection{The boundary of twistor space}

We begin by blowing up the line at infinity $I\subset\F$ to obtain the variety\footnote{This blowup is not a twistor space because it is compact and K\"ahler, and according to a theorem of Hitchin \cite{hitchinkahler}, the only compact K\"ahler twistor spaces are $\CP^3$ and the flag variety $\F$.  We thank Claude LeBrun for the warning! Nonetheless, it \emph{is} an important building block in the Donaldson-Friedman construction of twistor spaces of connected sums of $\CP^2$ \cite{donaldson1989connected}.}
\be
\wt\F = \text{Bl}_{I}\F\,.
\ee
Away from $V^3=W_3=0$, this is biholomorphic to $\F-I$. But the line at infinity has now been replaced by the exceptional divisor
\be
E=\CP^1\times\CP^1 = I\times\CS^2\,.
\ee
The first $\CP^1$ is a copy of the line at infinity. The second $\CP^1$ has been suggestively named $\CS^2$, as it will turn out to be the correct boundary on which our holographic dual lives. In this sense, we interpret it as our Euclidean analogue of the ``celestial sphere''. $E$ consists of a union of projective lines $\CP^1\times z$, one for each point $z\in\CS^2$.\footnote{In the terminology familiar from work on scattering amplitudes, $\CS^2$ is the sphere of undotted spinor-helicity variables $\lambda_\al=(1,z)$, while the deprojectivization of $I$ is ``Fourier conjugate'' to the 2-plane of dotted spinor-helicity variables $\tilde\lambda_{\dal}$ (one expects this to be a precise correspondence in split signature).} 

$\til\F$ is a compactification of the twistor space of the Burns geometry that will encode the asymptotic structure at infinity.  
The complement of $\SL_2(\C)$ in $\til{\F}$ now consists of three divisors. Two of them are the preimages $\br D_0$, $\br D_\infty$ of the divisors $D_0,D_\infty$ by the blow-up morphism $\til\F\to\F$, but now they no longer intersect. The third divisor is $E$, which intersects $\br{D}_0$ along $\CP^1\times0$ and $\br{D}_\infty$ along $\CP^1\times\infty$. 

One may define the blow-up in coordinates by working in affine patches. Because $V^{\dot1}W_{\dot1}+V^{\dot2}W_{\dot2}=0$ along $V^3=W_3=0$, the line $I$ is contained within the union of the two patches $\{V^{\dot2}\neq0, W_{\dot1}\neq0\}$ and $\{V^{\dot1}\neq0, W_{\dot2}\neq0\}$ of $\F$. Let us work in the first of these. In this patch, the blow-up is described by the set of points $([V],[W],z)\in\CP^2\times\CP^2\times\CP^1$ satisfying
\be\label{Ftildef}
 V\cdot W = 0\,,\qquad \frac{W_3}{W_{\dot1}}=z\,\frac{V^3}{V^{\dot2}}\,.
\ee
The $\C^\times$ action $(V^3,W_3)\mapsto (s^{-1}V^3,sW_3)$ extends to $\CS^2$ as the $\C^\times$ defect conformal symmetry $z\mapsto s^2z$. 

The line at infinity $I$ is acted on by the right-handed $\SL_2(\C)$ symmetry as usual. Since the right-handed $\SU(2)$ action collapses $I$ to a point, the exceptional divisor $E$ projects down to $\CS^2$ under the $\SU(2)$ quotient. In the patch under consideration, $V^{\dal}W_{\dal}=0$ is solved to find $V^{\dal}=-V^{\dot2}W^{\dal}/W_{\dot1}$. Using this alongside $V^3=W_3=0$, the Poincar\'e coordinates \eqref{xy} reduce to
\be
x = z\,,\qquad  \vho = 0
\ee
on the image of $E$ in the $\SU(2)$ quotient. A similar analysis may be performed in the other patch $\{V^{\dot1}\neq0, W_{\dot2}\neq0\}$ by setting $W_3/W_{\dot2}=-\tilde z\,V^3/V^{\dot1}$ for some $\tilde z\in\CP^1$, where the sign is fixed by demanding that $\tilde z = z$ along $I$ on the overlap of the two patches.

This confirms that the boundary of AdS$_3$ coincides with the celestial sphere $\CS^2$ created by the blow-up. So $\CS^2$ is the $z$-plane where the chiral algebra lives. 
In this way, twisted holography makes precise the bottom-up intuition that celestial conformal field theory lives on a 2-sphere at infinity \cite{Pasterski:2016qvg}. In our Euclidean setting, although we cannot access the standard notion of a spacetime celestial sphere, we have just shown that our holographic plane is nevertheless a 2-sphere that acts as the boundary of the twistor space of Burns space.

\medskip

Using these facts, we may recast Burns holography as an example of AdS$_3$ holography. To see this explicitly, we need to understand how compactification of $\SL_2(\C)\simeq\AdS_3\times S^3$ along the twistor lines affects the AdS$_3$ factor. 

Let us rewrite the twistor lines corresponding to the $\C^2-0$ patch of Burns space in compact notation,
\be
v^{\dal}(\sigma) = u^{\dal} + \sigma\vartheta\hat u^{\dal}\,,\qquad w^{\dal}(\sigma) = \hat u^{\dal} + \sigma^{-1}\vartheta u^{\dal}\,,\qquad \vartheta \vcentcolon=\sqrt{1+\frac{N}{\|u\|^2}}\,.
\ee
The quantity $\vartheta$ is the fourth root of the metric determinant of the Burns metric. Under a quotient by the right-handed rotations $(v,w)\mapsto(Lv,Lw)$, $L\in\SU(2)$, a curve associated to $u$ is identified with the curve associated to $Lu$. So in the $\SU(2)$ quotient, the curves do not drop in dimension, but all the curves along any $\SU(2)$ orbit get identified with each other. Every curve in the quotient corresponds to an $S^3$ worth of curves in $\SL_2(\C)$. 

We can compute their images in the orbit space AdS$_3=\SL_2(\C)/\SU(2)$. In the Poincar\'e coordinates introduced in \eqref{xy}, the curves project down to
\be\label{xyproj}
x(\sigma) = \frac{\sigma\vartheta(1+|\sigma|^2)}{\vartheta^2+|\sigma|^2}\,,\qquad \vho(\sigma) = \frac{|\sigma|^2}{\|u\|^2(\vartheta^2+|\sigma|^2)}\,.
\ee
Plugging in $\sigma=0$, we observe that each of these curves starts at $x=0$, $\vho=0$. Setting $\sigma=\infty$ shows that each curve also ends at $x=\infty$, with the curve corresponding to $u^{\dal}$ asymptoting to $\vho=\|u\|^{-2}$. 

These curves are best visualized by working in the Poincar\'e disk coordinates
\be
(y_1,y_2,y_3) = \left(\frac{x+\bar x}{|x|^2+(N\vho+1)^2},\frac{\im(\bar x-x)}{|x|^2+(N\vho+1)^2},\frac{|x|^2+N^2\vho^2-1}{|x|^2+(N\vho+1)^2}\right)\,.
\ee
For positive $N$, this change of coordinates maps AdS$_3$ to the unit ball $y_1^2+y_2^2+y_3^2<1$. The equations \eqref{xyproj} are now recognized as parametrizing the spheroids
\be
\vartheta^{2}(y_1^2+y_2^2)+y_3^2=1
\ee
with $1/\vartheta$ acting as the half-length of two of the principal axes.

Similarly, setting $u^{\dal}=t\,\zeta^{\dal}$, the twistor lines associated to points $[\zeta^{\dal}]\in\CP^1$ on the exceptional divisor $t=0$ of Burns space take the form
\be
v^{\dal}(\sigma) = \frac{\sqrt{N}\hat\zeta^{\dal}}{\|\zeta\|}\,\sigma\,,\qquad w^{\dal}(\sigma) = \frac{\sqrt{N}\zeta^{\dal}}{\|\zeta\|}\,\sigma^{-1}\,.
\ee
All points on a 2-sphere can be mapped to each other by its rotational symmetry, so this exceptional divisor lies in a single $\SU(2)$ orbit. Consequently, all of these twistor lines get identified with each other. In Poincar\'e coordinates, they descend to a single one-real-dimensional line in the $\SU(2)$ quotient which reads
\be
x(\sigma) = 0\,,\qquad \vho(\sigma) = \frac{|\sigma|^2}{N}\,.
\ee
This line shoots straight across the Poincar\'e disk. 

\begin{figure}
    \centering
    \includegraphics[scale=0.35]{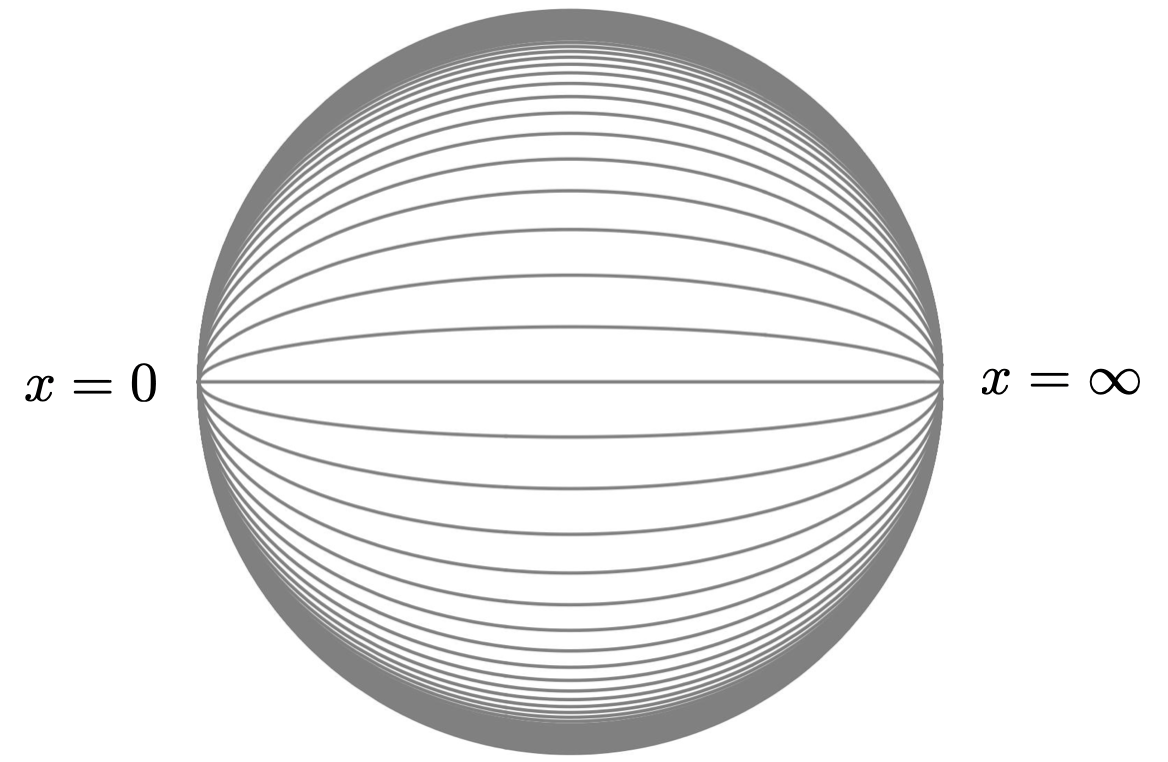}
    \caption{A foliation of the interior of AdS$_3$ by the twistor lines of Burns space in the Poincar\'e disk model. All curves stretch from $x=0$ on the left to the point $x=\infty$ on the right. The boundary of AdS$_3$ is \emph{not} a twistor line but is instead identified with $\CS^2$ via the transformation $x=z,\vho=0$.}
    \label{fig:onion}
\end{figure}

Putting these curves together, one finds a foliation of the entire interior of AdS$_3$. This creates the image of an onion-like foliation of the Poincar\'e disk as displayed in figure \ref{fig:onion}. The construction of our holographic duality consists of compactification along the twistor lines. This compactifies the interior of AdS$_3$ down to a copy of $\R_+\equiv(0,\infty)$ parametrized by $\|u\|$. Globally, the B-model on $\SL_2(\C)\simeq\AdS_3\times S^3$ compactifies down to a theory on $\til{\C}^2$. 
The final outcome is a duality between WZW$_4$ + Mabuchi gravity on Burns space, and our holographic chiral algebra living on the AdS$_3$ boundary that consists of the ``celestial sphere'' $\CS^2$.

We emphasize that this procedure is essentially the opposite of de Boer and Solodukhin's approach to flat holography. In \cite{deBoer:2003vf}, they conjectured that flat holography might arise from foliating Minkowski spacetime by AdS$_3$ and dS$_3$ slices and ``uplifting'' AdS$_3$ holography to embedding space. Our procedure instead starts with a relatively standard holographic duality in AdS$_3\times S^3$ and compactifies AdS$_3$ along a foliation by Riemann spheres. In this way, it compactifies AdS$_3$/CFT$_2$ to a celestial holographic duality. Nevertheless, in agreement with de Boer and Solodukhin, we have shown that the holographic dual of our four-dimensional theory does indeed live on a two-dimensional sphere.

Finally, to complete the construction of $\til\F/\SU(2)$, let us also mention that the divisors $\br D_0$ and $\br D_{\infty}$ project down to copies of $[0,\infty]$ that respectively live entirely at the points $x=0$ and $x=\infty$ on the AdS$_3$ boundary. They partially resemble a real blow-up of these boundary points, in the sense of replacing $x=0,\infty$ at the boundary with copies of $[0,\infty]$ in the above picture. For example, to zoom into $\br D_0$ in the quotient, we replace the quotient by the set of points $(x,\vho,t)\in\AdS_3\times[0,\infty]$ satisfying $|x|^2 = t\vho$. Away from $x=\vho=0$, this is diffeomorphic to AdS$_3-\{x=\vho=0\}$. But just as for a blow-up, the idea is that the point $x=\vho=0$ is replaced by a copy of $[0,\infty]$ with coordinate $t$. A similar analysis works for $\br D_\infty$ if one uses coordinates adapted to $x=\infty$. 

An important point is that before the $\SU(2)$ quotient, all twistor lines were disjoint. But from figure \ref{fig:onion}, the projections of the twistor lines to $\til\F/\SU(2)$ naively appear to all intersect at $x=0,x=\infty$. They are only seen to be disjoint in the quotient geometry after accounting for the projections of the divisors $\br D_0,\br D_\infty$. E.g., the projected twistor line \eqref{xyproj} intersects $x=\vho=0$ at $t=\lim_{\sigma\to0}|x(\sigma)|^2/\vho(\sigma) = \|u\|^2$. The extra divisors might support extra sources or intial data of interest, but we leave their study to future work. 


\section{The holographic dictionary on twistor space}
\label{sec:twistor_dictionary}

Celestial holography has greatly benefited from the study of extrapolate dictionaries in the past \cite{Pasterski:2021dqe,Donnay:2022sdg}. Before we start building a holographic dictionary between Burns space states and celestial operators, in this section we will flesh out an extrapolate dictionary directly on twistor space. This crucially uses the AdS$_3\times S^3$ geometry of the twistor space of Burns space that was described in the previous section. We will identify operators living on the boundary of AdS$_3$ with states in the bulk of twistor space. The twistor space states are always in one-to-one correspondence with spacetime states via the Penrose transform.

Quite dramatically, this will enable a proof that at tree level, the $ \SU(2)\times\U(1)$ isometries of Burns space get enhanced to a full $\SU(2)\times\SU(2)$ Lorentz symmetry (and in fact its complexification) in the scattering amplitudes of our bulk theory.  At loop level we do not provide a rigorous proof of this statement. However, we conjecture that it remains true at loop level, and our tree level argument provides very strong evidence.    This enhancement will be concretely observed in the amplitudes computed in section \ref{sec:tests}, but proving it by direct calculation on Burns space turns out to be nigh-on impossible.


\subsection{Relating two different compactifications}

Let us recall the holographic dictionary of \cite{CG}. The bulk geometry is   $\SL_2(\C)$, which we coordinatize as the set of pairs of vectors $v^{\dal}$, $w^{\dal}$ satisfying $[v\,w] = N$. 

The whole space has a symmetry of $\SL^+_2 \times \SL_2^-$, acting on the left and on the right. This  is a copy of the Lorentz group. If we build a four-vector $X^{\al\dal}$ with 
\begin{equation} 
	X^{1 \dal } =  v^{\dal}  \,, \quad X^{2 \dal} =   w^{\dal}  
\end{equation}
then $\SL_2^+$ rotates the dotted indices of $X^{\alpha \dot{\alpha}}$ and $\SL_2^-$ rotates the undotted indices. In terms of the group manifold $\SL_2(\C)$, $\SL_2^+$ acts by left multiplication and $\SL_2^-$ acts by right multiplication. The Cartan of $\SL_2^-$ acts on $v^{\dal}$ with weight $1$ and $w^{\dal}$ with weight $-1$.   

 In \cite{CG} the group manifold was compactified to the quadric 
\begin{equation} \label{SLbar}
	\br{\SL}_2(\C) =  \{[v^{\dal}, w^{\dal}, t] \mid  [v\,w]  = Nt^2\} \subset \CP^4.  
\end{equation}
$\SL_2(\C)$ is embedded as the patch $t=1$. The boundary divisor $t=0$ is a quadric in $\CP^3$, so is a $\CP^1_+ \times \CP^1_-$. The spheres $\CP^1_{\pm}$ are acted on by $\SL_2^{\pm}$. The boundary chiral algebra consists of operators which live at a point $z \in \CP^1_-$, and wrap $\CP^1_+$.  The holographic chiral algebra is defined to be the set of such boundary operators equipped with the OPE in the $z$-plane (i.e.\ the $\CP^1_-$ direction).

As explained in the previous section, the twistor construction leads us to a different compactification of $\SL_2(\C)$, that inside the blow-up $\til\F = \text{Bl}_I\F$ of the flag variety $\F$ along the line at infinity $I$. In that case, the complement of the interior locus $\SL_2(\C)$ consists of three divisors $\br D_0$, $\br D_\infty$, and $E$. The first two divisors are each a copy of $\til{\CP}{}^2$, the blow-up of $\CP^2$ at a point. They are not part of the holographic boundary but instead behave like surface defects. The third divisor $E=\CP^1_+\times\CP^1_- = I\times\CS^2$ is the holographic boundary. It admits an action of the $\SU(2)\times\U(1)$ isometries of Burns space, where $\SU(2)$ acts on $\CP^1_+=I$ and $\U(1)$ on $\CP^1_-=\CS^2$. The second sphere $\CP^1_-$ is the $z$-plane that supports the celestial CFT. 

There is a rational map\footnote{That is, a map from a dense open subset of $\F$ to $\br\SL_2(\C)$ that need not extend to all of $\F$.}  of algebraic varieties
\begin{equation} 
\F \to \br{\SL}_2(\C)  
\end{equation}
defined by 
\begin{equation} 
	\big([V^{\dal}, V^3] , [W_{\dal}, W_3 ]\big) \mapsto [V^{\dal} W_3, W^{\dal} V^3, V^3 W_3].  
	\label{eqn:rational_map} 
\end{equation}
This is an isomorphism on the interior $\SL_2(\C)$.  It is regular on the loci $\{V^3= 0, W_3\neq0\}$ and $\{W_3 = 0,V^3\neq0\}$. However, it has singularities when $V^3 = W_3 = 0$.  This is because on the flag variety $\F$ the coordinates $V^{\dal}$ and $W^{\dal}$ are defined up to separate rescalings by $\C^\times$, whereas on $\br{\SL}_2(\C)$ only one rescaling is used, scaling $V^{\dal}, W^{\dal}$ at the same time.  Thus, when $V^3, W_3 \to 0$, the relative coefficient of $V^{\dal}$ and $W^{\dal}$ in \eqref{eqn:rational_map} is undetermined.

This becomes a regular map (i.e.\ everywhere well-defined) if we lift it to the blow-up, giving us a map
\begin{equation} 
	\til{\F} \to \br{\SL}_2(\C)\,.  
\end{equation}
This maps the exceptional divisor $E$ to the boundary divisor $\CP^1_+\times\CP^1_-$ of $\br\SL_2(\C)$, while collapsing $\br{D}_0$ and $\br{D}_\infty$ down to copies of $\CP^1$. The fiber of the map $\til{\F} \to \br{\SL}_2(\C)$ is a point, except over the locus $z = 0$, $z = \infty$ in the boundary where the fiber is a copy of $\CP^1$.  

Using this map, we can modify the holographic dictionary of \cite{CG}.  This dictionary was defined by considering boundary conditions for the fields, and then modifying them on a curve $\CP^1 \times z$ in the boundary of $\br{\SL}_2(\C)$.  We can implement the same procedure on $\til{\F}$, again modifying the boundary conditions on a $\CP^1 \times z$.


\subsection{Building an extrapolate dictionary}

To do this in detail, we need to specify the boundary conditions on $\til{\F}$ and relate them to those on $\br{\SL}_2(\C)$.

Let us first discuss the open string field. This is required to vanish on the boundary divisor of $\br{\SL}_2(\C)$. If we take such   a $(0,1)$-form on $\br{\SL}_2(\C)$, and pull it back to $\til{\F}$, it still vanishes on the boundary divisors of $\til{\F}$.  So we can impose the same boundary condition on $\til{\F}$.  

For the closed string fields, things are slightly more tricky. The closed string field is a $(2,1)$ form, which in \cite{CG} was  allowed to have logarithmic poles on the boundary divisor. However, this was not entirely essential, and it is perhaps better to use the stronger boundary condition that requires that the closed string field is a regular $(2,1)$-form on the boundary.  This boundary condition was discussed in \cite{Costello:2021bah}.  If we use this boundary condition for the closed string field, then a field on $\br{\SL}_2(\C)$ obviously pulls back to one on $\til{\F}$.  

In this way, we can mimic the story of \cite{CG}, and define the holographic correlation functions to be the path integral of the theory where the boundary conditions are modified along a $\CP^1$ in the boundary. The $\CP^1$ must live in the divisor $\CP^1 \times \CP^1 = I\times\CS^2$, at a fixed value of $z\in\CS^2$. The modification of the boundary conditions encode the choice of single-trace local operators in the dual CFT. 

As discussed in \cite{CG}, sections 7 and 8, modifications of the boundary condition are the same as boundary operators placed along $\CP^1 \times z$.  On $\til\F$, these boundary operators (in the open string sector) are given by integrals of the holomorphic Chern-Simons gauge field $\cA \in \Omega^{0,1}(\til{\F}, \g \otimes \Oo(-D))$ along $\CP^1 \times z\subset E$, where $D = \br D_0\cup\br D_\infty\cup E$ is the total boundary divisor. In other words, they are holomorphic Wilson lines wrapping the line at infinity $I$ at a point $z\in\CS^2$.

\paragraph{Dictionary on $\br{\SL}_2(\C)$ for open string fields.} Before we build these, let us revisit the procedure of specifying such boundary operators on $\br{\SL}_2(\C)$. Focus on the patch $w^{\dot2}\neq0$ without loss of generality. We let $n=t/w^{\dot2}$ denote a coordinate normal to the boundary $\partial \br{\SL}_2(\C) = \CP^1 \times \CP^1$. Let $y,z$ be coordinates on this $\CP^1 \times \CP^1$, where $y = w^{\dot{1}} / w^{\dot{2}} $ and $z=-v^{\dot2}/w^{\dot2}$. The normal bundle to $\p\br{\SL}_2(\C)$ is $\CO(1,1)$. So the normal vector $\partial_n$ has a first order zero at $y = \infty$ and at $z = \infty$.    

Suppose we study the B-model on $\br{\SL}_2(\C)$. As a boundary condition, we assume that the holomorphic Chern-Simons gauge field $\mc{A} \in \Omega^{0,1}(\br{\SL}_2(\C), \mf{so}_8)$ vanishes at $n = 0$.  One can build boundary local operators by expressions like
\begin{equation} 
	\int_{\CP^1 \times z}  y^l\,  \partial_n^k \mc{A}\,\d y + \cdots \label{eqn:boundaryintegral0}  
\end{equation}
where $l \le k-2$ in order to ensure the measure has no pole at $y = \infty$.  The expression as written is invariant under linearized gauge transformations. The ellipsis indicates expressions we need to add on to ensure invariance under non-linear gauge transformations.  Because our boundary conditions force gauge transformations to vanish at $n = 0$, there are a finite number of such terms; they will not be relevant for our discussion.

The normal vector $\partial_n$ has spin $\frac{1}{2}$ under rotations of the $z$-plane.  Therefore the expression \eqref{eqn:boundaryintegral0} has spin $k/2$. Further, because $l$ ranges from $0$ to $k-2$, it lives in a representation of spin $(k-2)/2$ of the $\SU(2)$ rotating the dotted indices. In \cite{CG} the identification 
\begin{equation} 
	J[l, k-2-l] =	\int_{\CP^1 \times z}  y^l\,  \partial_n^k \mc{A}\, \d y + \cdots \label{eqn:boundaryintegral}  
\end{equation}
was made (up to normalization). The point is that both sides have the same quantum numbers. This provides an extrapolate dictionary for twisted holography on $\br\SL_2(\C)$. We wish to pull this back to a dictionary on $\til\F$.


\paragraph{Dictionary on $\til\F$ for open string fields.} How  does this analysis change when we consider the blow up $\til \F$?  To answer this question, we need to understand how the normal vector $\partial_n$ to the boundary divisor of $\br{\SL}_2(\C)$ behaves on $\til \F$. 

The boundary of $\til \F$ consists of three divisors. Two of them, $\br{D}_0$, $\br{D}_\infty$, come from the boundary of the flag variety $\F$. The third is the exceptional divisor of the blow up, which we called $E$.

The map 
\begin{equation} 
	\pi : \til \F \to \br{\SL}_2(\C) 
\end{equation}
has the feature that
\begin{equation} 
	\pi^{-1} \partial \br{\SL}_2(\C) = \br{D}_0 \cup \br{D}_\infty \cup E 
\end{equation}
where on the right hand side all three divisors appear with multiplicity one.  For the exceptional divisor $E$, this is clear, because $E$ maps isomorphically onto the boundary divisor of $\br{\SL}_2(\C)$.  For the other divisors, we note that according to \eqref{eqn:rational_map}, the coordinate $t$ in \eqref{SLbar} whose vanishing cuts out the boundary divisor of $\br{\SL}_2(\C)$ corresponds to the function $V^3 W_3$ on the flag variety $\F$, whose zero locus is $\br{D}_0 \cup \br{D}_\infty$ with multiplicity one.

Let us work in a local patch near the boundary of $\br{\SL}_2(\C)$ with coordinates $y,z,n$ as above, where we work in the region near $z = 0$, $n = 0$.  Similarly, on $\til{\F}$, we can work in a local patch near the exceptional divisor $E$, with coordinates $y,z,\til{n}$ where $\til{n} = 0$ is the exceptional divisor. Pulling back the $\br\SL_2(\C)$ coordinates $y=w^{\dot1}/w^{\dot2}$, $z=-v^{\dot2}/w^{\dot2}$ by the map \eqref{eqn:rational_map} yields the corresponding coordinates $y=W^{\dot1}/W^{\dot2}$, $z=-V^{\dot2}W_3/W^{\dot2}V^3$ on $\til\F$. Because $W^{\dot2}=-W_{\dot1}$, the latter coincides with the coordinate $z$ defined through the blow-up equations \eqref{Ftildef}. In these coordinates, the divisor $\br{D}_0$ is the locus $z = 0$.

The locus $n = 0$ on our patch of $\br{\SL}_2(\C)$ must correspond to the locus where either $\til{n}=0$ or $z=0$ on our patch of $\til{\F}$ (with multiplicity one). This is because the limit $\til n\to0$ at fixed $z$ picks out $E$, while $z\to0$ at fixed $\til n$ picks out $\br D_0$. Thus, the map between the varieties takes the form
\begin{equation} 
	n = \til{n} z\,. 
\end{equation}
Alternatively, we can take this as the definition of $\til n$. In homogeneous coordinates, one finds that $\til n = -V^3/V^{\dot2}$. In the patch $\{V^{\dot2}\neq0,W_{\dot1}\neq0\}$, sending $V^3\to0$ while keeping $z$ fixed requires simultaneously sending $W_3\to0$, which lands on exactly the divisor $E$.

On the locus $\til{n} = 0$, we therefore have
\begin{equation} 
	\partial_n = z^{-1} \partial_{\til{n}}\,. 
\end{equation}
Now consider the pullback of the expression for $J[m,n]$:
\begin{equation} 
	J[m,n](z) =	\int_{\CP^1 \times z} y^m z^{-m-n-2}\, \partial_{\til{n}}^{m+n+2} \mc{A}\,\d y + \cdots 
\end{equation}
Noting that our boundary conditions for holomorphic Chern-Simons give $\mc{A}$ a zero at $z = 0$, we see that $J[m,n](z)$ has a pole of order $m+n+1$ at $z = 0$, exactly as desired. 

\paragraph{Dictionary for closed string fields.} Let us now do the corresponding analysis for the closed string field, which is a closed $(2,1)$ form $\eta$.  In the coordinates $y,n,z$ above, $\eta$ has three components $\eta_{yn}$, $\eta_{nz}$, $\eta_{zy}$ related by
\begin{equation} 
	\partial_z \eta_{yn} +  \partial_y \eta_{nz} +  \partial_n \eta_{zy}  = 0. 
\end{equation}
This equation tells us that only two of the three components will give us independent operators -- say $\eta_{nz}$ and $\eta_{yn}$.  Quantities built by integrating $\eta_{zy}$ over the $y$-plane will generally\footnote{There is a single operator built from $\eta_{zy}$ which is independent, namely $\int_{y} \iota_{\partial_z} \eta$.  This is of spin $1$ and is $\SU(2)$ invariant. }  not be independent, as they will have an $n$-derivative and so can be transformed into expressions built from $\eta_{yn}$ and $\eta_{zy}$.

The boundary operators we can build are then the integrals
\begin{equation} 
	\begin{split} 
		E[r,s] &= \int_{y} y^{r} \,\iota_{\partial_n}\partial_n^{r+s-1} \eta + \cdots  \\
		F[r,s] &= \int_{y}  \d y\,y^{r}\,  \partial_n^{r+s+1}  \iota_{\partial_z}  \iota_{\partial_n} \eta + \cdots \label{eqn:boundary_closedstring}
	\end{split}
\end{equation}
where $\iota_\xi\eta\equiv\xi\ip\eta$ denotes interior product with a vector $\xi$. In each equation, we integrate over $n = 0$ at a fixed value of $z$.  In the expression for $E[r,s]$, we must have $r+s > 0$, which is consistent with $E[0,0]$ being the identity operator. The fact that $\iota_{\partial_n}\partial_n^{r+s-1} \eta$ is a $1$-form in the $y$-plane vanishing to order $r+s$ at $y = \infty$ guarantees that the integrand has no poles.  Because $\partial_n$ is of spin $\frac{1}{2}$, this expression has spin $\frac{1}{2}(r+s)$, giving us the correct quantum numbers for $E[r,s]$ as defined on the CFT side.

Similarly, the expression for $F[r,s]$ has spin $\frac{1}{2}(r+s) + 2$, and so has the same quantum numbers as the CFT operator.

As we see from equations \eqref{eqn:boundary_closedstring}, when we work on the flag variety and replace $\partial_n$ by $z^{-1} \partial_{\til{n}}$, $E[r,s]$ has a pole at $z = 0$ of order $r+s$, and $F[r,s]$ has a pole of order $r+s+2$, exactly as desired.

\subsection{Matching correlators on $\til{\F}$ and $\br{\SL}_2(\C)$}
\label{sec:enhance}

One can ask whether the holographic correlators defined using the compactification $\til\F$ coincide with those for the compactification $\br{\SL}_2(\C)$.  In this section we will prove this statement for all tree level holographic correlators, and argue that at loop level the statement remains true. 

Consider placing boundary operators such as those in \eqref{eqn:boundaryintegral} at points $z_i$ where $z_i \neq 0,\infty$.  These source linear fields in the topological string on $\SL_2(\C)$, and the CFT correlation function is given by the topological string scattering amplitude of these field configurations.

It turns out that the fields sourced by these operators are the same whether we construct it using the $\br{\SL}_2(\C)$ boundary condition or the $\til{\F}$ boundary condition. Let us phrase this result as a theorem:
\begin{theorem}
	There is a unique (up to gauge equivalence) field configuration on either $\til \F$ or $\br{\SL}_2(\C)$ sourced by any boundary operator $J[r,s]$, $E[r,s]$, $F[r,s]$ placed at $z \neq 0,\infty$ on $\til \F$. Further, the field configuration on $\til\F$ is gauge-equivalent to the field configuration pulled back from $\br{\SL}_2(\C)$. 
\end{theorem}
\begin{proof}
	Let us first check this for the open string states. We will show that the field sourced by a boundary operator $J[r,s]$ on either $\br{\SL}_2(\C)$ or $\til\F$ is unique up to gauge equivalence. The ambiguity in finding a solution to the equations of motion with a given source term, up to gauge equivalence, is an element of Dolbeault $H^1$ with coefficient in the sheaf of functions vanishing on the boundary. The obstruction to finding a global field sourcing a given boundary operator is an element of Dolbeault $H^2$.  Thus,  it suffices to know that
\begin{equation} 
	\begin{split} 
		H^\ast_{\dbar}(\br{\SL}_2(\C), \Oo ( - \partial\br{\SL}_2(\C)) ) &= 0 \\
		H^\ast_{\dbar}(\til\F, \Oo ( - \br{D}_0 -\br{D}_\infty - E)) &= 0\,. 
	\end{split}
\end{equation}
We check this (easy) vanishing result in Appendix \ref{app:vanishing}.

	For the closed string fields, the argument is a little more complicated.  We let $\Omega^2_\text{cl}$ be the complex of sheaves $\Omega^2 \xrightarrow{\partial} \Omega^3$.  The closed string fields live in 
	\begin{equation} 
		H^1( X, \Omega^2_\text{cl} )  
	\end{equation}
	where $X$ is either $\til{\F}$ or $\br{\SL}_2(\C)$. This is the same as 
	\begin{equation} 
		H^{2,1}(X) \oplus H^{3,0}(X) 
	\end{equation}
	which is easily seen to be zero, because in either case $H^3(X) = 0$.  This tells us that, if a field sourced by a boundary operator exists, it is necessarily unique up to gauge transformations.

	The obstruction to existing lives in
	\begin{equation} 
		H^2(X, \Omega^{2}_\text{cl} ) = H^{2,2}(X) \oplus H^{3,1}(X)  
	\end{equation}
	(and $H^{3,1}(X) = 0$). By the Lefschetz hyperplane theorem and Poincar\'e duality, 
	\begin{equation} 
		H^{2,2}(\br{\SL}_2(\C)) = \C 
	\end{equation}
	with a basis given by the location $\CP^1 \times z$ where we insert a boundary operator. This group contains the obstruction to finding a field sourced by a boundary operator. Similarly, one can show that $H^{2,2}(\til\F) = \C^3$. 

	To complete the proof, we need to check that the obstruction vanishes for the boundary operators $E[r,s]$, $F[r,s]$ in equation \eqref{eqn:boundary_closedstring}.  Let us perform the analysis for $E[r,s]$; the $F[r,s]$ case is similar.    

	Using the kinetic term $\partial^{-1} \eta \wedge \dbar \eta$, we see that the field sourced by $E[r,s](z)$ in \eqref{eqn:boundary_closedstring} satisfies 
	\begin{equation} 
	\dbar \eta =    \partial \left( y^r \iota_{\partial_n} \partial_n^{r+s-1} \bar\delta_{\CP^1\times z}      \right)  \,.\label{eqn:closedsource}	\end{equation}
	This can be solved as long as the $(2,2)$ form on the right hand side is zero in $H^{2,2}(\br{\SL}_2(\C))$.  This is automatic, however, by Hodge theory: the map
	\begin{equation} 
		H^{2,1}(\br{\SL}_2(\C))  \xrightarrow{\partial}  H^{2,2}(\br{\SL}_2(\C))  
	\end{equation}
	is zero, and the right hand side of \eqref{eqn:closedsource} is in the image of this map. 

	This argument applies on $\til\F$ in the same way, so that in each case there is a unique field configuration up to gauge equivalence sourced by a given operator.  Uniqueness implies that the pullback of the field sourced by an operator on $\br{\SL}_2(\C)$ is (up to a gauge transformation) the field sourced by an operator on $\til \F$. 
\end{proof}

This almost immediately tells us that holographic correlators are the same whether computed on $\br{\SL}_2(\C)$ or on $\til{\F}$.  Let us explain this in detail for the open string sector; the argument in the closed string and mixed sectors is identical.    

Let us start with the holographic two-point function of $J[r,s](z_1)$ with $J[s,r](z_2)$.  Let $\mc{F}[r,s](z_1)$ be the field sourced by $J[r,s](z_1)$ on $\br{\SL}_2(\C)$.  Then, on $\br{\SL}_2(\C)$, the holographic two-point function is
\begin{equation} 
	\left\langle J[r,s](z_1) J[s,r](z_2) \right\rangle_{\br{\SL}_2(\C)} = \int_{\CP^1 \times z_2}\d y\; y^s\,  \partial_{n}^{s+m+2} \mc{F}[r,s](z_1)  \,.
\end{equation}
In this expression, we have inserted the field configuration $\mc{F}[r,s](z_1)$ into the linear term of the operator $J[s,r](z_2)$ defined in equation \eqref{eqn:boundaryintegral}.  (Here, we are not concerned with details of normalization: only the structure of the expression is needed to compare the results on $\til\F$ and on $\br{\SL}_2(\C)$).

The corresponding expression on $\til\F$ is almost exactly the same, except that everything has been pulled back along the map $\pi : \til\F \to \br{\SL}_2(\C)$. We get
\begin{equation} 
	\left\langle J[r,s](z_1) J[s,r](z_2) \right\rangle_{\til\F} = \int_{\CP^1 \times z_2} \d y\; y^s\,  \partial_{n}^{s+m+2} \pi^\ast \mc{F}[r,s](z_1)  \,.
\end{equation}
Since the map $\pi$ is an isomorphism on the $\CP^1 \times z_2$ we are considering, these expressions are the same.

Next, let us consider the $3$-point function of open string operators placed at $z_1,z_2,z_3$ where $z_i \neq 0$, $z_i \neq \infty$.  We can take the operators to be $J[r_i,s_i]$, as in equation \eqref{eqn:boundaryintegral}.  We let $\mc{F}[r_i,s_i](z_i)$ be the fields they source. These are field configurations on $\br{\SL}_2(\C)$: 
\begin{equation} 
	\mc{F}[r_i,s_i](z_i) \in \br{\Omega}^{0,1}(\br{\SL}_2(\C), \mf{so}_8 \otimes \Oo(-\partial\br{\SL}_2(\C)) )\,. 
\end{equation}
The notation $\br{\Omega}^{0,1}$ indicates that these are distributional $(0,1)$ forms. It is important that these distributions are regular away from $\CP^1 \times z_i$ in the boundary. These field configurations are uniquely specified by the gauge, linearized equations of motion, boundary conditions, and source term.

The holographic $3$-point function defined in $\br{\SL}_2(\C)$ is
\begin{multline}
	\left\langle J[r_1,s_1](z_1)J[r_2,s_2](z_2)J[r_3,s_3](z_3) \right\rangle _{\br{\SL}_2(\C) } \\= 	\int_{\br{\SL}_2(\C) }  \mc{F}[r_1,s_1](z_1)  \mc{F}[r_2,s_2](z_2)  \mc{F}[r_3,s_3](z_3)\, \Omega 	\label{eqn:3ptsl2} 
\end{multline}
where $\Omega$ is the meromorphic volume form on $\br{\SL}_2(\C)$. 

The holographic $3$-point function on $\til\F$ is defined in exactly the same way:
\begin{multline}
	\left\langle J[r_1,s_1](z_1)J[r_2,s_2](z_2)J[r_3,s_3](z_3)\right\rangle_{\til\F } \\= 	\int_{\til\F }  \pi^\ast \mc{F}[r_1,s_1](z_1) \pi^\ast  \mc{F}[r_2,s_2](z_2)  \pi^\ast\mc{F}[r_3,s_3](z_3)\,\pi^\ast \Omega \,.\label{eqn:3ptflag} 
\end{multline}
All quantitites in the integrand, including the meromorphic volume form, have been pulled back to $\til\F$.

It is almost evident that the expresions \eqref{eqn:3ptsl2} and \eqref{eqn:3ptflag} are the same.  If the $\mc{F}[r_i,s_i](z_i)$ were smooth, then both expressions are given by absolutely convergent integrals on the interior $\SL_2(\C)$, and so they must coincide.   

A slight subtlety is engendered by the fact that $\mc{F}[r_i,s_i](z_i)$ has distributional singularities in the boundary on the locus $\CP^1 \times z_i$.  This is not a serious problem: each integral can be computed on the open region 
\begin{equation} 
	U = \til\F - \br{D}_0 \cup \br{D}_\infty 
\end{equation}
where the locus $z = 0,\infty$ in the boundary has been removed. The map $\pi$ is an isomorphism from $U$ onto its image in $\br{\SL}_2(\C)$, and the integrands in \eqref{eqn:3ptflag} and \eqref{eqn:3ptsl2} are quantities that are smooth and absolutely convergent away from a compact set in $U$. It follows that the integrals agree.

Now let us generalize this to show that the holographic $n$-point functions at tree level coincide whether we use $\til\F$ or $\br{\SL}_2(\C)$, as long as the insertions are of operators with $z \neq 0,\infty$.   We will present the details for the $4$-point function; the general case is similar.  Consider operators $J[r_i,s_i](z_i)$ sourcing fields $\mc{F}[r_i,s_i](z_i)$ on $\br{\SL}_2(\C)$, as above.   The four-point function
\begin{equation} 
	\left\langle J[r_1,s_1](z_1)J[r_2,s_2](z_2)J[r_3,s_3](z_3) J[r_4, s_4](z_4) \right\rangle _{\br{\SL}_2(\C) } 
\end{equation}
has, as usual, $3$ channels.  We will focus on the contribution of one channel, and see that the computation of this channel on $\br{\SL}_2(\C)$ or $\til \F$ yields the same result.

In one channel, the four-point function is
\begin{equation} 
	\int_{\br{\SL}_2(\C) } \dbar^{-1} \bigl\{ \mc{F}[r_1,s_1](z_1) \mc{F}[r_2,s_2](z_2) \bigr\}\, \mc{F}[r_3,s_3](z_3) \mc{F}[r_4,s_4](z_4)\, \Omega\,. 
\end{equation}
The four-point function on $\til\F$ in the same channel is
\begin{equation} 
	\int_{\til\F} \dbar^{-1} \bigl\{ \pi^\ast \mc{F}[r_1,s_1](z_1) \pi^\ast \mc{F}[r_2,s_2](z_2) \bigr\}\, \pi^\ast\mc{F}[r_3,s_3](z_3) \pi^\ast \mc{F}[r_4,s_4](z_4) \,\pi^*\Omega\,.  
\end{equation}
The two integrals will coincide, as long as we can show that
\begin{equation} 
	\dbar^{-1} \left( \pi^\ast \mc{F}[r_1,s_1](z_1) \pi^\ast \mc{F}[r_2,s_2](z_2) \right) = \pi^\ast  \dbar^{-1} \left( \mc{F}[r_1,s_1](z_1) \mc{F}[r_2,s_2](z_2) \right) .
\end{equation}
Here, by $\dbar^{-1}$ we mean any solution of the $\dbar$ equation with a source; we have seen that any two solutions are gauge equivalent. 

This equation is easily seen to hold by applying $\dbar$ to both sides. The point is that $\dbar$ commutes with $\pi^\ast$, as long as $\pi^\ast$ is only applied to distributions whose singularities lie in the locus where the map $\pi$ is a diffeomorphism.

In this way we see that the four-point functions on $\br{\SL}_2(\C)$ and on $\til\F$ match channel by channel. It is straightforward to generalize this to show that holographic $n$-point functions at tree level coincide whether we use $\til\F$ or $\br{\SL}_2(\C)$, as long as the insertions are of operators with $z \neq 0,\infty$.  

It is also natural to conjecture:
\begin{conjecture}
	The holographic correlators for insertions at $z \neq 0,\infty$ for $\til\F$ and $\br{\SL}_2(\C)$ coincide at loop level as well as tree level. \label{conjecture:looplevel}
\end{conjecture}
One can argue for this as follows. As we have seen, the states sourced by boundary operators are the same whether we work on $\til \F$ or on $\br{\SL}_2(\C)$.  The holographic correlator is the scattering amplitude of these states on the non-compact geometry $\SL_2(\C)$.  This scattering amplitude should be independent of the boundary structure, except for the region of the boundary on which the states we are scattering localize.  Therefore, the scattering amplitude should be the same whether we use $\br{\SL}_2(\C)$ or $\til \F$. Of course, a careful proof of this conjecture would require one to find counterterms at loop level that respect such enhanced symmetries.

This conjecture has some quite remarkable consequences.  Firstly, by the Penrose transform, it is immediate that the holographic correlators on twistor space are the same as scattering amplitudes of WZW$_4$ plus Mabuchi gravity on Burns space.  After all, the Penrose transform of the states sourced by the boundary operators on twistor space are single-particle  states on Burns space, and the  topological string Lagrangian on twistor space corresponds as we have seen to the Lagrangian on Burns space. 
Now, on twistor space, if we use the compactification $\br{\SL}_2(\C)$, the holographic dual chiral algebra has \emph{no defects}.  It is simply the BRST reduction of certain symplectic bosons by $\Sp(N)$.  At infinite $N$, this algebra has no conformal blocks.\footnote{This is because all the operators $J[r,s]$, $E[r,s]$, $F[r,s]$ have positive spin. As we saw in Proposition \ref{prop:conformalblocks}, the full space of conformal blocks is the symmetric algebra of the space of possible one-point functions. For operators  of positive spin, one-point functions must vanish, as they have no poles and yet must vanish at $\infty$. }  This means correlation functions are uniquely defined.   

We see that conjecture \ref{conjecture:looplevel} now implies the following:
\begin{conjecture}
	At all orders in $1/N$, the scattering amplitudes for WZW$_4$ plus Mabuchi gravity on Burns space coincide with the (uniquely-defined) correlators for the large $N$ chiral algebra with no defects. 
\end{conjecture}
At tree level, this is a theorem.   This conjecture implies a remarkable simplicity in the scattering amplitudes on Burns space: after all, correlators of the large $N$ chiral algebra are essentially combinatorial in nature, since it is the BRST reduction of a free theory. 

At tree level, this result implies that all $+$ amplitudes for Yang-Mills theory on Burns space are precisely the correlators of open string states of the planar chiral algebra.

There are a few more important statements one can prove at tree level. In \cite{CG}, it was shown in the context of the type II topological B-model that one can match:
\begin{enumerate} 
	\item Single-trace bosonic modes in the planar chiral algebra which preserve the vacuum at $0$ and $\infty$.  This Lie algebra, sometimes called the wedge algebra, is the Lie algebra of modes which preserve all planar correlators.
	\item The Lie algebra $\op{Vect}_0(\SL_2(\C))$ of divergence free holomorphic vector fields on $\SL_2(\C)$.  These are infinitesimal diffeomorphisms of $\SL_2(\C)$ which preserve the complex structure and volume form. As such, they are symmetries of all holographic correlators at tree level. 
\end{enumerate}
In \cite{CG} it was shown that all planar correlation functions in both boundary and bulk theories are essentially completely determined by these symmetries.

A small variant of the arguments of \cite{CG} shows that, in the type I topological string, in both bulk and boundary systems, we have similar infinite dimensional symmetry algebras arising as symmetries of planar correlators. In the open string sector, it is the Lie algebra of holomorphic maps from $\SL_2(\C)$ to $\mf{so}_8$, and in the closed string sector it is again the Lie algebra of divergence-free holomorphic vector fields on $\SL_2(\C)$.  

This tells us that scattering amplitudes on Burns space have not just the $\SU(2) \times \SU(2)$ symmetry mentioned above, but an infinite-dimensional enhancement to $\op{Vect}_0(\SL_2(\C))$.


\section{States in the 4d theory}
\label{sec:dict}

To test our holographic duality in the planar limit, in the sections that follow, we will describe a correspondence between scattering states on Burns space and single trace, gauge invariant operators of the large $N$ chiral algebra. The former will be built as solutions of linearized field equations on Burns space, while the latter have been listed in section \ref{sec:dual} by means of BRST cohomology.


\subsection{States and propagator of WZW$_4$}
\label{sec:wzw4}

To compute scattering amplitudes, one starts with solutions of linearized free field equations to build a space of states of the bulk theory and proceeds to include interactions between the states perturbatively via propagators and interaction vertices. Such linearized solutions are commonly known as \emph{scattering wavefunctions}. To perform this analysis for the ``positive helicity gluon'' states of the WZW$_4$ model, we will solve the Laplace equation for the adjoint-valued scalar $\phi$ defined in \eqref{phidef}. Demanding plane wave asymptotics, we discover the remarkably compact solutions displayed in \eqref{burnsmom} below. The corresponding linearized spin 1 wavefunction may be found by applying $A = -\dbar\phi$. In the latter half of this section, we will also derive the Green's function of the Laplacian on Burns space, which acts as the scalar propagator for $\phi$. The result for this is equation \eqref{phiprop}.

\paragraph{Wavefunctions in WZW$_4$.} As seen from \eqref{phieom}, the linearized free field equation of the adjoint-valued scalar $\phi$ is $\omega\wedge\p\dbar\phi=0$. This is equivalent to the Laplace equation:
\be\label{kahlerlap}
\lap\phi = 4\im\,\frac{\omega\wedge\p\dbar\phi}{\omega^2} = 0\,,
\ee
having recognized the commonly used expression for the Laplacian on K\"ahler 4-manifolds. The Laplacian associated to the Burns metric \eqref{burns} has a simple explicit expression,
\be\label{burnslap}
\lap = \frac{2}{\sqrt{|g|}}\bigg(\eps^{\dal\dot\beta}+\frac{Nu^{\dal}\hat u^{\dot\beta}}{\|u\|^4}\bigg)\frac{\p^2}{\p u^{\dal}\p\hat u^{\dot\beta}}\,,
\ee
where $\sqrt{|g|}=1+N/\|u\|^2$ is the square root of the metric determinant. This inherits the $\U(2)$ symmetries of Burns space.

In what follows, our strategy will be to construct solutions $\phi$ that are expressible as formal Taylor series in $N$ with summands spanning an ``integer basis'' of polynomial solutions. Our polynomial solutions will be the spacetime analogue of the discrete basis of states of holomorphic Chern Simons theory on the deformed conifold found in \cite{CG}. We will start with constructing such a polynomial basis, then assemble infinite linear combinations of the polynomial solutions that resum into a wavefunction with plane-wave-like asymptotics. This is somewhat in the spirit of recent work on celestial holography \cite{Freidel:2022skz}, though our polynomial solutions won't correspond to soft modes by themselves. Instead, they provide a Burns space analogue of the integer basis of flat space states introduced in \cite{Cotler:2023qwh}, which were themselves analogous to quasinormal modes in de Sitter space.

Of course, in practice there are many ways to solve such a Laplace equation. For example, neglecting the prefactor of $\sqrt{|g|}$ that drops from the Laplace equation, the Burns space Laplacian \eqref{burnslap} splits cleanly into a flat space Laplacian plus an order $N$ deformation. Therefore, a more systematic approach would be to decompose $\phi$ into a flat space solution plus an $N$ dependent correction, and solve for the latter perturbatively in $N$ using the Green's function of the flat space Laplacian. This works in greater generality but can be computationally cumbersome and relatively unintuitive. Another alternative would be to apply the Penrose transform \cite{Hitchin:1980hp,Woodhouse:1985id,Ward:1990vs}, which we anticipate to be a powerful tool but leave to future work.

Let us also remark that the perturbative expansions in ``small'' $N$ that we employ below may appear unconventional at first, as we want to take the \emph{large} $N$ limit at the end of the day. However, they will generally either be finite polynomials in $N$, or will take the form of Taylor expansions in $N$ that can be explicitly resummed into expressions analytic in $N$. For instance, our brane backreaction was already seen to be polynomial in $N$, both at the level of the twistor complex structure \eqref{musol} and the Burns metric \eqref{burns}. Subsequently, the coordinate transformation $u^{\dal}\mapsto\sqrt{N}u^{\dal}$ mentioned in \eqref{utoNu} can be applied to systematically map our small $N$ perturbation theory to a $1/\sqrt N$ expansion. The fact that this works appears to be a simplifying characteristic of twisted holography -- as opposed to standard holography -- and has been key to the Koszul duality based determinations of backreaction effects in past work on AdS$_3$ \cite{Costello:2020jbh}. We will see this happen concretely when we come to the computation of scattering amplitudes.

\medskip

For now, let us start by looking for solutions that are polynomials in the complex coordinates $u^{\dal}, \hat u^{\dal}$, diagonalize the left-handed $\U(1)$ isometry generated by the Killing vector $u\cdot\p_u-\hat u\cdot\p_{\hat u}$, and transform in irreps of the unbroken right-handed $\SU(2)$ isometries. To discover such solutions, let us momentarily return to flat space. On $\R^4$, solutions of the Laplace equation with the above properties are given by the totally symmetrized polynomials
\be\label{polsolflat}
u^{(\dal_1}\cdots u^{\dal_k}\hat u^{\dot\beta_1}\cdots\hat u^{\dot\beta_l)}\,,\qquad k,l\in\N\,.
\ee
Since $u^{\dal}=x^{1\dal}$ and $\hat u^{\dal}=x^{2\dal}$ transform in the same spin $\frac12$ representation of the right-handed $\SU(2)$, this set of wavefunctions spans the spin $(k+l)/2$ irrep of $\SU(2)$. For identical values of $k+l$, solutions with differing $k,l$ are distinguished from each other by their eigenvalue $k-l$ under $u\cdot\p_u-\hat u\cdot\p_{\hat u}$. 

To verify that the expressions in \eqref{polsolflat} are annihilated by the flat space Laplacian
\be
\lap_0\equiv2\,\eps^{\dal\dot\beta}\frac{\p^2}{\p u^{\dal}\p\hat u^{\dot\beta}}\,,
\ee
it is useful to introduce an auxiliary spinor $\tilde\lambda_{\dal}\in\C^2$ and democratically contract it into all the $k+l$ indices of \eqref{polsolflat}. This yields a complex wavefunction
\be\label{flatkl}
[u\,\tilde\lambda]^k[\hat u\,\tilde\lambda]^l
\ee
which satisfies the flat space Laplace equation because $\lap_0([u\,\tilde\lambda]^k[\hat u\,\tilde\lambda]^l)\propto[\tilde\lambda\,\tilde\lambda]=0$. Unsurprisingly, it will turn out that this auxiliary spinor is the same as the spinor-helicity variable $\tilde\lambda_{\dal}$ that enters expressions for scattering amplitudes \cite{Elvang:2013cua}.

It is now a straightforward exercise to correct these order-by-order in the backreaction to obtain the analogous solutions on Burns space. We demand that the corrections transform in the same representations of $\U(2)$ as \eqref{flatkl}, so the corrections can necessarily only differ from \eqref{flatkl} by $\U(2)$ invariant factors. The most basic $\U(2)$ invariant is the Euclidean norm of the coordinates, $\|u\|^2 = \frac12\, x^2$, so the new solution should equal \eqref{flatkl} times a factor depending on $\|u\|$. 

These considerations -- along with a bit of hindsight -- motivate us to use the ansatz
\be\label{phiklF}
\phi_{k,l}[\tilde\lambda](x) = \msf{T}_a\,[u\,\tilde\lambda]^k[\hat u\,\tilde\lambda]^l F(\chi)\,,\qquad\chi \vcentcolon= \frac{\|u\|^2}{N}\,,
\ee
where $\msf{T}_a$ is a generator of the Lie algebra $\g$ occurring in WZW$_4$. Plugging this into the Laplace equation on Burns space, we successfully extract an ODE for $F(\chi)$ that is sufficient for our ansatz to work,
\be
\chi^2(\chi+1)\frac{\d^2F}{\d\chi^2} + \chi\big((k+l+2)(\chi+1)-1\big)\frac{\d F}{\d\chi} + klF = 0\,.
\ee
When $k,l\in\N$, this ODE can be solved in terms of Jacobi polynomials,
\be
F(\chi) = \begin{cases}
C\chi^{-l} P_l^{k-l,0}(1+2\chi)\,,\qquad k\geq l\,,\\
    \vspace{-1.2em}\\
C\chi^{-k} P_k^{l-k,0}(1+2\chi)\,,\qquad k<l\,,
\end{cases}
\ee
up to a constant of integration $C$. As we are working with a second order ODE, in principle we also find a second set of solutions. But they involve $\log\chi$, so are not well-behaved in the $\|u\|\to0$ limit. We suspect that they are Burns space analogues of the shadow transformed massless wavefunctions studied in \cite{Pasterski:2017kqt}. We leave their inclusion to future work. 

The Jacobi polynomials $P_n^{a,b}(t)$ are a well-known class of classical orthogonal polynomials. The ones occurring here are given by
\be\label{jackl}
\begin{split}
    P_l^{k-l,0}(1+2\chi) &= \sum_{j=0}^l\frac{(k+j)!}{(k-l+j)!}\,\frac{\chi^j}{j!(l-j)!}\\
    &= \sum_{j=0}^l\frac{(k+l-j)!}{j!(k-j)!(l-j)!}\,\chi^{l-j}\,,
\end{split}
\ee
for $k\geq l$, with the case $k<l$ obtained by exchanging $k$ and $l$. Jacobi polynomials famously enter the Wigner D-matrices of $\SU(2)$ representation theory \cite{Biedenharn:1981er}, so their appearance here is not a total surprise. Curiously, the specific polynomials $P_l^{k-l,0}$ are also related to the Zernike polynomials $R_p^q$ used in optical imaging, e.g.,
\be
R_{k+l}^{k-l}(\sqrt{-\chi}) = (-1)^l(-\chi)^{\frac{k-l}{2}}P_l^{k-l,0}(1+2\chi)\,.
\ee
It is tempting to speculate that Burns space must literally feel like a lens or a holographic medium for the positive helicity photons or gluons traversing through it!

To constrain the integration constant, we impose the boundary condition that our solution reduces to the flat space value \eqref{flatkl} as $\|u\|\to\infty$. At the level of $F(\chi)$, this is the boundary condition
\be
\lim_{\chi\to\infty}F(\chi) = 1\,.
\ee
It is easily solved for $C$ to find
\be
C = \frac{k!\,l!}{(k+l)!}\,,
\ee
valid for either $k\geq l$ or $k<l$. With this choice, the full solutions also happen to admit a uniform expression for all values of $k,l\in\N$,
\be\label{phiklburns}
\phi_{k,l}[\tilde\lambda](x) = \frac{k!\,l!}{(k+l)!}\sum_{j=0}^{\min(k,l)}\frac{(k+l-j)!}{j!(k-j)!(l-j)!}\left(\frac{N}{\|u\|^2}\right)^j[u\,\tilde\lambda]^k[\hat u\,\tilde\lambda]^l\,\msf{T}_a\,,
\ee
having reinstated $\chi=\|u\|^2/N$. As promised, this reduces to \eqref{flatkl} asymptotically as $\|u\|\to\infty$ or in the limit $N\to0$ of no backreaction.

This result gives us another example of the phenomenon where the backreaction turns out to be polynomial in $N$, at least at the level of such building blocks. We now come to more involved wavefunctions that are Taylor series in $N$ but asymptotically behave like the familiar momentum eigenstates from flat space.

\medskip

The most commonly used scattering wavefunctions in flat space are momentum eigenstates labeled by a null momentum $p_{\al\dal}$,
\be
\e^{\im p\cdot x}\,,\qquad p_{\al\dal} = \lambda_\al\tilde\lambda_{\dal}\,.
\ee
As shown here, the nullity constraint $p^2=0$ is trivialized in terms of a pair of constant spinors $\lambda_\al$, $\tilde\lambda_{\dal}$ of opposite chirality. These are known variously as \emph{momentum spinors} or \emph{spinor-helicity variables}. In Euclidean signature, any null momentum must be complex, so the spinor-helicity variables are also generically complex-valued. They are defined up to little group scaling,
\be\label{lgsc}
(\lambda_\al,\tilde\lambda_{\dal})\sim(s\lambda_\al,s^{-1}\tilde\lambda_{\dal})\,,\qquad s\in\C^\times.
\ee
When taken up to an overall energy scale, they act as homogeneous coordinates on the complexified celestial sphere. 

To compute scattering amplitudes on Burns space, we want to construct solutions that asymptote to such plane wave states.
We will do this by expanding $\e^{\im p\cdot x}$ as a double series in $[u\,\tilde\lambda]$ and $[\hat u\,\tilde\lambda]$ and turning on the backreaction term by term. Starting with
\be
\begin{split}
    p\cdot x &= p_{1\dal}x^{1\dal}+p_{2\dal}x^{2\dal}\\
    &= \lambda_1[u\,\tilde\lambda] + \lambda_2[\hat u\,\tilde\lambda]\,,
\end{split}
\ee
we can Taylor expand
\be
\e^{\im p\cdot x} = \sum_{k,l\geq0}\frac{\im^{k+l}\lambda_1^k\lambda_2^l}{k!\,l!}\,[u\,\tilde\lambda]^k[\hat u\,\tilde\lambda]^l\,.
\ee
To deform this into a wavefunction for the adjoint-valued scalar on Burns space, all we need to do is replace each factor of $[u\,\tilde\lambda]^k[\hat u\,\tilde\lambda]^l$ by the polynomial solution $\phi_{k,l}[\tilde\lambda]$ found in \eqref{phiklburns} and resum the resulting series,
\be
\phi(x) = \msf{T}_a\sum_{k,l\geq0}\frac{\im^{k+l}\lambda_1^k\lambda_2^l}{(k+l)!}\sum_{j=0}^{\min(k,l)}\frac{(k+l-j)!}{j!(k-j)!(l-j)!}\left(\frac{N}{\|u\|^2}\right)^j[u\,\tilde\lambda]^k[\hat u\,\tilde\lambda]^l\,.
\ee
Exchanging the sum over $j$ with the sums over $k,l$, then shifting $k\mapsto k+j$, $l\mapsto l+j$, we can reduce this to sums of hypergeometric type,
\be
\begin{split}
    \phi(x) &= \msf{T}_a\sum_{j\geq0}\frac{1}{j!}\biggl(-N\lambda_1\lambda_2\frac{[u\,\tilde\lambda][\hat u\,\tilde\lambda]}{\|u\|^2}\biggr)^j\sum_{k,l\geq0}\frac{(k+l+j)!}{(k+l+2j)!}\frac{(\im\lambda_1[u\,\tilde\lambda])^k(\im\lambda_2[\hat u\,\tilde\lambda])^l}{k!\,l!}\\
    &= \msf{T}_a\sum_{j\geq0}\frac{1}{j!}\biggl(-N\lambda_1\lambda_2\frac{[u\,\tilde\lambda][\hat u\,\tilde\lambda]}{\|u\|^2}\biggr)^j\sum_{n\geq0}\frac{(n+j)!}{(n+2j)!}\frac{(\im p\cdot x)^n}{n!}\,.
\end{split}
\ee
To get the second line, we've partially performed the $k,l$ sum using the binomial theorem and are left with an effective sum over the combination $n=k+l$. The sum over $n$ yields confluent hypergeometric functions,
\be\label{phi1f1}
\phi(x) = \msf{T}_a\sum_{j\geq0}\frac{1}{(2j)!}\biggl(-N\lambda_1\lambda_2\frac{[u\,\tilde\lambda][\hat u\,\tilde\lambda]}{\|u\|^2}\biggr)^j{}_1F_1(j+1,2j+1\,|\,\im p\cdot x)\,.
\ee
This is the form of the wavefunction that we will most commonly use for practical calculations. As $N\to0$, only the $j=0$ term survives which reduces to a plane wave due to the identity ${}_1F_1(1,1\,|\,\im p\cdot x) = \e^{\im p\cdot x}$.

As anticipated, our wavefunction has taken the form of a formal Taylor series in $N$. But we can keep going and completely resum this series to find a result that is analytic in $N$ as well as the spinor-helicity variables. To do this, we first recast the confluent hypergeometric functions occurring in \eqref{phi1f1} in terms of spherical Bessel functions,
\be
{}_1F_1\!\left(j+1,\,2j+1\,|\,\im p\cdot x\right) = \frac{1}{2}\,\frac{(2j)!}{j!}\,\frac{\e^{\frac{\im p\cdot x}{2}}}{(p\cdot x)^{j-1}}\left\{j_{j-1}\biggl(\frac{p\cdot x}{2}\biggr)+\im\,j_{j}\biggl(\frac{p\cdot x}{2}\biggr)\right\}\,.
\ee
This breaks \eqref{phi1f1} into a pair of sums,
\be\label{momstate}
\phi(x) = \msf{T}_a\,\e^{\frac{\im p\cdot x}{2}}\left(S_1+\im\,S_2\right)\,.
\ee
The first term in the brackets here is the sum
\be
S_1 = \frac{p\cdot x}{2}\sum_{j\geq0}\frac{1}{j!}\biggl(-\frac{N\lambda_1\lambda_2}{p\cdot x}\frac{[u\,\tilde\lambda][\hat u\,\tilde\lambda]}{\|u\|^2}\biggr)^j\,j_{j-1}\biggl(\frac{p\cdot x}{2}\biggr)\,.
\ee
Using the generating function of spherical Bessel functions,
\be\label{besgen}
\frac{1}{y}\,\cos\sqrt{y^2-2yt} = \sum_{j\geq0}\frac{t^j}{j!}\,j_{j-1}(y)\,,
\ee
we can resum $S_1$ to find a simple, closed-form expression
\be\label{S1}
S_1 = \cos\sqrt{\frac{(p\cdot x)^2}{4}+N\lambda_1\lambda_2\frac{[u\,\tilde\lambda][\hat u\,\tilde\lambda]}{\|u\|^2}}\,.
\ee
The second term is the sum
\be
S_2 = \frac{p\cdot x}{2}\sum_{j\geq0}\frac{1}{j!}\biggl(-\frac{N\lambda_1\lambda_2}{p\cdot x}\frac{[u\,\tilde\lambda][\hat u\,\tilde\lambda]}{\|u\|^2}\biggr)^j\,j_{j}\biggl(\frac{p\cdot x}{2}\biggr)\,.
\ee
Differentiating \eqref{besgen} with respect to $t$ generates the identity
\be\label{besgen1}
\frac{\sin\sqrt{y^2-2yt}}{\sqrt{y^2-2yt}} = \sum_{j\geq0}\frac{t^j}{j!}\,j_{j}(y)\,.
\ee
Using this, we can resum $S_2$ to find
\be\label{S2}
S_2 = \frac{p\cdot x}{2}\bigg(\frac{(p\cdot x)^2}{4}+N\lambda_1\lambda_2\frac{[u\,\tilde\lambda][\hat u\,\tilde\lambda]}{\|u\|^2}\bigg)^{-\frac{1}{2}}\sin\sqrt{\frac{(p\cdot x)^2}{4}+N\lambda_1\lambda_2\frac{[u\,\tilde\lambda][\hat u\,\tilde\lambda]}{\|u\|^2}}\;.
\ee
Note that neither $S_1$ nor $S_2$ have branch cuts since both $\cos(\theta)$ and $\sin(\theta)/\theta$ are even functions of $\theta$. So we can analytically continue to all values of the argument of the square root.

Putting \eqref{momstate}, \eqref{S1} and \eqref{S2} together gives us our ``quasi-momentum eigenstate'' in closed form,
\be\label{burnsmom}
\phi(x) = \msf{T}_a\,\e^{\frac{\im p\cdot x}{2}}\left\{\cos\biggl(\frac{\psi\, p\cdot x}{2}\biggr)+\frac{\im}{\psi}\,\sin\biggl(\frac{\psi\, p\cdot x}{2}\biggr)\right\}\,,
\ee
where we have introduced the dressing factor
\be\label{psi}
\psi(x) = \sqrt{1+\frac{4N\lambda_1\lambda_2}{(p\cdot x)^2}\frac{[u\,\tilde\lambda][\hat u\,\tilde\lambda]}{\|u\|^2}}
\ee
as a convenient abbreviation. As $N\to0$ or $\|u\|\to\infty$, we see that $\psi\to1$ and our state asymptotes to $\msf{T}_a\,\e^{\im p\cdot x}$. By substituting $x^{\al\dal}=(u^{\dal},\hat u^{\dal}) = (t\,\zeta^{\dal},\bar t\,\hat\zeta^{\dal})$, where $\zeta^{\dal}$ are homogeneous coordinates on $\CP^1$, it is easy to check that this solution is also finite in the $t\to0$ limit,
\be
\lim_{t\to0}\phi(x) = \msf{T}_a\cos\sqrt{\frac{N\lambda_1\lambda_2[\zeta\,\tilde\lambda][\hat\zeta\,\tilde\lambda]}{\|\zeta\|^2}}\,.
\ee
That is, it extends to $\widetilde{\C}^2$, the blow-up of $\C^2$ at the origin. The dressing \eqref{psi} and wavefunction \eqref{burnsmom} are also invariant under the little group scalings \eqref{lgsc}, so we have found a 3-complex parameter family of solutions analogous to plane waves in flat space.

These are the states that we will scatter on Burns space. They are the asymptotically flat analogues of AdS bulk-to-boundary propagators. We will also be able to establish a dictionary matching these states to gauge invariant operators in our dual chiral algebra. That we find such a simple, closed-form solution of the wave equation is a testament to the simplicity that comes with working on self-dual backgrounds; see \cite{Adamo:2020yzi,Adamo:2022mev,Bittleston:2023bzp} for other recent examples of this phenomenon. In fact, it would be desirable to find a simpler, group theoretic derivation of \eqref{burnsmom}, somewhat akin to Witten's derivation of AdS bulk-to-boundary propagators from conformal inversions of power-law wavefunctions \cite{Witten:1998qj}.

\paragraph{A puzzle regarding normalizable modes.} For completeness, let us also mention a curiosity associated to our solutions. Contrary to expectations, we have noticed that states like \eqref{burnsmom} are not the only solutions asymptotic to $\e^{\im p\cdot x}$. One can find further solutions of the Burns space Laplace equation that are not captured by our analysis. 

For example, we have discovered that \eqref{burnsmom} can actually be split into a linear combination of two ``simpler'' solutions,
\be
\phi = \phi_+ + \phi_-\,.
\ee
$\phi_\pm$ are found by replacing the sines and cosines by exponentials,
\begin{align}
    \phi_+(x) &= \msf{T}_a\,\frac{\psi+1}{2\psi}\;\e^{\frac{\im}{2}(1+\psi)p\cdot x}\,,\\
    \phi_-(x) &= \msf{T}_a\,\frac{\psi-1}{2\psi}\;\e^{\frac{\im}{2}(1-\psi)p\cdot x}\,.
\end{align}
It can be verified by direct calculation that both of these separately solve the Burns space Laplace equation.

The first mode $\phi_+$ asymptotes to $\msf{T}_a\,\e^{\im p\cdot x}$ as $\|u\|\to\infty$ or $N\to0$ and behaves like another Burns space analogue of a momentum eigenstate. But interestingly, the second solution $\phi_-$ decays to $0$ as $\|u\|\to\infty$ or $N\to0$. So we can take any solution $\phi$ and add to it any amount of $\phi_-$ without affecting its asymptotic boundary conditions. 

We will be unable to determine any operators dual to the $\phi_\pm$ modes in this work. The main reason for this is that unlike \eqref{burnsmom}, these states acquire a branch cut in the argument of the square root that enters $\psi$, thereby lacking analyticity and any clear twistorial interpretation. The branch cut only disappears in the combination $\phi_++\phi_-$. It would be interesting if the $\phi_-$ states were analogous to the normalizable modes encountered in AdS, although we do not expect them to arise from twistor space. More generally, it would be very interesting to complete such early considerations into a full-fledged harmonic analysis of linearized field equations on Burns space.

\paragraph{Green's function of the Laplacian.} There exists a beautiful trick to derive solutions and Green's functions of conformally invariant equations on Burns space: simply map the equations to $\CP^2$ using the conformal diffeomorphism $u^{\dal}\mapsto\sqrt{N}u^{\dal}/\|u\|^2$ discussed in \eqref{burnstocp2}. This allows us to recycle much of the literature on solving such PDEs on $\CP^2$ and apply it to Burns space. 

On any 4-manifold, the conformally covariant completion of the Laplacian is the \emph{Yamabe operator},
\be
\msf Y_g = \lap - \frac{R}{6}\,,
\ee
also known as the \emph{conformal Laplacian}. It satisfies the conformal transformation law
\be\label{Yamabetrans}
\msf Y_{f^2g}(f^{-1}\phi) = f^{-3}\msf Y_g\phi
\ee
for all positive conformal rescalings $g\mapsto f^2g$. On a scalar-flat geometry like the Burns metric, the Yamabe operator coincides with the ordinary Laplacian. On $\CP^2$, the scalar curvature of the Fubini-Study metric is $R=24$, so the Yamabe operator becomes $\lap-4$. Hence, solutions of $\lap\phi=0$ on Burns space can be found by instead solving $(\lap-4)\phi=0$ on $\CP^2$.

Denote the Burns and Fubini-Study metrics by $g_\text{B}$ and $g_\text{FS}$ respectively. We saw in \eqref{burnstocp2} that the inversion diffeomorphism $ u^{\dal}\mapsto\sqrt{N}u^{\dal}/\|u\|^2$ from $\widetilde\C^2$ to $\CP^2$ minus a point induces the conformal equivalence
\be\label{BFS}
g_\text{B} = f^2g_\text{FS}\,,\qquad f = \sqrt{2N}(1+\|u\|^{-2})
\ee
where $f$ has been expressed in the $u^{\dal}$ coordinates obtained after the inversion. These are affine coordinates on $\CP^2$, in which the Fubini-Study metric reads
\be\label{fubini}
g_\text{FS} = \frac{\|\d u\|^2+|[u\,\d u]|^2}{(1+\|u\|^2)^2}\,.
\ee
A systematic procedure to construct Green's functions of wave equations on projective spaces is described in \cite{Warner:1982fv}. The main idea is to express the Green's function in terms of the geodesic distance. Let us review this briefly for the case of the Yamabe operator.

Let $u_1^{\dal},u_2^{\dal}$ be two points on $\CP^2$. Using the Fubini-Study metric, the geodesic distance between them is calculated to be
\be
\text{dist}(u_1,u_2) = \cos^{-1}\sqrt{L}\,,
\ee
expressed in terms of the quantity
\be
    L = \frac{|1+\bar u^1_1u_2^{\dot1}+\bar u_1^2 u_2^{\dot2}|^2}{(1+\|u_1\|^2)(1+\|u_2\|^2)} = \frac{(1+[\hat u_1u_2])(1+[\hat u_2u_1])}{(1+\|u_1\|^2)(1+\|u_2\|^2)}
\ee
where $\hat u_i^{\dal} = (-\bar u_i^2,\bar u_i^1)$ as before. Motivated by this, if we use the ansatz
\be
G_\text{FS}(u_1,u_2) \equiv \mathcal{G}(L)
\ee
for the Green's function, then away from $u_1^{\dal}=u_2^{\dal}$ the conformally coupled Laplace equation $(\lap_{u_1}-4)G_\text{FS}(u_1,u_2)=0$ reduces to an ODE,
\be
L(L-1)\frac{\d^2\mathcal{G}}{\d L^2} + (3L-1)\frac{\d\mathcal{G}}{\d L} + \mathcal{G} = 0\,.
\ee
This has the general solution
\be
\mathcal{G}(L) = \frac{C_1}{1-L}+\frac{C_2\log L}{1-L}\,,
\ee
with $C_1,C_2$ being constants of integration. 

The only singularity that a genuine Green's function can have is the short distance singularity at $L=1$ when $u_1^{\dal}=u_2^{\dal}$. So we drop the logarithmic solution which is non-singular at $L=1$ but has an unphysical singularity at $L=0$. Choosing the normalization $C_1=-1/4\pi^2$ then ensures that it solves $(\lap_{u_1}-4)G_\text{FS}(u_1,u_2)=|g_\text{FS}|^{-\frac12}\delta^4(u_1-u_2)$. This gives us the Green's function of the Yamabe operator associated to the Fubini-Study metric,
\be
G_\text{FS}(u_1,u_2) = -\frac{1}{4\pi^2}\frac{1}{1-L}\,.
\ee
This was originally found in \cite{Hawking:1979pi} and used to study transition amplitudes of conformally coupled scalars traversing through spacetime foam built from $\CP^2$ vacuum bubbles. Due to conformal covariance, this happens to be very closely related to scattering of ordinary scalars on Burns space.

To return to Burns space, one simply dresses the conformally coupled scalar with the conformal factor $f(u)=\sqrt{2N}(1+\|u\|^{-2})$ in accordance with the transformation law \eqref{Yamabetrans}. Following this, one undoes the conformal diffeomorphism. Being an inversion, the inverse of this diffeomorphism is again $u^{\dal}\mapsto\sqrt{N}u^{\dal}/\|u\|^2$. So, to obtain the Green's function of the Laplacian on Burns space, we start with the conformally rescaled scalar solution
\be
\frac{G_\text{FS}(u_1,u_2)}{f(u_1)f(u_2)}
\ee
and map $u^{\dal}\mapsto\sqrt{N}u^{\dal}/\|u\|^2$. Dressing the result with a color-trace, we land on the propagator for the adjoint-valued scalar $\phi$ on Burns space,
\be\label{phiprop}
G_{ab}(u_1,u_2) = -\frac{\tr(\sT_a\sT_b)}{8\pi^2}\left(\|u_1-u_2\|^2+\frac{N|[u_1u_2]|^2}{\|u_1\|^{2}\|u_2\|^{2}}\right)^{-1}\,,
\ee
where the factor of $|[u_1u_2]|^2=[u_1u_2][\hat u_1\hat u_2]$ has been obtained by writing $\|u_i\|^2=[\hat u_i u_i]$ and applying Schouten's identity $[\hat u_1u_1][\hat u_2u_2]-[\hat u_1u_2][\hat u_2u_1] = [u_1u_2][\hat u_1\hat u_2]$. 

In the limit $u_1^{\dal}\to u_2^{\dal}$, \eqref{phiprop} has the expected short distance singularity
\be
G_{ab}(u_1,u_2) \sim
-\frac{\tr(\sT_a\sT_b)}{4\pi^2}\,\bigl(g_{\mu\nu}(x_2)x_{12}^\mu x_{12}^\nu\bigr)^{-1} + \text{non-singular in }x_{12}^\mu\,,
\ee
where $g$ now denotes the Burns metric, $x_{12}^\mu \vcentcolon= x_1^\mu - x_2^\mu$, and we have reinstated the original double null coordinates $x_i^{\al\dal}=(u_i^{\dal},\hat u_i^{\dal})$ using \eqref{utox}. Since it solves the Laplace equation everywhere other than $u_1^{\dal}=u_2^{\dal}$, this boundary condition is then sufficient to conclude that it solves
\be
\lap_{u_1}G_{ab}(u_1,u_2) = \frac{1}{\sqrt{|g|}}\,\delta^4(x_1-x_2)\,\tr(\sT_a\sT_b)
\ee
on Burns space, where the metric determinant can be evaluated at either point as usual.


\subsection{States of Mabuchi gravity}
\label{sec:gravity}

Similarly, we can determine the scattering states of our gravitational sector. The linearized field equation of the K\"ahler scalar perturbation $\rho$ follows from varying just the quadratic terms in \eqref{rhoac},
\be\label{rholin}
\lap^2\rho - \frac{8\im P\wedge\p\dbar\rho}{\omega^2} = 0\,.
\ee
This involves the \emph{Paneitz operator}: a conformally covariant fourth order operator which reduces to the squared Laplacian on $\R^4$ \cite{paneitz:2008}; see \cite{Bailey:1990qn} for a review. On a general 4-manifold with Riemannian metric $g$, it can be expressed in terms of the Laplacian $\lap$, the Levi-Civita connection $\nabla_\mu$, and the Ricci and scalar curvatures $R_{\mu\nu}$, $R$ as
\be
\mathsf{P}_g = \lap^2 + 2\,\nabla_\mu\biggl(R^{\mu\nu}-\frac{R}{3} \,g^{\mu\nu}\biggr)\nabla_\nu\,.
\ee
Under a conformal transformation $g\mapsto f^2g$, it obeys
\be\label{Paneitztrans}
\msf{P}_{f^2g}\rho = f^{-4} \msf{P}_g\rho\,.
\ee
The linearized equation \eqref{rholin} is precisely the conformally invariant equation
\be
\msf{P}_g\rho=0
\ee
specialized to a scalar-flat K\"ahler 4-manifold with K\"ahler form $\omega$, Ricci form $P$ and Laplacian $\lap$. 

\paragraph{Wavefunctions for K\"ahler perturbations.} In flat space, \eqref{rholin} reduces to $\lap_0^2\rho=0$. This is trivially solved by the usual momentum eigenstate $\rho=\e^{\im p\cdot x}$ when $p_{\al\dal}$ is null. But, being a fourth order equation, it also admits an additional solution $\rho=x^2\e^{\im p\cdot x}$. Our gravitational modes will be described by the analogous solutions on Burns space.

The procedure for solving \eqref{rholin} is technically identical to the case of the Laplace equation dealt with in the previous section. 
The Laplacian on Burns space was displayed in \eqref{burnslap}, while the K\"ahler and Ricci forms can be found in \eqref{omegaburns} and \eqref{ricciburns} respectively. But this presentation of the linearized field equation is somewhat cumbersome to work with. In fact, it will prove much wiser to start with the Paneitz operator on $\CP^2$ and conformally transform it back to Burns space to obtain a more tractable version of \eqref{rholin}.

Since the Fubini-Study metric \eqref{fubini} is Einstein, its Ricci curvature is pure trace
\be
R^{\mu\nu} = 6 g_\text{FS}^{\mu\nu}\,.
\ee
As mentioned before, the corresponding scalar curvature is $R=24$. Therefore, the Paneitz operator on $\CP^2$ is very simple and can be expressed in terms of its Yamabe operator,
\begin{align}
    &\msf{P}_{g_\text{FS}} = \lap^2-4\lap = (\lap-4)(\lap-4+4)\nonumber\\
    &\implies\msf{P}_{g_\text{FS}} = \msf{Y}_{g_\text{FS}}\bigl(\msf{Y}_{g_\text{FS}}+4\bigr)\,.\label{PtoY}
\end{align}
Following \eqref{BFS}, denote the conformal equivalence between the Burns and Fubini-Study metrics by $g_\text{B} = f^2 g_\text{FS}$. In the coordinates $u^{\dal}$ on $\CP^2$, we had $f = \sqrt{2N}(1+\|u\|^{-2})$. Under the inverse map $u^{\dal}\mapsto\sqrt{N}u^{\dal}/\|u\|^2$ from $\CP^2$ to Burns space, this maps to
\be\label{rescaling}
f = \sqrt{2N}\left(1+\frac{\|u\|^2}{N}\right)\,.
\ee
Next, the left hand side of \eqref{PtoY} can be mapped to Burns space using the conformal transformation law for the Paneitz operator, equation \eqref{Paneitztrans}, whereas its right hand side can be independently transformed using the transformation law of the Yamabe operator, equation \eqref{Yamabetrans}. This leads to the identity
\be
f^{4}\msf{P}_{g_\text{B}}\rho = f^3\msf{Y}_{g_\text{B}}\!\left\{f^{-1}\cdot f^3\msf{Y}_{g_\text{B}}(f^{-1}\rho) + 4f^{-1}\rho\right\}\,.
\ee
Since the Yamabe operator on Burns space is just the Laplacian, this produces a simpler, factorized expression for the free field equation \eqref{rholin}:
\be\label{rholin1}
\msf{P}_{g_\text{B}}\rho = \frac{1}{f}\,
\lap\bigg[f^2\lap\bigg(\frac{\rho}{f}\bigg) + \frac{4\rho}{f}\bigg] = 0
\ee
with $f$ given by \eqref{rescaling}. Of course, \eqref{rholin1} can also be directly verified to match \eqref{rholin}.

Generating solutions to this is now straightforward. We rewrite \eqref{rholin1} as a pair of second order PDEs in terms of $\rho$ and an auxiliary scalar field $\tilde\rho$,
\begin{align}
    &f^2\lap\bigg(\frac{\rho}{f}\bigg) + \frac{4\rho}{f} = \tilde\rho\,,\label{rhoeq1}\\
    &\lap\tilde\rho = 0\,.\label{rhoeq2}
\end{align}
To find quasi-momentum eigenstates for $\rho$, it is natural to take $\tilde\rho$ to be the solution of the Laplace equation discovered in \eqref{burnsmom},
\be
\tilde\rho(x) = 4\sqrt{2N}\,\e^{\frac{\im p\cdot x}{2}}\left\{\cos\biggl(\frac{\psi\, p\cdot x}{2}\biggr)+\frac{\im}{\psi}\,\sin\biggl(\frac{\psi\, p\cdot x}{2}\biggr)\right\}\,,
\ee
having dressed it with a judicious normalization while dropping the color factor $\sT_a$. Plugging this into the second order PDE \eqref{rhoeq1} should give rise to a pair of solutions for $\rho$. At least one of these can be found trivially: if $\tilde\rho$ solves the Laplace equation \eqref{rhoeq2}, then the profile $\rho=f\tilde\rho/4$ automatically solves \eqref{rhoeq1}. This shows that
\be\label{rhophi}
\rho(x) = 2\,\e^{\frac{\im p\cdot x}{2}}\left(\|u\|^2+N\right)\left\{\cos\biggl(\frac{\psi\, p\cdot x}{2}\biggr)+\frac{\im}{\psi}\,\sin\biggl(\frac{\psi\, p\cdot x}{2}\biggr)\right\}
\ee
is one of the solutions of \eqref{rholin} that we seek. As we send $\|u\|\to\infty$ with $N$ fixed, this solution asymptotes to $\rho\sim x^2\e^{\im p\cdot x}$. It also follows from \eqref{phi1f1} that this solution admits an expansion in $N$ with ${}_1F_1$ coefficients, which may be useful for future calculations.

We now look for a second solution that asymptotes to just $\e^{\im p\cdot x}$. Let us again use an ansatz of the form \eqref{phiklF}, now involving two undetermined functions of $\chi = \|u\|^2/N$,
\begin{align}
    \rho_{k,l}[\tilde\lambda](x) &= [u\,\tilde\lambda]^k[\hat u\,\tilde\lambda]^lF(\chi)\,,\\
    \tilde\rho_{k,l}[\tilde\lambda](x) &= [u\,\tilde\lambda]^k[\hat u\,\tilde\lambda]^l\tilde F(\chi)\,.
\end{align}
Substituting these in \eqref{rhoeq1}, \eqref{rhoeq2} yields a coupled system of second order ODEs,
\begin{align}
    &\chi^2(\chi+1)\frac{\d^2F}{\d\chi^2} + \chi\big((k+l)(\chi+1)+1\big)\frac{\d F}{\d\chi} + \big(kl-(k+l)\chi\big) F = \frac{\sqrt{2N}}{4}\,\chi\tilde F\,,\label{Feq}\\
    &\chi^2(\chi+1)\frac{\d^2\tilde F}{\d\chi^2} + \chi\big((k+l+2)(\chi+1)-1\big)\frac{\d \tilde F}{\d\chi} + kl\tilde F = 0\,.\label{Ftildeeq}
\end{align}
As mentioned in the previous section, the second of these is solved by the polynomial profiles
\be\label{Ftsol}
\tilde F(\chi) = \begin{cases}
\tilde C\chi^{-l} P_l^{k-l,0}(1+2\chi)\,,\qquad k\geq l\,,\\
    \vspace{-1.2em}\\
\tilde C\chi^{-k} P_k^{l-k,0}(1+2\chi)\,,\qquad k<l\,,
\end{cases}
\ee
with $\tilde C$ a constant of integration, and $P_l^{k-l,0}$ the Jacobi polynomials displayed in \eqref{jackl}.

When \eqref{Ftsol} is seeded into the right hand side of \eqref{Feq}, the resulting ODE for $F(\chi)$ can be solved by series expansions or the method of Green's functions. To be brief, we only quote the result. The general polynomial solution is found to be\footnote{Again, since we are really dealing with a fourth order equation, there also exist a pair of logarithmic solutions which we have neglected here to ensure finiteness as $\|u\|\to0$.}
\be\label{Fsol}
F(\chi) = \begin{cases}
\displaystyle C\chi^{-l}P_{l+1}^{k-l,-2}(1+2\chi) - \frac{\sqrt{2N}}{4k}\,\tilde C\chi^{-l} P_l^{k-l,-1}(1+2\chi)\,,\qquad k\geq l\,,\\
    \vspace{-1em}\\
\displaystyle C\chi^{-k}P_{k+1}^{l-k,-2}(1+2\chi)-\frac{\sqrt{2N}}{4l}\,\tilde C\chi^{-k} P_k^{l-k,-1}(1+2\chi)\,,\qquad k<l\,,
\end{cases}
\ee
where $C$ is another independent constant of integration that comes up in solving the homogeneous part of \eqref{Feq}. 
The Jacobi polynomials constituting these solutions are
\begin{align}
     P_l^{k-l,-1}(1+2\chi) &= k\sum_{j=0}^l\frac{(k+l-j-1)!}{j!(k-j)!(l-j)!}\,\chi^{l-j}\,,\label{P-1}\\
     P_{l+1}^{k-l,-2}(1+2\chi) &= k(k+1)\sum_{j=0}^{l+1}\frac{(k+l-j)!}{j!(k+1-j)!(l+1-j)!}\,\chi^{l+1-j}\,.\label{P-2}
\end{align}
We find the expected behaviors. If we set $C=0$, then $F(\chi)\sim \mathrm{O}(1)$ as $\chi\to\infty$. On the other hand, if we set $\tilde C=0$, then $F(\chi)\sim \mathrm{O}(\chi)$ in this limit (unless $k=l=0$).\footnote{When $k=l=0$, both solutions shown in \eqref{Fsol} degenerate to the constant solution, and by some simple guesswork a second non-logarithmic solution of \eqref{rholin} is independently found to be $\rho=\|u\|^2$.} The complete solutions then asymptote to $[u\,\tilde\lambda]^k[\hat u\,\tilde\lambda]^l$ and $\|u\|^2[u\,\tilde\lambda]^k[\hat u\,\tilde\lambda]^l$ respectively.

Let us focus on the former class of solutions, as they will allow us to deform $\e^{\im p\cdot x}$ into a solution for $\rho$ on Burns space. The boundary condition $\lim_{\chi\to\infty}F(\chi) = 1$ fixes
\be
C=0\,,\qquad \tilde C = -\frac{4}{\sqrt{2N}}\frac{k!\,l!}{(k+l-1)!}\,.
\ee
To be clear, we are choosing $\tilde C=0$ for $k=l=0$ by using $(k+l-1)!=\Gamma(k+l)$. Putting this together with \eqref{P-1} and \eqref{Fsol}, we obtain a uniform solution for the wavefunction asymptotic to $[u\,\tilde\lambda]^k[\hat u\,\tilde\lambda]^l$ for all $k,l\in\N$,
\be
\rho_{k,l}[\tilde\lambda](x) = \frac{k!\,l!}{(k+l-1)!}\sum_{j=0}^{\min(k,l)}\frac{(k+l-j-1)!}{j!(k-j)!(l-j)!}\left(\frac{N}{\|u\|^2}\right)^j[u\,\tilde\lambda]^k[\hat u\,\tilde\lambda]^l\,.
\ee
By our arguments, this solves the fourth order linearized field equation \eqref{rholin1}. Yet the difference between this and the solution \eqref{phiklburns} of the Laplace equation is remarkably minimal. Nevertheless, being careful here will have dramatic effects when we come to computing amplitudes involving gravity. 

Resumming these into a solution with plane wave asymptotics is also procedurally identical to the WZW$_4$ case. As a Taylor series in $N$, such a solution is given by
\be
\begin{split}
    \rho(x) &= \sum_{k,l\geq0}\frac{\im^{k+l}\lambda_1^k\lambda_2^l}{k!\,l!}\,\rho_{k,l}[\tilde\lambda](x)\\
    &= \sum_{k,l\geq0}\frac{\im^{k+l}\lambda_1^k\lambda_2^l}{(k+l-1)!}\sum_{j=0}^{\min(k,l)}\frac{(k+l-j-1)!}{j!(k-j)!(l-j)!}\left(\frac{N}{\|u\|^2}\right)^j[u\,\tilde\lambda]^k[\hat u\,\tilde\lambda]^l\,.
\end{split}
\ee
Exchanging the $k,l$ and $j$ sums, followed by the shift $k\mapsto k+j$, $l\mapsto l+j$, yields the rearrangement
\be
    \rho(x) = \sum_{j\geq0}\frac{1}{j!}\biggl(-N\lambda_1\lambda_2\frac{[u\,\tilde\lambda][\hat u\,\tilde\lambda]}{\|u\|^2}\biggr)^j\sum_{k,l\geq0}\frac{(k+l+j-1)!}{(k+l+2j-1)!}\frac{(\im\lambda_1[u\,\tilde\lambda])^k(\im\lambda_2[\hat u\,\tilde\lambda])^l}{k!\,l!}\,.
\ee
Summing over $k,l$ now yields an expansion in confluent hypergeometric functions,
\be\label{Erhoexp}
    \rho(x) = \e^{\im p\cdot x} + \sum_{j\geq1}\frac{2}{(2j)!}\biggl(-N\lambda_1\lambda_2\frac{[u\,\tilde\lambda][\hat u\,\tilde\lambda]}{\|u\|^2}\biggr)^j{}_1F_1(j,2j\,|\,\im p\cdot x)\,.
\ee
We have written out the $j=0$ term of the sum separately to avoid any notational confusion about the meaning of ${}_1F_1(0,0\,|\,\im p\cdot x)$. 

To perform the final sum over $j$, we use the formula
\be
{}_1F_1(j,2j\,|\,\im p\cdot x) = \frac{1}{2}\,\frac{(2j)!}{j!}\,\frac{\e^{\frac{\im p\cdot x}{2}}}{(p\cdot x)^{j-1}}\;j_{j-1}\biggl(\frac{p\cdot x}{2}\biggr)
\ee
to express the summands in terms of spherical Bessel functions:
\be
\rho(x) = \e^{\im p\cdot x} + \e^{\frac{\im p\cdot x}{2}}p\cdot x\sum_{j\geq1}\frac{1}{j!}\biggl(-\frac{N\lambda_1\lambda_2}{p\cdot x}\frac{[u\,\tilde\lambda][\hat u\,\tilde\lambda]}{\|u\|^2}\biggr)^j\;j_{j-1}\bigg(\frac{p\cdot x}{2}\bigg)\,.
\ee
Using the generating functional \eqref{besgen} of spherical Bessel functions, this resums to
\be
\rho(x) = \e^{\im p\cdot x} + 2\,\e^{\frac{\im p\cdot x}{2}}\biggl(\cos\sqrt{\frac{(p\cdot x)^2}{4}+N\lambda_1\lambda_2\frac{[u\,\tilde\lambda][\hat u\,\tilde\lambda]}{\|u\|^2}} - \cos\frac{p\cdot x}{2}\biggr)\,.
\ee
Writing $\cos(p\cdot x/2) = \frac{1}{2}(\e^{\im p\cdot x/2}+\e^{-\im p\cdot x/2})$ simplifies this to its final form
\be\label{Erhosol0}
\rho(x) = 2\,\e^{\frac{\im p\cdot x}{2}}\cos\biggl(\frac{\psi\, p\cdot x}{2}\biggr)-1\,,
\ee
where $\psi$ is the same dressing factor encountered in \eqref{psi}. As $N\to0$ or $\|u\|\to\infty$, the dressing $\psi\to1$ and this wavefunction limits to $\e^{\im p\cdot x}$. At the same time, setting $u^{\dal}=t\,\zeta^{\dal}$ and taking $t\to0$ yields a finite limit, confirming that this mode extends to the blow-up $\widetilde\C^2$.

Finally, recall that $\rho$ is only defined up to constant shifts. We choose to subtract an extra factor of $1$ from it and work with the solution
\be\label{Erhosol}
\rho(x) = 2\,\e^{\frac{\im p\cdot x}{2}}\cos\biggl(\frac{\psi\, p\cdot x}{2}\biggr)-2\,,
\ee
This asymptotes to $\e^{\im p\cdot x}-1$ which has no constant term in its ``soft'' expansion in $p\cdot x$. This will later be interpreted as the lack of an $E[0,0]$ current in the holographic dual. In total, the two solutions \eqref{Erhosol} and \eqref{rhophi} together capture linear perturbations of the scalar-flat K\"ahler geometry of Burns space.



\section{The holographic dictionary on spacetime}
\label{sec:ops}

Precision tests of any holographic duality require a \emph{holographic dictionary}: a map from states of the bulk theory to operators of the boundary dual. Once such a map is set up, holography tells us that scattering amplitudes of bulk states equal correlators of the boundary operators. In the previous sections, we have described both the gauge invariant local operators of our chiral algebra in the large $N$ limit, as well as wavefunctions of gluon and K\"ahler perturbations on Burns space. In this section, we will link them with a holographic dictionary.

\subsection{Open string modes} 

On one hand, string states of the B-model on twistor space descend to states of the WZW$_4$ model and Mabuchi gravity on spacetime. On the other hand, holography relates them to chiral algebra operators. Combining these facts, we can associate spacetime states with their dual operators.

Classically, the group-valued field $\sg=\e^\phi$ of the WZW$_4$ model encodes the gauge field $A=-\dbar\sg\,\sg^{-1}$ on spacetime. In linear theory, scalar wavefunctions for $\phi$ give rise to positive helicity gluon wavefunctions on Burns space, $A=-\dbar\phi + \mathrm{O}(\phi^2)$. In the flat space limit $N\to0$, they should reduce to the positive helicity gluon momentum eigenstates
\be
A = -\frac{\im\tilde\lambda_{\dal}\d\hat u^{\dal}}{\lambda_1}\,\e^{\im p\cdot x}\,\sT_a
\ee
where $p_{\al\dal}=\lambda_\al\tilde\lambda_{\dal}$ is a flat space null momentum, and $\sT_a$ is a color factor as before. Such a state scales with weight $-2$ under little group scalings \eqref{lgsc}, as desired from positive helicity gluons. The self-dual part of its linearized field strength is canonically normalized,
\be
\tilde F_{\dal\dot\beta} = \p_{\al(\dal}A^\al{}_{\dot\beta)}
 = \tilde\lambda_{\dal}\tilde\lambda_{\dot\beta}\,\e^{\im p\cdot x}\,\sT_a\,.
\ee
The flat space solution of $\lap_0\phi=0$ that gives rise to this spin 1 wavefunction via $A=-\dbar\phi$ has the profile
\be\label{flatphi}
\phi = \frac{\sT_a\,\e^{\im p\cdot x}}{\lambda_1\lambda_2}\,.
\ee
Therefore, on Burns space, we find it natural to normalize the quasi-momentum eigenstates such that they reduce to \eqref{flatphi} as $N\to0$ or $\|u\|\to\infty$:
\be\label{phidict}
\phi(x) = \phi_a(x|z,\tilde\lambda) \equiv \frac{\sT_a\,\e^{\frac{\im p\cdot x}{2}}}{\lambda_1\lambda_2}\left\{\cos\biggl(\frac{\psi\, p\cdot x}{2}\biggr)+\frac{\im}{\psi}\,\sin\biggl(\frac{\psi\, p\cdot x}{2}\biggr)\right\}\,,
\ee
where the dressing factor $\psi$ was written out in \eqref{psi}. Just as for $A$, these states for $\phi$ (in flat as well as Burns space) scale with little group weight $-2$, which provides a useful check of many amplitude calculations.

For our holographic dictionary, we propose to associate the gauge invariant currents $J_{rs}[k,l](z)$ displayed in \eqref{CFT} to elements in a soft expansion of \eqref{phidict}. It is easiest to describe this by fixing the little group freedom to set
\be
\lambda_\al = (1,z)\,,\qquad\tilde\lambda_{\dal} = \omega(1,\bz)\,.
\ee
Here, $z,\bz$ are independent complex variables in Euclidean signature (they only become complex conjugates of each other when analytically continued to Lorentzian signature). $\omega\in\C$ is the energy of the state. A soft expansion of \eqref{phidict} is a Taylor expansion in $\omega,\bz$ of the form\footnote{With our choice of little group fixing, there is no Weinberg soft pole of order $\omega^{-1}$. Instead, the expansion of \eqref{phidict} around $\omega=0$ starts at order $\omega^0$.}
\be\label{phisoft}
\begin{split}
    \phi_a(x|z,\tilde\lambda) &= \sum_{p=0}^\infty\omega^p\sum_{k=0}^p\frac{\bz^{p-k}}{k!(p-k)!}\,\phi[k,p-k](x|z)\,\sT_a\\
    &= \sum_{k,l=0}^\infty\frac{\tilde\lambda_{\dot1}{}^k\tilde\lambda_{\dot2}{}^l}{k!\,l!}\,\phi[k,l](x|z)\,\sT_a\,.
\end{split}
\ee
The wavefunctions $\phi[k,l](x|z)\,\sT_a$ are solutions of the Burns space Laplace equation and will be referred to as \emph{soft gluon currents} in analogy with positive helicity soft gluons in flat space. We will not obtain a general formula for these Taylor coefficients, but content ourselves with listing out a few leading examples
\be\label{softexamples}
\begin{split}
    &\phi[0,0] = \frac{1}{z}\,,\quad\phi[1,0] = \frac{\im}{z}\,(u^{\dot1}-z\bar u^{2})\,,\quad \phi[0,1] = \frac{\im}{z}\,(u^{\dot2}+z\bar u^{1})\,,\\
    &\phi[2,0] = z\,\phi[1,0]^2 + \frac{Nu^{\dot1}\bar u^2}{\|u\|^2}\,,\quad\phi[0,2] = z\,\phi[0,1]^2 - \frac{Nu^{\dot2}\bar u^1}{\|u\|^2}\,, \\
    &\phi[1,1] = z\,\phi[1,0]\,\phi[0,1] - \frac{N}{2}\,\frac{|u^{\dot1}|^2-|u^{\dot2}|^2}{\|u\|^2}\,,\quad\text{etc.}
\end{split}
\ee
More generally, one finds $\phi[k,l]=z^{k+l-1}\phi[1,0]^k\phi[0,1]^l + \text{corrections polynomial in $N$}$. 

These give a taste of how positive helicity soft gluons get deformed when backreaction is turned on. Twisted holography again ensures that the deformations of such soft modes are polynomial in $N$. It is also curious to note that deformations only start from the sub-subleading order. The leading and subleading soft gluon wavefunctions receive no backreaction!

Trading adjoint indices $a$ for composite indices $rs$ of the antisymmetric square of the fundamental of $\SO(8)$, we observe that the currents $J_{rs}(z) = I_r(X^{\dot1})^{(k}(X^{\dot2})^{l)}I_s$ are in one-to-one correspondence with these soft gluon wavefunctions. It then seems natural to identify the holographic dictionary to be
\be
\phi[k,l](x|z)\,\sT_{rs}\quad\overset{?}{\longleftrightarrow}\quad J_{rs}[k,l](z)\,.
\ee
This happens to be a bit naive, since we are allowing the currents on the right to have poles of order up to $k+l+1$ in a general 2d vacuum, whereas the wavefunctions on the left clearly only have simple poles in $z$. The latter are displayed in \eqref{softexamples}. 

A simple way to account for this is to cancel the poles on the right by multiplying with a factor of $z^{k+l}$. This leads us to propose the dictionary\footnote{If one works in homogeneous coordinates $\lambda_\al=(1,z)$, the current $J[k,l]$ has homogeneity $-k-l-2$ in $\lambda_\al$. Then one can also motivate multiplication with $(\lambda_1\lambda_2)^{k+l}=z^{k+l}$ by matching little group weights.}
\be
\phi[k,l](x|z)\,\sT_{rs}\quad\longleftrightarrow\quad z^{k+l}\,J_{rs}[k,l](z)\,.
\ee
We can also define a (non-conformal-primary) ``hard'' operator in analogy with the soft expansion \eqref{phisoft}
\be
\begin{split}
    J_{rs}(z,\tilde\lambda) &= \sum_{k,l=0}^\infty\frac{\tilde\lambda_{\dot1}{}^k\tilde\lambda_{\dot2}{}^l}{k!\,l!}\,z^{k+l}\,J_{rs}[k,l](z)\\
    &= I_r\e^{z[X\,\tilde\lambda]}I_s(z)\,,
\end{split}
\ee
having resummed it in terms of $[X\,\tilde\lambda]=X^{\dot1}\tilde\lambda_{\dot1}+X^{\dot2}\tilde\lambda_{\dot2}$ using the binomial theorem. The soft gluon dictionary then implies that this operator must be dual to the hard gluon wavefunction in the bulk:
\be\label{opendict}
\phi_{rs}(x|z,\tilde\lambda)\quad\longleftrightarrow\quad J_{rs}(z,\tilde\lambda)\,.
\ee
This dictionary will identify scattering amplitudes of hard states on the left with correlators of the operators on the right in the large $N$ limit. 

We expect that a very similar dictionary will continue to hold also at finite $N$, although the basis of gauge invariant operators at finite $N$ may need to be reevaluated to account for trace relations.

We also remark that what we have called soft gluon currents $J_{rs}[k,l]$ are not immediately the same as the $S$-algebra currents studied by Strominger \cite{Strominger:2021mtt}. There is an important distinction to be made with regards their conformal weights: the $S$-algebra currents have mostly negative conformal weights (see also \cite{Costello:2022wso}), whereas the $J[k,l]$ have positive weights. In our CCFT dual, the currents $J[k,l] = I(X^{\dot1})^{(k}(X^{\dot2})^{l)}I$ have weights $h=\frac12(k+l+2)$. In the presence of our defects, these can be mapped to currents $z^{k+l+1}J[k,l]$ that transform as primaries of weight $h=-\frac12(k+l+2)$ under the defect conformal group $z\mapsto az$, $a\in\C^\times$. It is these currents that enter the soft expansion of the hard modes, and it is these currents that can genuinely be identified as Burns space analogues of the $S$-algebra currents by a simple relabeling:
\be
S^p_m(z) = z^{2p-1}J[p-1+m,p-1-m](z)\,,\qquad p\in\left\{1,\frac{3}{2},2,\cdots\right\}\,,\;|m|\leq p-1\,.
\ee
The presence of the defects means that the OPEs of $S^p_m$ can now contain increasingly singular terms in $z_{ij}$, as long as we compensate for the scaling dimension by multiplying with appropriate factors of $z_i,z_j$ in the numerator. In particular, the chiral algebra of our currents will be able to display interesting central terms.

\subsection{Closed string modes} 

In previous sections, we have also determined the wavefunctions for K\"ahler perturbations of Mabuchi gravity. Their dual operators can be built out of the $E[k,l]$ and $F[k,l]$ currents of equation \eqref{CFT}. 

We choose to describe this dictionary in the case where we do not set the center of mass modes $E[1,0] = \text{Tr}\,X^{\dot1}$ and $E[0,1]=\text{Tr}\,X^{\dot2}$ to zero. As discussed below equation \eqref{tracefree}, this is analogous to working with $\U(N)$ rather than $\SU(N)$ gauge group in Maldacena's AdS$_5$/CFT$_4$ duality. Doing this will allow us to match three point functions like $\la E[1,0]J[0,1]J[0,0]\ra$, etc.\ to 4d amplitudes, thereby providing the simplest test of the gravitational sector of our duality in section \ref{sec:gravope}.

On $\R^4$, positive helicity metric perturbations in gravity are conventionally given little group weight $-4$. The associated K\"ahler potential perturbations must then also have the same little group weight. For example, the K\"ahler potential perturbation around flat space,
\be
\rho = \frac{\e^{\im p\cdot x}}{\lambda_1^2\lambda_2^2}\,,
\ee
solves $\lap_0^2\rho=0$ (and in fact $\lap_0\rho=0$) and gives rise to a standard positive helicity graviton wavefunction on $\R^4$,
\be
h = 2\,\frac{\p^2\rho}{\p u^{\dal}\p\hat u^{\dot\beta}}\,\d u^{\dal}\,\d\hat u^{\dot\beta} = -\frac{2\,\e^{\im p\cdot x}}{\lambda_1\lambda_2}\,[\tilde\lambda\,\d u][\tilde\lambda\,\d\hat u]\,.
\ee
This motivates us to normalize our quasi-momentum eigenstates for $\rho$ by a factor of $1/\lambda_1^2\lambda_2^2$. The reader may feel free to normalize them differently.

Tacking on this normalization, let us collect the two solutions \eqref{Erhosol} and \eqref{rhophi} below for the reader's convenience:
\begin{align}
    \rho^E(x|z,\tilde\lambda) &= \frac{2}{\lambda_1^2\lambda_2^2}\left\{\e^{\frac{\im p\cdot x}{2}}\cos\biggl(\frac{\psi\, p\cdot x}{2}\biggr)-1\right\}\,,\\
    \rho^F(x|z,\tilde\lambda) &= \frac{2\,\e^{\frac{\im p\cdot x}{2}}}{\lambda_1^2\lambda_2^2}\left(\|u\|^2+N\right)\left\{\cos\biggl(\frac{\psi\, p\cdot x}{2}\biggr)+\frac{\im}{\psi}\,\sin\biggl(\frac{\psi\, p\cdot x}{2}\biggr)\right\}\,.
\end{align}
We associate the currents $E[k,l]$ to soft modes of $\rho^E$, and $F[k,l]$ to the soft modes of $\rho^F$. Since $E[k,l]$ can have poles of order $k+l$, while $\rho^E$ by construction only has a second order pole in $z$, we postulate that $z^{k+l-2}E[k,l]$ are dual to soft modes of $\rho^E$. Similarly, $F[k,l]$ can have poles of order $k+l+2$, so we say that $z^{k+l}F[k,l]$ are dual to soft modes of $\rho^F$.

The operators dual to the corresponding hard states are then built from summing the soft operators:
\be\label{Ehard}
\begin{split}
    E(z,\tilde\lambda) &= -
    \frac{1}{z^2}+\sum_{k,l=0}^\infty\frac{\tilde\lambda_{\dot1}{}^k\tilde\lambda_{\dot2}{}^l}{k!\,l!}\,z^{k+l-2}E[k,l](z)\\
    &= \frac{1}{z^2}\left(\text{Tr}\,\bigl(\e^{z[X\,\tilde\lambda]}\bigr)-1\right)\,,
\end{split}
\ee
along with
\be\label{Fhard}
\begin{split}
    F(z,\tilde\lambda) &= \sum_{k,l=0}^\infty\frac{\tilde\lambda_{\dot1}{}^k\tilde\lambda_{\dot2}{}^l}{k!\,l!}\,z^{k+l}F[k,l](z)\\
    &= \text{Tr}\left([X\,\p X]\,\e^{z[X\,\tilde\lambda]}\right) + \text{terms with ghosts}\,.
\end{split}
\ee
The sum over $E[k,l]$ is designed to miss the trivial current $E[0,0]=1$. In terms of these, we postulate the gravitational holographic dictionary
\begin{align}
    &\rho^E(x|z,\tilde\lambda)\quad\longleftrightarrow\quad E(z,\tilde\lambda)\,,\\
    &\rho^F(x|z,\tilde\lambda)\quad\longleftrightarrow\quad F(z,\tilde\lambda)\,,
\end{align}
i.e., correlators of the gravitational operators on the right should reproduce amplitudes of Mabuchi gravity on Burns space.

In section \ref{sec:gravope}, we will content ourselves with testing this dictionary by computing an amplitude involving two gluons and one $\rho^E$ K\"ahler perturbation. The dictionary for $\rho^F$ is also a conjecture that we find natural, but tests of our duality that involve the $\rho^F$ states are left to another occasion.



\section{Tests of the duality}
\label{sec:tests}

Past work on celestial holography has produced a treasure trove of bottom-up correspondences between perturbative QFT on 4d asymptotically flat spacetimes and correlators of putative 2d CFTs. With our top-down model of holography on Burns space, we have completed this into a full-fledged duality by providing independent definitions of the bulk and boundary theories. In the sections that follow, we show that our proposed dual chiral algebra explicitly reproduces some simple low-lying amplitudes of WZW$_4$ + Mabuchi gravity on Burns space. 

Unlike in flat space, these amplitudes are non-vanishing on Burns space. They vanish on flat space because momentum eigenstates on flat space uplift to states on twistor space that are localized on different fibers of $\CO(1)\oplus\CO(1)\to\CP^1$ (see \cite{Adamo:2017qyl} for the explicit form of such states). Because they are separated in the $z$-plane, they cannot interact with each other through the local interactions present in the holomorphic Chern-Simons or BCOV twistor actions. In contrast, the twistor space of Burns space no longer fibers over $\CP^1$, and the states will scatter non-trivially for generic kinematics.

We begin by computing examples of 2-gluon tree amplitudes in WZW$_4$.  These give rise to a central extension for the celestial chiral algebra \cite{Guevara:2021abz, Strominger:2021mtt,Costello:2022wso} of positive helicity gluons on Burns space. We also compute collinear limits of 3-gluon amplitudes at leading and subleading order in the brane backreaction, and show that they match the OPE coefficients of the dual chiral algebra.\footnote{A direct calculation of the 3-gluon tree amplitude of WZW$_4$ on Burns space to leading order in the backreaction can also be found in the supplemental material to 
\cite{Costello:2022jpg}. We will not repeat it here, but its collinear limit agrees with the OPE that we report.} A similar calculation is also performed for the 3-point scattering of two gluons and one gravitational $E$-mode to leading order.


\subsection{2-point amplitudes in WZW$_4$}
\label{sec:2gluon}

Scattering amplitudes are the primary physical observables of perturbative gauge theory and gravity in asymptotically flat spacetimes. It is well-known that in Minkowski space, tree level amplitudes of massless particles are analytic functions of spinor-helicity variables. So they can be analytically continued to split signature, where the spinor-helicity variables $\lambda_\al,\tilde\lambda_{\dal}$ can be real and independent, or to Euclidean signature where they are necessarily complex due to the absence of real null momenta. As we are working on a Euclidean signature, self-dual spacetime, we will define scattering amplitudes on Burns space via similar analytic continuation. 

A preliminary understanding of amplitudes on Euclidean signature instanton backgrounds has been developed in \cite{HawkingYM, Hawking:1979hw, Hawking:1979pi}. Alternatively, one can think of our Euclidean signature amplitudes as the analogue of holographic correlators computed by summing Witten diagrams in Euclidean AdS$_3$, which is how such observables were originally computed in earlier examples of twisted holography \cite{CG, Costello:2020jbh}.

\paragraph{Boundary computation of 2-point amplitude.} The first test that we present is a match between the 2-point amplitude of the WZW$_4$ model on Burns space and a 2-point correlator in the boundary dual. Before we dive into the computation on Burns space, let us start by describing the calculation in the dual chiral algebra. This happens to be drastically simpler than the calculation in the bulk because the boundary dual is the BRST reduction of a free theory.

Let us work with the standard basis of $\mf{so}(k)$ provided by the antisymmetric matrices
\be
(\sT_{pq})_e{}^f = \delta_{pe}\delta_q{}^f - \delta_p{}^f\delta_{qe}\,,\qquad p,q,e,f\in\{1,\dots,k\}\,.
\ee
For us, the case of interest is $k=8$. But since we are working in the planar limit, our calculations will be agnostic to the value of $k$. The $\mf{so}(k)$ algebra is given by
\be
[\sT_{pq},\sT_{rs}] \equiv f_{pq,rs}{}^{tu}\sT_{tu} = \delta_{qr}\sT_{ps} - \delta_{pr}\sT_{qs} - \delta_{qs}\sT_{pr} + \delta_{ps}\sT_{qr}\,,
\ee
while the trace in its fundamental representation may be evaluated to find
\be\label{sotrace}
\tr(\sT_{pq}\sT_{rs}) = -2\left(\delta_{pr}\delta_{qs}-\delta_{ps}\delta_{qr}\right)\,.
\ee
These fix the conventions for our color factors.

The currents $J(z)\equiv J[0,0](z)$ form a non-unitary $\mf{so}(8)$ current algebra at level $-2N$,
\be
J_{pq}(z)\,J_{rs}(z') \sim -\frac{N\,\tr(\sT_{pq}\sT_{rs})}{(z-z')^2} + \frac{f_{pq,rs}{}^{tu}J_{tu}(z')}{z-z'}\,.
\ee
This is computed using the free field OPEs \eqref{IIope}. The first term on the right determines the 2-point function $\la J[0,0]J[0,0]\ra$. This is a Burns space analogue of the soft-gluon Kac-Moody algebra \cite{He:2015zea}. Its non-unitarity is consistent with general expectations in celestial holography. Similar central extensions of the soft gluon algebra have also been observed in other closely related theories and classical backgrounds \cite{Costello:2022wso,Melton:2022fsf}.

The 2-point functions of more general currents $J[k,l]$ may be computed along the same lines. Since we are working in the vacuum conformal block in which $X,I$ have vanishing one-point functions, the only terms that survive are those in which all the fields $X,I$ have been Wick contracted away. Hence, $J[k,l]$ can only have a 2-point function with $J[l,k]$. 

Also, we only keep planar contractions, since we wish to reproduce a tree level 2-point amplitude in the bulk. Each two point function $\la J_{pq}[k,l]J_{rs}[l,k]\ra$ then receives contributions from two planar orderings, one in which $I_p$ is contracted with $I_r$, and another in which $I_p$ is contracted with $I_s$. The rest of the contractions are then fixed by planarity. An example of such a configuration of planar contractions is displayed in figure \ref{contractions}.

\begin{figure}
\centering
\begin{subfigure}[t]{0.45\textwidth}
    \centering
    \includegraphics[scale=0.3]{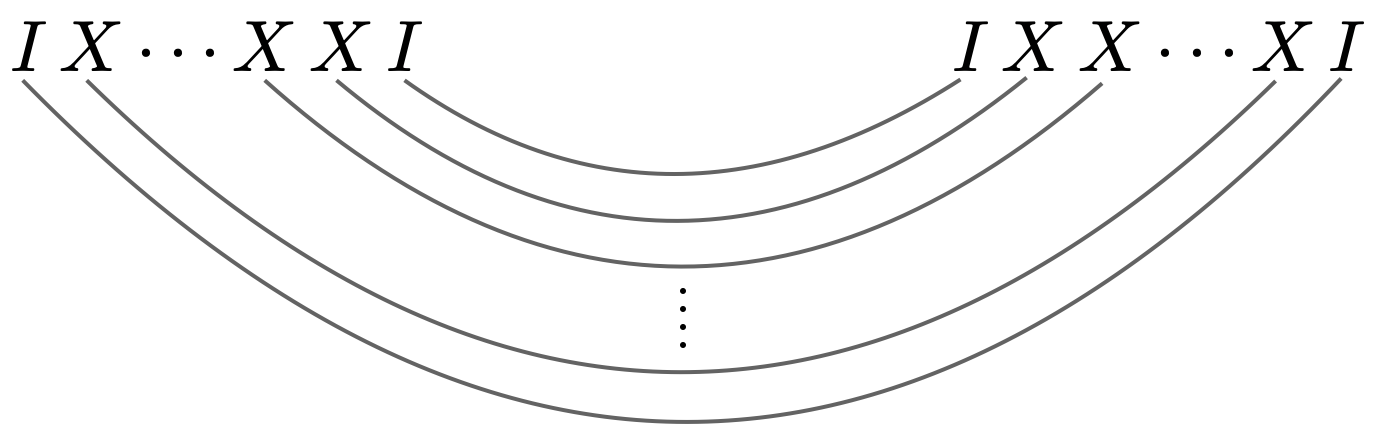}
    \caption{Wick contractions}
    \label{contractions}
\end{subfigure}
\hspace{2cm}
\begin{subfigure}[t]{0.25\textwidth}
    \centering
    \includegraphics[scale=0.3]{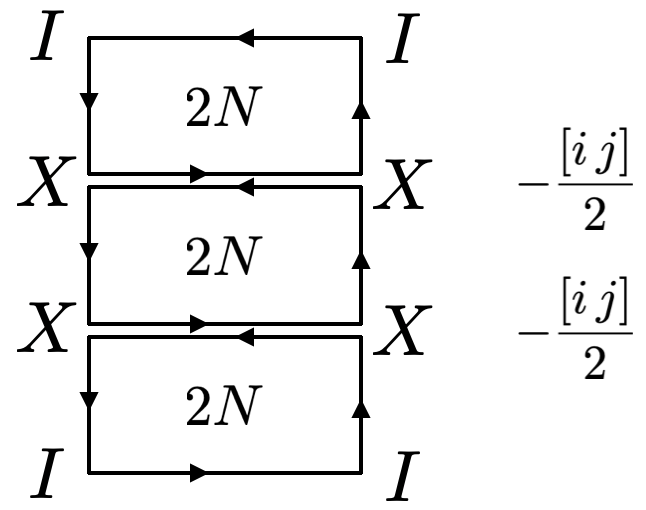}
    \caption{Planar diagrams}
    \label{planar}
\end{subfigure}
\caption{Planar Wick contractions of $\mf{so}(8)$ currents containing equal number of $X$'s.}
\end{figure}

To produce a prediction for the 2-point amplitude, we use the dictionary \eqref{opendict}. This says that the 2-point amplitude of two WZW$_4$ states $\phi_i\equiv\phi_{pq}(x|z_i,\tilde\lambda_i)$ and $\phi_j\equiv\phi_{rs}(x|z_j,\tilde\lambda_j)$ is given by the chiral algebra correlator
\be\label{2ptchiralcalc}
\begin{split}
    A(\phi_i,\phi_j) &= \bigl\la J_{pq}(z_i,\tilde\lambda_i)\,J_{rs}(z_j,\tilde\lambda_j)\bigr\ra\\
    &= \sum_{m,n\geq0}\frac{z_i^mz_j^n}{m!n!}\,\Bigl\la :\!I_p[X\,i]^mI_q\!:\!(z_i)\,:\!I_r[X\,j]^nI_s\!:\!(z_j)\Bigl\ra\\
    &= \sum_{n\geq0}\frac{(z_iz_j)^n}{(n!)^2}\,\Bigl\la :\!I_p[X\,i]^nI_q\!:\!(z_i)\,:\!I_r[X\,j]^nI_s\!:\!(z_j)\Bigl\ra
    \,.
\end{split}
\ee
In getting the third line, we have discarded terms for which $m\neq n$ as they do not contain non-zero 2-point functions.

Note that in our conventions
\be
\bigl\la [X^k{}_l\,i](z_i)\,[X^m{}_n\,j](z_j)\bigr\ra = \frac{1}{2}\,(\delta^k_n\delta^m_l-\omega^{km}\omega_{ln})\cdot\frac{-[i\,j]}{z_{ij}}
\ee
because $\eps^{\dal\dot\beta}\tilde\lambda_{i\dal}\tilde\lambda_{j\dot\beta}=-[i\,j]$. Each closed loop in a planar diagram generates a factor of $\omega^{kl}\omega_{lk}=\delta^k{}_k=2N$. The $n^\text{th}$ term in the sum \eqref{2ptchiralcalc} will receive contributions from planar diagrams with $n+1$ closed loops, like the ones in figure \ref{planar}. 
This generates a factor of $(2N)^{n+1}$. The $n$ $XX$ OPEs also generate a factor $(\frac{1}{2})^n$ through the antisymmetrizations $\frac{1}{2}(\delta^k_n\delta^m_l-\omega^{km}\omega_{ln})$, since only $\delta^k_n\delta^m_l$ contributes to planar diagrams. They also generate a factor of $(-[i\,j])^n$. So each summand in \eqref{2ptchiralcalc} comes with $(2N)^{n+1}(-[i\,j])^n/2^n = 2N(-N[i\,j])^n$.

Thus, performing the OPEs and dropping non-planar contributions generated in each summand, the 2-point correlator can be expressed in terms of a Bessel function of the first kind:
\be\label{2ptprediction}
\begin{split}
    A(\phi_i,\phi_j) &= 2N(\delta_{pr}\delta_{qs}-\delta_{ps}\delta_{qr})\sum_{n\geq0}\frac{1}{(n!)^2}\,\frac{(-Nz_iz_j[i\,j])^n}{z_{ij}^{n+2}}\\
    &= -\frac{N\,\tr(\sT_{pq}\sT_{rs})}{z_{ij}^2}\,J_0\bigg(\sqrt{4Nz_iz_j\frac{[i\,j]}{z_{ij}}}\bigg)\,.
\end{split}
\ee
Let's see how this prediction stacks up against a bulk computation of the 2-point amplitude of the WZW$_4$ model. In fact, we will see that the tree level computation in the bulk does not generate any defect singularities $z_i^{-1},z_j^{-1}$, so the match with the vacuum block will be perfect. This will act as a spacetime corroboration of the twistorial proof of symmetry enhancement described in section \ref{sec:enhance}.

\paragraph{Bulk computation of 2-point amplitude.} Let $\phi_i$ be solutions of the linearized field equations of a field theory described by an action $S[\phi]$. Such solutions represent scattering states of particle-like excitations, and $i$ is the particle label. Independently of signature, one can compute the tree level scattering amplitudes of these states from a suitably constructed on-shell action \cite{Arefeva:1974jv,Jevicki:1987ax}. This is also closely related to the perturbiner method \cite{Rosly:1996vr,Rosly:1997ap}. 

In this approach, one starts by constructing a solution of the form
\be
\phi = \sum_{i}\veps_i\phi_i + \mathrm{O}(\veps_i\veps_j)
\ee
to the fully nonlinear field equations of $S[\phi]$, where the formal variables $\veps_i$ are classical analogues of creation operators. Having found such a solution, one evaluates its on-shell action as a power series in the $\veps_i$. The amplitude for these states is then obtained by extracting the term in the on-shell action that is multilinear in the $\veps_i$:
\be
A(\{\phi_i\}) = \bigg(\prod_{j}\frac{\p}{\p\veps_j}\bigg)\,S[\phi]\,\bigg|_{\veps_i=0}\,.
\ee
In Lorentzian signature spacetimes with canonical notions of asymptotic infinities and positive vs negative frequency states, this reduces to the standard definition of the S-matrix. More generally, this is the prescription that is commonly employed in computing holographic correlators in AdS/CFT \cite{Witten:1998qj}.

The evaluation of 2-point amplitudes is particularly simple. Starting with the action \eqref{phiac} of the WZW$_4$ model, one simply evaluates its kinetic term on the linear combination
\be
\phi = \veps_i\phi_i + \veps_j\phi_j
\ee
of two solutions of the adjoint-valued Laplace equation. Following this, one extracts the amplitude as the coefficient of $\veps_i\veps_j$:
\be\label{2ptgluon}
A(\phi_i,\phi_j) = -\frac{\im}{8\pi^2}\int_M\omega\wedge\tr\left(\p\phi_i\wedge\dbar\phi_j+\p\phi_j\wedge\dbar\phi_i\right)\,.
\ee
Let us determine this for a pair of quasi-momentum eigenstates of the form \eqref{burnsmom} scattering on $M=\widetilde\C^2$ equipped with the Burns metric. 

Before embarking on the calculation, it is worth remarking that this 2-point amplitude has also been computed in a beautiful and simpler manner via the Klein-Gordon inner product on $\CP^2$ minus a point in \cite{Hawking:1979pi}. This is accomplished through analytic continuation to a Lorentzian slice equipped with the complexified Burns or Fubini-Study metric. This method trades the spacetime integral in \eqref{2ptgluon} for an integral over the conformal boundary, and hence is only applicable at 2 points. In contrast, the methods we follow below -- albeit more cumbersome -- will also generalize to the computation of higher point amplitudes.

\medskip

The K\"ahler form of the Burns metric was found in \eqref{omegaburns} and can be written as
\be
\omega = \im\left(\eps_{\dal\dot\beta}+\frac{Nu_{\dal}\hat u_{\dot\beta}}{\|u\|^4}\right)\d u^{\dal}\wedge\d\hat u^{\dot\beta}\,.
\ee
Let's begin by plugging this into the amplitude \eqref{2ptgluon}. To simplify the resulting expression, recall the definitions
\be
\p\phi_i = \frac{\p\phi_i}{\p u^{\dal}}\,\d u^{\dal}\,,\qquad \dbar\phi_i = \frac{\p\phi_i}{\p\hat u^{\dal}}\,\d\hat u^{\dal}\,.
\ee
We can use our orientation convention \eqref{orientation}, the definition of the quaternionic conjugate $\hat u^{\dal}=(-\bar u^2,\bar u^1)$, and standard spinor identities to show that
\be
\d u^{\dal}\wedge\d\hat u^{\dot\beta}\wedge\d u^{\dot\gamma}\wedge\d\hat u^{\dot\delta} = \eps^{\dal\dot\gamma}\eps^{\dot\beta\dot\delta}\,\d^4x\,.
\ee
Using these relations, the 2-gluon amplitude \eqref{2ptgluon} takes the form
\be\label{2ptgluon1}
A(\phi_i,\phi_j) = \frac{1}{8\pi^2}\int_{\R^4}\d^4x\;\bigg( \eps^{\dal\dot\beta}+\frac{Nu^{\dal}\hat u^{\dot\beta}}{\|u\|^4}\bigg)\,\tr\!\left(\frac{\p\phi_i}{\p u^{\dal}}\frac{\p\phi_j}{\p\hat u^{\dot\beta}} + (i\leftrightarrow j)\right)\,.
\ee
In writing this, we are computing the integral over $\widetilde\C^2$ as an integral over the coordinate patch $x^{\al\dal}=(u^{\dal},\hat u^{\dal})$. Strictly speaking, the integral is over $\C^2-0$ due to the singularities at the origin. But it can be regularized via standard procedures involved in computing Fourier transforms of Coulomb-type potentials. So we will continue to express it as an integral over $\C^2 = \R^4$.

The wavefunctions for the states $\phi_i,\phi_j$ were derived in section \ref{sec:wzw4}. Let us repeat them for convenience. A quasi-momentum eigenstate $\phi_i(x)$ that solves the adjoint-valued Laplace equation is given by the expansion in \eqref{phi1f1},
\be
\phi_i(x) = \frac{\msf{T}_{a_i}}{\lambda_{i1}\lambda_{i2}}\sum_{n=0}^\infty\frac{1}{(2n)!}\biggl(-N\lambda_{i1}\lambda_{i2}\frac{[u\,i][\hat u\,i]}{\|u\|^2}\biggr)^n{}_1F_1(n+1,2n+1\,|\,\im p_i\cdot x)\,,
\ee
where $[u\,i]\equiv[u\,\tilde\lambda_i]$, etc., and $\lambda_{i\al},\tilde\lambda_{i\dal}$ are spinor-helicity variables labelling the state. The null momentum $p_i$ has spinor components $p_{i\al\dal}=\lambda_{i\al}\tilde\lambda_{i\dal}$, so that its product with $x^{\al\dal}$ can be expressed as $p_i\cdot x = \lambda_{i1}[u\,i]+\lambda_{i2}[\hat u\,i]$. The state $\phi_j(x)$ is obtained by replacing the particle label $i$ with the label $j$.

Let us denote the coefficient of $\sT_{a_i}N^n$ in this wavefunction by $\phi_i^n(x)$:
\be\label{phin}
\phi_i^n(x) = \frac{1}{\lambda_{i1}\lambda_{i2}}\,\frac{1}{(2n)!}\biggl(-\lambda_{i1}\lambda_{i2}\frac{[u\,i][\hat u\,i]}{\|u\|^2}\biggr)^n{}_1F_1(n+1,2n+1\,|\,\im p_i\cdot x)\,.
\ee
The confluent hypergeometric function occurring here has a very useful expansion on flat space momentum eigenstates,
\be\label{1f1int}
{}_1F_1(n+1,2n+1\,|\,\im p_i\cdot x) = \frac{\Gamma(2n+1)}{\Gamma(n+1)\Gamma(n)}\int_0^1\d s_i\,s_i^n(1-s_i)^{n-1}\e^{\im s_i p_i\cdot x}\,.
\ee
For $n=0$, we recover $\e^{\im p_i\cdot x}$ by taking the integral to mean an $n\to0$ limit, understood via the following identity valid within integrals over $s_i\in(0,1)$,
\be\label{sdelta}
\lim_{n\to0}\frac{(1-s_i)^{n-1}}{\Gamma(n)} = \delta(1-s_i)\,.
\ee
This representation of the states will allow us to perform the spacetime integrals as standard Fourier integrals. 

Inserting the expansions
\be\label{phiijexp}
\phi_i = \sT_{a_i}\sum_{m=0}^\infty N^m\phi_i^m\,,\qquad \phi_j = \sT_{a_j}\sum_{n=0}^\infty N^n\phi_j^n
\ee
into \eqref{2ptgluon1} yields the amplitude as a formal power series in $N$,
\be\label{A2series}
A(\phi_i,\phi_j) = \tr(\sT_{a_i}\sT_{a_j})\sum_{m,n=0}^\infty N^{m+n}\left(N\,\cI_{m,n}+\cJ_{m,n}\right) + (i\leftrightarrow j)\,.
\ee
In writing this, we have split our spacetime integrals into two ``irreducible'' terms that do not contain any further dependence on $N$:
\begin{align}
    \cI_{m,n} &= \frac{1}{8\pi^2}\int\d^4x\;\frac{u^{\dal}\hat u^{\dot\beta}}{\|u\|^4}\,\frac{\p\phi_i^m}{\p u^{\dal}}\,\frac{\p\phi_j^n}{\p\hat u^{\dot\beta}}\,,\label{Imn}\\
    \cJ_{m,n} &= \frac{1}{8\pi^2}\int\d^4x\;\eps^{\dal\dot\beta}\,\frac{\p\phi_i^m}{\p u^{\dal}}\,\frac{\p\phi_j^n}{\p\hat u^{\dot\beta}}\,.\label{Jmn}
\end{align}
Both these integrals can be performed using very similar techniques. Let us see how to compute them one by one.

\paragraph{Computation of $\cI_{m,n}$.} We will compute the first of these in some detail, as it happens to be the simpler case. 

One of the reasons this is easier is because of the following identities
\be
u^{\dal}\frac{\p}{\p u^{\dal}}\frac{[u\,i][\hat u\,i]}{\|u\|^2} = \hat u^{\dal}\frac{\p}{\p\hat u^{\dal}}\frac{[u\,i][\hat u\,i]}{\|u\|^2} = 0\,.
\ee
These hold because the vector fields $u\cdot\p_u$ and $\hat u\cdot\p_{\hat u}$ are the Euler vector fields that extract homogeneities in $u^{\dal}$ and $\hat u^{\dal}$ respectively, and $[u\,i][\hat u\,i]/\|u\|^2 = [u\,i][\hat u\,i]/[\hat u\,u]$ has zero homogeneity in both of these variables. So the derivatives displayed in \eqref{Imn} will only act on the dependence in $\im p_i\cdot x$ and $\im p_j\cdot x$. These can be computed by using the formula for $\phi_i^n$ given in \eqref{phin}, combined with the integral representation \eqref{1f1int} for the ${}_1F_1$ factors,
\begin{align}
    u^{\dal}\frac{\p\phi_i^m}{\p u^{\dal}} &= \frac{1}{\lambda_{i1}\lambda_{i2}}\biggl(-\lambda_{i1}\lambda_{i2}\frac{[u\,i][\hat u\,i]}{\|u\|^2}\biggr)^m\int_0^1\frac{\d s_i\,s_i^{m}(1-s_i)^{m-1}}{m!(m-1)!}\,\e^{\im s_ip_i\cdot x}\times\im s_i\lambda_{i1}[u\,i]\,,\label{uduphi}\\
    \hat u^{\dot\beta}\frac{\p\phi_j^n}{\p\hat u^{\dot\beta}} &= \frac{1}{\lambda_{j1}\lambda_{j2}}\biggl(-\lambda_{j1}\lambda_{j2}\frac{[u\,j][\hat u\,j]}{\|u\|^2}\biggr)^n\int_0^1\frac{\d s_j\,s_j^{n}(1-s_j)^{n-1}}{n!(n-1)!}\,\e^{\im s_jp_j\cdot x}\times\im s_j\lambda_{j2}[\hat u\,j]\,.\label{uhduhphi}
\end{align}
In deriving these, we have remembered that $p_i\cdot x = \lambda_{i1}[u\,i]+\lambda_{i2}[\hat u\,i]$, etc., using which we found $u\cdot\p_u(p_i\cdot x) = \lambda_{i1}[u\,i]$ and so on.

Next note the extremely helpful intertwining relations
\be\label{intertwine}
\begin{split}
    \im [u\,i]\e^{\im s_ip_i\cdot x} &= \frac{1}{s_i}\frac{\p}{\p\lambda_{i1}}\,\e^{\im s_ip_i\cdot x}\,,\\
    \im [\hat u\,i]\e^{\im s_ip_i\cdot x} &= \frac{1}{s_i}\frac{\p}{\p\lambda_{i2}}\,\e^{\im s_ip_i\cdot x}\,.
\end{split}
\ee
We can use these to simplify our integrals by trading factors of $[u\,i]$, $[\hat u\,i]$, etc.\ for derivatives in the external data like the spinor-helicity variables. Doing this, \eqref{uduphi} and \eqref{uhduhphi} can be recast as
\begin{align}
    u^{\dal}\frac{\p\phi_i^m}{\p u^{\dal}} &= \frac{\lambda_{i1}^{m}\lambda_{i2}^{m-1}}{m!(m-1)!}\int_0^1\frac{\d s_i}{s_i^m}\,(1-s_i)^{m-1}\times\p_{\lambda_{i1}}^{m+1}\p_{\lambda_{i2}}^m\bigg(\frac{\e^{\im s_ip_i\cdot x}}{\|u\|^{2m}}\bigg)\,,\label{uduphi1}\\
    \hat u^{\dot\beta}\frac{\p\phi_j^n}{\p\hat u^{\dot\beta}} &= \frac{\lambda_{j1}^{n-1}\lambda_{j2}^{n}}{n!(n-1)!}\int_0^1\frac{\d s_j}{s_j^{n}}\,(1-s_j)^{n-1}\times\p_{\lambda_{j1}}^{n}\p_{\lambda_{j2}}^{n+1}\bigg(\frac{\e^{\im s_jp_j\cdot x}}{\|u\|^{2n}}\bigg)\,.\label{uhduhphi1}
\end{align}
With these representations for the derivatives in hand, we are finally ready to start integrating.

Substituting \eqref{uduphi1}, \eqref{uhduhphi1} into the definition \eqref{Imn} of $\cI_{m,n}$, we can segregate the $s_i,s_j$ integrals from the spacetime integrals,
\begin{multline}\label{Imn1}
    \cI_{m,n} = \frac{\lambda_{i1}^{m}\lambda_{i2}^{m-1}}{m!(m-1)!}\,\frac{\lambda_{j1}^{n-1}\lambda_{j2}^{n}}{n!(n-1)!}\int_0^1\frac{\d s_i}{s_i^m}\,(1-s_i)^{m-1}\int_0^1\frac{\d s_j}{s_j^{n}}\,(1-s_j)^{n-1}\\
    \times\p_{\lambda_{i1}}^{m+1}\p_{\lambda_{i2}}^m\p_{\lambda_{j1}}^{n}\p_{\lambda_{j2}}^{n+1}\mathcal{S}_{m+n+2}\,,
\end{multline}
where the final factor contains the most basic spacetime integral that we will repeatedly encounter:
\be\label{Smn}
\mathcal{S}_{\ell} = \int\frac{\d^4x}{8\pi^2}\,\frac{\e^{\im(s_ip_i+s_jp_j)\cdot x}}{(x^2/2)^{\ell}}\,.
\ee
In writing this, we have recalled the expression $\|u\|^2=x^2/2$ for the Euclidean norm in our conventions $u^{\dal}=\frac{1}{\sqrt 2}(x^0+\im x^3,x^2+\im x^1)$ for the complex coordinates. The spacetime integral has now manifestly reduced to a Coulomb-type Fourier transform in four dimensions. We will perform it by Fourier transforming with respect to real momenta and analytically continuing the result to generic complex momenta like $s_ip_i+s_jp_j$.

It is useful for future reference to set this up in some generality. Let $k_\mu$ be a real momentum and consider the Fourier integral
\be
\int\frac{\d^4x\;\e^{\im k\cdot x}}{(x^2/2)^\ell}\,.
\ee
Since $x^2>0$ in Euclidean signature, we can introduce the Schwinger parametrization
\be
\frac{1}{(x^2/2)^\ell} = \frac{1}{\Gamma(\ell)}\int_0^\infty\d t\,t^{\ell-1}\,\e^{-tx^2/2}\,,
\ee
where we are assuming that $\ell$ is not a negative integer. Doing this allows us to perform the $x^\mu$ integral as a Fourier transform of a Gaussian. This yields a slick identity for obtaining Coulomb-like Fourier transforms:
\be\label{fourierid}
\int\frac{\d^4x\;\e^{\im k\cdot x}}{(x^2/2)^\ell} = \frac{4\pi^2}{\Gamma(\ell)}\int_0^\infty\frac{\d t}{t^{\ell-1}}\,\e^{-t k^2/2}\,.
\ee
In arriving at this, we have inverted $t\mapsto 1/t$ after doing the $x^\mu$ integral. A useful sanity check is that when $\ell=1$, this gives the expected Fourier transform $-8\pi^2/k^2$ of a 4d Coulomb potential $2/x^2$. 

Analytically continuing to complex $k_\mu$ away from the singular locus $k^2=0$ and setting $k=s_ip_i+s_jp_j$ yields
\be\label{Sval}
\mathcal{S}_{\ell} = \frac{1}{2}\,\frac{1}{\Gamma(\ell)}\int_0^\infty\frac{\d t}{t^{\ell-1}}\,\e^{-ts_is_jp_i\cdot p_j}\,.
\ee
If $\ell$ is a positive integer greater than $1$, then this integral diverges. But this is precisely the case we're in: $\ell=m+n+2$ for \eqref{Imn1}! What will save this integral from diverging are the $\lambda_{i\al},\lambda_{j\al}$ derivatives present in \eqref{Imn1}. Plugging the value of $\mathcal{S}_{m+n+2}$ into \eqref{Imn1} and performing just the $\lambda_{i2}, \lambda_{j2}$ derivatives is enough to leave us with a convergent integral
\begin{multline}\label{Imn2}
    \cI_{m,n} = \frac{1}{2}\,\frac{1}{(m+n+1)!}\,\frac{\lambda_{i1}^{m}\lambda_{i2}^{m-1}}{m!(m-1)!}\,\frac{\lambda_{j1}^{n-1}\lambda_{j2}^{n}}{n!(n-1)!}\int_0^1\frac{\d s_i}{s_i^m}\,(1-s_i)^{m-1}\int_0^1\frac{\d s_j}{s_j^{n}}\,(1-s_j)^{n-1}\\
    \times\p_{\lambda_{i1}}^{m+1}\p_{\lambda_{j1}}^{n}\int_0^\infty\d t\,\e^{-ts_is_j\la i\,j\ra[i\,j]}\,(-s_is_j\lambda_{j1}[i\,j])^m(s_is_j\lambda_{i1}[i\,j])^{n+1}\,,
\end{multline}
having used $p_i\cdot p_j = \la i\,j\ra[i\,j]$. We are abbreviating $\la i\,j\ra\equiv\la\lambda_i\lambda_j\ra$, $[i\,j]\equiv[\tilde\lambda_i\tilde\lambda_j]$ as is common practice, and have computed the derivatives using
\be
\la i\,j\ra = \lambda_{i2}\lambda_{j1}-\lambda_{i1}\lambda_{j2}\implies \p_{\lambda_{i2}}\la i\,j\ra = \lambda_{j1}\,,\quad\p_{\lambda_{j2}}\la i\,j\ra = -\lambda_{i1}\,.
\ee
\eqref{Imn2} is now in a form amenable to straightforward integration.

We can integrate over $t$ by analytically continuing the standard Mellin integral
\be
\int_0^\infty\d t\,\e^{-tk^2/2} = \frac{2}{k^2}
\ee
to the non-vanishing complex value $k^2 = (s_ip_i+s_jp_j)^2 = 2s_is_j\la i\,j\ra[i\,j]$. After some simplifications, this produces
\begin{multline}\label{Imn3}
    \cI_{m,n} = \frac{1}{2}\,\frac{(-1)^m[i\,j]^{m+n}}{(m+n+1)!}\,\frac{\lambda_{i1}^{m}\lambda_{i2}^{m-1}}{m!(m-1)!}\,\frac{\lambda_{j1}^{n-1}\lambda_{j2}^{n}}{n!(n-1)!}\;\p_{\lambda_{i1}}^{m+1}\p_{\lambda_{j1}}^{n}\bigg(\frac{\lambda_{i1}^{n+1}\lambda_{j1}^m}{\la i\,j\ra}\bigg)\\
    \times\int_0^1\d s_i\,s_i^{n}\,(1-s_i)^{m-1}\int_0^1\d s_j\,s_j^m\,(1-s_j)^{n-1}\,.
\end{multline}
We can integrate over $s_i$ and $s_j$ using the Euler Beta integral to find
\be
\int_0^1\d s_i\,s_i^{n}\,(1-s_i)^{m-1}\int_0^1\d s_j\,s_j^m\,(1-s_j)^{n-1} = \frac{n!(m-1)!}{(m+n)!}\cdot\frac{m!(n-1)!}{(m+n)!}\,.
\ee
As a result, we are able to perform all the integrals in $\cI_{m,n}$ and conclude that
\be\label{Imn4}
\cI_{m,n} = \frac{1}{2}\,\frac{(-1)^m[i\,j]^{m+n}}{(m+n+1)!}\,\frac{\lambda_{i1}^{m}\lambda_{i2}^{m-1}\lambda_{j1}^{n-1}\lambda_{j2}^{n}}{((m+n)!)^2}\;\p_{\lambda_{i1}}^{m+1}\p_{\lambda_{j1}}^{n}\bigg(\frac{\lambda_{i1}^{n+1}\lambda_{j1}^m}{\la i\,j\ra}\bigg)\,.
\ee
In the edge cases $m=0$ or $n=0$, the same result is obtained by remembering the identity \eqref{sdelta}.

The last step is to compute the $\lambda_{i1}, \lambda_{j1}$ derivatives. Although this step looks deceptively hard, it will display a few more combinatorial simplifications that precipitate a delightfully compact final answer. In detail, we want to compute
\be
\p_{\lambda_{i1}}^{m+1}\p_{\lambda_{j1}}^{n}\bigg(\frac{\lambda_{i1}^{n+1}\lambda_{j1}^m}{\la i\,j\ra}\bigg) = \p_{\lambda_{i1}}^{m+1}\p_{\lambda_{j1}}^{n}\bigg(\frac{\lambda_{i1}^{n+1}\lambda_{j1}^m}{\lambda_{i2}\lambda_{j1}-\lambda_{i1}\lambda_{j2}}\bigg)\,.
\ee
Suppose first that $m\geq n$. Substituting $\lambda_{i1}^{n+1} = (\lambda_{i2}\lambda_{j1}-\la i\,j\ra)^{n+1}/\lambda_{j2}^{n+1}$ and expanding using the binomial theorem yields
\be
\p_{\lambda_{i1}}^{m+1}\p_{\lambda_{j1}}^{n}\bigg(\frac{\lambda_{i1}^{n+1}\lambda_{j1}^m}{\la i\,j\ra}\bigg) = \p_{\lambda_{i1}}^{m+1}\p_{\lambda_{j1}}^{n}\sum_{k=0}^{n+1}{n+1\choose k}\frac{\lambda_{i2}^{n+1-k}\lambda_{j1}^{m+n+1-k}}{\lambda_{j2}^{n+1}}\,(-1)^k\la i\,j\ra^{k-1}\,.
\ee
Since $\lambda_{i1}$ only enters the right hand side through $\la i\,j\ra$, the $\lambda_{i1}$ derivatives only act on the factors of $\la i\,j\ra$. Moreover, $\la i\,j\ra$ is linear in $\lambda_{i1}$. So in the case $m\geq n$, only the $k=0$ term of the sum survives the $m+1$ $\lambda_{i1}$ derivatives because $k-1$ is bounded above by $n$. 

This simplification allows us to compute the $\lambda_{i1}$ derivatives in closed-form:
\begin{align}
\p_{\lambda_{i1}}^{m+1}\p_{\lambda_{j1}}^{n}\bigg(\frac{\lambda_{i1}^{n+1}\lambda_{j1}^m}{\la i\,j\ra}\bigg) &= \frac{\lambda_{i2}^{n+1}}{\lambda_{j2}^{n+1}}\,\p_{\lambda_{i1}}^{m+1}\p_{\lambda_{j1}}^{n}\bigg(\frac{\lambda_{j1}^{m+n+1}}{\la i\,j\ra}\bigg)\nonumber\\
&= (m+1)!\,\frac{\lambda_{i2}^{n+1}}{\lambda_{j2}^{n-m}}\,\p_{\lambda_{j1}}^{n}\bigg(\frac{\lambda_{j1}^{m+n+1}}{\la i\,j\ra^{m+2}}\bigg)\,.
\end{align}
Distributing the remaining $\lambda_{j1}$ derivatives among the factors of $\lambda_{j1}^{m+n+1}$ and $1/\la i\,j\ra^{m+2}$ produces
\begin{align}
&\p_{\lambda_{i1}}^{m+1}\p_{\lambda_{j1}}^{n}\bigg(\frac{\lambda_{i1}^{n+1}\lambda_{j1}^m}{\la i\,j\ra}\bigg) = (m+1)!\,\frac{\lambda_{i2}^{n+1}}{\lambda_{j2}^{n-m}}\sum_{k=0}^n{n\choose k}\,\p_{\lambda_{j1}}^{n-k}\bigl(\lambda_{j1}^{m+n+1}\bigr)\,\p_{\lambda_{j1}}^{k}\bigg(\frac{1}{\la i\,j\ra^{m+2}}\bigg)\nonumber\\
&= (m+1)!\,\frac{\lambda_{i2}^{n+1}}{\lambda_{j2}^{n-m}}\sum_{k=0}^n{n\choose k}\,\frac{(m+n+1)!}{(m+k+1)!}\,\lambda_{j1}^{m+k+1}\times\frac{(m+k+1)!}{(m+1)!}\,\frac{(-\lambda_{i2})^k}{\la i\,j\ra^{m+k+2}}\nonumber\\
&= (m+n+1)!\,\frac{\lambda_{i2}^{n+1}\lambda_{j1}^{m+1}\lambda_{j2}^{m-n}}{\la i\,j\ra^{m+2}}\sum_{k=0}^n{n\choose k}\,\frac{(-\lambda_{i2}\lambda_{j1})^k}{\la i\,j\ra^{k}}\,.
\end{align}
Resumming the last line using the binomial theorem lands us on the identity
\begin{align}\label{derid}
\p_{\lambda_{i1}}^{m+1}\p_{\lambda_{j1}}^{n}\bigg(\frac{\lambda_{i1}^{n+1}\lambda_{j1}^m}{\la i\,j\ra}\bigg) = (-1)^n\,(m+n+1)!\,\frac{\lambda_{i1}^n\lambda_{i2}^{n+1}\lambda_{j1}^{m+1}\lambda_{j2}^m}{\la i\,j\ra^{m+n+2}}\,.
\end{align}
One obtains the same result for the case $0\leq m<n$ by initially eliminating the factor of $\lambda_{j1}^m$ instead of $\lambda_{i1}^{n+1}$.

With this identity, we can finish off the calculation of $\cI_{m,n}$,
\be\label{Imnmain}
\cI_{m,n} = \frac{1}{2}\,\frac{1}{((m+n)!)^2}\,\frac{(-\lambda_{i1}\lambda_{i2}\lambda_{j1}\lambda_{j2}[i\,j])^{m+n}}{\la i\,j\ra^{m+n+2}}\qquad \forall\; m,n\geq 0\,.
\ee
It is very interesting to note that this only depends on the sum $m+n$, which will come in handy when resumming the series expansion \eqref{A2series} of the 2-point amplitude. Before we turn to resummation, let us also similarly evaluate $\cJ_{m,n}$.

\paragraph{Computation of $\cJ_{m,n}$.} The second integral occurring in the 2-gluon amplitude was defined in \eqref{Jmn}. In the case $m=n=0$, upon symmetrization in $i,j$, $\cJ_{0,0}$ reduces simply to the flat space 2-point amplitude:
\begin{align}
    &\cJ_{0,0} + (i\leftrightarrow j) = \frac{1}{8\pi^2}\int\d^4x\;\eps^{\dal\dot\beta}\,\frac{\p\phi_i^0}{\p u^{\dal}}\,\frac{\p\phi_j^0}{\p\hat u^{\dot\beta}} + (i\leftrightarrow j)\nonumber\\
    &= \frac{1}{8\pi^2}\int\d^4x\;\eps^{\dal\dot\beta}\,\frac{(\im\lambda_{i1}\tilde\lambda_{i\dal})(\im\lambda_{j2}\tilde\lambda_{j\dot\beta})}{\lambda_{i1}\lambda_{i2}\lambda_{j1}\lambda_{j2}}\,\e^{\im(p_i+p_j)\cdot x} + (i\leftrightarrow j)\nonumber\\
    &= -\frac{\la i\,j\ra[i\,j]\,\delta^4(p_i+p_j)}{8\pi^2\lambda_{i1}\lambda_{i2}\lambda_{j1}\lambda_{j2}} = 0\,,
\end{align}
which vanishes for generic momenta due to $0=(p_i+p_j)^2 = 2\la i\,j\ra[i\,j]$ on the support of the delta function. For all other values of $m,n$, we expect to get a nontrivial answer arising from working on the Burns background.

Previously in \eqref{uduphi}, \eqref{uhduhphi}, we obtained expressions for $u\cdot\p_u\phi_i^m$ and $\hat u\cdot\p_{\hat u}\phi_j^n$. To evaluate $\cJ_{m,n}$, we need the more general derivatives $\p_{u^{\dal}}\phi_i^m$, $\p_{\hat u^{\dot\beta}}\phi_j^n$. To find these, one first establishes the derivatives
\be
\frac{\p}{\p u^{\dal}}\frac{[u\,i][\hat u\,i]}{\|u\|^2} = \frac{[\hat u\,i]^2u_{\dal}}{\|u\|^4}\,,\qquad\frac{\p}{\p\hat u^{\dot\beta}}\frac{[u\,j][\hat u\,j]}{\|u\|^2} = -\frac{[u\,j]^2\hat u_{\dot\beta}}{\|u\|^4}\,,
\ee
simplified using $\|u\|^2=[\hat u\,u]$ and Schouten's identities like $[\hat u\,u]\tilde\lambda_{i\dal}+[u\,i]\hat u_{\dal}=[\hat u\,i]u_{\dal}$, etc. Along with the by now familiar expressions $\p_{u^{\dal}}(p_i\cdot x) = \lambda_{i1}\tilde\lambda_{i\dal}$ and $\p_{\hat u^{\dot\beta}}(p_j\cdot x) = \lambda_{j2}\tilde\lambda_{j\dot\beta}$, these can be used to determine
\be\label{duphi}
\begin{split}
    \frac{\p\phi_i^m}{\p u^{\dal}} = -\left(-\lambda_{i1}\lambda_{i2}\frac{[u\,i][\hat u\,i]}{\|u\|^2}\right)^{m-1}\int_0^1&\frac{\d s_i\,s_i^{m}(1-s_i)^{m-1}}{m!(m-1)!}\,\e^{\im s_ip_i\cdot x}\\
    &\times\frac{[\hat u\,i]}{\|u\|^2}\,\bigg(\im s_i\lambda_{i1}[u\,i]\tilde\lambda_{i\dal} + \frac{m[\hat u\,i]u_{\dal}}{\|u\|^2}\bigg)
\end{split}
\ee
and similarly
\be\label{duhphi}
\begin{split}
    \frac{\p\phi_j^n}{\p\hat u^{\dot\beta}} = -\left(-\lambda_{j1}\lambda_{j2}\frac{[u\,j][\hat u\,j]}{\|u\|^2}\right)^{n-1}\int_0^1&\frac{\d s_j\,s_j^{n}(1-s_j)^{n-1}}{n!(n-1)!}\,\e^{\im s_jp_j\cdot x}\\
    &\times\frac{[u\,j]}{\|u\|^2}\,\bigg(\im s_j\lambda_{j2}[\hat u\,j]\tilde\lambda_{j\dot\beta} - \frac{n[u\,j]\hat u_{\dot\beta}}{\|u\|^2}\bigg)\,.
\end{split}
\ee
A useful check on these expressions is the fact that contracting them with $u^{\dal}$ and $\hat u^{\dot\beta}$ respectively leads to \eqref{uduphi} and \eqref{uhduhphi}. Another good check is that as $m\to0$ or $n\to0$, application of the delta function identity \eqref{sdelta} implies that these reduce to derivatives of the plane waves $\e^{\im p_i\cdot x}$ or $\e^{\im p_j\cdot x}$. Since we are working with spinor-helicity variables, covariance under the little group scalings \eqref{lgsc} also remains available as a consistency check at each step of the calculation.

The next step is to substitute \eqref{duphi} and \eqref{duhphi} into the definition \eqref{Jmn} of $\cJ_{m,n}$, which involves a straightforward contraction using $\eps^{\dal\dot\beta}$. Following this, one can once again use the intertwining relations \eqref{intertwine} to replace factors of $[u\,i],[\hat u\,i]$, etc.\ by derivatives in the external spinor-helicity data. When the dust settles, one obtains
\begin{multline}\label{Jmn1}
    \cJ_{m,n} = \frac{(\lambda_{i1}\lambda_{i2})^{m-1}}{m!(m-1)!}\,\frac{(\lambda_{j1}\lambda_{j2})^{n-1}}{n!(n-1)!}\int_0^1\frac{\d s_i}{s_i^{m}}\,(1-s_i)^{m-1}\int_0^1\frac{\d s_j}{s_j^{n}}\,(1-s_j)^{n-1}\\
    \hspace{-3cm}\times\Big\{s_is_j\lambda_{i1}\lambda_{j2}[i\,j]\,\p_{\lambda_{i1}}^m\p_{\lambda_{i2}}^m\p_{\lambda_{j1}}^n\p_{\lambda_{j2}}^n\cS_{m+n}\\
     - \big(n\lambda_{i1}\p_{\lambda_{i1}}+m\lambda_{j2}\p_{\lambda_{j2}}+mn\big)\,\p_{\lambda_{i1}}^{m-1}\p_{\lambda_{i2}}^{m+1}\p_{\lambda_{j1}}^{n+1}\p_{\lambda_{j2}}^{n-1}\cS_{m+n+1}\Big\}
\end{multline}
where the integral $\cS_\ell$ was defined in \eqref{Smn}. Now we are in the same ballpark as \eqref{Imn1}. Actually, the edge cases $m=0, n\geq1$ and $n=0, m\geq1$ are somewhat delicate, because they will require nontrivial cancellations between divergences of the $s_i,s_j$ integrals. So we will first work out $\cJ_{m,n}$ for $m,n\geq1$ and deal with these edge cases separately at the end.

The spacetime integral $\cS_\ell$ was converted to an integral over a single Schwinger parameter $t$ in \eqref{Sval}. If we plug this into \eqref{Jmn1} for the values $\ell=m+n$ and $\ell=m+n+1$, take all the $\lambda_{i2}$ derivatives, and take $n-1$ of the $\lambda_{j2}$ derivatives in both terms, the integrals over $t, s_i,s_j$ become convergent for $m,n\geq1$ and can be performed to give
\begin{align}\label{Jmn2}
    \cJ_{m,n} = \frac{1}{2mn}\,&\frac{(-1)^m[i\,j]^{m+n-1}}{(m+n)!}\,\frac{(\lambda_{i1}\lambda_{i2})^{m-1}(\lambda_{j1}\lambda_{j2})^{n-1}}{((m+n-1)!)^2}\nonumber\\
    &\times\bigg\{(m+n)\lambda_{i1}\lambda_{j2}\p_{\lambda_{j2}}\p_{\lambda_{i1}}^m\p_{\lambda_{j1}}^n\bigg(\frac{\lambda_{i1}^{n-1}\lambda_{j1}^m}{\la i\,j\ra}\bigg)\nonumber\\
    &\qquad+\, \big(n\lambda_{i1}\p_{\lambda_{i1}}+m\lambda_{j2}\p_{\lambda_{j2}}+mn\big)\,\p_{\lambda_{i1}}^{m-1}\p_{\lambda_{j1}}^{n+1}\bigg(\frac{\lambda_{i1}^{n-1}\lambda_{j1}^{m+1}}{\la i\,j\ra}\bigg)\bigg\}\,.
\end{align}
This has the same structure as the corresponding result \eqref{Imn4} for $\cI_{m,n}$. Thankfully, we will not need to repeat the computations of the derivatives in any great detail, as the identity \eqref{derid} suffices to simplify \eqref{Jmn2} as well.

Rearranging $\la i\,j\ra = \lambda_{i2}\lambda_{j1}-\lambda_{i1}\lambda_{j2}$ gives the relation
\be
\lambda_{j1} = \frac{\lambda_{i1}\lambda_{j2}}{\lambda_{i2}} + \frac{\la i\,j\ra}{\lambda_{i2}}
\ee
which can be used to reduce the derivatives in the last line of \eqref{Jmn2} to
\be\label{derstep}
\p_{\lambda_{i1}}^{m-1}\p_{\lambda_{j1}}^{n+1}\bigg(\frac{\lambda_{i1}^{n-1}\lambda_{j1}^{m+1}}{\la i\,j\ra}\bigg) = \frac{\lambda_{j2}}{\lambda_{i2}}\,\p_{\lambda_{i1}}^{m-1}\p_{\lambda_{j1}}^{n+1}\bigg(\frac{\lambda_{i1}^{n}\lambda_{j1}^{m}}{\la i\,j\ra}\bigg) + \frac{1}{\lambda_{i2}}\,\p_{\lambda_{i1}}^{m-1}\p_{\lambda_{j1}}^{n+1}\big(\lambda_{i1}^{n-1}\lambda_{j1}^{m}\big)\,.
\ee
For $m,n\geq1$, the second term on the right vanishes. For the $\lambda_{j1}$ derivatives to give a non-vanishing answer, we would need $m\geq n+1$. But this implies that $n-1\leq m-2<m-1$, whence the $\lambda_{i1}$ derivatives necessarily vanish. This leaves us with the first term on the right hand side of \eqref{derstep}. This is a special case of \eqref{derid} obtained by exchanging $(i,m)\leftrightarrow(j,n)$ and shifting $m\mapsto m-1$. Therefore,
\be
\p_{\lambda_{i1}}^{m-1}\p_{\lambda_{j1}}^{n+1}\bigg(\frac{\lambda_{i1}^{n-1}\lambda_{j1}^{m+1}}{\la i\,j\ra}\bigg) = (-1)^{n+1}(m+n)!\,\frac{\lambda_{i1}^{n+1}\lambda_{i2}^{n-1}\lambda_{j1}^{m-1}\lambda_{j2}^{m+1}}{\la i\,j\ra^{m+n+1}}\,.
\ee
Acting on this with $\big(n\lambda_{i1}\p_{\lambda_{i1}}+m\lambda_{j2}\p_{\lambda_{j2}}+mn\big)$ yields the third line of \eqref{Jmn2},
\begin{multline}\label{Jmnl3}
\big(n\lambda_{i1}\p_{\lambda_{i1}}+m\lambda_{j2}\p_{\lambda_{j2}}+mn\big)\,\p_{\lambda_{i1}}^{m-1}\p_{\lambda_{j1}}^{n+1}\bigg(\frac{\lambda_{i1}^{n-1}\lambda_{j1}^{m+1}}{\la i\,j\ra}\bigg)\\
= (-1)^{n+1}(m+n)!\,\frac{\lambda_{i1}^{n+1}\lambda_{i2}^{n-1}\lambda_{j1}^{m-1}\lambda_{j2}^{m+1}}{\la i\,j\ra^{m+n+2}}\\
\times\Bigl(\bigl[m(m+1)+n(n+1)+mn\bigr]\,\lambda_{i2}\lambda_{j1} + mn\lambda_{i1}\lambda_{j2}\Bigr)\,.
\end{multline}
Similarly, to evaluate the derivatives in the second line of \eqref{Jmn2}, we can use the identity
\be
\p_{\lambda_{i1}}^{m-1}\p_{\lambda_{j1}}^n\bigg(\frac{\lambda_{i1}^{n-1}\lambda_{j1}^m}{\la i\,j\ra}\bigg) = (-1)^{n}(m+n-1)!\,\frac{\lambda_{i1}^{n}\lambda_{i2}^{n-1}\lambda_{j1}^{m-1}\lambda_{j2}^{m}}{\la i\,j\ra^{m+n}}
\ee
that follows by exchanging $(i,m)\leftrightarrow(j,n)$ and shifting $(m,n)\mapsto(m-1,n-1)$ in \eqref{derid}. Acting on this with $(m+n)\lambda_{i1}\lambda_{j2}\p_{\lambda_{i1}}\p_{\lambda_{j2}}$ generates the second line of \eqref{Jmn2},
\begin{multline}\label{Jmnl2}
(m+n)\lambda_{i1}\lambda_{j2}\p_{\lambda_{j2}}\p_{\lambda_{i1}}^m\p_{\lambda_{j1}}^n\bigg(\frac{\lambda_{i1}^{n-1}\lambda_{j1}^m}{\la i\,j\ra}\bigg)
= (-1)^n(m+n)!\,\frac{\lambda_{i1}^{n}\lambda_{i2}^{n-1}\lambda_{j1}^{m-1}\lambda_{j2}^{m}}{\la i\,j\ra^{m+n+2}}\\
\times\Bigl(mn\,\bigl(\lambda_{i1}^2\lambda_{j2}^2+\lambda_{i2}^2\lambda_{j1}^2\bigr) + \bigl[m(m+1)+n(n+1)\bigr]\,\lambda_{i1}\lambda_{i2}\lambda_{j1}\lambda_{j2}\Bigr)\,.
\end{multline}
With these in hand, we have collected all the ingredients we need.

The sum of \eqref{Jmnl3} and \eqref{Jmnl2} collapses to
\be\label{Jmnsimple}
(-1)^n(m+n)!\,mn\,\frac{(\lambda_{i1}\lambda_{i2})^n(\lambda_{j1}\lambda_{j2})^m}{\la i\,j\ra^{m+n+1}}\,.
\ee
This showed some pretty dramatic cancellations! Plugging this into \eqref{Jmn2} reduces $\cJ_{m,n}$ down to
\be\label{Jmn3}
\cJ_{m,n} = -\frac{1}{2}\,\frac{1}{((m+n-1)!)^2}\,\frac{(-\lambda_{i1}\lambda_{i2}\lambda_{j1}\lambda_{j2}[i\,j])^{m+n-1}}{\la i\,j\ra^{m+n+1}}\qquad \forall\; m,n\geq 1\,.
\ee
Yet again, this happens to depend only on the sum $m+n$. Comparing this with the corresponding formula for $\cI_{m,n}$ found in \eqref{Imnmain} shows that the two are related by a shift $m+n\mapsto m+n-1$, but only up to an overall minus sign. This sign will be crucial for further cancellations to come.

Lastly, we come to the edge cases $\cJ_{m,0}$ and $\cJ_{0,n}$. At first glance, they appear to result in ill-defined expressions of the form $0\times\infty$. Indeed, $\cJ_{m,0}$ can be naively evaluated by repeating the steps used to obtain \eqref{Jmn1} but with $n=0$ chosen from the start. This yields 
\begin{multline}\label{Jm0}
    \cJ_{m,0} = \frac{(\lambda_{i1}\lambda_{i2})^{m-1}}{m!(m-1)!}\,\frac{1}{\lambda_{j1}}\int_0^\infty\frac{\d s_i}{s_i^{m}}\,(1-s_i)^{m-1}\\
    \times\Big\{s_i\lambda_{i1}[i\,j]\,\p_{\lambda_{i1}}^m\p_{\lambda_{i2}}^m\cS_{m}- m\,\p_{\lambda_{i1}}^{m-1}\p_{\lambda_{i2}}^{m+1}\p_{\lambda_{j1}}\cS_{m+1}\Big\}\,.
\end{multline}
Substituting for $\cS_m$ and $\cS_{m+1}$ from \eqref{Sval}, by now standard manipulations show that
\begin{align}
    \p_{\lambda_{i1}}^m\p_{\lambda_{i2}}^m\cS_m &= \frac{m(m+1)}{2}\,\frac{(-\lambda_{j1}\lambda_{j2})^m(s_i[i\,j])^{m-2}}{\la i\,j\ra^{m+2}}\,,\\
    \p_{\lambda_{i1}}^{m-1}\p_{\lambda_{i2}}^{m+1}\p_{\lambda_{j1}}\cS_{m+1} &= \frac{m+1}{2}\,\frac{\lambda_{i1}(-\lambda_{j1}\lambda_{j2})^m(s_i[i\,j])^{m-1}}{\la i\,j\ra^{m+2}}\,.
\end{align}
From these, it becomes immediately clear that both terms in the integrand of \eqref{Jm0} lead to divergent $s_i$ integrals, but also that they cancel against each other. So we are led to try and make sense of the result $\cJ_{m,0}=0\times\int_0^1\d\log s_i = 0\times\infty$. It turns out that if one naively sets this to $0$, 
one finds an answer that matches Hawking et al.'s calculation \cite{Hawking:1979pi} as well as the chiral algebra prediction, but only up to an overall sign error!

It so happens that our computations suggest a natural regularization procedure to overcome this hurdle.\footnote{There may exist other more systematic regularization schemes which we leave for future analyses.} We notice that the first line of \eqref{Jmn2} has an apparent pole at $mn=0$. But the residue at that pole vanishes because the coefficient of the pole simplifies to \eqref{Jmnsimple}: an expression with a zero at $mn=0$. This cancels the pole and results in \eqref{Jmn3} which is finite and nonvanishing as $n\to0$ with $m\neq0$ or vice versa. This resolves our $0\times\infty$ singularity.

More precisely, one analytically continues the order $N^n$ term \eqref{phin} to arbitrary $n$ (away from non-negative integers) by replacing factorials with gamma functions:
\be
\phi_j^n(x) = \frac{1}{\lambda_{j1}\lambda_{j2}}\biggl(-\lambda_{j1}\lambda_{j2}\frac{[u\,j][\hat u\,j]}{\|u\|^2}\biggr)^n\frac{1}{\Gamma(n+1)\Gamma(n)}\int_0^1\d s_j\,s_j^n(1-s_j)^{n-1}\e^{\im s_j p_j\cdot x}\,.
\ee
$n$ itself now takes up the role of a regulator for $\cJ_{m,0}$. Using \eqref{sdelta}, we understand $\phi_j^0 = \e^{\im p_j\cdot x}/\lambda_{j1}\lambda_{j2}$ as its $n\to0$ limit. Plugging $\phi_i^m$ and $\phi_j^n$ in $\cJ_{m,n}$ results in an expression analogous to \eqref{Jmn1}, but the $n^\text{th}$ order derivatives are no longer well-defined for non-integral $n$. Instead, one employs Mellin identities of the form
\be
[u\,j]^a = \frac{1}{\Gamma(-a)}\int_0^\infty\frac{\d\omega}{\omega^{a+1}}\,\e^{-\omega[u\,j]}\,,\qquad a\not\in\Z_{\leq0}\,,
\ee
to bring factors like $[u\,j]^a,[\hat u\,j]^b$ in the exponential. The evaluation then proceeds as before using the Fourier identity \eqref{fourierid}. 

But since the integrand is analytic in $n$ and the integrals are convergent, we do not need to repeat the calculation. To find the answer, one simply analytically continues \eqref{Jmn3} to non-integral $n$ by replacing $(m+n-1)!\mapsto\Gamma(m+n)$. Following which, one takes the limit $n\to0$. This results in a finite limit when $m\neq0$. So we conclude that, with this choice of regularization, $\cJ_{m,0}$ coincides with the limit of $\cJ_{m,n}$ as $n\to0$. Hence,
\be\label{Jmn4}
\cJ_{m,n} = -\frac{1}{2}\,\frac{1}{((m+n-1)!)^2}\,\frac{(-\lambda_{i1}\lambda_{i2}\lambda_{j1}\lambda_{j2}[i\,j])^{m+n-1}}{\la i\,j\ra^{m+n+1}}\qquad \forall\; m,n\neq0\,.
\ee
This expression also subsumes the case $\cJ_{0,0}=0$ because of $1/(-1)!=1/\Gamma(0)=0$.

\paragraph{Matching the 2-point amplitude.} For the reader's convenience, let's reprise the fruits of our calculation in the bulk. The final results for $\cI_{m,n}$ and $\cJ_{m,n}$ are
\be
\begin{split}
    \cI_{m,n} &= \frac{1}{2}\,\frac{1}{((m+n)!)^2}\,\frac{(-\lambda_{i1}\lambda_{i2}\lambda_{j1}\lambda_{j2}[i\,j])^{m+n}}{\la i\,j\ra^{m+n+2}}\,,\\
    \cJ_{m,n} &= -\frac{1}{2}\,\frac{1}{((m+n-1)!)^2}\,\frac{(-\lambda_{i1}\lambda_{i2}\lambda_{j1}\lambda_{j2}[i\,j])^{m+n-1}}{\la i\,j\ra^{m+n+1}}\,,
\end{split}
\ee
Both of these have turned out to be trivially symmetric in $m$ and $n$.

With these expressions, we can resum the series \eqref{A2series} for the 2-gluon amplitude. Since $\cJ_{m,n}$ is only non-zero for $m,n\geq1$, we first rewrite this series expansion as
\be
    A(\phi_i,\phi_j) = \,\tr(\sT_{a_i}\sT_{a_j})\,\bigg\{2\sum_{m,n=0}^\infty N^{m+n+1}\cI_{m,n} + 2\sum_{m,n=0}^\infty N^{m+n}\cJ_{m,n}\bigg\}\,.
\ee
The factors of $2$ arise from symmetrizing over $i$ and $j$. The sums can be performed analytically. The first sum yields
\begin{align}\label{sumImn}
    &2\sum_{m,n=0}^\infty N^{m+n+1}\cI_{m,n} = \sum_{m,n=0}^\infty\frac{N}{((m+n)!)^2}\,\frac{(-N\lambda_{i1}\lambda_{i2}\lambda_{j1}\lambda_{j2}[i\,j])^{m+n}}{\la i\,j\ra^{m+n+2}}\nonumber\\
    &= \frac{N}{\la i\,j\ra^2}\,J_0\bigg(\sqrt{4N\lambda_{i1}\lambda_{i2}\lambda_{j1}\lambda_{j2}\frac{[i\,j]}{\la i\,j\ra}}\bigg) \nonumber\\
    &\hspace{3cm}- \frac{N}{\la i\,j\ra^{5/2}}\,\sqrt{N\lambda_{i1}\lambda_{i2}\lambda_{j1}\lambda_{j2}[i\,j]}\,J_1\bigg(\sqrt{4N\lambda_{i1}\lambda_{i2}\lambda_{j1}\lambda_{j2}\frac{[i\,j]}{\la i\,j\ra}}\bigg)\,.
\end{align}
where $J_0$, $J_1$ are Bessel functions of the first kind. Similarly, the second sum produces
\begin{align}\label{sumJmn}
    &2\sum_{m,n=0}^\infty N^{m+n}\cJ_{m,n} = -\sum_{m,n=0}^\infty\frac{N}{((m+n-1)!)^2}\,\frac{(-N\lambda_{i1}\lambda_{i2}\lambda_{j1}\lambda_{j2}[i\,j])^{m+n-1}}{\la i\,j\ra^{m+n+1}}\nonumber\\
    &= -\frac{2N}{\la i\,j\ra^2}\,J_0\bigg(\sqrt{4N\lambda_{i1}\lambda_{i2}\lambda_{j1}\lambda_{j2}\frac{[i\,j]}{\la i\,j\ra}}\bigg) \nonumber\\
    &\hspace{3cm}+ \frac{N}{\la i\,j\ra^{5/2}}\,\sqrt{N\lambda_{i1}\lambda_{i2}\lambda_{j1}\lambda_{j2}[i\,j]}\,J_1\bigg(\sqrt{4N\lambda_{i1}\lambda_{i2}\lambda_{j1}\lambda_{j2}\frac{[i\,j]}{\la i\,j\ra}}\bigg)\,. 
\end{align}
The importance of finding opposite signs in $\cI_{m,n}$ and $\cJ_{m,n}$ now becomes manifest.

Adding \eqref{sumImn} to \eqref{sumJmn} engenders a final cancellation between the $J_1$ functions. One thus finds the total 2-point amplitude
\be
    A(\phi_i,\phi_j) = -\frac{N}{\la i\,j\ra^2}\,J_0\bigg(\sqrt{4N\lambda_{i1}\lambda_{i2}\lambda_{j1}\lambda_{j2}\frac{[i\,j]}{\la i\,j\ra}}\bigg)\,\tr(\sT_{a_i}\sT_{a_j})\,,
\ee
This amplitude has little group weight $-2$ in each of the two particles, which is consistent with $A(\phi_i,\phi_j)$ being an amplitude of two positive helicity gluons. It agrees with the result of Hawking et al.\ \cite{Hawking:1979pi} if one divides their scalar amplitude by a factor of $1/\lambda_{i1}\lambda_{i2}\lambda_{j1}\lambda_{j2}$ (that came from the gluon polarizations) and dresses it with a color trace.

Finally, fix the little group freedom by setting $\lambda_{i\al}=(1,z_i)$, $\lambda_{j\al}=(1,z_j)$. This yields $\la i\,j\ra=z_{ij}$. If one also identifies $\sT_{a_i}\leftrightarrow\sT_{pq}$, $\sT_{a_j}\leftrightarrow\sT_{rs}$, then this result for the amplitude coincides perfectly with the prediction \eqref{2ptprediction} from celestial CFT. This is one of the main results of this paper. It is notable how drastically simpler the chiral algebra calculation was, perhaps proving the merit of the search for more celestial CFTs!


\subsection{Celestial OPE in WZW$_4$}
\label{sec:gluonope}

Combinatorially, higher-point amplitudes in curved backgrounds like Burns space are controlled by the same Feynman diagrammatics encountered in flat space (for example, Witten diagrams arise as Feynman diagrams for AdS). But calculating them in closed-form is much harder than the calculation of 2-point amplitudes. For this reason, in holographic setups, one often resorts to computing only their collinear limits. These are the limits in which boundary insertions approach each other. So they can be matched onto operator product expansions in the dual theory, providing another independent verification of the holographic duality.

\paragraph{Celestial OPE from the chiral algebra.} Let's start by obtaining the chiral algebra prediction for the $\phi\phi$ collinear limit of WZW$_4$. It will again turn out to be a dramatically simpler calculation on the 2d side, and will take up a good chunk of time to describe on the 4d side.

We will compute the OPE of $J_{pq}(z_i,\tilde\lambda_i)$ and $J_{rs}(z_j,\tilde\lambda_j)$ to zeroth and first order in $N$. It is the latter that corresponds to a genuine backreaction effect. The contribution at zeroth order in $N$ comes from single contractions. In the planar limit, it follows from the ADHM constraint \eqref{adhm} that we are free to reorder the factors of $[X\,i]$ and $[X\,j]$ -- viewed as $N\times N$ matrices -- inside the normal ordering. So one finds, up to regular terms,
\be
\begin{split}
    J_{pq}(z_i,\tilde\lambda_i)\,J_{rs}(z_j,\tilde\lambda_j) &\sim \frac{f_{pq,rs}{}^{tu}}{z_{ij}}\,:\!I_t\e^{z_j([X\,i]+[X\,j])}I_u\!:\!(z_j)
\end{split}
\ee
This follows because only contractions between the $I$'s survive in the planar limit at order $1/z_{ij}$, and a typical contribution looks like the contraction in figure \ref{opeN0}.
\begin{figure}
\centering
\begin{subfigure}[t]{0.4\textwidth}
    \centering
    \includegraphics[scale=0.3]{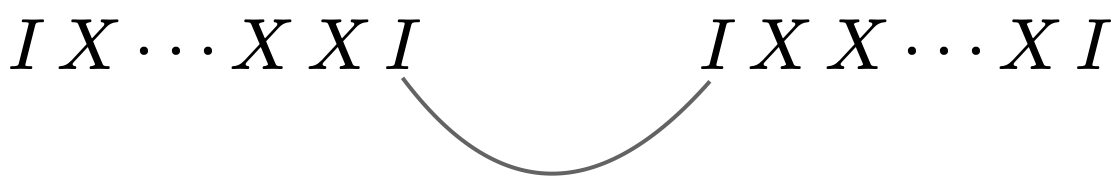}
    \caption{Single contraction}
    \label{opeN0}
\end{subfigure}
\hspace{1.5cm}
\begin{subfigure}[t]{0.4\textwidth}
    \centering
    \includegraphics[scale=0.3]{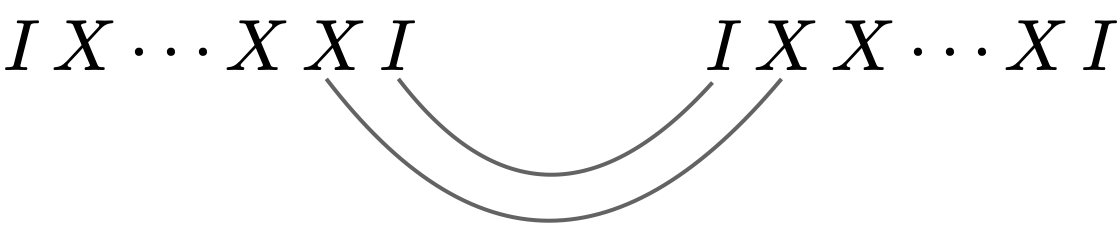}
    \caption{Adjacent double contractions}
    \label{opeN1}
\end{subfigure}
\caption{Planar contractions contributing to the OPE of two gluon currents at zeroth and first order in the backreaction strength $N$ respectively.}
\end{figure}
Abbreviating the composite indices $pq,rs,\dots$ by adjoint indices $a,b,\dots$, this reduces to
\be\label{jjN0}
J_{a}(z_i,\tilde\lambda_i)\,J_{b}(z_j,\tilde\lambda_j) \sim \frac{f_{ab}{}^{c}}{z_{ij}}\,J_c(z_j,\tilde\lambda_i+\tilde\lambda_j)
\ee
which is the momentum space version of the like-helicity gluon celestial OPE that one encounters in flat space \cite{Costello:2022wso}.

At the next order in $N$, one Wick contracts a pair of $I$'s, then further contracts a pair of adjacent $X$'s as shown in figure \ref{opeN1}. The result may be written as a double sum
\begin{multline}
    J_{pq}(z_i,\tilde\lambda_i)\,J_{rs}(z_j,\tilde\lambda_j) \\
    \sim -\frac{Nz_iz_j[i\,j]f_{pq,rs}{}^{tu}}{z_{ij}^2}\sum_{m,n\geq1}\frac{1}{m!n!}\,:\!I_t\bigl(z_i)(z_i[X(z_i)\,i]\bigr)^{m-1}\bigl(z_j[X(z_j)\,j]\bigr)^{n-1}I_u(z_j)\!:
\end{multline}
where we have suppressed the $1/z_{ij}$ term for clarity. The right hand side is no longer a single hard gluon state. But using the identity
\be
\int_0^1\d\omega_i\int_0^1\d\omega_j\,\e^{\omega_ix+\omega_jy} = \sum_{m,n\geq1}\frac{x^{m-1}y^{n-1}}{m!n!}\,,
\ee
it can be expressed as a smearing over hard gluon states and their descendants,
\begin{multline}\label{jjN1full}
    J_{pq}(z_i,\tilde\lambda_i)\,J_{rs}(z_j,\tilde\lambda_j) \sim -\frac{Nz_j^2[i\,j]f_{pq,rs}{}^{tu}}{z_{ij}^2}\int_0^1\d\omega_i\int_0^1\d\omega_j:\!I_t\e^{z_j\omega_i[X\,i]+z_j\omega_j[X\,j]}I_u\!:\!(z_j)\\
    - \frac{Nz_j[i\,j]f_{pq,rs}{}^{tu}}{z_{ij}}\int_0^1\d\omega_i\int_0^1\d\omega_j:\!\p_{z_j}\Bigl(z_jI_t\e^{z_j\omega_i[X\,i]}\Bigr)\e^{z_j\omega_j[X\,j]}I_u\!:\!(z_j)\,.
\end{multline}
The second line of this expression may be expressed as a combination of conformal and Kac-Moody descendants along the lines of \cite{Adamo:2022wjo,Ren:2023trv}. The integration variables $\omega_i$ are reminiscent of ``Feynman parameters''. We will soon see how they arise as genuine Feynman parameters when computing the celestial OPE through bulk Feynman diagrams.

The first term in \eqref{jjN1full} is the more interesting primary term. It can be succinctly summarized as
\be\label{jjN1}
J_{a}(z_i,\tilde\lambda_i)\,J_{b}(z_j,\tilde\lambda_j) \sim -\frac{Nz_j^2[i\,j]f_{ab}{}^{c}}{z_{ij}^2}\int_0^1\d\omega_i\int_0^1\d\omega_j\,J_c(z_j,\omega_i\tilde\lambda_i+\omega_j\tilde\lambda_j)\,.
\ee
This is what we will spend most of our time reproducing from the bulk. The second term in \eqref{jjN1full} is proportional to $[i\,j]/z_{ij}$, so it is easily disentangled from the contributions \eqref{jjN0}, \eqref{jjN1} that sit at orders $1/z_{ij}$ and $[i\,j]/z_{ij}^2$ respectively. Reproducing it from the bulk is left to future work.

\paragraph{Celestial OPE from collinear limits in the bulk.} The bulk quantity that computes the celestial OPE is what is known as (in the scattering amplitudes literature) the \emph{tree level splitting function}.  The tree level splitting function tells us the collinear singularities in any amplitude or, more generally, in any form factor. Its relevance for celestial holography was realized in \cite{Fan:2019emx,Pate:2019lpp}. It can also be thought of as a Berends-Giele current \cite{Berends:1987me}.

Since WZW$_4$ is a gauge-fixed version of self-dual Yang-Mills, we can identify (for rather formal reasons) the tree level splitting function of WZW$_4$ with the terms in the tree level splitting function of self-dual Yang-Mills theory which only involve positive helicity states.  One very concrete statement of one aspect of  our duality is then:
\begin{conjecture}
	The tree level splitting function of self-dual Yang-Mills on Burns space, for gauge group $\SO(8)$, is precisely the OPE of the large $N$ chiral algebra. 
\end{conjecture}

Let us now explain how to compute the celestial OPE (or splitting function) from the bulk. The 3-point interactions of WZW$_4$ on a scalar-flat K\"ahler manifold $M$ were displayed in \eqref{phiac}. If one starts with three solutions $\phi_i,\phi_j,\phi_k$ of the adjoint-valued Laplace equation, their scattering amplitude in WZW$_4$ is found by evaluating the on-shell 3-point vertex
\be
\frac{\im}{24\pi^2}\int_{M}\omega\wedge\tr\left(\phi\,[\p\phi,\dbar\phi]\right)
\ee
on the combination $\phi=\veps_i\phi_i+\veps_j\phi_j+\veps_k\phi_k$ and extracting the coefficient of $\veps_i\veps_j\veps_k$. Recall that we are using the notation $[\p\phi,\dbar\phi]\equiv\p\phi\wedge\dbar\phi + \dbar\phi\wedge\p\phi$ as is appropriate for the Lie bracket of 1-forms. Thus, the 3-point amplitude reads
\be\label{3ptgluon}
A(\phi_i,\phi_j,\phi_k) = \frac{\im}{24\pi^2}\int_{M}\omega\wedge\tr\,\Bigl(\phi_k\bigl([\p\phi_i,\dbar\phi_j]+[\p\phi_j,\dbar\phi_i]\bigr) + \text{cyclic}\Bigr)\,.
\ee
Instead of trying to determine this 3-point amplitude in all its glory, we confine ourselves to computing its expansion in small $\la i\,j\ra$. This defines the so-called \emph{holomorphic collinear limit}. 

We also remark that in past work, we actually have computed this 3-point amplitude to leading order in the backreaction, and the result matches the correlator of three leading soft gluon currents $J[0,0]$ in the dual chiral algebra. We will omit repeating the details of that approach here, but the interested reader is referred to the supplemental materials of \cite{Costello:2022jpg} for the computation. Alternatively, an analogous computation involving two gluons $\phi_i,\phi_j$ and a gravitational $E$-type mode $\rho^E_k$ will be performed in the next section.

What we will compute now is the \emph{holographic OPE}. The holographic OPE of two states $\phi_i\equiv\phi_{a_i}(x|\lambda_i,\tilde\lambda_i)$, $\phi_j$ is defined to be the ``off-shell state'' $\phi_{ij}$ such that the linear combination
\be
\veps_i\phi_i+\veps_j\phi_j + \veps_i\veps_j\phi_{ij}
\ee
solves the nonlinear field equation \eqref{phieom} to first order in the product $\veps_i\veps_j$. To this order, the field equation reads
\be\label{opepde}
\omega\wedge\p\dbar\phi_{ij} = -\frac{1}{2}\,\omega\wedge\left([\p\phi_i,\dbar\phi_j]+[\p\phi_j,\dbar\phi_i]\right)\,.
\ee
From this, one observes that the ``2-point amplitude'' of an on-shell state $\phi_k$ and the off-shell state $\phi_{ij}$, a priori defined to be
\be\label{A3from2}
A(\phi_{ij},\phi_k) = -\frac{\im}{8\pi^2}\int_M\omega\wedge\tr\left(\p\phi_{ij}\wedge\dbar\phi_{k}+\p\phi_{k}\wedge\dbar\phi_{ij}\right)\,,
\ee
computes the on-shell 3-point amplitude \eqref{3ptgluon} on integrating by parts and summing over cyclic permutations of the particle labels $i,j,k$. 
This idea is again adapted from the perturbiner formalism \cite{Rosly:1996vr,Rosly:1997ap}. In the flat space limit $N\to0$, the OPE $\phi_{ij}$ reduces to the usual notion of a (holomorphic) collinear splitting function.

In what follows, we will solve the field equation for $\phi_{ij}$ order-by-order in the backreaction. Crucially, this takes the form of a series in $\la i\,j\ra$ whose coefficients can be expressed in terms of \emph{on-shell} states. Plugging this expansion into \eqref{A3from2} and summing over permutations, one obtains an expansion of the 3-point amplitude $A(\phi_i,\phi_j,\phi_k)$ as a series in $\la i\,j\ra$ whose coefficients are on-shell 2-point amplitudes. This is exactly the way a holographic correlator is supposed to behave. And indeed, following this computation, we will apply our holographic dictionary to convert this into a celestial OPE. An independent computation in the chiral algebra will then match this celestial OPE on the nose. 

\medskip

Remembering the expression \eqref{kahlerlap} for the Laplacian on K\"ahler manifolds, along with the normalization convention $\omega^2=-2\sqrt{|g|}\,\d^4x$, one can invert \eqref{opepde} to find the solution
\be\label{opeinvert}
\phi_{ij}(x) = \im\int_{M}\omega\wedge\left([\p\phi_i,\dbar\phi_j]+[\p\phi_j,\dbar\phi_i]\right)\Bigr|_{x'}\,G(x,x')\,,
\ee
where the integral is over spacetime points $x'\in M$, and $G(x,x')$ is the Green's function of the Laplacian satisfying
\be
\lap_{x}G(x,x') = \frac{1}{\sqrt{|g|}}\,\delta^4(x-x')\,.
\ee
We anticipate that the result of the integral \eqref{opeinvert} will be expandable on universal singularities in $\la i\,j\ra=z_{ij}$, with coefficients that are on-shell states. This ``splitting'' idea is depicted in figure \ref{splitting}, with the right hand side showing the rough structure of the holographic OPE.

\begin{figure}
    \centering
    \includegraphics[scale=0.4]{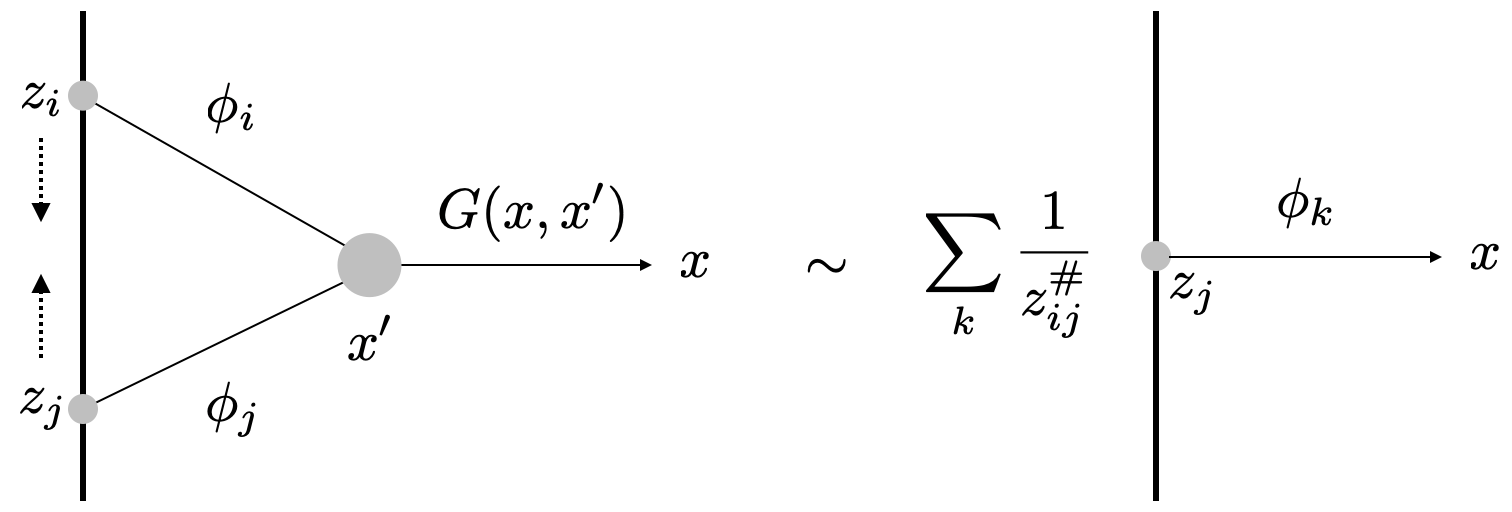}
    \caption{The holographic OPE as the collinear expansion of a partially off-shell 3-point vertex. The vertical line represents the celestial sphere.}
    \label{splitting}
\end{figure}

To be clear, the Green's function here is just the scalar propagator and does not contain any color factors. On Burns space $M=\widetilde\C^2$, an explicit expression for this propagator can be read off by stripping the color factor from the adjoint-valued propagator that we arrived at in \eqref{phiprop}:
\be\label{scalarprop}
G(x,x') = -\frac{1}{8\pi^2}\left(\|u-u'\|^2+\frac{N|[u\,u']|^2}{\|u\|^{2}\|u'\|^{2}}\right)^{-1}\,,
\ee
where $x^{\al\dal}=(u^{\dal},\hat u^{\dal})$ and $x'^{\al\dal}=(u'^{\dal},\hat u'^{\dal})$, etc. This is the Green's function of the Laplacian on Burns space. We will use this in conjunction with formula \eqref{opeinvert} to determine our holographic OPE.

In the integrand of \eqref{opeinvert}, partition the K\"ahler form \eqref{omegaburns} into its flat space part and the backreaction:
\be
\omega = \omega_0+N\omega_1\,,\qquad \omega_0 = \im\,\d u_{\dal}\wedge\d\hat u^{\dal}\,,\qquad\omega_1 = \frac{\im\,[u\,\d u]\wedge[\hat u\,\d\hat u]}{\|u\|^4}\,.
\ee
At the same time, the Green's function admits the expansion
\be\label{greenexp}
G(x,x') = \sum_{n=0}^\infty N^nG_n(x,x')\,,\qquad G_n(x,x') = G_0(x,x')\left(\frac{-1}{\|u-u'\|^2}\frac{|[u\,u']|^2}{\|u\|^2\|u'\|^2}\right)^n\,,
\ee
with $G_0(x,x')$ being the flat space propagator
\be
G_0(x,x') = -\frac{1}{8\pi^2}\frac{1}{\|u-u'\|^2} = -\frac{1}{4\pi^2}\frac{1}{(x-x')^2}\,.
\ee
Also recall the expansions \eqref{phiijexp} of the quasi-momentum eigenstates, with coefficients given by \eqref{phin}. 

Using these, let's expand $\phi_{ij}(x)$ as a power series in $N$,
\be
\phi_{ij}(x) = [\sT_{a_i},\sT_{a_j}]\sum_{p=0}^\infty N^p\phi_{ij}^{p}(x)\,.
\ee
Matching $\phi_{ij}^p$ with the term of order $N^p$ on right hand side of \eqref{opeinvert} yields
\begin{multline}\label{opeterm}
\phi_{ij}^p(x) = \sum_{\ell+m+n=p}\im\int_{\widetilde\C^2}\omega_0\wedge\left(\p\phi_i^m\wedge\dbar\phi_j^n -\p\phi^n_j\wedge\dbar\phi^m_i\right)\Bigr|_{x'}\,G_{\ell}(x,x')\\
+ \sum_{\ell+m+n=p-1}\im\int_{\widetilde\C^2}\omega_1\wedge\left(\p\phi_i^m\wedge\dbar\phi_j^n -\p\phi^n_j\wedge\dbar\phi^m_i\right)\Bigr|_{x'}\,G_{\ell}(x,x')\,,
\end{multline}
where the sums run over non-negative integers $\ell,m,n$ as usual. We will content ourselves with working these out for $p=0,1$ and matching the result with predictions from the holographic dual.

\paragraph{Flat space celestial OPE.} The leading term $\phi^0_{ij}$ should reproduce the momentum space celestial OPE of positive helicity hard gluons that one obtains in flat space \cite{Fan:2019emx}. Equivalently, its collinear expansion should reduce to the flat space tree level splitting function.\footnote{A useful exposition of collinear limits of flat space amplitudes may be found in \cite{Taylor:2017sph}.}

Let us verify this. To this order, only the terms $\phi_i^0 = \e^{\im p_i\cdot x}/\lambda_{i1}\lambda_{i2}$ contribute. Setting $p=0$ in \eqref{opeterm}, one can drop the second sum, whereas the first sum reduces to a single Fourier integral
\be\label{phiij0}
\phi^0_{ij} = -\frac{(\lambda_{i1}\lambda_{j2}+\lambda_{i2}\lambda_{j1})[i\,j]}{\lambda_{i1}\lambda_{i2}\lambda_{j1}\lambda_{j2}}\int\d^4x'\,\e^{\im(p_i+p_j)\cdot x'}G_0(x,x')\,.
\ee
Because the flat space propagator solves $(\lap_{0})_{x}G_0(x,x') = \delta^4(x-x')$, one notes that this is just the solution of a sourced flat space Laplace equation,
\be
\lap_0\phi^0_{ij} = -\frac{(\lambda_{i1}\lambda_{j2}+\lambda_{i2}\lambda_{j1})[i\,j]}{\lambda_{i1}\lambda_{i2}\lambda_{j1}\lambda_{j2}}\,\e^{\im(p_i+p_j)\cdot x}\,,
\ee
where $\lap_0 = \p^{\al\dal}\p_{\al\dal}$ is the Laplacian on $\R^4$. Solving this does not actually require performing the integral in \eqref{phiij0}, as one can easily guess the value of $\lap_0^{-1}\e^{\im(p_i+p_j)\cdot x}$:
\begin{align}
\phi^0_{ij} &= -\frac{(\lambda_{i1}\lambda_{j2}+\lambda_{i2}\lambda_{j1})[i\,j]}{\lambda_{i1}\lambda_{i2}\lambda_{j1}\lambda_{j2}}\,\frac{\e^{\im(p_i+p_j)\cdot x}}{\im^2(p_i+p_j)^2}\nonumber\\
&= \frac{(\lambda_{i1}\lambda_{j2}+\lambda_{i2}\lambda_{j1})}{2\lambda_{i1}\lambda_{i2}\lambda_{j1}\lambda_{j2}}\,\frac{\e^{\im(p_i+p_j)\cdot x}}{\la i\,j\ra}\nonumber\\
&\sim \frac{\lambda_{j1}}{\lambda_{i1}}\,\frac{1}{\la i\,j\ra}\,\frac{1}{\lambda_{j1}\lambda_{j2}}\,\exp\left(\im x^{\al\dal}\lambda_{j\al}\bigg(\tilde\lambda_{j\dal}+\frac{\lambda_{i1}}{\lambda_{j1}}\tilde\lambda_{i\dal}\bigg)\right)\,.\label{phi0ijsol}
\end{align}
The second line is the exact result. The third line displays the singular term in its holomorphic collinear expansion around $\la i\,j\ra=0$. It has been obtained by eliminating $\lambda_{i2}$ using $\lambda_{i2}=(\la i\,j\ra + \lambda_{i1}\lambda_{j2})/\lambda_{j1}$, then expanding in $\la i\,j\ra$. The weird factors of $\lambda_{i1}/\lambda_{j1}$ in various places are present purely to ensure correct little group scalings. We also remark that, as the Fourier transform of $1/(x-x')^2$ is well-known, one can also obtain \eqref{phi0ijsol} by direct integration of \eqref{phiij0} for real $p_i^\mu+p_j^\mu$ and analytic continuation to complex momenta.

To recast this as an OPE expansion of a 2d CFT, it is useful to fix the little group scalings by setting $\lambda_{i\al}=(1,z_i)$, $\lambda_{j\al}=(1,z_j)$. Expanding in small $z_{ij}$ gives the singular part of the leading OPE
\be\label{N0ope}
\phi_{ij} \sim \frac{[\sT_{a_i},\sT_{a_j}]}{z_{ij}}\,\phi^0(x|z_j,\tilde\lambda_i+\tilde\lambda_j) + \mathrm{O}(N)\,.
\ee
We have denoted by $\phi^0(x|z,\tilde\lambda)$ the flat space momentum eigenstate $\e^{\im\la\lambda|x|\tilde\lambda]}/\lambda_1\lambda_2$ for the little group fixing $\lambda_\al=(1,z)$. This is the expected Kac-Moody OPE of two positive helicity hard gluons. It encodes the celestial OPE of infinitely many soft gluons, as was first realized in \cite{Guevara:2021abz}. Having obtained consistency with the bottom-up expectations of celestial holography, we are emboldened to embark on computing some backreaction effects!

\paragraph{Absence of branch cuts in the celestial OPE.}
Before we embark on the rather challenging computation of the OPE in the backreacted geometry, let us explain why the celestial OPE must always yield a rational functional of the kinematic variables $[i\,j]$, $\la i\,j\ra$.  From the CFT point of view this is of course obvious.  However, certain intermediate steps in the OPE computation can lead to logarithms, and it is useful to have a general conceptual argument for why they must all cancel in the end.

Recall that the celestial OPE is defined by 
\be
\phi_{ij}(x) = \im\int_{M}\omega\wedge\left([\p\phi_i,\dbar\phi_j]+[\p\phi_j,\dbar\phi_i]\right)\Bigr|_{x'}\,G(x,x')\,,
\ee
We then take the limit as $p_i \cdot p_j \to 0$ and retain the singular part.

However, we don't have to use the full Green's function to compute the celestial OPE.  The Green's function can be written as an integral of the heat kernel: 
\begin{equation} 
	G(x,x') = \int_{0}^\infty \d t\, K_t (x,x')  
\end{equation}
(This is a version of the Schwinger parameterization).  We can consider the Green's function with an infrared cutoff:
\begin{equation} 
	G_T(x,x') =   \int_{0}^T \d t\, K_t (x,x')  
\end{equation}
Unlike the actual Green's function, the IR cutoff Green's function has exponential decay for $\norm{x - x'}$ large:
\begin{equation} 
	G_T(x,x') \sim e^{-\norm{x-x'}^2 / T}  
\end{equation}
We will show that the celestial OPE defined using the infra-red cutoff is the same as the OPE defined using the Green's function.  The OPE is only sensitive to the UV singularities in the Green's function.

Although we present the argument with the heat-kernel cutoff, it will be clear that any reasonable IR cutoff can be used.

The celestial OPE with the IR cutoff is
\begin{equation} 
	\phi^T_{ij}(x) = \im\int_{M}\omega\wedge\left([\p\phi_i,\dbar\phi_j]+[\p\phi_j,\dbar\phi_i]\right)\Bigr|_{x'}\,G_T(x,x')\,. \label{OPE_cutoff} 
\end{equation}
If send $T \to \infty$, obviously it becomes our original expression for the celestial OPE. On the other hand, we can differentiate it with respect to $T$ to get
\begin{equation} 
	\partial_T	 \phi^T_{ij}(x) = \im\int_{M}\omega\wedge\left([\p\phi_i,\dbar\phi_j]+[\p\phi_j,\dbar\phi_i]\right)\Bigr|_{x'}\, K_T(x,x')   \,. 
\end{equation}
The heat kernel $K_T(x,x')$ has rapid decay when $\norm{x-x'}$ is large, and (unlike the Green's function) has no singularities when $\norm{x-x'}$ is small. This implies that the integral
\begin{equation} 
	\int_x F(x) K_T(x,x') 
\end{equation}
is absolutely convergent for any function $F$ which grows like $P(x) e^{C \norm{x}}$ for any polynomial $P$ and constant $C$. In particular, the integral
\begin{equation} 
	 \im\int_{M}\omega\wedge\left([\p\phi_i,\dbar\phi_j]+[\p\phi_j,\dbar\phi_i]\right)\Bigr|_{x'}\, K_T(x,x')   \,, 
\end{equation}
converges absolutely even in the collinear limit $p_i \cdot p_j \to 0$. 

This tells us that the \emph{singular part} of the celestial OPE, defined with the IR cutoff propagator $G_T(x,x')$, is independent of $T$. It remains the same in the $T\to\infty$ limit. The singular part is, of course, what we are ultimately interested in.

Let us now verify that the celestial OPE (defined using the IR cutoff propagator) is an entire analytic function of the complex momenta $p_i$, and so has no branch cuts. Let us look again at equation \eqref{OPE_cutoff}.   For any value of the momenta, the term depending on the $\phi_i$ has growth bounded by an exponential: 
\begin{equation} 
	\norm { [\p\phi_i,\dbar\phi_j]+[\p\phi_j,\dbar\phi_i] }  \le P(x) e^{C \norm{x}}  
\end{equation}
for some polynomial $P$ and constant $C$, which both depend on the momenta $p_i$.   

For  large $\norm{x-x'}$, we have
\begin{equation} 
	G_T(x,x') \le C e^{ - \norm{x-x'}^2 / T} 
\end{equation}
for some constant $C$.  Therefore, the integral \eqref{OPE_cutoff} converges absolutely in the IR region $\norm{x - x'} \gg 0$.  It also converges in the UV region, simply because $G_T(x,x')$ is a distribution.  This convergence guarantees analyticity of \eqref{OPE_cutoff} as a function of the momenta.  

This argument is a version of the familiar Paley-Weiner theorem, in its distributional incarnation proved by Schwartz.  This result states that the Fourier transform of a distribution (in this case $G_T(x,x')$) with sufficiently rapid decay is always an entire analytic function of the momenta.   

\paragraph{Celestial OPE with backreaction.} At $p=1$, one encounters three terms in the sum \eqref{opeterm}. Let us list them in order of increasing complexity that we will encounter when evaluating them,
\begin{align}
    \phi^p_{ij} = \mathcal{I}^1_{ij} + \mathcal{I}^2_{ij} + \mathcal{I}^3_{ij}\,.
\end{align}
The first of these comes from setting $\ell=m=n=p-1=0$ in the second line of \eqref{opeterm},
\be
\mathcal{I}^1_{ij} = -\int_{\R^4}\frac{u'^{\dal}\hat u'^{\dot\beta}}{\|u'\|^4}\,\bigg(\frac{\p\phi_i^0}{\p u'^{\dal}}\frac{\p\phi_j^0}{\p\hat u'^{\dot\beta}} - \frac{\p\phi_j^0}{\p u'^{\dal}}\frac{\p\phi_i^0}{\p\hat u'^{\dot\beta}}\bigg)\,G_{0}(x,x')\,\d^4x'\,.
\ee
The $\phi_i^m, \phi_j^n$ are all implicitly evaluated at $u'$. The integral has again been written as an integral over $\R^4$, meant in the sense of appropriate regularizations at the origin. The second term arises by setting $\ell=1,m=n=0$ in the first line of \eqref{opeterm},
\be
\mathcal{I}^2_{ij} = -\int_{\R^4}\eps^{\dal\dot\beta}\,\bigg(\frac{\p\phi_i^0}{\p u'^{\dal}}\frac{\p\phi_j^0}{\p\hat u'^{\dot\beta}} - \frac{\p\phi_j^0}{\p u'^{\dal}}\frac{\p\phi_i^0}{\p\hat u'^{\dot\beta}}\bigg)\,G_{1}(x,x')\,\d^4x'\,.
\ee
The final term is the most involved to compute, but also the term that will give rise to the actual backreaction effects in the OPE,
\be
\mathcal{I}^3_{ij} = -\int_{\R^4}\eps^{\dal\dot\beta}\,\bigg(\frac{\p\phi_i^1}{\p u'^{\dal}}\frac{\p\phi_j^0}{\p\hat u'^{\dot\beta}} - \frac{\p\phi_j^0}{\p u'^{\dal}}\frac{\p\phi_i^1}{\p\hat u'^{\dot\beta}} - (i\leftrightarrow j)\bigg)\,G_{0}(x,x')\,\d^4x'\,.
\ee
We will evaluate all three as series in $\la i\,j\ra$. Let us also recall the leading and subleading terms in the scattering wavefunctions that we will need:
\be
\phi_i^0(x) = \frac{\e^{\im p_i\cdot x}}{\lambda_{i1}\lambda_{i2}}\,,\qquad\phi_i^1(x) = -\frac{[u\,i][\hat u\,i]}{\|u\|^2}\int_0^1\d s_i\,s_i\,\e^{\im s_ip_i\cdot x}
\ee
and similarly for $\phi_j^0,\phi_j^1$.

In the course of evaluating these terms, we will encounter spacetime integrals of the form
\be\label{Smnkx}
\cS_{m,n}(k,x) = \frac{1}{8\pi^2}\int\frac{\d^4x'\,\e^{\im k\cdot x'}}{[(x'-x)^2/2]^m(x'^2/2)^n}\,,\qquad m,n\in\Z_+
\ee 
For instance, plugging in the values of $\phi_i^0$, $\phi_j^0$ in $\cI^1_{ij}$, and simplifiying using $\|u'\|^2 = x'^2/2$, $\|u'-u\|^2=(x'-x)^2/2$, along with the intertwining relations $\p_{\lambda_{i1}}\e^{\im p_i\cdot x'} = \im [u'i]\e^{\im p_i\cdot x'}$, etc., we find that
\be\label{I1ij1}
\cI^1_{ij} = \frac{1}{\lambda_{i1}\lambda_{i2}\lambda_{j1}\lambda_{j2}}\left(\lambda_{i1}\lambda_{j2}\p_{\lambda_{i1}}\p_{\lambda_{j2}}-\lambda_{i2}\lambda_{j1}\p_{\lambda_{i2}}\p_{\lambda_{j1}}\right)\cS_{1,2}(p_i+p_j,x)
\ee
So it is worthwile computing $\cS_{m,n}(k,x)$ in some generality.

The integration technique that we apply to compute it is a 4d adaptation of methods developed in \cite{PhysRevA.39.5062} to Fourier transform products of 3d Coulomb potentials with differing centers. It is culturally appealing to see that, in this way, our calculations happen to be souped-up cousins of standard calculations performed in quantum chemistry.

One begins by performing a Feynman parametrization to combine the two propagator denominators:
\be
\cS_{m,n}(k,x) = \frac{1}{8\pi^2}\,\frac{\Gamma(m+n)}{\Gamma(m)\Gamma(n)}\int_0^1\d s\,s^{m-1}\,(1-s)^{n-1}\int\frac{\d^4x'\,\e^{\im k\cdot(x'+s x)}}{[(x'^2+s(1-s)x^2)/2]^{m+n}}\,,
\ee
having shifted $x'\mapsto x'+sx$ for further simplification. A further Schwinger parametrization allows for the evaluation of the spacetime integral as a Fourier transform of a Gaussian:
\begin{align}
    &\cS_{m,n}(k,x) = \frac{1}{8\pi^2}\,\frac{1}{\Gamma(m)\Gamma(n)}\int_0^1\d s\,s^{m-1}\,(1-s)^{n-1}\,\e^{\im sk\cdot x}\nonumber\\
    &\hspace{4cm}\times\int_0^\infty\d t\,t^{m+n-1}\int\d^4x'\,\e^{\im k\cdot x'-t(x'^2+s(1-s)x^2)/2}\nonumber\\
    &= \frac{1}{2\,\Gamma(m)\Gamma(n)}\int_0^1\d s\,s^{m-1}\,(1-s)^{n-1}\,\e^{\im sk\cdot x}\int_0^\infty\d t\,t^{m+n-3}\,\e^{-ts(1-s)x^2/2}\,\e^{-k^2/2t}\,.\label{Smnform}
\end{align}
In practice, this is the formula for $\cS_{m,n}$ that we use when evaluating the $\lambda_{i\al},\lambda_{j\al}$ derivatives in expressions like \eqref{I1ij1}.

The $t$ integral can also be performed explicitly using the Schl\"afli integral formula for modified Bessel functions of the second kind,
\be\label{Kschlafli}
K_\nu(z) = \frac{1}{2}\,\bigg(\frac{z}{2}\bigg)^\nu\int_0^\infty\frac{\d t}{t^{\nu+1}}\,\e^{-t-z^2/4t}\,,
\ee
valid if the real part of $z^2$ is positive. This yields
\be
    \cS_{m,n}(k,x) = \frac{(k^2/x^2)^{\frac{m+n-2}{2}}}{\Gamma(m)\Gamma(n)}\int_0^1\d s\,s^{\frac{m-n}{2}}\,(1-s)^{\frac{n-m}{2}}\,\e^{\im sk\cdot x}K_{m+n-2}\!\left(\sqrt{s(1-s)k^2x^2}\right)
\ee
where the $t$ integral has been performed for $x^2>0,k^2>0$. As always, we analytically continue the final result to complex $k^2$, at least as long as the $s$ integral converges. It does not seem possible to perform the $s$ integral in closed form, but it is amenable to a series expansion in $k^2$ (which is what we ultimately seek).

To find \eqref{I1ij1}, we need $\cS_{1,2}(p_i+p_j,x)$. Reading this off from \eqref{Smnform} and evaluating the derivatives, one finds
\begin{align}
    &\cI^1_{ij} = \frac{1}{2\lambda_{i1}\lambda_{i2}\lambda_{j1}\lambda_{j2}}\int_0^1\d s\,(1-s)\,\e^{\im s(p_i+p_j)\cdot x}\int_0^\infty\d t\,\e^{-ts(1-s)\|u\|^2}\,\e^{-\la i\,j\ra[i\,j]/t}\nonumber\\
    &\times\left[s^2\bigl(\lambda_{i2}\lambda_{j1}[\hat u\,i][u\,j]-\lambda_{i1}\lambda_{j2}[u\,i][\hat u\,j]\bigr) - \frac{\la i\,j\ra[i\,j]^2}{{t^2}}\,(\lambda_{i1}\lambda_{j2}+\lambda_{i2}\lambda_{j1})\right.\nonumber\\
    &\left. +\; \frac{[i\,j]}{t}\,\Bigl(\lambda_{i1}\lambda_{j2}\bigl(\im s\lambda_{i1}[u\,i]+\im s\lambda_{j2}[\hat u\,j]+1\bigr) + \lambda_{i2}\lambda_{j1}\bigl(\im s\lambda_{i2}[\hat u\,i]+\im s\lambda_{j1}[u\,j]+1\bigr)\Bigr)\right]\,.
\end{align}
Integrating over $t$ using the integral representation \eqref{Kschlafli} converts this to
\begin{align}\label{I1ij2}
    &\cI^1_{ij} = \frac{1}{\lambda_{i1}\lambda_{i2}\lambda_{j1}\lambda_{j2}}\int_0^1\d s\,(1-s)\,\e^{\im s(p_i+p_j)\cdot x}\nonumber\\
    &\times\left[s^2\bigl(\lambda_{i2}\lambda_{j1}[\hat u\,i][u\,j]-\lambda_{i1}\lambda_{j2}[u\,i][\hat u\,j]\bigr)\sqrt{\frac{\la i\,j\ra[i\,j]}{s(1-s)\|u\|^2}}\,K_1\bigl(\sqrt{4s(1-s)\|u\|^2\la i\,j\ra[i\,j]}\bigr)\right.\nonumber\\
    &- \la i\,j\ra[i\,j]^2\,(\lambda_{i1}\lambda_{j2}+\lambda_{i2}\lambda_{j1})\sqrt{\frac{s(1-s)\|u\|^2}{\la i\,j\ra[i\,j]}}\,K_1\bigl(\sqrt{4s(1-s)\|u\|^2\la i\,j\ra[i\,j]}\bigr)\nonumber\\
    & + [i\,j]\,\Bigl(\lambda_{i1}\lambda_{j2}\bigl(\im s\lambda_{i1}[u\,i]+\im s\lambda_{j2}[\hat u\,j]+1\bigr) + \lambda_{i2}\lambda_{j1}\bigl(\im s\lambda_{i2}[\hat u\,i]+\im s\lambda_{j1}[u\,j]+1\bigr)\Bigr)\nonumber\\
    &\left.\hspace{4cm}\times\; K_0\bigl(\sqrt{4s(1-s)\|u\|^2\la i\,j\ra[i\,j]}\bigr)\right]\,.
\end{align}
As this cannot be reduced any further, we must resort to series expanding in small $\la i\,j\ra$.

At integer arguments, the leading terms in the series expansions of $K_0(z)$, $K_1(z)$ and $K_2(z)$ for small $z$ are given by
\begingroup
\allowdisplaybreaks
\begin{align}
    K_0(z) &=  - \gamma_E -\log\frac{z}{2}  + \frac{z^2}{4}\left(1-\gamma_E-\log\frac{z}{2}\right) + \mathrm{O}(z^3)\label{K0exp}\\
    K_1(z) &= \frac{1}{z} - \frac{z}{4}\left(1-2\gamma_E-2\log\frac{z}{2}\right) + \mathrm{O}(z^3)\label{K1exp}\\
    K_2(z) &= \frac{2}{z^2} - \frac{1}{2} + \frac{z^2}{32}\left(3-4\gamma_E-4\log\frac{z}{2}\right) + \mathrm{O}(z^3)\,,\label{K2exp}
\end{align}
\endgroup
where $\gamma_E$ denotes the Euler-Mascheroni constant. Terms like $z^n\log z$ with $n>0$ vanish in the $z\to0$ limit, so will not contribute to the singular part of our OPE computation. Terms like $\log z$ will give rise to $\log\,\la i\,j\ra$ singularities, but they are destined to cancel among themselves to ensure the absence of branch cuts.\footnote{Although this holds quite generally due to our heat kernel arguments, it is good to see that very similar evaluations of the celestial OPE on Eguchi-Hanson display similar cancellations of the logs \cite{Bittleston:2023bzp}. We thank Roland Bittleston for bringing this to our attention.} The remaining contributions of interest will be generated by the terms containing non-logarithmic poles. 

Let's use these expansions to study $\cI^1_{ij}$. If we eliminate $\lambda_{i2}$ in favor of $\la i\,j\ra$ by substituting $\lambda_{i2} = (\la i\,j\ra+\lambda_{i1}\lambda_{j2})/\lambda_{j1}$, and expand around $\la i\,j\ra=0$, we find that the second and third lines of \eqref{I1ij2} do not generate any singular terms in $\la i\,j\ra$. On the other hand, the third line only generates a mild logarithmic singularity, leading to the asymptotics
\begin{align}
    \cI^1_{ij} &\sim -\frac{1}{2}\,\frac{[i\,j]\log\,\la i\,j\ra}{\lambda_{i1}\lambda_{j2}}\int_0^1\d s\,(1-s)\,\e^{\im s p_{ij}\cdot x}\left(\im s p_{ij}\cdot x + 2\right)\nonumber\\
    &= -\frac{1}{2}\,\frac{[i\,j]\log\,\la i\,j\ra}{\lambda_{i1}\lambda_{j2}}\,{}_1F_1(1,2\,|\,\im p_{ij}\cdot x) \sim \mathrm{O}(\log\,\la i\,j\ra)\,,\label{I1fin}
\end{align}
where we have only displayed terms that are singular in the $\la i\,j\ra\to0$ limit. As a matter of abbreviation, we have also introduced an effective null momentum
\be
p_{ij,\al\dal} = \lambda_{j\al}\left(\tilde\lambda_{j\dal}+\frac{\lambda_{j1}}{\lambda_{i1}}\tilde\lambda_{i\dal}\right)
\ee
that any state on the right hand side of our OPEs would carry. It equals the exchanged momentum $p_i+p_j$ modulo terms of order $\la i\,j\ra$.

Since it only contains a logarithmic singularity, this result for $\cI^1_{ij}$ does not contribute to the celestial OPE. This is because, by our heat kernel arguments earlier, the final answer for the celestial OPE cannot have branch cuts. And more generally, we expect the logarithmic singularities to not contribute to on-shell observables. Hence, we choose to neglect terms involving $\log\,\la i\,j\ra$ and only focus on rational singularities. The twistorial arguments of section \ref{sec:enhance} show that any branch cuts must also drop out from actual 3-point amplitudes, but it will be interesting to verify this explicitly.

\medskip

The second term to compute is $\cI^2_{ij}$. Plugging in the plane wave states corresponding to $\phi_i^0,\phi_j^0$, one observes that it is just a Fourier transform of $G_1(x,x')$:
\be
\cI^2_{ij} = -\frac{(\lambda_{i1}\lambda_{j2}+\lambda_{i2}\lambda_{j1})[i\,j]}{\lambda_{i1}\lambda_{i2}\lambda_{j1}\lambda_{j2}}\int\d^4x'\,\e^{\im(p_i+p_j)\cdot x'}G_1(x,x')\,.
\ee
The first order in $N$ term in the Green's function is read off from \eqref{greenexp} to be
\be
G_1(x,x') = \frac{1}{8\pi^2}\,\frac{[u\,u'][\hat u\,\hat u']}{\|u-u'\|^4\|u\|^2\|u'\|^2}\,,
\ee
having remembered that $|[u\,u']|^2 = [u\,u']\overline{[u\,u']} = [u\,u'][\hat u\,\hat u']$. We would like to replace the factor of $[u\,u'][\hat u\,\hat u']$ in the numerator with derivatives in external data. To this end, let us expand $u^{\dal}$, $u'^{\dal}$, $\hat u^{\dal}$, $\hat u'^{\dal}$ in the natural constant basis of spinors $\tilde\lambda_i^{\dal},\tilde\lambda_j^{\dal}$. This allows us to rewrite the numerator as
\be
[u\,u'][\hat u\,\hat u'] = \frac{([u\,i][u'j]-[u\,j][u'i])([\hat u\,i][\hat u'j]-[\hat u\,j][\hat u'i])}{[i\,j]^2}\,.
\ee
Making the replacements $[u'i]\leftrightarrow -\im\p_{\lambda_{i1}}, [\hat u'i]\leftrightarrow-\im\p_{\lambda_{i2}}$, etc., we then get
\be
\cI^2_{ij} = \frac{(\lambda_{i1}\lambda_{j2}+\lambda_{i2}\lambda_{j1})}{\lambda_{i1}\lambda_{i2}\lambda_{j1}\lambda_{j2}[i\,j]\|u\|^2}\,\bigl([u\,i]\p_{\lambda_{j1}}-[u\,j]\p_{\lambda_{i1}}\bigr)\bigl([\hat u\,i]\p_{\lambda_{j2}}-[\hat u\,j]\p_{\lambda_{i2}}\bigr)\cS_{2,1}(p_i+p_j,x)\,.
\ee
The rest of the calculation is similar to $\cI^1_{ij}$, so we will be brief.

Reading off the value of $\cS_{2,1}(p_i+p_j,x)$ from \eqref{Smnform}, computing the derivatives, and integrating over $t$, one finds
\begin{align}
    \cI^2_{ij} &= -\frac{(\lambda_{i1}\lambda_{j2}+\lambda_{i2}\lambda_{j1})[i\,j]}{\lambda_{i1}\lambda_{i2}\lambda_{j1}\lambda_{j2}}\int_0^1\d s\,s\,\e^{\im s(p_i+p_j)\cdot x}\bigg[K_0\bigl(\sqrt{4s(1-s)\|u\|^2\la i\,j\ra[i\,j]}\bigr)\nonumber\\
    &\hspace{2cm}+\bigl(\lambda_{i2}[u\,i]+\lambda_{j2}[u\,j]\bigr)\bigl(\lambda_{i1}[\hat u\,i]+\lambda_{j1}[\hat u\,j]\bigr)\sqrt{\frac{s(1-s)}{\la i\,j\ra[i\,j]\|u\|^2}}\nonumber\\
    &\hspace{5cm}\times\,K_1\bigl(\sqrt{4s(1-s)\|u\|^2\la i\,j\ra[i\,j]}\bigr)\bigg]\,.
\end{align}
Expanding the integrand in small $\la i\,j\ra$ by eliminating $\lambda_{i2}$ as before, and integrating over $s$, we obtain a relevant singularity only from the $K_1$ Bessel function (neglecting terms containing $\log\,\la i\,j\ra$ by invoking the absence of branch cuts)
\be\label{I2fin}
    \cI^2_{ij} 
    \sim - \frac{1}{\la i\,j\ra}\,\frac{1}{\lambda_{i1}\lambda_{j2}}\,\frac{1}{2}\,\frac{(u^{\dal}p_{ij,1\dal})(\hat u^{\dot\beta}p_{ij,2\dot\beta})}{\|u\|^2}\,{}_1F_1(2,3\,|\,\im p_{ij}\cdot x) + \mathrm{O}(\log\,\la i\,j\ra)\,.
\ee
Comparing with \eqref{phin}, this is easily recognized to be $1/\la i\,j\ra$ times the order $N$ term in a quasi-momentum eigenstate of momentum $p_{ij}$. It combines with the leading flat space OPE \eqref{phi0ijsol} to inject backreaction into the state accompanying the $1/\la i\,j\ra$ singularity. But this is still just the leading OPE. The nontrivial OPE at order $1/\la i\,j\ra^2$ only gets generated through the third term $\cI^3_{ij}$.

\medskip

Substituting the expressions for $\phi_i^1,\phi_j^0$, etc., let's write out $\cI^3_{ij}$ in full:
\begin{multline}
    \cI^3_{ij} = \frac{1}{\lambda_{j1}\lambda_{j2}}\int_0^1\d s_i\,\bigg\{(\lambda_{i1}\lambda_{j2}+\lambda_{i2}\lambda_{j1})[i\,j]\,\p_{\lambda_{i1}}\p_{\lambda_{i2}}\cS_{1,1}(s_ip_i+p_j,x)\\
    + \frac{1}{s_i}\bigl(\lambda_{j1}\p_{\lambda_{i1}}^2\p_{\lambda_{j2}}-\lambda_{j2}\p_{\lambda_{i2}}^2\p_{\lambda_{j1}}\bigr)\cS_{1,2}(s_ip_i+p_j,x)\bigg\} - (i\leftrightarrow j)\,.
\end{multline}
Plugging in the values of $\cS_{1,1}$ and $\cS_{1,2}$ produces the following combination of Bessel functions
\begingroup
\allowdisplaybreaks
\begin{align}
    &\cI^3_{ij} = -\int_0^1\d s\int_0^1\d s_i\,\e^{\im s(s_ip_i+p_j)\cdot x}\nonumber\\
    &\times\left\{\frac{[i\,j]^2}{\la i\,j\ra}s^2(1-s) s_i\|u\|^2(\lambda_{i1}\lambda_{j2}+\lambda_{i2}\lambda_{j1})K_2\bigl(\sqrt{4s(1-s)\|u\|^2s_i\la i\,j\ra[i\,j]}\bigr)\right.\nonumber\\
    &+ \im s s_i[i\,j]^2\bigg(\frac{s_i\la i\,j\ra}{\lambda_{j1}\lambda_{j2}}\bigl(\lambda_{j1}[u\,i]+\lambda_{j2}[\hat u\,i]\bigr) + 2ss_i\bigl(\lambda_{i1}[u\,i]-\lambda_{i2}[\hat u\,i]\bigr)+ (1-s)\bigl(\lambda_{j1}[u\,j]-\lambda_{j2}[\hat u\,j]\bigr)\bigg)\nonumber\\
    &\hspace{3cm}\times\sqrt{\frac{s(1-s)\|u\|^2}{s_i\la i\,j\ra[i\,j]}}K_1\bigl(\sqrt{4s(1-s)\|u\|^2s_i\la i\,j\ra[i\,j]}\bigr)\nonumber\\
    &+ \frac{ss_i[i\,j]}{\lambda_{j1}\lambda_{j2}}\bigg(ss_i[u\,i][\hat u\,i](\lambda_{i1}\lambda_{j2}+\lambda_{i2}\lambda_{j1})+ s(1-s)\bigl(s_i\lambda_{i1}\lambda_{j1}[u\,i]^2+s_i\lambda_{i2}\lambda_{j2}[\hat u\,i]^2\bigr)\nonumber\\
    &\qquad\qquad+2s(1-s)\lambda_{j1}\lambda_{j2}\bigl([u\,i][\hat u\,j]+[u\,j][\hat u\,i]\bigr)-2\im(1-s)\bigl(\lambda_{j1}[u\,i]+\lambda_{j2}[\hat u\,i]\bigr) \bigg)\nonumber\\
    &\hspace{3cm}\times K_0\bigl(\sqrt{4s(1-s)\|u\|^2s_i\la i\,j\ra[i\,j]}\bigr)\nonumber\\
    &\left.+\; \frac{\im s^3(1-s)s_i(\lambda_{j1}[u\,i]^2[\hat u\,j]-\lambda_{j2}[\hat u\,i]^2[u\,j])}{\lambda_{j1}\lambda_{j2}}\sqrt{\frac{s_i\la i\,j\ra[i\,j]}{s(1-s)\|u\|^2}}K_1\bigl(\sqrt{4s(1-s)\|u\|^2s_i\la i\,j\ra[i\,j]}\bigr)\right\}\nonumber\\
    &-(i\leftrightarrow j)\,.\label{I3ishorrible}
\end{align}
\endgroup
The main difference between this term and the previous two terms is that now the derivatives have produced a Bessel $K_2$. It is this $K_2$ function whose series expansion will finally reproduce the anticipated $1/\la i\,j\ra^2$ singularity in the OPE.

To show this, use the series expansion $K_2(z)\sim 2/z^2$ around $z=0$ mentioned in \eqref{K2exp}. It is clear from \eqref{I3ishorrible} that the terms containing $K_1$ or $K_0$ would expand into singularities of the type $[i\,j]/\la i\,j\ra$ or $\log\,\la i\,j\ra$. As we are only interested in extracting the $[i\,j]/\la i\,j\ra^2$ and $1/\la i\,j\ra$ singularities, we suppress the rest to find the small $\la i\,j\ra$ expansion
\begin{multline}
\cI^3_{ij} \sim -\frac{\lambda_{i1}\lambda_{j2}[i\,j]}{\la i\,j\ra^2}\left\{\int_{[0,1]^2}\d s_i\,\d s\,s\,\e^{\im s(s_ip_i+p_j)\cdot x}+\int_{[0,1]^2}\d s_j\,\d s\,s\,\e^{\im s(p_i+s_jp_j)\cdot x}\right\}\\
+ \mathrm{O}\bigg(\frac{[i\,j]}{\la i\,j\ra},\,\log\,\la i\,j\ra\bigg)\,.
\end{multline}
Performing the change of variables $\omega_i = ss_i,\omega_j=s$ in the first integral, and $\omega_i = s,\omega_j=ss_j$ in the second integral, one can combine the two integrals to find
\be\label{I3fin}
\cI^3_{ij} \sim -\frac{\lambda_{i1}\lambda_{j2}[i\,j]}{\la i\,j\ra^2}\int_{[0,1]^2}\d\omega_i\,\d\omega_j\,\e^{\im (\omega_ip_i+\omega_jp_j)\cdot x}
+ \mathrm{O}\bigg(\frac{[i\,j]}{\la i\,j\ra},\,\log\,\la i\,j\ra\bigg)\,.
\ee
Substituting the little group fixed spinor-helicity variables $\lambda_{i\al}=(1,z_i)$, etc., this reduces to the expected collinear singularity
\be
\cI^3_{ij} \sim -\frac{z_j^2[i\,j]}{z_{ij}^2}\int_{[0,1]^2}\d\omega_i\,\d\omega_j\,\frac{\e^{\im (\omega_ip_i+\omega_jp_j)\cdot x}}{z_j}
+ \mathrm{O}\bigg(\frac{[i\,j]}{z_{ij}},\,\log\,\la i\,j\ra\bigg)\,.
\ee
Since the result \eqref{I1fin} for $\cI^1_{ij}$ only contained a log singularity, all that's left is to put $\cI^2_{ij}$ and $\cI^3_{ij}$ together.

Combining the result \eqref{N0ope} for the order $N^0$ OPE with the results \eqref{I2fin} and \eqref{I3fin}, and fixing the little group scalings, we obtain the backreacted celestial OPE,
\be
\phi_{ij} \sim \frac{f_{a_ia_j}{}^b}{z_{ij}}\,\phi_b(x|z_j,\tilde\lambda_i+\tilde\lambda_j) - \frac{Nz_j^2[i\,j]f_{a_ia_j}{}^b}{z_{ij}^2}\int_0^1\d\omega_i\int_0^1\d\omega_j\,\phi_b(x|z_j,\omega_i\tilde\lambda_i+\omega_j\tilde\lambda_j)\,.
\ee
valid up to terms of order $N^2$ or terms proportional to $[i\,j]/z_{ij}$. Comparing this with the OPEs \eqref{jjN0}, \eqref{jjN1} obtained from 2d Wick contractions, we find a perfect match with the predictions of our celestial CFT.


\subsection{Tests involving gravitational modes}
\label{sec:gravope}

Up till now, we have only computed observables involving gluon degrees of freedom $\phi$. A similar calculation of two and three point amplitudes involving only the K\"ahler perturbations $\rho$ is tough, so for the purposes of this work, we focus on 3-point scattering of two gluon modes against an $E$-type gravitational mode. 

In what follows, we will compute the OPE of a gravitational mode $\rho^E_i$ with a gluon mode $\phi_j$ using the dual chiral algebra. We will show that this matches the OPE computed from solving the equations of motion of WZW$_4$ coupled to Mabuchi gravity, although we will not venture beyond the simple pole $1/z_{ij}$ for now. Following this, we will explicitly compute the 3-point amplitude to order $N$ and show that it takes the expected form.

\paragraph{$\rho\phi\to\phi$ OPE and 3-point amplitude from the chiral algebra.} The operator dual to $\rho_i^E$ is
\be
E(z_i,\tilde\lambda_i) = \frac{1}{z_i^2}\bigl\{\text{Tr}\exp\!\left(z_i[X(z_i)\,i]\right)-1\bigr\}\,.
\ee
Its OPE with a hard gluon operator $J(z_j,\tilde\lambda_j)$ to zeroth order in $N$ is computed by performing a single $XX$ Wick contraction,
\be
E(z_i,\tilde\lambda_i)\,J_{rs}(z_j,\tilde\lambda_j) \sim \frac{z_iz_j}{z_i^2}\,\frac{[i\,j]}{z_{ij}}:\!I_r(z_j)\e^{z_i[X(z_i)\,i]+z_j[X(z_j)\,j]}I_s(z_j)\!: + \cdots\,.
\ee
An example of this is depicted in double-line notation in figure \ref{JEtoJ}.
\begin{figure}
    \centering
    \includegraphics[scale=0.3]{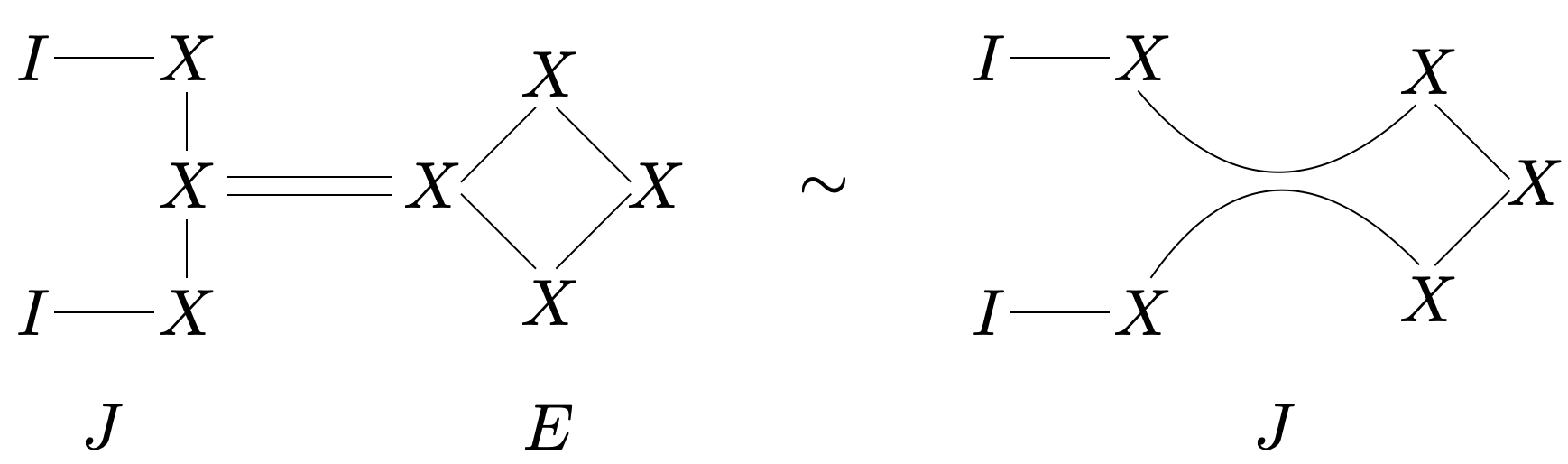}
    \caption{Scattering gluons off Mabuchi gravitons: a view from the boundary. Such $XX$ Wick contractions contribute to the $E\,J\sim J$ OPE. Each line denotes a contraction on the $\Sp(N)$ indices. There are no closed loops or factors of $N$ generated at the displayed order.}
    \label{JEtoJ}
\end{figure}

In simplifying this, we have freely commuted factors of $[X(z_i)\,i]$ past factors of $[X(z_j)\,j]$ inside the normal ordering. This uses the fact that $X^{\dot1}, X^{\dot2}$ -- viewed as $N\times N$ matrices -- commute in the large $N$ limit due to the ADHM constraint \eqref{adhm}. Expanding around $z_{ij}=0$ yields the OPE
\be
E(z_i,\tilde\lambda_i)\,J_{rs}(z_j,\tilde\lambda_j) \sim \frac{[i\,j]}{z_{ij}}\,J_{rs}(z_j,\tilde\lambda_i+\tilde\lambda_j) + \mathrm{O}(N)\,.
\ee
One can also compute the OPE to higher order in $N$, but reproducing it from the bulk promises to be a daunting task.

It is a straightforward matter to evaluate the 3-point amplitude $A(\rho_i^E,\phi_j,\phi_k)$ to leading order in $N$ using such Wick contractions. Since we subtracted off $E[0,0]=1$ when defining $E(z,\tilde\lambda)$ in \eqref{Ehard}, the 3-point correlator again starts at first order in $N$,
\begin{multline}
    \bigl\la E(z_i,\tilde\lambda_i)\,J_{pq}(z_j,\tilde\lambda_j)\,J_{rs}(z_k,\tilde\lambda_k)\bigr\ra \\
    = \frac{z_j}{z_i}\,\bigl\la\text{Tr}\,[X\,i](z_i)\,:\!I_p[X\,j]I_q\!:\!(z_j)\,:\!I_rI_s\!:\!(z_k)\bigr\ra\\
    + \frac{z_k}{z_i}\,\bigl\la\text{Tr}\,[X\,i](z_i)\,:\!I_pI_q\!:\!(z_j)\,:\!I_r[X\,k]I_s\!:\!(z_k)\bigr\ra + \mathrm{O}(N^2)\,.
\end{multline}
Performing maximal Wick contractions using the OPE \eqref{XXcom} that keeps the center of mass modes in play, one finds the following prediction for the amplitude:
\be\label{rhophi2pred}
A(\rho_i^E,\phi_j,\phi_k) = -\frac{N\,\tr(\sT_{a_j}\sT_{a_k})}{z_{jk}^2}\left\{\frac{[i\,j]}{z_{ij}}\,\frac{z_j}{z_i}+\frac{[i\,k]}{z_{ik}}\,\frac{z_k}{z_i}\right\} + \mathrm{O}(N^2)\,,
\ee
having identified $\sT_{a_j}\leftrightarrow\sT_{pq}$, $\sT_{a_k}\leftrightarrow\sT_{rs}$, and used \eqref{sotrace} to identify the color factor as a trace. This brings us to the bulk computation.

\paragraph{$\rho\phi\to\phi$ OPE from the bulk.} To see this OPE emerging from the collinear splitting function in the bulk, we simply repeat the strategy of the previous section. As was explained in section \ref{sec:closed}, the WZW$_4$ model couples to K\"ahler perturbations by the shift $\omega\mapsto\omega+\im\,\p\dbar\rho$:
\be
S_{\text{WZW}_4}[\phi,\rho] = -\frac{\im}{8\pi^2}\int\left(\omega+\im\,\p\dbar\rho\right)\wedge\tr\left(\p\phi\wedge\dbar\phi\right) + \mathrm{O}(\phi^3) \,,
\ee
where $\omega$ is the K\"ahler form of Burns space. 
The equation of motion of $\phi$, including the $\rho\phi^2$ coupling, now takes the form
\be
\omega\wedge\p\dbar\phi = -\im\,\p\dbar\rho\wedge\p\dbar\phi + \mathrm{O}(\phi^2)\,.
\ee
The right hand side gives rise to the $\rho\phi$ OPE.

The holographic OPE $\phi_{ij}$ of a WZW$_4$ state $\phi_j$ with a K\"ahler perturbation $\rho^E_i$ is defined by demanding that the linear combinations
\be
\rho = \veps_i\rho_i^E\,,\qquad\phi = \veps_j\phi_j+\veps_i\veps_j\phi_{ij}
\ee
solve the $\phi$ equation of motion at order $\veps_i\veps_j$. This leads to the PDE
\be
\lap\phi_{ij} = -2\,\frac{\p^2\rho^E_i}{\p u^{\dal}\p\hat u^{\dot\beta}}\frac{\p^2\phi_j}{\p u_{\dal}\p\hat u_{\dot\beta}}\,,
\ee
which can be solved iteratively in $N$ as in the previous section.

We will compute the OPE only to zeroth order in $N$, so we are free to replace $\lap$ and $\rho_i^E,\phi_j$ by their flat space counterparts. Plugging in
\be
\rho^E_i = \frac{\e^{\im p_i\cdot x}}{\lambda_{i1}^2\lambda_{i2}^2}\,,\qquad \phi_j = \frac{\sT_{a_j}\e^{\im p_j\cdot x}}{\lambda_{j1}\lambda_{j2}}\,,
\ee
one immediately finds the solution
\be
\phi_{ij} = \frac{[i\,j]}{\la i\,j\ra}\,\frac{\sT_{a_j}\e^{\im(p_i+p_j)\cdot x}}{\lambda_{i1}\lambda_{i2}}\,.
\ee
Setting $\lambda_{i\al}=(1,z_i)$, $\lambda_{j\al}=(1,z_j)$, and expanding in small $z_{ij}$, one finds the holographic OPE
\be
\phi_{ij} \sim \frac{[i\,j]}{z_{ij}}\,\phi_{a_j}(x|z_j,\tilde\lambda_i+\tilde\lambda_j)\,.
\ee
This matches the prediction of the celestial chiral algebra. 

More singular terms in $z_{ij}$ will enter at higher powers of $N$ (with each term receiving its own set of $1/N$ loop corrections as usual). Their determination will be a very useful test of our duality, but they lie beyond the scope of this paper. In the rest of this section, we instead focus on obtaining this OPE in a different way: by directly computing the order $N$ term in the 3-point $\rho\phi\phi$ amplitude.

\paragraph{$\rho\phi\phi$ 3-point amplitude.} Using the perturbiner formalism as before, the amplitude for scattering a gravitational mode $\rho^E_i$ against two gluon modes $\phi_j,\phi_k$ is given by the symmetrized on-shell interaction vertex \cite{Adamo:2017nia}
\be
A(\rho^E_i,\phi_j,\phi_k) = \frac{1}{8\pi^2}\int_{\widetilde\C^2}\p\dbar\rho^E_i\wedge\tr\left(\p\phi_j\wedge\dbar\phi_k+\p\phi_k\wedge\dbar\phi_j\right)\,.
\ee
If we evaluate this as a series in $N$, we find that just like the 3-point $\phi\phi\phi$ amplitude, the zeroth order term is the flat space 3-point amplitude which is distributional (see the supplemental material to \cite{Costello:2022jpg}). It only finds support at $z_i=z_j=z_k$. We will work at generic kinematics $z_i\neq z_j\neq z_k$ and drop this term.\footnote{As always with distributional terms, it is unclear how to obtain it as a local chiral algebra correlator.} The first term of interest is then the term of order $N$. This is what we will compute below.

Expand $\rho_i^E = \rho_i^{E,0}+N\rho_i^{E,1}+\mathrm{O}(N^2)$, $\phi_j = \sT_{a_j}(\phi_j^0 + N\phi_j^1+\mathrm{O}(N^2))$, etc. For the K\"ahler perturbation, the coefficients are read off from \eqref{Erhoexp} (up to the usual subtraction by $1$ to remove $E[0,0]$),
\be\label{rhoE01}
\rho_i^{E,0}(x) = \frac{\e^{\im p_i\cdot x}-1}{\lambda_{i1}^2\lambda_{i2}^2}\,,\qquad\rho_i^{E,1}(x) = -\frac{1}{\lambda_{i1}\lambda_{i2}}\,\frac{[u\,i][\hat u\,i]}{\|u\|^2}\,\int_0^1\d s_i\,\e^{\im s_ip_i\cdot x}\,.
\ee
The first is the flat space wavefunction, while the second is the leading order backreaction. The ${}_1F_1(1,2\,|\,\im p_i\cdot x)$ factor in the latter has been expressed through its integral representation as usual. For the reader's convenience, let us also repeat the corresponding terms in the $\phi$ wavefunctions,
\be\label{phi01}
\phi_j^0(x) = \frac{\e^{\im p_j\cdot x}}{\lambda_{j1}\lambda_{j2}}\,,\qquad\phi_j^1(x) = -\frac{[u\,j][\hat u\,j]}{\|u\|^2}\int_0^1\d s_j\,s_j\,\e^{\im s_jp_j\cdot x}\,.
\ee
Similar expressions for $\phi_k^0,\phi_k^1$ are obtained by replacing $j$ with $k$.

The 3-point amplitude can be broken down into two types of terms,
\be\label{Arhophi2}
A(\rho^E_i,\phi_j,\phi_k) = N\,\tr(\sT_{a_j}\sT_{a_k})\,\bigl\{\cI_{ijk} + \cJ_{ijk} + (j\leftrightarrow k)\bigr\} + \mathrm{O}(N^2)\,.
\ee
The first term is an integral containing the backreaction associated to $\rho^E_i$,
\be\label{rhoI}
\cI_{ijk} = \frac{1}{8\pi^2}\int_{\widetilde\C^2}\p\dbar\rho^{E,1}_i\wedge\p\phi_j^0\wedge\dbar\phi_k^0\,.
\ee
The second term encodes the backreaction coming from the WZW$_4$ fields,
\be\label{rhoJ}
\cJ_{ijk} = \frac{1}{8\pi^2}\int_{\widetilde\C^2}\p\dbar\rho^{E,0}_i\wedge\left(\p\phi_j^1\wedge\dbar\phi_k^0 + \p\phi_k^0\wedge\dbar\phi_j^1\right)\,.
\ee
As we have already developed quite a bit of technology while computing the 2-point amplitude, evaluating these integrals proves to be relatively straightforward.

Let us start with \eqref{rhoI}. On the coordinate patch $u^{\dal}\in\C^2-0$, it expands out to
\be
\cI_{ijk} = \frac{1}{8\pi^2}\int\d^4x\,\frac{\p^2\rho_i^{E,1}}{\p u^{\dal}\p\hat u^{\dot\beta}}\,\frac{\p\phi_j^0}{\p u_{\dal}}\,\frac{\p\phi_k^0}{\p\hat u_{\dot\beta}}\,.
\ee
Plugging in the various wavefunctions from \eqref{rhoE01} and \eqref{phi01}, and remembering the intertwining trick \eqref{intertwine}, this can be reduced to
\begin{multline}
    \cI_{ijk} = \frac{1}{\lambda_{i1}\lambda_{i2}\lambda_{j2}\lambda_{k1}}\int_0^1\d s_i\left\{\lambda_{i1}\lambda_{i2}[i\,j][i\,k]\p_{\lambda_{i1}}\p_{\lambda_{i2}}\cS_1(P) + \frac{\lambda_{i1}[i\,j]}{s_i}\,\p_{\lambda_{i1}}^2\p_{\lambda_{k2}}\cS_2(P) \right.\\
    \left.- \frac{\lambda_{i2}[i\,k]}{s_i}\,\p_{\lambda_{i2}}^2\p_{\lambda_{j1}}\cS_2(P) - \frac{2}{s_i^2}\,\p_{\lambda_{i1}}\p_{\lambda_{i2}}\p_{\lambda_{j1}}\p_{\lambda_{k2}}\cS_3(P) \right\}\biggr|_{P=s_ip_i+p_j+p_k}\,.
\end{multline}
We have again abbreviated the spacetime integrals by defining
\be
\cS_\ell(P) = \frac{1}{8\pi^2}\int\frac{\d^4x\;\e^{\im P\cdot x}}{(x^2/2)^\ell} = \frac{1}{2\,\Gamma(\ell)}\int_0^\infty\frac{\d t}{t^{\ell-1}}\,\e^{-t P^2/2}\,.
\ee
Their value has been obtained using \eqref{fourierid}.

In our case, $\frac12P^2=\frac12(s_ip_i+p_j+p_k)^2 = s_i\la i\,j\ra[i\,j]+s_i\la i\,k\ra[i\,k]+\la j\,k\ra[j\,k]$. As in the previous sections, the $t$ integral is understood in the sense of analytic continuation: evaluated for $P^2>0$, then continued to complex $P^\mu$ away from the singular locus $P^2=0$. Thereby performing the spacetime integrals as well as the $t$ and $s_i$ integrals, one finds a simple if somewhat obscure expression
\be\label{Iijkval}
    \cI_{ijk} = -\frac{(\lambda_{j1}[i\,j]+\lambda_{k1}[i\,k])(\lambda_{j2}[i\,j]+\lambda_{k2}[i\,k])}{\lambda_{i1}\lambda_{i2}\,\la j\,k\ra^2}\,\frac{(p_i+p_j+p_k)^2+(p_j+p_k)^2}{(p_i+p_j+p_k)^4}\,.
\ee
We will soon see that it drastically simplifies when combined with $\cJ_{ijk}$.

The second term \eqref{rhoJ} can also be written in coordinates as
\be
\cJ_{ijk} = \frac{1}{8\pi^2}\int\d^4x\,\frac{\p^2\rho_i^{E,0}}{\p u^{\dal}\p\hat u^{\dot\beta}}\bigg(\frac{\p\phi_j^1}{\p u_{\dal}}\,\frac{\p\phi_k^0}{\p\hat u_{\dot\beta}} + \frac{\p\phi_k^0}{\p u_{\dal}}\,\frac{\p\phi_j^1}{\p\hat u_{\dot\beta}}\bigg)\,.
\ee
With the aid of \eqref{rhoE01}, \eqref{phi01}, this is found to reduce to
\begin{multline}
    \cJ_{ijk} = \frac{1}{\lambda_{i1}^2\lambda_{i2}^2\lambda_{k1}\lambda_{k2}}\int_0^1\d s_j\,\bigg\{[i\,j][i\,k](\lambda_{j1}\lambda_{k2}+\lambda_{j2}\lambda_{k1})\p_{\lambda_{j1}}\p_{\lambda_{j2}}\cS_1(P)\\
    + \frac{[i\,k]}{s_j}\big(\lambda_{k2}\p_{\lambda_{i1}}\p_{\lambda_{j2}}^2-\lambda_{k1}\p_{\lambda_{i2}}\p_{\lambda_{j1}}^2\big)\cS_2(P)\bigg\}\biggr|_{P=p_i+s_jp_j+p_k}\,.
\end{multline}
Performing all the integrals, we are left with
\be\label{Jijkval}
\cJ_{ijk} = -\frac{4\,[i\,k]\,(\lambda_{i1}[j\,i]+\lambda_{k1}[j\,k])(\lambda_{i2}[j\,i]+\lambda_{k2}[j\,k])}{\lambda_{i1}\lambda_{i2}\,\la i\,k\ra\,(p_i+p_j+p_k)^4}\,.
\ee
With this, we have all the ingredients we need.

Substituting \eqref{Iijkval} and \eqref{Jijkval} into \eqref{Arhophi2}, and symmetrizing over $j\leftrightarrow k$ as instructed, one finds dramatic cancellations. The final result for the 3-point amplitude collapses to
\be\label{Arhoijval}
A(\rho^E_i,\phi_j,\phi_k) = -\frac{N\,\tr(\sT_{a_j}\sT_{a_k})}{\la j\,k\ra^2}\left\{\frac{[i\,j]}{\la i\,j\ra}\frac{\lambda_{j1}\lambda_{j2}}{\lambda_{i1}\lambda_{i2}} + \frac{[i\,k]}{\la i\,k\ra}\frac{\lambda_{k1}\lambda_{k2}}{\lambda_{i1}\lambda_{i2}}\right\} + \mathrm{O}(N^2)\,.
\ee
An easy check is that this has little group weight $-4$ in $i$, and weight $-2$ in $j,k$ each. Yet again, this matches the chiral algebra 3-point prediction \eqref{rhophi2pred} if one sets $\lambda_{i\al}=(1,z_i)$, etc.

It is remarkable that the 3-particle singularities $1/(p_i+p_j+p_k)^4$ completely drop out, leaving one with only 2-particle singularities that a CFT can generate. Of course, one also finds natural singularities at the defect loci $z_i=0,\infty$, but at this stage they result purely from dressing the boundary operators with factors of $z_i$ to ensure the correct scaling dimensions.\footnote{The expression \eqref{Arhoijval} is reminiscent of Hodges' formula \cite{Hodges:2012ym} for MHV graviton amplitudes in flat space, except here the reference spinors are the defect locations, and they do not drop out because we no longer have momentum conservation.} It would be interesting to find genuine defect effects by matching nontrivial defect conformal blocks of the chiral algebra to form-factors of operators wrapping the core of Burns space.




\section*{Acknowledgements}

We are grateful to Roland Bittleston, Eduardo Casali, Maciej Dunajski, Joel Fine, Davide Gaiotto, Manki Kim, Claude LeBrun, Lionel Mason, Walker Melton, Rashmish Mishra, Prahar Mitra, Sujay Nair, Shruthi Narayanan, Sabrina Pasterski, David Skinner and Andy Strominger for enlightening discussions. K.C. is supported by the NSERC Discovery Grant program and by the Perimeter Institute for Theoretical Physics. Research at Perimeter Institute is supported by the Government of Canada through Industry Canada and by the Province of Ontario through the Ministry of Research and Innovation.  NP is supported by the University of Washington and the DOE Early Career Research Program under award DE-SC0022924. AS is supported by a Black Hole Initiative fellowship, funded by the Gordon and Betty Moore Foundation and the John Templeton Foundation. He has also received support from the ERC grant GALOP ID: 724638 during earlier stages of this work.


\begin{appendix}

\section{More on curved twistor spaces}
\label{app:twistor}

\subsection{Twistors for self-dual spacetimes}

Let $M$ be an oriented 4-manifold equipped with a Riemannian metric $g$. Take $\theta^a$, $a=0,1,2,3$, to be a local frame of the cotangent bundle of $M$ satisfying $g = \delta_{ab}\theta^a\theta^b$. We define its spinor equivalent $\theta^{\al\dal}$, $\al=1,2$, $\dal=\dot1,\dot2$, as follows
\be
\theta^{\al\dal} = \frac{1}{\sqrt2}\begin{pmatrix}\theta^0+\im\,\theta^3&&\theta^2+\im\,\theta^1\\-\theta^2+\im\,\theta^1&&\theta^0-\im\,\theta^3\end{pmatrix}\,.
\ee
In terms of $\theta^{\al\dal}$, the metric can be expressed as
\be
g = \eps_{\al\beta}\,\eps_{\dal\dot\beta}\,\theta^{\al\dal}\,\theta^{\beta\dot\beta}
\ee
where $\eps_{\al\beta}$, $\eps_{\dal\dot\beta}$ are $2\times2$ Levi-Civita symbols. Since $\theta^a$ are real, we observe that 
\be\label{realtheta}
\theta^{2\dot1}=-\overline{\theta^{1\dot2}}\,,\qquad\theta^{2\dot2}=\overline{\theta^{1\dot1}}\,.
\ee
The dual frame of vector fields will be denoted $e_{\al\dal}$, so that $e_{\al\dal}\ip\theta^{\beta\dot\beta}=\delta_\al^\beta\delta_{\dal}^{\dot\beta}$. Let $\Lambda^-$ and $\Lambda^+$ be the rank 3 bundles of 2-forms on $M$ that are respectively ASD or SD with respect to $g$. Denote by $\Sigma^{\al\beta}\equiv\Sigma^{(\al\beta)}$ and $\tilde\Sigma^{\dal\dot\beta}\equiv\tilde\Sigma^{(\dal\dot\beta)}$ the following choice of local frames for $\Lambda^-$ and $\Lambda^+$ respectively:
\be\label{sigmabasis}
\Sigma^{\al\beta} = \eps_{\al\beta}\,\theta^{\al\dal}\wedge\theta^{\beta\dot\beta}\,,\qquad\tilde\Sigma^{\dal\dot\beta} = \eps_{\dal\dot\beta}\,\theta^{\al\dal}\wedge\theta^{\beta\dot\beta}\,.
\ee
This choice is standard in the literature \cite{Capovilla:1991qb}. These 2-forms obey the reality conditions $\overline{\Sigma^{12}}=-\Sigma^{12}$, $\overline{\Sigma^{22}}=\Sigma^{11}$, etc.

The Levi-Civita connection $\Gamma_{\al\beta}\equiv\Gamma_{(\al\beta)}$ on $\Lambda^-$ obeys the structure equation
\be\label{sigmastruc}
\d\Sigma^{\al\beta} = 2\,\Gamma^{(\al}{}_\gamma\wedge\Sigma^{\beta)\gamma}\,,
\ee
where spinor indices are raised or lowered using the conventions \eqref{conventions}. It is called the ASD spin connection. Rotations of the frame $\Sigma^{\al\beta}$ induce $\SU(2)$ gauge transformations on $\Gamma_{\al\beta}$. The curvature of the ASD spin connection is given by
\be
R_{\al\beta} = \d\Gamma_{\al\beta} + \Gamma_\al{}^\gamma\wedge\Gamma_{\gamma\beta}\,.
\ee
This computes the ASD part of the Riemann tensor 2-form. It admits a well-known decomposition into irreps of $\Spin(4)\simeq\SU(2)\times\SU(2)$:
\be
R_{\al\beta} = \Psi_{\al\beta\gamma\delta}\Sigma^{\gamma\delta} + \Phi_{\al\beta\dal\dot\beta}\tilde\Sigma^{\dal\dot\beta} + \frac{R}{12}\,\Sigma_{\al\beta}
\ee
where $\Psi_{\al\beta\gamma\delta}$ is the spinor equivalent of the ASD Weyl curvature, $\Phi_{\al\beta\dal\dot\beta}$ the trace-free Ricci tensor, and $R$ the Ricci scalar.  

Similarly, one can also define the SD spin connection $\tilde\Gamma_{\dal\dot\beta}$ via the structure equation $\d\tilde\Sigma^{\dal\dot\beta} = 2\,\tilde\Gamma^{(\dal}{}_{\dot\gamma}\wedge\tilde\Sigma^{\dot\beta)\dot\gamma}$, and its curvature tensors are defined analogously. Putting $\Gamma_{\al\beta}$ and $\tilde\Gamma_{\dal\dot\beta}$ together, one obtains the full Levi-Civita connection $\Gamma_{\al\dal\beta\dot\beta} = \eps_{\al\beta}\tilde\Gamma_{\dal\dot\beta}+\eps_{\dal\dot\beta}\Gamma_{\al\beta}$. This fits into Cartan's structure equation for the coframe, $\d\theta^{\al\dal} = -\Gamma^{\al\dal}{}_{\beta\dot\beta}\wedge\theta^{\beta\dot\beta}$. We will work chirally and mainly use the ASD spin connection for the construction of twistor spaces.

\medskip

Twistor space $Z$ can be defined as the projective bundle of undotted (left-handed) Weyl spinors on $M$. This is well-defined even if $M$ is not spin, because spinor representations of $\Spin(4)$ descend to well-defined projective representations of $\SO(4)$. Alternatively, it is diffeomorphic to the sphere bundle $S(\Lambda^-)$ of the bundle of ASD 2-forms on $M$.

Let $\sigma_\al$ be homogeneous coordinates on the $\CP^1$ fibers of the projection $\pi:Z\to M$. At a given $x\in M$, a point $\sigma_\al$ on its twistor line $L_x = \pi^{-1}(\{x\})$ labels the almost complex structure on $T^*_xM\otimes\C$ whose $(1,0)$-forms are spanned by $\sigma_\al\theta^{\al\dal}$. It is straightforward to verify that this is  metric and orientation compatible. Also introduce the quaternionic conjugates $\hat\sigma_\al = (-\overline{\sigma_2},\overline{\sigma_1})$ that transform in the same $\SU(2)$ representation as $\sigma_\al$. They satisfy the useful property $\la\hat\sigma\,\sigma\ra=\|\sigma\|^2\equiv|\sigma_1|^2+|\sigma_2|^2$. The map $\sigma_\al\mapsto\hat\sigma_\al$ is just the familiar antipodal map at the level of the $\CP^1$ fibers. It lifts to $Z$ as an involution without fixed points, thereby inducing a real structure on $Z$.

Now, whereas horizontal forms on $Z$ are canonically defined to be the span of pullbacks of $\theta^{\al\dal}$, we need a connection on $Z$ to define vertical 1-forms. This is provided by the ASD spin connection, in terms of which we take
\be
\nabla\sigma^\al = \d\sigma^\al-\Gamma^\al{}_\beta\sigma^\beta\,,\qquad\nabla\hat\sigma^\al = \d\hat\sigma^\al-\Gamma^\al{}_\beta\hat\sigma^\beta
\ee
as a frame for the bundle of vertical 1-forms. The reality conditions on $\Gamma^{\al}{}_{\beta}$ resulting from those on $\Sigma^{\al\beta}$ ensure that $\nabla\hat\sigma^\al = (\widehat{\nabla\sigma})^\al$. The Atiyah-Hitchin-Singer almost complex structure is defined by declaring
\be\label{10basis}
\tau = \sigma_\al\nabla\sigma^\al\,,\qquad\theta^{\dal} = \sigma_\al\theta^{\al\dal}
\ee
to be a local basis of $(1,0)$-forms $\Omega^{1,0}(Z)$ and
\be\label{01basis}
\bar\tau = \frac{\hat\sigma_\al\nabla\hat\sigma^\al}{\|\sigma\|^4}\,,\qquad\bar\theta^{\dal} = -\frac{\hat\sigma_\al\theta^{\al\dal}}{\|\sigma\|^2}
\ee
as a standard local basis of the complex conjugate $(0,1)$-forms $\Omega^{0,1}(Z)$. They are designed to have zero scaling weight in $\hat\sigma_\al$. In this almost complex structure, the canonical bundle $K_Z=\Omega^{3,0}(Z)$ admits the local frame
\be
\tau\wedge\theta^{\dot1}\wedge\theta^{\dot2}\,.
\ee
This has scaling weight $+4$ in the homogeneous coordinates $\sigma_\al$. As a result, local sections of $K_Z$ are given by $\varphi\,\tau\wedge\theta^{\dot1}\wedge\theta^{\dot2}$ for smooth functions $\varphi(x,\sigma,\hat\sigma)$ of weight $-4$ in $\sigma_\al$ and $0$ in $\hat\sigma_\al$. In particular, $K_Z$ restricts to the bundle $\CO(-4)\to\CP^1$ on every twistor line.

Let $\pi_{p,q}$ denote the projection of a $(p+q)$-form onto its $(p,q)$ part. Starting with the exterior derivative $\d$ on $Z$, the holomorphic and antiholomorphic exterior derivatives of smooth forms $\al\in\Omega^{p,q}(Z)$ are defined to be the projections $\p\al\vcentcolon=\pi_{p+1,q}(\d\al)$, $\dbar\al\vcentcolon=\pi_{p,q+1}(\d\al)$. The almost complex structure is said to be integrable if $\d=\p+\dbar$. This happens if and only if $\tau,\theta^{\dal}$ span a closed differential ideal, or equivalently when the kernel of $\tau,\theta^{\dal}$ forms an involutive distribution of vector fields. It is a textbook calculation to verify that
\begin{align}
    \d\theta^{\dal} &= \tau\wedge\bar\theta^{\dal}+\tilde\Gamma^{\dal}{}_{\dot\beta}\wedge\theta^{\dot\beta} + \frac{\la\hat\sigma\,\nabla\sigma\ra}{\|\sigma\|^2}\wedge\theta^{\dal}\nonumber\\
    &\equiv 0\quad\text{mod span}\{\tau,\theta^{\dal}\}\,,\\
    \d\tau &= \sigma^\al\sigma^\beta R_{\al\beta} + 2\,\frac{\la\hat\sigma\,\nabla\sigma\ra}{\|\sigma\|^2}\wedge\tau\nonumber\\
    &\equiv \sigma^\al\sigma^\beta\sigma^\gamma\sigma^\delta\Psi_{\al\beta\gamma\delta}\,\bar\theta_{\dal}\wedge\bar\theta^{\dal}\quad\text{mod span}\{\tau,\theta^{\dal}\}\,.
\end{align}
Hence, the Atiyah-Hitchin-Singer almost complex structure is integrable if and only if the ASD Weyl curvature vanishes \cite{Penrose:1976js,Atiyah:1978wi}, i.e., $(M,g)$ is self-dual. See \cite{Mason:1991rf,Dunajski:2010zz,Ward:1990vs} for helpful reviews.


\subsection{Applications to scalar-flat K\"ahler geometry}

Next, let $g$ be scalar-flat K\"ahler, or equivalently self-dual and K\"ahler. The associated K\"ahler form $\omega$ gives a classic example of a globally defined, nowhere vanishing ASD 2-form (in our conventions). With the choice of frame \eqref{sigmabasis} for $\Lambda^-$, it can be expanded as
\be\label{omegaexp}
\omega = \im\,\omega_{\al\beta}\Sigma^{\al\beta}\,.
\ee
In our orientation, this should satisfy $\frac{1}{2!}\,\omega^2 = -\vol_g$. Using $\Sigma^{\al\beta}\wedge\Sigma^{\gamma\delta} = 4\eps^{\al(\gamma}\eps^{\delta)\beta}\vol_g$, this constrains the determinant of the coefficient matrix $\omega_{\al\beta}$ to be
\be
\det\omega_{\al\beta} = \omega_{11}\omega_{22}-\omega_{12}^2 = -\frac14\,.
\ee
Demanding $\d\omega=0$ on $M$ yields the twistor equation (see lemma 1.1 in \cite{pontecorvo})
\be\label{tweq}
\nabla_{(\gamma|\dot\gamma}\omega_{|\al\beta)}=0\,,
\ee
where $\nabla_{\gamma\dot\gamma}\omega_{\al\beta} = e_{\gamma\dot\gamma}\omega_{\al\beta} + 2\,\omega_{\delta(\al}\Gamma^\delta{}_{\beta)\gamma\dot\gamma}$ is the covariant derivative. We can pull back $\omega_{\al\beta}$ to $Z$ and, following \cite{pontecorvo}, use it to define a section of the bundle $K_Z^{-1/2}$,
\be\label{omegacech}
\check{\omega} = \omega_{\al\beta}\sigma^\al\sigma^\beta\,.
\ee
As expected, this is a section of $\CO(2)$ on every twistor line $L_x$. Computing its antiholomorphic exterior derivative in the complex structure of $Z$ yields
\be\label{dbaromegacech}
    \dbar\check\omega = \sigma^\al\sigma^\beta\,\pi_{0,1}(\nabla\omega_{\al\beta})= \sigma^\al\sigma^\beta\sigma^\gamma\nabla_{\gamma\dot\gamma}\omega_{\al\beta}\,\bar\theta^{\dot\gamma} = 0\,,
\ee
where $\nabla\omega_{\al\beta}\equiv\theta^{\gamma\dot\gamma}\,\nabla_{\gamma\dot\gamma}\omega_{\al\beta}$. In this way, the K\"ahler form on $M$ lifts to a holomorphic global section of $K_Z^{-1/2}$. Due to the reality condition $\overline\omega=\omega$, this section is pure imaginary under the real structure of $Z$, i.e., it satisfies $\overline{\check{\omega}(x,\sigma)} = -\omega(x,\hat\sigma)$. 

Since $\check{\omega}$ and $\tau\wedge\theta^{\dot1}\wedge\theta^{\dot2}$ are both holomorphic, we conclude that scalar-flat K\"ahler metrics give rise to a meromorphic 3-form on their twistor spaces:
\be
\Omega = \frac{\tau\wedge\theta^{\dot1}\wedge\theta^{\dot2}}{\check{\omega}^2}\,.
\ee
This is manifestly invariant under scalings of $\sigma_\al$, so is projectively well-defined. Since $\check{\omega}$ is a global section, $\Omega$ is globally defined on $Z-\{\check{\omega}=0\}$. Its only singularities are double poles at the locations of the zeroes of $\check{\omega}$. Let us verify that we can recover the K\"ahler form \eqref{omegaexp} from this $3$-form by contour integrating around one of these poles as in \eqref{omegacirc}. 
Introduce affine coordinates $\sigma_\al=(1,\sigma)$. At fixed $x\in M$, we can factorize 
\be\label{omegapol}
\check{\omega} = \omega_{11}\sigma^2-2\omega_{12}\sigma+\omega_{22} = \omega_{11}(\sigma-\sigma_+)(\sigma-\sigma_-)
\ee
in terms of $x$-dependent roots $\sigma_\pm(x)$ given by
\be
\sigma_\pm = \frac{\omega_{12}\pm\sqrt{\omega_{12}^2-\omega_{11}\omega_{22}}}{\omega_{11}}\,.
\ee
Similarly, we can write $\theta^{\dal} = \theta^{1\dal}+\sigma\theta^{2\dal}$ and $\tau = \d\sigma+\Gamma_{11}\sigma^2-2\Gamma_{12}\sigma+\Gamma_{22}$. Performing a contour integral around the pole $\sigma=\sigma_-$, we obtain
\be
\frac{1}{2\pi}\oint_{|\sigma-\sigma_-|=\veps}\Omega = \frac{\im\,[\Sigma^{11}+(\sigma_++\sigma_-)\Sigma^{12}+\sigma_+\sigma_-\Sigma^{22}]}{\omega_{11}^2(\sigma_+-\sigma_-)^3} = \frac{\im\,\omega_{\al\beta}\Sigma^{\al\beta}}{(-4\det\omega_{\al\beta})^{3/2}}\,.
\ee
Remembering that $\det\omega_{\al\beta}=-\frac14$, this reduces precisely to \eqref{omegaexp}. This also fixes the overall normalization $\omega_{11}$ to be $\omega_{11}=(\sigma_+-\sigma_-)^{-1}$.

Geometrically, the roots $\sigma_\pm(x)$ define almost complex structures on $M$ with $(1,0)$-forms spanned by $\theta^{1\dal}+\sigma_\pm\theta^{2\dal}$. The reality condition $\overline{\check{\omega}(x,\sigma)}=-\check{\omega}(x,\hat\sigma)$ requires that $\sigma_+=-1/\overline{\sigma_-}$, i.e., $\sigma_+$ is antipodal to $\sigma_-$. So the two almost complex structures are conjugate to each other. As a result,
\be\label{omega+-}
\omega = \frac{\im\,\eps_{\dal\dot\beta}\,(\theta^{1\dal}+\sigma_-\theta^{2\dal})\wedge(\theta^{1\dot\beta}+\sigma_+\theta^{2\dot\beta})}{(\sigma_+-\sigma_-)}
\ee
is of type $(1,1)$ in either of these almost complex structures. Without loss of generality, we can then identify $\sigma_-$ with the complex structure with respect to which $\omega$ is the K\"ahler form on $M$. The associated $(0,1)$-vector fields on $M$ are spanned by $\sigma_-e_{1\dal}-e_{2\dal}$ and define the antiholomorphic exterior derivative of smooth functions $f$ on $M$,
\be\label{dbarM}
\dbar f = \frac{(\theta^{2\dal}-\overline{\sigma_-}\theta^{1\dal})}{1+|\sigma_-|^2}\,(e_{2\dal}-\sigma_-e_{1\dal})f\,.
\ee
The holomorphic exterior derivative $\p$ on $M$ is defined by complex conjugation.

\medskip

The construction of $\check{\omega}$ from $\omega$ is a special case of the Penrose transform relating solutions of the twistor equation \eqref{tweq} to elements of $H^0(Z,K_Z^{-1/2})$ \cite{Hitchin:1980hp}. In particular, the converse of this correspondence also holds \cite{lebrun1992twistors}: any holomorphic global section $\check{\omega}\in H^0(Z,K_Z^{-1/2})$ vanishing at exactly two distinct points on every $L_x$ and satisfying $\overline{\check{\omega}(x,\sigma)} = -\omega(x,\hat\sigma)$ gives rise to a K\"ahler metric on $M$ that is conformal to its self-dual metric $g$.

Indeed, since $\check{\omega}|_{L_x}$ is a global section of $\CO(2)\to\CP^1$, one can use Liouville's theorem to conclude that it is a degree 2 polynomial in $\sigma$ of the form $\check{\omega}|_{L_x} = \omega_{\al\beta}(x)\sigma^\al\sigma^\beta$. Let $\sigma_\al=(1,\sigma_\pm(x))$ be the roots of $\check{\omega}|_{L_x} = 0$, and use them to define a normalized global section of the projective spinor bundle of $M$,
\be
\psi_\al = \frac{(1,\sigma_-)}{\sqrt{1+|\sigma_-|^2}}\,,\qquad\la\hat\psi\,\psi\ra = \|\psi\|^2 = 1\,.
\ee
The reality condition $\sigma_+=-1/\overline{\sigma_-}$ allows us to express the coefficients $\omega_{\al\beta}$ as
\be
\omega_{\al\beta} = \varphi\psi_{(\al}\hat\psi_{\beta)}
\ee
where $\varphi(x)$ is a real and nowhere vanishing normalization factor that will play the role of conformal scale. 

As seen via \eqref{dbaromegacech}, the holomorphicity of $\check{\omega}$ is equivalent to the twistor equation \eqref{tweq} for $\omega_{\al\beta}$. Contracting the latter with $\psi^\al\psi^\beta\psi^\gamma$ yields
\be\label{covpsi}
\psi^\al\psi^\beta\psi^\gamma\nabla_{\al\dal}\omega_{\beta\gamma} = 0\implies\psi^\al\psi^\beta\nabla_{\al\dal}\psi_\beta = 0\,.
\ee
At the same time, a short calculation shows that
\be
\begin{split}
    \d(\psi_\al\theta^{\al\dal}) &= \nabla_{\beta\dot\beta}\psi_\al\,\theta^{\beta\dot\beta}\wedge\theta^{\al\dal} + \tilde\Gamma^{\dal}{}_{\dot\beta}\wedge\psi_\beta\theta^{\beta\dot\beta}\\
    &\equiv \psi^\al\psi^\beta\nabla_{\beta\dot\beta}\psi_\al\,\hat\psi_{\delta}\theta^{\delta\dot\beta}\wedge\hat\psi_{\gamma}\theta^{\gamma\dal}\quad\text{mod span}\{\psi_\al\theta^{\al\dal}\}\,.
\end{split}
\ee
Hence, \eqref{covpsi} is sufficient to ensure that the 1-forms $\psi_\al\theta^{\al\dal}$ span a closed differential ideal and give rise to a complex structure on $M$. The spacetime 2-form $\omega=\im|\varphi|\psi_\al\hat\psi_\beta\Sigma^{\al\beta}$ is then closed by virtue of \eqref{tweq}, positive, and is of type $(1,1)$ in this complex structure, so it yields a K\"ahler form on $M$. The associated K\"ahler metric is
\be
\begin{split}
    2\,\omega_{\al\beta}\eps_{\dal\dot\beta}\,\theta^{\al\dal}\theta^{\beta\dot\beta} &= 2|\varphi|\cdot\psi_\al\theta^\al{}_{\dal}\cdot\hat\psi_\beta\theta^{\beta\dal} \\
    &= |\varphi|\,\la\hat\psi\,\psi\ra\,\theta^{\al\dal}\theta_{\al\dal} = |\varphi|\, g\,,
\end{split}
\ee
which is conformal to the SD metric $g$ that we originally used to construct twistor space. Choosing the conformal scale to be $\varphi=1$ is then equivalent to normalizing the K\"ahler form to satisfy $\omega^2 = -2\,\vol_g$. Lastly, scalar-flatness follows from self-duality of $g$, as proven for instance in \cite{lebrun1986topology}.


\section{Reducing twistor actions to spacetime}
\label{app:twac}

In sections \ref{sec:open} and \ref{sec:closed}, we mentioned certain holomorphic theories on twistor space $Z$ and the 4-dimensional integrable field theories that they give rise to on a scalar-flat K\"ahler spacetime $(M,g)$. In this appendix, we provide more details on reducing our twistor actions to spacetime via compactification along the $\CP^1$ fibers of the twistor fibration $\pi:Z\to M$. In performing such reductions, we will use the local presentation of twistor spaces described in appendix \ref{app:twistor}.


\subsection{Holomorphic Chern-Simons theory}
\label{app:hcsred}

Let $E\to Z$ be a complex vector bundle with structure group $G$, and let $\cA$ be a partial connection on it as in section \ref{sec:open}. We start by studying the holomorphic Chern-Simons action given in \eqref{hcs}. Its reduction to spacetime will be performed by generalizing the methods of \cite{Mason:2005zm} to twistor actions on curved self-dual backgrounds.

The partial connection $\cA$ can be decomposed into vertical and horizontal $(0,1)$-forms on $Z$ in the basis \eqref{01basis}:
\be
\cA = \cA_0\bar\tau + \cA_{\dal}\bar\theta^{\dal}\,.
\ee
The vertical part $\cA_0\bar\tau$ acts as a partial connection on the restriction of $E$ to the twistor lines $L_x$. $\cA$ is weightless in the homogeneous coordinates $\sigma_\al$, so $\cA_0$ must have weight $+2$ and $\cA_{\dal}$ must have weight $+1$. Similarly, all of these have weight $0$ in the complex conjugates $\hat\sigma_\al$. A short computation using $\Psi_{\al\beta\gamma\delta}=R=0$ establishes that
\be\label{dbara}
\dbar \cA = \left(\dbar_0\cA_{\dal}-\dbar_{\dal}\cA_0\right)\bar\tau\wedge\bar\theta^{\dal} + \frac{1}{2}\left(\dbar_{\dal}\cA^{\dal} + \sigma^\al\tilde\Gamma_{\dal\dot\beta\al}{}^{\dot\beta} \cA^{\dal}\right)\bar\theta_{\dot\gamma}\wedge\bar\theta^{\dot\gamma}\,,
\ee
written in terms of a basis of $(0,1)$-vector fields on $Z$,
\be
\dbar_0 = -\|\sigma\|^2\,\sigma_\al\frac{\p}{\p\hat\sigma_\al}\,,\qquad\dbar_{\dal} = \sigma^\al\bigg(e_{\al\dal} + \sigma^\beta\Gamma_{\beta\gamma\al\dal}\frac{\p}{\p\sigma_\gamma}+\hat\sigma^\beta\Gamma_{\beta\gamma\al\dal}\frac{\p}{\p\hat\sigma_\gamma}\bigg)\,.
\ee
This can be further simplified by noticing that $\sigma^\al e_{\al\dal}\cA^{\dal} + \sigma^\al\tilde\Gamma_{\dal\dot\beta\al}{}^{\dot\beta} \cA^{\dal} = \sigma^\al\nabla_{\al\dal}\cA^{\dal}$, where $\nabla_{\al\dal}$ is the Levi-Civita connection of $g$. This is useful when performing manipulations involving integration-by-parts in actions like \eqref{hcs}.

Plugging \eqref{dbara} in the action \eqref{hcs} shows that $\cA_{\dal}$ occurs quadratically. Integrating out $\cA_{\dal}$ is then equivalent to imposing its equation of motion, which is seen to be
\be\label{adaleq}
\dbar_0\cA_{\dal} - \dbar_{\dal}\cA_0 + [\cA_0,\cA_{\dal}] = 0\,.
\ee
This is the gauge constraint associated to the vertical part of the gauge symmetry \eqref{asym}. To eliminate it, we start by defining a frame $H\in\text{Maps}(Z,G)$ on the vector bundle $E\to Z$ that trivializes the partial connection $\cA$ on every twistor line,
\be\label{avH}
\cA_0 = H^{-1}\dbar_0 H\,.
\ee
This is possible because $E$ was assumed to be trivial along the twistor lines. Note that $H$ has homogeneity $0$ in $\sigma_\al,\hat\sigma_\al$. In this frame, \eqref{adaleq} allows us to show that
\be
\dbar_0\!\left(H(\dbar_{\dal}+\cA_{\dal})H^{-1}\right) = H[\dbar_0,\dbar_{\dal}]H^{-1} = 0\,,
\ee
having used the fact $[\dbar_0,\dbar_{\dal}]=0$ that can be verified by direct computation by remembering that $\Gamma_{\al\beta\gamma\dot\gamma}=\Gamma_{\beta\al\gamma\dot\gamma}$. Hence, $H(\dbar_{\dal}+\cA_{\dal})H^{-1}$ are a pair of weight $1$, globally holomorphic functions on each $\CP^1$ fiber. By Liouville's theorem, they must be linear in $\sigma^\al$, so that
\be\label{ahH}
H(\dbar_{\dal}+\cA_{\dal})H^{-1}=\sigma^\al A_{\al\dal}
\ee
for some Lie algebra valued covector $A_{\al\dal}$ on $M$. 

Putting \eqref{avH} and \eqref{ahH} together, we obtain the partial solution
\be\label{asol}
\cA = H^{-1}(\dbar+\pi_{0,1}A)H
\ee
of the equations of motion of holomorphic Chern-Simons theory. Here, 
\be
\pi_{0,1}A = \sigma^\al A_{\al\dal}\bar\theta^{\dal}
\ee
is the $(0,1)$ part of the pullback of $A = A_{\al\dal}\theta^{\al\dal}$ to $Z$, and 
\be
\dbar H = \dbar_0H\,\bar\tau + \dbar_{\dal}H\,\bar\theta^{\dal}
\ee
can be obtained from the definition of $\dbar$. To physically interpret $A$, one constructs a vector bundle $E_M\to M$ from the bundle $E\to Z$ satisfying $\pi^*E_M=E$.\footnote{Again, such a bundle exists because $E$ is trivial along the fibers of $\pi:Z\to M$. The fiber of $E_M$ over $x\in M$ is taken to be the vector space of sections $\Gamma(L_x,E)$ that are holomorphic with respect to $(\dbar+a)|_{L_x}$.} Then $A$ provides a connection on $E_M$. 

The boundary conditions \eqref{abdry} on $\cA$ can be used to further constrain $A$. To do this, let us specialize to affine coordinates $\sigma_\al=(1,\sigma)$. In \eqref{omegapol}, we denoted the zeroes of $\check{\omega}$ by $\sigma=\sigma_\pm$. We gauge fix $\cA$ by imposing the conditions \cite{Bittleston:2020hfv,Penna:2020uky}
\begin{itemize}
    \item $H$ depends on $\sigma$ only through $|\sigma|$ and is independent of arg$(\sigma)$,
    \item $H = \mathbbm{1}$ in a small neighborhood of $\sigma=\sigma_+$ on each $L_x$,
    \item $H = \sg$ in a small neighborhood of $\sigma=\sigma_-$ on each $L_x$,
\end{itemize}
where $\sg(x)$ is a Lie algebra valued function on $M$ containing the dynamical degree of freedom in our theory. Evaluating \eqref{asol} at $\sigma=\sigma_\pm$ in this gauge, and using the boundary conditions $\cA|_{\sigma=\sigma_\pm}=0$, we can solve for the spacetime gauge field in terms of $\sg$ to find
\be
A = -\dbar\sg\,\sg^{-1}\,,
\ee
where the spacetime antiholomorphic derivative $\dbar\sg$ is defined as in \eqref{dbarM} by setting $f=\sg$.

Substituting \eqref{asol} in the holomorphic Chern-Simons action \eqref{hcs} and using the fact that $\pi_{0,1}A$ is purely horizontal, we are left with
\be\label{hcs1}
S_\text{hCS}[\cA] = \frac{1}{2(2\pi\im)^3}\int_Z\Omega\wedge\dbar\,\tr\!\left(\dbar H\,H^{-1}\wedge \pi_{0,1}A\right)-\frac{1}{6(2\pi\im)^3}\int_Z\Omega\wedge\tr\!\left(H^{-1}\dbar H\right)^3\,.
\ee
Integrating the $\dbar$ in the first term by parts generates a factor of
\be
    \dbar\Omega 
    = -2\pi\im\,(\sigma_+-\sigma_-)^2\left\{\frac{\p_\sigma\delta(\sigma-\sigma_-)}{(\sigma-\sigma_+)^2} + \frac{\p_\sigma\delta(\sigma-\sigma_+)}{(\sigma-\sigma_-)^2}\right\}(1+|\sigma|^2)^2\,\bar\tau\wedge\tau\wedge\theta^{\dot1}\wedge\theta^{\dot2}\,,
\ee
having used $\dbar_0=(1+|\sigma|^2)^2\,\p_{\bar\sigma}$, the identity $\p_{\bar\sigma}(\sigma-z)^{-1}=2\pi\im\,\delta(\sigma-z)$, and the normalization condition $\omega_{11}(\sigma_+-\sigma_-)=\sqrt{-4\det\omega_{\al\beta}}=1$ encountered in appendix \ref{app:twistor}. Using this to simplify the first term of \eqref{hcs1} after integration by parts, we see that the contribution localized at $\sigma=\sigma_+$ vanishes in the aforementioned gauge. We are left with the contribution localized at $\sigma=\sigma_-$, which quickly reduces to
\begin{align}
    &\frac{1}{2(2\pi\im)^3}\int_Z\dbar\Omega\wedge\tr\!\left(\dbar H\,H^{-1}\wedge \pi_{0,1}A\right)\nonumber\\
    &= -\frac{1}{8\pi^2}\int\frac{\vol_g}{\sigma_+-\sigma_-}\,\tr\left\{\sg^{-1}(e_{2\dal}\sg-\sigma_+e_{1\dal}\sg)\,\sg^{-1}(e_{2\dal}\sg-\sigma_-e_{1\dal}\sg)\right\}\nonumber\\
    &= -\frac{\im}{8\pi^2}\int\omega\wedge\tr\left(\sg^{-1}\p\sg\wedge\sg^{-1}\dbar\sg\right)\,,
\end{align}
where we have used $(1+|\sigma|^2)^2\,\bar\tau\wedge\tau\wedge\theta^{\dot1}\wedge\theta^{\dot2}\wedge\bar\theta^{\dot1}\wedge\bar\theta^{\dot2} = -\d\bar\sigma\wedge\d\sigma\wedge\vol_g$ for performing the $\sigma$ integral. In going from the second to the third line, we have used the identity $\vol_g = -\frac12\,\omega^2$, the expression \eqref{omega+-} for the K\"ahler form, and the definition \eqref{dbarM} for the spacetime antiholomorphic derivative $\dbar\sg$ (and its complex conjugate $\p\sg$).

The second term in \eqref{hcs1} reduces to a five-dimensional Wess-Zumino term. Firstly, because $\Omega$ is a $(3,0)$-form, we can trivially replace $H^{-1}\dbar H$ by the de Rham differential $H^{-1}\d H$ on twistor space. Since $H$ is gauge fixed to be independent of the phase of $\sigma$, the integral over $\sigma$ now factorizes into a 5d integral times a contour integral of $\Omega$ over $\arg(\sigma)\in S^1$,
\be
-\frac{1}{6(2\pi\im)^3}\int_Z\Omega\wedge\tr\!\left(H^{-1}\d H\right)^3 = \frac{-\im}{48\pi^3}\int_{M\times\R_+}\bigg(\oint_{S^1}\Omega\bigg)\wedge\tr(H^{-1}\d H)^3\,.
\ee
The contour integral extracts the residue at one of the two poles of $\Omega$. Recalling \eqref{omegacirc}, this reproduces $-2\pi\omega$ (the sign arises from a choice of orientation). The integral over $\R_+$ can be converted to one over the interval $[0,1]$ by the change of variable $|\sigma|=-\log t$. Relabeling $H\equiv\tilde\sg$, one reproduces the Wess-Zumino term
\be
\frac{\im}{24\pi^2}\int_{M\times[0,1]}\omega\wedge\tr(\tilde\sg^{-1}\d\tilde\sg)^3
\ee
displayed in \eqref{Swzw}.


\subsection{BCOV theory}
\label{app:bcovred}

The reduction of gravitational theories like BCOV theory to spacetime is much harder in practice. Usually, one argues equivalence of BCOV to Mabuchi gravity by showing that their equations of motion coincide. The normalization of the spacetime action is then evaluated in terms of the normalization of the twistor action by working in the free limit and KK reducing just the kinetic terms along the twistor lines. 

In the rest of this appendix, we provide a brief summary of the on-shell equivalence. The interested reader is referred to \cite{Costello:2021bah} for more details, as well as for a determination of the normalization from a Green-Schwarz mechanism.\footnote{This is done in a flat space background, but the normalization of the action is of course background independent.} A similar but more explicit calculation may also be found in \cite{Mason:2005zm}, where the off-shell reduction of the twistor action of self-dual conformal gravity is worked out. This is relevant to Mabuchi gravity because Mabuchi gravity is just a gauge fixed K\"ahler subsector of self-dual conformal gravity.

Let $Z$, with Dolbeault operator $\dbar$, be the twistor space of a scalar-flat K\"ahler 4-manifold $(M,g)$ as usual. BCOV theory \cite{Bershadsky:1993cx} deals with complex structure deformations $\dbar\mapsto\dbar+\mu$ that are integrable (on-shell) and divergence-free (via an off-shell constraint):
\be
\dbar\mu + \frac{1}{2}\,[\mu,\mu] = 0\,,\qquad\p(\mu\ip\Omega)=0\,,
\ee
where $\Omega=\check{\omega}^{-2}\tau\wedge\theta^{\dot1}\wedge\theta^{\dot2}$ as before. Assuming that $\mu$ is on-shell, we can construct a deformed K\"ahler potential $K\mapsto K+\rho$ as follows.

Let $L_x$ be the twistor line of a point $x\in M$ in the undeformed twistor space $Z$. Let $L_x[\mu]$ denote the deformed twistor line. This is constructed as a rational curve in $Z$ that is holomorphic with respect to $\dbar+\mu$. If $\mu$ is small (but finite), then $L_x[\mu]$ will act as a small perturbation of $L_x$. In particular, if $L_x$ has normal bundle $\CO(1)\oplus\CO(1)$, then so will $L_x[\mu]$. In this case, $M$ continues to act as the moduli space of the deformed twistor lines.

Additionally, as is standard in the theory of deformations of Calabi-Yau structures, one can construct a deformed volume form from $\Omega$ using the Tian-Todorov lemma \cite{Bershadsky:1993cx,todorov1989weil,Tian1987SmoothnessOT,Barannikov:1997dwc}:
\be
\Omega[\mu] = \e^{\mu}\!\ip\Omega\,.
\ee
This is holomorphic with respect to $\dbar+\mu$, away from its poles. The polar divisor of $\Omega[\mu]$ continues to be the same as the polar divisor of $\Omega$: the quadric $\check{\omega}=0$. This is because $\mu$ has second order zeroes at $\check{\omega}=0$ and is regular everywhere else.

This is precisely the data that goes into the construction of a scalar-flat K\"ahler metric on $M$. Following Pontecorvo's theorem, the deformed K\"ahler form may be found from a contour integral of the deformed volume form around one of the zeroes of $\check{\omega}|_{L_x[\mu]}$:
\be
\varpi(x) = \frac{1}{2\pi}\oint_{\Gamma\subset L_x[\mu]}\Omega[\mu]\,.
\ee
Since the polar quadric $\check{\omega}$ is preserved, the complex structure $\p,\dbar$ on spacetime is undeformed. Only the K\"ahler structure gets deformed: $\omega\mapsto\varpi$. This allows us to introduce a scalar field $\rho$ that captures the deformation of the K\"ahler potential, $K\mapsto\cK=K+\rho$. That is, one sets
\be
\varpi = \im\,\p\dbar\cK = \omega + \im\,\p\dbar\rho
\ee
where $\omega=\im\,\p\dbar K$ is the background K\"ahler form.

The deformed K\"ahler metric is scalar-flat by the Penrose transform \cite{pontecorvo,lebrun1992twistors}. Hence, we conclude that on-shell Beltrami differentials $\mu$ map to K\"ahler perturbations $\rho$ on spacetime that give rise to scalar-flat metrics. The latter are precisely the solutions of the equations of motion of Mabuchi gravity. This confirms, at least on-shell, that BCOV theory reduces to Mabuchi gravity on spacetime. Unlike WZW$_4$, it remains an important open problem to prove this directly at the level of the off-shell BCOV action. As of now, going beyond the preliminary results of \cite{Costello:2021bah} in this direction has been prohibitively difficult.


\section{BRST cohomology at large $N$} \label{app:BRST_computation}
Here we will verify that the BRST cohomology of the chiral algebra at large $N$ has a basis given by normally ordered products of modes of the currents $J[r,s]$, $E[r,s]$, $F[r,s]$ described before.

This computation is only concerned with the size of the BRST cohomology, and says nothing about the operator product.  As such, we can compute the BRST cohomology using a spectral sequence whose second page computes the \emph{classical} BRST cohomology.   The classical BRST cohomology is the Lie algebra cohomology of a certain infinite dimensional Lie algebra.

First, we build a graded Lie algebra we call $\mf{g}_N$ which encodes the gauge algebra and the symplectic reduction. This algebra is non-zero in degrees $0,1,2$:  
\begin{equation} 
	\mf{g}_N = \begin{cases} 
		\mf{sp}_{N}  &\text{ in degree } 0 \\
		\C^{8} \otimes F \oplus \wedge^2_0 F \otimes \C^2  & \text{ in degree } 1\\
		\mf{sp}_{N} & \text{ in degree } 2.
	\end{cases} 
\end{equation}
The copy of $\mf{sp}_{N}$ in degree zero acts on everything in the natural way, and the bracket of two degree $1$ elements is the moment map viewed as an element of $\mf{sp}_{N}$ in degree $2$. All other brackets are zero.

Then, the classical BRST cohomology is the relative Lie algebra cohomology\footnote{Here, and below, when we discuss Lie algebra homology or cohomology of an algebra like $\mf{g}_N[[z]]$, we mean in the topological sense.  Sometimes this is called Gelfand-Fuchs cohomology.}  of 
\begin{equation} 
	C^\ast(\mf{g}_N[[z]] \mid \mf{sp}_{N} ). 
\end{equation}
The use of relative Lie algebra cohomology corresponds to the fact that in the BRST complex, we should not include the $\c$-ghost, only its derivatives; invariance under constant gauge transformations in $\Sp(N)$ is imposed by hand.

At large $N$, this Lie algebra cohomology can be readily computed using a variant of the techniques developed by Loday-Quillen-Tsygan \cite{loday1984cyclic, tsygan1983homology}. The underlying cochain complex of the relative Lie algebra cohomology is
\begin{equation} 
	\Sym^\ast \left(   (\mf{g}_N[[z]] / \mf{sp}_{N} )^\vee[-1] \right)^{\Sp(N)} . 
\end{equation}
Classical invariant theory tells us that at large $N$, the ring of invariants under $\Sp(N)$ is freely generated by the open strings (products of fields starting and ending with $I$), and the closed strings. 

The differential does not necessarily preserve the number of open and closed strings.   This is because, inside any open or closed string expression, one might have the operator $\b$, and the BRST differential can turn this into $:\!II\!:$.   This is the only term in the differential which can change the number of open and closed strings.    This tells us that the different possibilities are:
\begin{equation} 
	\begin{split} 
		\text{open string} &\mapsto \text{one or two open strings} \\
		\text{closed string} & \mapsto \text{a closed string or an open string}.
	\end{split} 
\end{equation}
We can filter the complex by the number of open strings.  By taking the spectral sequence associated to this filtration,  we bring ourselves to a complex where the differential preserves (separately) the number of open and closed strings.

The cohomology of this complex is the graded symmetric algebra on the subcomplex consisting of states which are either a single open string, or a single closed string.   We will compute the cohomology of this complex and find that the open string part is spanned by the operator $J[r,s]$ and derivatives, and the closed string part by $E[r,s]$, $F[r,s]$ and derivatives.  

This will complete the proof. Indeed, all these expressions are in ghost number zero so that any further pages of the spectral sequences we have used must have zero differential. 

The computation of the cohomology for a single open or closed string uses the techniques of \cite{loday1984cyclic, tsygan1983homology}. We let $A$ be the graded commutative algebra
\begin{equation} 
	A = \C[[z,\eps_1,\eps_2]]
\end{equation}
where $\eps_1$, $\eps_2$ are anti-commuting variables of ghost number $1$. 

We let $\iota: A \to A$ be the  $\Z/2$ action which sends $\eps_i \to -\eps_i$.     Define an involution
\begin{equation} 
	\begin{split} 
		\rho : \mf{gl}_{2N} \otimes A & \mapsto \mf{gl}_{2N} \otimes A\\ 
		 Z & \mapsto -\iota ( Z^{t}	) 
	\end{split}	
\end{equation}
where $M^t$ is the symplectic transpose of $M$, defined so that $\mf{sp}(N)$ is the subalgebra $M + M^t = 0$.

Let
\begin{equation} 
	A_N \subset A \otimes \mf{gl}_{2N} 
\end{equation}
be the subspace of elements $Z \in \mf{gl}_{2N} \otimes A$ satisfying $\rho(Z) = Z$.  The space of such elements is
\begin{equation} 
	\begin{split} 
		\mf{sp}_{N}[[z]] & \text{ in degree } 0 \\
		\wedge^2 F \otimes \C^2 [[z]] & \text{ in degree } 1 \\ 
		\mf{sp}_{N}[[z]] & \text{ in degree } 2  
	\end{split}
\end{equation}
Thus, these elements match the $\b$, $X$ and $\c$ fields, except that we have not removed the trace from the $X$ fields.

We will use this description to compute the closed string part of the BRST cohomology, which is essentially the same as the single-trace part of the Lie algebra cohomology of $A_N$.  A small variant of the results of \cite{loday1984cyclic, tsygan1983homology} on dihedral homology tells us the following.
\begin{lemma}
	Let $HC^\ast(A)$ be the cyclic cohomology of $A$.   The single closed string part of the Lie algebra homology of $A_N$, at large $N$, is the $\Z/2$ invariants in $HC^\ast(A)[-1]$, where the $\Z/2$ action comes from an action on cyclic cochains described in the proof. 
\end{lemma}
\begin{proof}
	We let $\op{Cyc}(A)$ be the cyclic cochain complex of $A$.  As a graded vector space, without the differential, this is defined as follows. 

	First, we let $\op{Cyc}_n(A)$ be the quotient of the space of linear functionals	
	\begin{equation} 
		\phi : A^{\otimes n } \to \C 
	\end{equation}
	be the relation
	\begin{equation} 
		\phi = (-1)^n \sigma \phi 
	\end{equation}
	where $\sigma$ is a cyclic permutation.  Then,
	\begin{equation} 
		\op{Cyc}(A) = \prod_{n \ge 1} \op{Cyc}_n(A)[-n]. 
	\end{equation}
	We will not need right now the explicit formula for the differential.

Loday-Quillen-Tsygan defined a cochain map
	\begin{equation} 
		\op{Cyc}(A) \to C^\ast(A \otimes \mf{gl}_{2N})^{\GL(2N)} 
	\end{equation}
	which at large $N$ is an isomorphism onto the single trace part.  The explicit formula is the following. If $\phi : A^{\otimes n} \to \C$, then we define	
	\begin{equation} 
		\what{\phi} \in C^n(A \otimes \mf{gl}_{2N})^{\GL(2N) } 
	\end{equation}
	by the expression
	\begin{equation} 
		\what{\phi} ( (M_1 \otimes a_1) \dots (M_n \otimes a_n) ) =\frac{1}{n!} \sum_{\sigma \in S_n} (-1)^{\op{sign}(\sigma)}  \op{tr}(M_{\sigma_1} \dots M_{\sigma_n}) \phi(a_{\sigma_1} \otimes \dots \otimes a_{\sigma_n}). \label{eqn:whatphi}  
	\end{equation}
	Clearly $\what{\phi}$ only depends on $\phi \in \op{Cyc}_n(A)$.

	The Lie algebra map $A_N \to A \otimes \mf{gl}_{2N}$ gives rise to a map
	\begin{equation} 
		C^\ast(A \otimes \mf{gl}_{2 N})^{\GL(2N)}  \to C^\ast(A_N)^{\Sp(N)} .\label{eqn:cochain_map } 
	\end{equation}
	By composing these two maps we get a cochain map
	\begin{equation} 
		\op{Cyc}(A) \to C^\ast(A_N)^{\Sp(N)}.\label{eqn:cochain_map2} 
	\end{equation}
	Explicitly, this is given by the formula \eqref{eqn:whatphi}, where the $M_i$ are in $\mf{sp}_{N}$ if $a_i \in A^{even}$ and in $\wedge^2 F$ if $a_i \in A^{odd}$.

	We need to characterize the image and kernel of this map at large $N$, on single-trace expressions. 
	By considering invariant theory for the group $\Sp(N)$, $N$ large, it is clear that the map \eqref{eqn:cochain_map2} is surjective on single-trace expressions.

	However, it is not an isomorphism, because contractions of indices under $\Sp(N)$ have a dihedral symmetry not present for $\GL(2N)$. Consider a collection of elements $
 \mathcal{X}_{i,(r)} \in \mf{gl}_{2N}$, where the index $r$ takes value $0$ or $1$, and $\mathcal{X}_{i,(r)}$ satisfies
	\begin{equation} 
		\mathcal{X}_{i,(r)} + (-1)^r \mathcal{X}_{i,(r)}^t = 0 
	\end{equation}
	where $\mathcal{X}^t$ is the symplectic transpose. Thus,  $\mathcal{X}_{i,(r)}$ is in $\mf{sp}(N)$ if $r = 0$ and in $\wedge^2 F$ if $r = 1$. Then, 
	\begin{equation} 
		\op{tr}(\mathcal{X}_{1,r_1} \dots \mathcal{X}_{n,r_n}) = (-1)^{n + \sum r_i } \op{tr}(\mathcal{X}_{n,r_n} \dots \mathcal{X}_{1,r_1}) 
	\end{equation}
	This is the only relation (on top of cyclic (anti)-symmetry) we need to describe single trace expressions in $A_N$.  

	From this we find that the single trace part of $C^\ast(A_N)^{\Sp(N)}$ is isomorphic to the coinvariants of the cyclic cochain complex under the $\Z/2$ action which, on $\op{Cyc}_n(A)$, reverse the order of the elements, acts by $(-1)^n$, and also applies the automorphism of $A$ sending $\eps_i \to - \eps_i$. 
\end{proof}
The single trace part of the cyclic cohomology of $A$ was computed in \cite{CG}: it is given by the single-trace local operators of the vertex algebra with $\mf{gl}_N$ BRST reduction considered there.  There it was found that it is given by four towers: the two bosonic towers $E[r,s]$, $F[r,s]$ we have discussed before, and two new fermionic towers
\begin{equation} 
	\begin{split} 
		B[r,s]  &= \op{Tr} (  \b X_{\dot1}^{(r} X_{\dot2}^{s)} )  \\			C[r,s]  &= \op{Tr}(\partial c X_{\dot1}^{(r} X_{\dot2}^{s)} )   
	\end{split}
\end{equation}
It remains to check that the $B,C$ towers are odd under the $\Z/2$ action on cyclic cohomology.  This is not hard to verify.

The final thing to do is to compute the open string states.  We can do this by a similar argument.  Let $A$ be the algebra as above,  and let $M$ be the graded $A$-module
\begin{equation} 
	M = \C^8 \otimes \C[[z]] 
\end{equation}
Define the bar complex by
\begin{equation} 
	\begin{split} 
		\op{Bar}_n(M,A,M) = \op{Hom} (M \otimes A^{\otimes n} \otimes M   ,\C) \\
		\op{Bar}(M,A,A) = \prod_{n \ge 0} \op{Bar}_n(M,A,M) [-n]
	\end{split}
\end{equation}
where $[-n]$ indicates a cohomological shift up by $n$. When equipped with the standard bar differential, we have
\begin{equation} 
	 \op{Bar}_n(M,A,M)  \simeq  \op{Hom} ( M \otimes^{\mbb{L}}_{A} M, \C)  
\end{equation}
where on the right hand side we have the dual of the derived tensor product.  Thus, the cohomology of the bar complex can be computed as the dual of the Tor groups of $M$ with itself over $A$.

There is a natural map
\begin{equation} 
	\op{Bar}(M,A,M) \to C^\ast(\mf{g}_N)  
\end{equation}
which at infinite $N$ maps surjetively onto the single-string open string states.   The expression is the obvious one: an element of $\op{Bar}_n(M,A,M)$ maps to an open string expression starting with $I$, with $n$ elements of $A^\vee$ (representing $X,\b,\c$ and derivatives) and ending with $I$.

If we were working with a $\mf{gl}_N$ gauge theory, this map would be an isomorphism onto the open string states.  It is not, however, for the same reason we saw when discussing closed string states: when working with a $\mf{sp}_{N}$ gauge theory, there is a symmetry which reverses the order in a trace (and introduces a sign).

This $\Z/2$ action acts on $\op{Bar}_n(M,A,M)$ by reversing the order,  sending $\eps_i \to -\eps_i$, and introducing a sign of $(-1)^{n+1}$.

We find, then, that the single open string states are the dual of the $\Z/2$ invariants in 
\begin{equation} 
	\op{Tor}_{A}(M,M).  
\end{equation}

The next step is to compute these Tor groups.  This is done by writing down a free resolution for $M$ as an $A$ module. The resolution is
\begin{equation} 
	F = \C[[z,\eps_i, \gamma_i]] \otimes \C^8 
\end{equation}
with differential
\begin{equation} 
	\d = \eps_i \frac{\partial}{\partial \gamma_i}.   
\end{equation}
The $\gamma_i$ are bosonic variables in ghost number $0$.  At the level of cohomology, the $\eps_i$ and $\gamma_i$ cancel and $F$ becomes $M$.

The Tor group is then the cohomology of the complex
\begin{equation} 
	F \otimes_{A} M 
\end{equation}
Tensoring over $A$ with $M$ has the effect of setting the $\eps_i$ to $0$ in $F$, so that this tensor product is $\C[[z,\gamma_i]] \otimes \C^8 \otimes \C^8$ with zero differential.  

This tells us that the cohomology of the bar complex is spanned by $z$-derivatives of expressions like
\begin{equation} 
	I_k X_{\dot1}^{(r} X_{\dot2}^{s)} I_l 
\end{equation}
where the indices $k,l$ run over a basis of $\C^8$, and thee expression is not anti-symmetric in $i,j$. The integers that $r,s$ that appear here corresond to the powers of $\gamma_1,\gamma_2$ in the dual vector space $\C[[z,\gamma_i]]$.  Passing to $\Z/2$ invariants means that we only retain the expressions which are anti-symmetric in $k,l$, as desired. 

\subsection{BRST cohomology in the presence of the boundary condition}\label{app:BRST_boundary}
Now let us repeat the analysis in the case that we have the boundary condition discussed in section \ref{sec:bdydef}. Recall that with this boundary condition, $X_{\dal}$ have first order poles, $\b$ has a second order pole and $I$ has a branch cut and a pole of order $\frac{1}{2}$. 

The poles in the $X$ and $\b$ field can be engineered very easily in the framework of the previous section, simply by multiplying the fermionic variables $\eps_i$ by $z^{-1}$.  This does not change the algebra in any way, except by changing the spin of the $\eps_i$.  Thus, the analysis goes through without any change.

For the open string sector, the analysis is also almost identical.  The module $z^{-\frac{1}{2}} \C[[z]] \otimes \C^8$ associated to giving $I$ a pole of order $\frac{1}{2}$ is isomorphic to the module $\C[[z]] \otimes \C^8$ we used in the case without boundary condition.  Thus, the same argument applies.

The result is that a basis of single-trace closed string operators at $z = 0$, in the presence of the boundary condition, is given by the same expressions we used before, but with extra powers of $z$:
\begin{equation} 
	\begin{split} 
		z^{r+s+2} 	\op{tr}	\left( \eps^{\dal\dot\beta} X_{\dal}\p X_{\dot\beta} X_{\dot1}^{(r} X_{\dot2}^{s)} \right) \\
z^{r+s} 	\op{tr}	\left( X_{\dot1}^{(r} X_{\dot2}^{s)} \right) \\
	\end{split}
\end{equation}
and the open string expressions are similar:
\begin{equation} 
	z^{r+s+1} I_k X_{\dot1}^{(r} X_{\dot2}^{s)} I_l  
\end{equation}
(Of course, we also must include derivatives of these expressions).

\section{A basis for the boundary module (when defined as a quotient of the mode algebra)} \label{app:boundary_size}
We have given two definitions of the boundary module: one in terms of boundary conditions for the fundamental fields in section \ref{sec:bdydef}, where we then take BRST cohomology; and another, where we simply declare what poles the currents $E,F,J$ have at $z = 0$ and $z = \infty$ as in section \ref{sec:chiralbdy}.

In this section, for completeness, we will verify that the second definition leads to a boundary module of the correct size.  
The boundary module is defined by the left ideal generated by 
\begin{equation} 
	\begin{split} 
		J_l[m,n] &  \text { for } 2l + m + n < 0 \\
		E_l[m,n] & \text { for } 2l +  m + n < 0 \\
		F_l[m,n] & \text { for } 2l  + m + n < 0.
 \label{eqn:vaccuum_ideal_app}
	\end{split}
\end{equation}
We will prove that this module has a basis given by ordered products of the modes
\begin{equation} 
	\begin{split} 
		J_l[m,n] & \text { for } 2l +  m + n \ge  0 \\
		E_l[m,n]   &  \text { for } 2l +  m + n \ge 0 \\
		F_l[m,n]   &  \text { for } 2l  + m + n \ge 0.
		\label{eqn:vaccuum_module_app}
	\end{split} 
\end{equation}

Before we check this, we need to introduce some notation.  Given a single-trace mode $J_l[m,n]$, $E_l[m,n]$, $F_l[m,n]$ we say its \emph{length} is $(m+n)/2$.  The length labels which $\SL_2(\C)$ representation the mode lives in.  

We define the \emph{weight} of a single-trace mode by the formula
\begin{equation} 
	\text{weight } = 2\text{ spin } + 2 \text{ length}. 
\end{equation}
Then the ideal defining the module is the left ideal generated by those single-trace modes of weight $<0$.

We say a product of single-trace modes $\mc{O}_1 \dots \mc{O}_n$ of weights $w_1,\dots,w_n$  satisfies the \emph{modified normal ordering} $w_1 > w_2 > \dots > w_n$.  Let us also choose arbitrarily a normal ordering perscription on those modes of the same weight.  Then, products of single-trace modes in modified normal ordering form a topological basis for the mode algebra. 

In this basis, we say the modified normally ordered product $\mc{O}_1 \dots \mc{O}_n$ is of weight $w_1 + \dots + w_n$ (where $w_i$ is the weight of $\mc{O}_i$).    

We have the following lemma.
\begin{lemma} 
	If $\mc{O}_1$, $\mc{O}_2$ are any two states of weight $w_1,w_2$, then the commutator $[\mc{O}_1, \mc{O}_2]$ is a sum of states of weight $\le w_1 + w_2$.  
\end{lemma}
\begin{proof} 
	This is immediate from the representation theory. 	Let $\mc{O}_i$ be of spin $s_i$ and length $l_i$, where $w_i = 2 s_i + 2 l_i$.  Then, $[\mc{O}_1,\mc{O}_2]$ is of spin $s_1 + s_2$.  It also transforms in some representations of $\SL_2(\C)$ which are contained in the tensor product of the representation of spin $l_1$ and that of spin $l_2$.  Such a representation is of spin at most $l_1 + l_2$. 
\end{proof}

\begin{proposition}
	A basis for the ideal defining the module is given by  modified  normally-ordered products of single-trace generators $\mc{O}_1 \dots \mc{O}_n$ where $\mc{O}_n$ has weight $< 0$.  
\end{proposition}
\begin{proof} 
	Clearly these expressions are in the ideal.  However, \emph{a priori}, the ideal could be larger.  It is spanned by elements of the form $(\mc{O}_1 \dots \mc{O}_{n-1} )  \mc{O}_n$, where $\mc{O}_i$ are single-trace operators of weights $w_i$, where $w_n < 0$ and $\mc{O}_1 \dots \mc{O}_{n-1}$ are in modified normal order $w_1 \ge \dots \ge w_{n-1}$.  

	We need to show that we can re-order an expression like this to make the whole expression in modified normal order.  Suppose that we need to move $\mc{O}_{k} \dots \mc{O}_{n-1}$ to the right of $\mc{O}_n$ to place the expression in modified normal order. This can only happen if $w_i \le w_n$ for $i = k,\dots,n-1$.  When we re-order in this way, we pick up the commutator
	\begin{equation} 
		[\mc{O}_n, \mc{O}_{k} \dots \mc{O}_{n-1} ]. 
	\end{equation}
	By the previous lemma, this commutator is a sum of expressions in modified normal order of total weight which is at most $w_k + \dots + w_n$.  Since $w_n < 0$ and $w_i \le w_n$ for $i = k,\dots,n-1$, this total weight is less than $0$.   Any expression in modified normal order of total weight $ < 0$ has rightmost element of weight $< 0$.   We conclude that we can re-write any element of the ideal as a sum of expressions in modified normal order where the rightmost element has weight $< 0$, as desired.
\end{proof}
\begin{corollary}
	A basis for the module which is a quotient of the mode algebra by the ideal described above is given by products of single-trace modes in modified normal order
	\begin{equation} 
		\mc{O}_{1} \dots\mc{O}_n 
	\end{equation}
	where the weights $w_i$ are $\ge 0$.
\end{corollary}
\begin{proof} 
	This is exactly the complement, in our chosen basis, of those words in modified normal order whose rightmost element has weight $< 0$.  
\end{proof}

\section{Vanishing of some Dolbeault cohomology groups }
\label{app:vanishing}
In this section we will show that
\begin{equation} 
	\begin{split} 
		H^\ast_{\dbar}(\br{\SL}_2(\C), \Oo ( - \partial\br{\SL}_2(\C)) ) &= 0 \\
		H^\ast_{\dbar}(\til\F, \Oo ( - \br{D}_0 -\br{D}_\infty - E) &= 0. 
	\end{split}
\end{equation}
These cohomology groups are easily seen to vanish.  Let us first check it for $\br{\SL}_2(\C)$, which is a quadric in $\CP^4$.  The boundary divisor in $\br{\SL}_2(\C)$ is the pull-back of $\Oo(-1)$. There is an exact sequence in the category of sheaves on $\CP^4$ 
\begin{equation} 
	0 \to \Oo_{\CP^4} (-3) \to \Oo_{\CP^4}(-1) \to \Oo_{\br{\SL}_2(\C)}  (-1) \to 0. 
\end{equation}
The coherent sheaf cohomology of the first two terms vanishes, therefore that of the last term vanishes too.

To prove the corresponding result for $\til\F$, we first note that 
\begin{equation} 
	R \pi_\ast \Oo_{\til{\F}} = \Oo_{\br{\SL}_2(\C) } 
\end{equation}
where $R \pi_\ast$ is the derived push-forward along the map $\pi : \til{\F} \to \br{\SL}_2(\C)$.   This holds essentially because all fibers of the map $\pi$ are either a point or $\CP^1$.  Then, since 
\begin{equation} 
	\pi^\ast \Oo_{\br{\SL}_2(\C)} (- \partial \br{\SL}_2(\C) ) =  \Oo_{\til\F} ( - \br{D}_0 -\br{D}_\infty - E) 
\end{equation}
we also have
\begin{equation} 
	 R \pi_\ast \Oo_{\til\F} ( - \br{D}_0 -\br{D}_\infty - E)  =  \Oo_{\br{\SL}_2(\C)} (- \partial \br{\SL}_2(\C) ) 
\end{equation}
from which the result follows.

\end{appendix}

\bibliographystyle{JHEP}
\bibliography{twist}

\end{document}